\documentclass[
aps,
pra,
reprint,
superscriptaddress,
nofootinbib,
twocolumn
]{revtex4-2}
\usepackage[T1]{fontenc}
\usepackage[utf8]{inputenc}
\usepackage{lmodern}
\usepackage{mathtools}
\usepackage{amssymb}
\usepackage{amsthm}
\usepackage{bm}
\usepackage{graphicx}
\usepackage{dcolumn}
\usepackage{xcolor}
\usepackage[
colorlinks=true,
citecolor=blue,
linktocpage=true
]{hyperref}
\usepackage[capitalise]{cleveref}

\newcommand{\nocontentsline}[3]{}
\newcommand{\tocless}[2]{\bgroup\let\addcontentsline=\nocontentsline#1{#2}\egroup}

\makeatletter
\renewcommand\onecolumngrid{% <<<<<<
\do@columngrid{one}{\@ne}%
\def\set@footnotewidth{\onecolumngrid}% <<<<<<<<<<<<<<<<
\def\footnoterule{\kern-6pt\hrule width 1.5in\kern6pt}%
}
\renewcommand\twocolumngrid{% <<<<<<
        \def\footnoterule{% restore rule
        \dimen@\skip\footins\divide\dimen@\thr@@
        \kern-\dimen@\hrule width.5in\kern\dimen@}
        \do@columngrid{mlt}{\tw@}
}%
\makeatother

\makeatletter
\newtheorem*{rep@theorem}{\rep@title}
\newcommand{\newreptheorem}[2]{%
\newenvironment{rep#1}[1]{%
 \def\rep@title{#2 \ref{##1}}%
 \begin{rep@theorem}}%
 {\end{rep@theorem}}}
\makeatother

\theoremstyle{plain}
\newtheorem{theorem}{Theorem}
\newtheorem{lemma}[theorem]{Lemma}
\newtheorem{proposition}[theorem]{Proposition}

\theoremstyle{definition}
\newtheorem{definition}[theorem]{Definition}
\newtheorem{assumptions}{Assumptions}

\renewcommand{\O}{\mathcal{O}}
\newcommand{\Ot}{\tilde{\mathcal{O}}}
\newcommand{\R}{\mathbb{R}}
\newcommand{\C}{\mathbb{C}}

\newcommand{\Z}{\mathbb{Z}}

\newcommand{\I}{\mathbb{I}}

\renewcommand{\i}{\mathrm{i}}
\newcommand{\Orth}{\mathrm{O}}
\newcommand{\SO}{\mathrm{SO}}
\newcommand{\SU}{\mathrm{SU}}
\newcommand{\U}{\mathrm{U}}
\newcommand{\Cl}{\mathrm{Cl}}
\newcommand{\Pauli}{\mathcal{P}}
\newcommand{\Sym}{\mathrm{Sym}}
\newcommand{\Alt}{\mathrm{Alt}}
\newcommand{\B}{\mathrm{B}}
\newcommand{\SymCl}[1]{\Cl(1)^{\otimes #1}_{\Sym}}

\DeclareMathOperator{\spn}{span}
\DeclareMathOperator*{\E}{\mathbb{E}}
\DeclareMathOperator{\V}{Var}

\DeclareMathOperator{\sgn}{sign}
\DeclareMathOperator{\tr}{tr}
\newcommand{\comb}[2]{\mathcal{C}_{#1, #2}}
\newcommand{\diags}[2]{\mathcal{D}_{#1, #2}}
\newcommand{\op}[2]{\ket{#1}\!\bra{#2}}
\newcommand{\ip}[2]{\langle #1 | #2 \rangle}
\newcommand{\ev}[3]{\langle #1 | #2 | #3 \rangle}
\newcommand{\ket}[1]{| #1 \rangle}
\newcommand{\bra}[1]{\langle #1 |}

\newcommand{\rrangle}{\rangle \! \rangle}
\newcommand{\llangle}{\langle \! \langle}
\newcommand{\kket}[1]{| #1 \rrangle}
\newcommand{\bbra}[1]{\llangle #1 |}
\newcommand{\vop}[2]{\kket{#1} \! \bbra{#2}}
\newcommand{\vip}[2]{\llangle #1 | #2 \rrangle}
\newcommand{\vev}[3]{\llangle #1 | #2 | #3 \rrangle}

\newcommand{\sn}[1]{\| #1 \|_{\mathrm{shadow}}}
\newcommand{\sns}[2]{\| #1 \|_{#2}}
\renewcommand{\l}[1]{\mathopen{}\left#1}
\renewcommand{\r}[1]{\right#1\mathclose{}}

\newcommand{\SWAP}{\mathrm{SWAP}}
\newcommand{\iSWAP}{\i\SWAP}
\newcommand{\fSim}{\mathrm{fSim}}
\newcommand{\PhXZ}{\mathrm{PhXZ}}

\begin{document}

\title{Group-theoretic error mitigation enabled by classical shadows and symmetries}

\author{Andrew Zhao}
\email{azhao@unm.edu}
\affiliation{Center for Quantum Information and Control, Department of Physics and Astronomy, University of New Mexico, Albuquerque, New Mexico 87106, USA}

\author{Akimasa Miyake}
\email{amiyake@unm.edu}
\affiliation{Center for Quantum Information and Control, Department of Physics and Astronomy, University of New Mexico, Albuquerque, New Mexico 87106, USA}

% \date{May 14, 2024}
\date{\today}

\begin{abstract}
Estimating expectation values is a key subroutine in quantum algorithms. Near-term implementations face two major challenges:~a limited number of samples required to learn a large collection of observables, and the accumulation of errors in devices without quantum error correction. To address these challenges simultaneously, we develop a quantum error-mitigation strategy called \emph{symmetry-adjusted classical shadows}, by adjusting classical-shadow tomography according to how symmetries are corrupted by device errors. As a concrete example, we highlight global $\mathrm{U(1)}$ symmetry, which manifests in fermions as particle number and in spins as total magnetization, and illustrate their group-theoretic unification with respective classical-shadow protocols. We establish rigorous sampling bounds under readout errors obeying minimal assumptions, and perform numerical experiments with a more comprehensive model of gate-level errors derived from existing quantum processors. Our results reveal symmetry-adjusted classical shadows as a low-cost strategy to mitigate errors from noisy quantum experiments in the ubiquitous presence of symmetry.
\end{abstract}

\maketitle

\let\oldaddcontentsline\addcontentsline
\renewcommand{\addcontentsline}[3]{}

\section{Introduction}

Quantum computers are highly susceptible to errors at the hardware level, posing a considerable challenge to realize meaningful applications in the so-called noisy intermediate-scale quantum (NISQ) era~\cite{preskill2018quantum,bharti2022noisy}. One particularly promising and natural candidate for NISQ applications is the simulation of quantum many-body physics and chemistry~\cite{feynman1982simulating,georgescu2014quantum,mcardle2020quantum,bauer2020quantum}. In order to minimize the accumulation of errors, such algorithms prioritize low-depth circuits, for instance, variational quantum circuits~\cite{peruzzo2014variational,mcclean2016theory,yuan2019theory,cerezo2021variational}. However, in order to exhibit quantum advantage, these circuits must also be beyond the capabilities of classical simulation~\cite{osborne2006efficient,bravyi2021classical,napp2022efficient,wild2023classical}, resulting in noise levels that nonetheless corrupt the calculations.

While quantum error correction is the long-term solution, current state-of-the-art hardware is still a few orders of magnitude from achieving scalable, fault-tolerant quantum computation~\cite{fowler2012surface,kelly2015state,egan2021fault,postler2022demonstration,zhao2022realization,sundaresan2023demonstrating,google2023suppressing,sivak2023real,ni2023beating}. In the meantime, there have been considerable theoretical and experimental efforts probing the beyond-classical potential of NISQ computers~\cite{omalley2016scalable,kandala2017hardware,colless2018computation,dumitrescu2018cloud,hempel2018quantum,kandala2019error,kokail2019self,nam2020ground,arute2019quantum,arute2020hartree,harrigan2021quantum,arute2020observation,zhong2020quantum,huggins2022unbiasing,kim2023scalable,huang2022quantum,stanisic2022observing,tazhigulov2022simulating,madsen2022quantum,motta2023quantum,obrien2023purification,morvan2023phase,kim2023evidence}. Should such an application be demonstrated, quantum error mitigation (QEM) is expected to play a crucial role. Broadly speaking, QEM aims to approximately recover the output of an ideal quantum computation, given only access to noisy quantum devices and offline classical resources. We refer the reader to Refs.~\cite{endo2021hybrid,cai2022quantum} for a review of prominent concepts and strategies in QEM.

A related but separate challenge for NISQ algorithms is the need to learn many observables in a rudimentary fashion, i.e., by repeatedly running and sampling from quantum circuits. The number of repetitions required can be immense, both to suppress shot noise and to handle the measurement of noncommuting observables~\cite{wecker2015progress,gonthier2022measurements}. One particularly promising approach is that of classical shadows~\cite{huang2020predicting,paini2021estimating}. In contrast to prior measurement strategies~\cite{cotler2020quantum,bonet2019nearly,cerezo2021variational,tilly2022variational}, classical shadows are remarkably simple to implement and have been shown to exhibit optimal sample complexity in certain important scenarios~\cite{huang2020predicting,zhao2023learning}.

Classical shadows were developed primarily from the union of two themes in quantum learning theory:~linear-inversion estimators for state tomography~\cite{sugiyama2013precision,guta2020fast} (closed-form solutions that admit fast postprocessing and rigorous guarantees) and the framework of shadow tomography~\cite{aaronson2020shadow,aaronson2019gentle} (predict only a subset of observables, not the entire density matrix). The result is a simple but powerful protocol that accurately estimates a large collection of observables from relatively few samples. In terms of quantum resources, classical shadows only require the ability to measure in randomly selected bases, making the protocol particularly amenable to NISQ constraints. These desirable features have inspired a wide range of extensions and applications, 
for example:~entanglement detection~\cite{elben2020mixed}, quantum Fisher information bounds~\cite{rath2021quantum,vitale2023estimation}, learning quantum processes~\cite{levy2021classical,kunjummen2023shadow}, navigating variational landscapes~\cite{sack2022avoiding,boyd2022training}, energy-gap estimation~\cite{chan2022algorithmic}, and applications to fermions~\cite{zhao2021fermionic,wan2023matchgate,ogorman2022fermionic,low2022classical,babbush2023quantum,denzler2023learning} and bosons~\cite{gu2023efficient,becker2022classical}. For an overview of classical shadows and randomized measurement strategies, see Ref.~\cite{elben2023randomized}.

Due to their experimental friendliness and versatile prediction power, classical shadows naturally have been considered for QEM as well. For example, Refs.~\cite{seif2023shadow,hu2022logical} used classical shadows to approximately project a noisy quantum state toward a target subspace via classical postprocessing, the subspaces being either the logical subspace of an error-correcting code~\cite{mcclean2020decoding} and/or the dominant eigenvector (purification) of the noisy mixed state~\cite{koczor2021exponential,huggins2021virtual}. These shadow-based ideas circumvent some of the difficulties of performing subspace projection, at the cost of an exponential sample complexity. Meanwhile, Ref.~\cite{jnane2023quantum} intertwined classical shadows with other popular QEM strategies, with a particular focus on probabilistic error cancellation~\cite{temme2017error}. They establish rigorous estimators and performance guarantees, assuming an accurate characterization of the noisy quantum device. Finally, Refs.~\cite{chen2021robust,koh2022classical} described modifications to the classical linear-inversion step in order to mitigate errors in the randomized measurements. In particular, robust shadow estimation~\cite{chen2021robust} assumes no prior knowledge of the noise, instead implementing a separate calibration experiment that learns the necessary noise features.

In this paper, we take this latter perspective~\cite{chen2021robust,koh2022classical}, with an eye on a more comprehensive mitigation of errors beyond readout errors. We introduce a QEM protocol, which we refer to as \emph{symmetry-adjusted classical shadows}, that takes advantage of known symmetries in the quantum system of interest. For example, in simulations of chemistry, the number of electrons is typically fixed. The corruption of such symmetries by noise informs us how to undo the effects of that noise. Crucially, because randomized measurements scramble the information, the other properties of the quantum system are corrupted (and therefore can be mitigated) in the same manner. Using these insights, symmetry-adjusted classical shadows appropriately modifies the linear-inversion based on the symmetry information alone.

A notable advantage of our protocol is that we do not run any extraneous calibration experiments. This has the added benefit of inherently accounting for errors that occur throughout the full quantum circuit, rather than the randomized measurements in isolation~\cite{karalekas2020quantum,chen2021robust,koh2022classical,van2022model,arrasmith2023development}. Also, the simplicity of the protocol allows for additional QEM techniques to be straightforwardly applied in tandem. Finally, in contrast to other symmetry-based ideas~\cite{bonet2018low,mcardle2019error,cai2021quantum,jnane2023quantum}, our approach goes beyond the concept of symmetry projection, instead utilizing a unified group-theoretic understanding of classical shadows in conjunction with symmetries. We expound on this distinction in Supplementary Section~\ref{sec:symmetry_qem_discussion}, wherein we review these prior symmetry-based QEM techniques.

This paper is structured as follows. In Section~\ref{sec:summary_of_results}, we begin by establishing preliminaries and background material (Sec.~\ref{sec:background}). We then introduce our main contribution, symmetry-adjusted classical shadows, and describe its key application for mitigating local fermionic and qubit observables (Sec.~\ref{subsec:EM_shadows_summary}). We follow by highlighting additional technical results:~a modification to random Pauli measurements required to tailor its irreps for use with common symmetries, called subsystem-symmetrized Pauli shadows (Sec.~\ref{subsec:sym_pauli_summary});~a symmetry adaptation to fermionic classical shadows which reduces the quantum resources required, applicable to fermionic systems with spin symmetry (Sec.~\ref{subsec:spin-adapt_summary});~and an improved design for compiling fermionic Gaussian unitaries with lower circuit depth and fewer gates than prior art (Section~\ref{subsec:improved_circuit_summary}). Finally, we close the section with a series of numerical experiments, demonstrating the effectiveness of our error-mitigation protocol under realistic scenarios (Sec.~\ref{sec:numerics}). This includes simulations of a noise model based on existing superconducting-qubit platforms~\cite{isakov2021simulations}. In Section~\ref{sec:discussion}, we summarize our findings and discuss future prospects. In Section~\ref{sec:methods}, we illustrate the general theory of symmetry-adjusted classical shadows, and we provide further technical details regarding the applications to fermion and qubit systems with global $\U(1)$ symmetries. Details regarding the mathematical proofs and numerical simulations are provided in the \hyperlink{supplementary_material}{Supplementary Information}, and code for the latter is available at our open-source repository ({\small\url{https://github.com/zhao-andrew/symmetry-adjusted-classical-shadows}}).

\section{Results}\label{sec:summary_of_results}

\subsection{Background}\label{sec:background}

First, we provide a review of classical shadows~\cite{huang2020predicting,paini2021estimating} and robust shadow estimation~\cite{chen2021robust} necessary to understand our technical results. Readers familiar with this background material can skip to Section~\ref{subsec:EM_shadows_summary}, after familiarizing themselves with the notation that we establish below.\\

\textbf{Notation and preliminaries.} For any integer $N > 1$, we define $[N] \coloneqq \{0, \ldots, N-1\}$ (note that we index starting from $0$). We use $\i \equiv \sqrt{-1}$ for the imaginary unit.

Throughout this paper, we consider an $ n $-qubit system with Hilbert space $ \mathcal{H} \coloneqq (\C^{2})^{\otimes n} $. Its dimension is denoted by $ d \equiv 2^n $ unless otherwise specified. We often work with the space of linear operators $ \mathcal{L}(\mathcal{H}) \cong \C^{d \times d} $ as a vector space, so it will be convenient to employ the Liouville representation:~for any operator $ A \in \mathcal{L}(\mathcal{H}) $, its vectorization $\kket{A} \in \C^{d^2}$ in some orthonormal operator basis $\{B_1, \ldots, B_{d^2} : \tr(B_i^\dagger B_j) = \delta_{ij}\}$ is defined by the components $\vip{B_i}{A} \coloneqq \tr(B_i^\dagger A)$. Under this representation, superoperators are mapped to $ d^2 \times d^2 $ matrices:~any $\mathcal{E} \in \mathcal{L}(\mathcal{L}(\mathcal{H}))$ can be specified by its matrix elements $ \mathcal{E}_{ij} \coloneqq \vev{B_i}{\mathcal{E}}{B_j} = \tr(B_i^\dagger \mathcal{E}(B_j)) $. We let $\mathcal{E}$ denote both the superoperator and its matrix representation, and in a similar fashion we sometimes write $\kket{A} = A$.

For systems of qubits, the normalized Pauli operators $\Pauli(n) / \sqrt{d}$ are a convenient basis for $\mathcal{L}(\mathcal{H})$, where
\begin{equation}
    \Pauli(n) \coloneqq \{\I, X, Y, Z\}^{\otimes n}.
\end{equation}
This choice is called the Pauli transfer matrix (PTM) representation. The weight, or locality, of a Pauli operator $P \in \Pauli(n)$ is the number of its nontrivial tensor factors, denoted by $|P|$. For each $i \in [n]$, we define $W_i \in \Pauli(n)$ which acts as $W \in \{X, Y, Z\}$ on the $i$th qubit and trivially on the rest of the system.

For fermions in second quantization, a natural choice of basis is the set of Majorana operators, defined as $\{ \Gamma_{\bm{\mu}} / \sqrt{d} : \bm{\mu} \subseteq [2n] \}$ where
\begin{equation}\label{eq:majorana_def}
    \Gamma_{\bm{\mu}} \coloneqq (-\i)^{\binom{|\bm{\mu}|}{2}} \prod_{\mu \in \bm{\mu}} \gamma_{\mu}.
\end{equation}
The Hermitian generators $\{ \gamma_\mu : \mu \in [2n] \} \subset \mathcal{L}(\mathcal{H})$ obey the anticommutation relation $\gamma_\mu \gamma_\nu + \gamma_\nu \gamma_\mu = 2\delta_{\mu\nu} \I$ (we will use $\I$ to denote any identity operator whose dimension is clear from context). They are related to the fermionic creation and annihilation operators $a_p^\dagger, a_p$ via
\begin{equation}
    \gamma_{2p} = a_p + a_p^\dagger, \quad \gamma_{2p+1} = -\i(a_p - a_p^\dagger).
\end{equation}
By convention, the elements of $\bm{\mu}$ and the product in Eq.~\eqref{eq:majorana_def} are in strictly ascending order. We call $|\bm{\mu}|$ the degree of $\Gamma_{\bm{\mu}}$, or equivalently refer to them as ($|\bm{\mu}|/2$)-body operators whenever the degree is even. It is straightforward to check that Majorana operators are isomorphic to Pauli operators, in particular satisfying the orthogonality relation $\vip{\Gamma_{\bm{\mu}}}{\Gamma_{\bm{\nu}}} = d \delta_{\bm{\mu} \bm{\nu}}$.

For any unitary $U \in \U(d)$, its corresponding channel is denoted by $\mathcal{U}(\cdot) \coloneqq U (\cdot) U^\dagger$. For any $\ket{\varphi} \in \mathcal{H}$, $\kket{\varphi}$ is the vectorization of $\op{\varphi}{\varphi}$. We use tildes to indicate objects affected by quantum noise, e.g., $\widetilde{\mathcal{U}}$ denotes a noisy implementation of the $\mathcal{U}$. Hats indicate statistical estimators, e.g., $\hat{o}$ denotes an estimate for $o = \tr(O \rho)$. Asymptotic upper and lower bounds are denoted by $\O(\cdot)$ and $\Omega(\cdot)$ respectively, and $f(x) = \Theta(g(x))$ means that $f(x)$ is both $\O(g(x))$ and $\Omega(g(x))$.\\

\textbf{Classical shadows.} We summarize the method of classical shadows as formalized by Huang \emph{et al.}~\cite{huang2020predicting}, borrowing the PTM language of Chen \emph{et al.}~\cite{chen2021robust} which will make the robust extension clear later. Our task is to estimate the expectation values $\tr(O_j \rho) = \vip{O_j}{\rho}$ of a collection of $L$ observables $O_1, \ldots, O_L \in \mathcal{L}(\mathcal{H})$, ideally using as few copies of $\rho$ as possible. Classical shadows is based on a simple measurement primitive:~for each copy of $ \rho $, apply a unitary $U$ randomly drawn from a distribution of unitaries and measure in the computational basis. This produces a sample $ b \in \{0,1\}^n $ with probability $ \vev{b}{\mathcal{U}}{\rho} $. One then inverts the unitary on the outcome $\ket{b}$ in postprocessing, which amounts to storing a classical representation of $ U^\dagger \ket{b} $.

The unitary distribution determines the efficiency of this protocol with respect to the properties of interest. Throughout this paper, we assume that the distribution is a finite group equipped with the uniform probability distribution (it is straightforward to generalize to compact groups, using their Haar measures). Specifically, let $ U : G \to \U(\mathcal{H}) $ be a unitary representation of a group $G$. The measurement primitives averaged over all random unitaries and measurement outcomes implement the quantum channel
\begin{equation}\label{eq:M_channel}
    \mathcal{M} \coloneqq \E_{g \sim G} \mathcal{U}_g^\dagger \mathcal{M}_Z \mathcal{U}_g \equiv \frac{1}{|G|} \sum_{g \in G} \mathcal{U}_g^\dagger \mathcal{M}_Z \mathcal{U}_g,
\end{equation}
where
\begin{equation}
    \mathcal{M}_Z = \sum_{b \in \{0,1\}^n} \vop{b}{b}
\end{equation}
describes the effective process of computational-basis measurements. The channel $ \mathcal{U}_g $ is the random unitary acting on the target state $ \rho $, while $ \mathcal{U}_g^\dagger $ is its classically computed inversion on the measurement outcomes $ \kket{b} $. Thus in expectation we produce the state
\begin{equation}\label{eq:M_rho_def}
    \mathcal{M}\kket{\rho} = \E_{g \sim G, b \sim \mathcal{U}_g \kket{\rho}} \mathcal{U}_g^\dagger \kket{b}.
\end{equation}
If $\mathcal{M}$ is invertible (corresponding to informational completeness of the measurement primitive), then applying $ \mathcal{M}^{-1} $ to Eq.~\eqref{eq:M_rho_def} recovers the state:
\begin{equation}
    \kket{\rho} = \mathcal{M}^{-1} \mathcal{M}\kket{\rho} = \E_{g \sim G, b \sim \mathcal{U}_g \kket{\rho}} \mathcal{M}^{-1} \mathcal{U}_g^\dagger \kket{b}.
\end{equation}
The objects $ \kket{\hat{\rho}_{g,b}} \coloneqq \mathcal{M}^{-1} \mathcal{U}_g^\dagger \kket{b} $ are called the classical shadows of $ \kket{\rho} $, for which they serve as unbiased estimators. Hence by construction they can predict expectation values,
\begin{equation}
    \E_{g \sim G, b \sim \mathcal{U}_g \kket{\rho}} \vip{O_j}{\hat{\rho}_{g,b}} = \vip{O_j}{\rho},
\end{equation}
as well as nonlinear functions of $\rho$~\cite{huang2020predicting}. While $ \mathcal{M}^{-1} $ is not a physical map (it is not completely positive), it only appears as classical postprocessing. Such a computation can be accomplished, for instance, by first deriving a closed-form expression for $ \mathcal{M} $.

One systematic approach to deriving such an expression is through the representation theory of $ G $. First, note that the $ d $-dimensional unitary $ U $ is promoted to a $ d^2 $-dimensional representation $ \mathcal{U} $. Equation~\eqref{eq:M_channel} reveals that $ \mathcal{M} $ is a twirl of $ \mathcal{M}_Z $ by the group $ G $ under the action of $ \mathcal{U} $. Such objects are well studied:~assuming that the irreducible components of $ \mathcal{U} $ have no multiplicities, an application of Schur's lemma implies that~\cite{fulton2004representation}
\begin{equation}\label{eq:channel_diag}
    \mathcal{M} = \sum_{\lambda \in R_G} f_\lambda \Pi_\lambda.
\end{equation}
Note that the general expression with multiplicities can be found in Ref.~\cite[Eq.~(A6)]{chen2021robust}. Here, $ R_G $ is the set of labels $ \lambda $ for the irreducible representations (irreps) of $ G $. The superoperators $ \Pi_\lambda $ are orthogonal projectors onto the irreducible subspaces $V_\lambda \subseteq \mathcal{L}(\mathcal{H}) $. Choosing an orthonormal basis $ \{\kket{B_\lambda^j} : j = 1,\ldots,\dim V_\lambda\} $ for each subspace, we can write the projectors as
\begin{equation}\label{eq:proj_lambda}
    \Pi_\lambda = \sum_{j=1}^{\dim V_\lambda} \vop{B_\lambda^j}{B_\lambda^j}.
\end{equation}
The eigenvalues $ f_\lambda $ of $ \mathcal{M} $ can be computed using the orthogonality of projectors:
\begin{equation}\label{eq:channel_eigval}
    f_\lambda = \frac{\tr(\mathcal{M}_Z \Pi_\lambda)}{\tr(\Pi_\lambda)}.
\end{equation}
Note that $\tr(\Pi_\lambda) = \dim V_\lambda$. From this diagonalization, we immediately acquire an expression for the desired inverse:
\begin{equation}\label{eq:inverse_channel}
    \mathcal{M}^{-1} = \sum_{\lambda \in R_G} f_\lambda^{-1} \Pi_\lambda.
\end{equation}
If some $f_\lambda = 0$, then we may instead define 
$\mathcal{M}^{-1}$ as the pseudoinverse on the subspaces where $ f_\lambda $ is nonvanishing. This implies that the measurement primitive is informationally complete only within those subspaces.

To analyze the sample efficiency of this protocol, suppose we have performed $T$ experiments, yielding a collection of independent classical shadows $\hat{\rho}_1, \ldots, \hat{\rho}_T$ where each $\kket{\hat{\rho}_\ell} = \mathcal{M}^{-1} \mathcal{U}_{g_\ell}^\dagger \kket{b_\ell}$. From this data we can construct estimates
\begin{equation}\label{eq:noiseless_estimator}
    \hat{o}_j(T) = \frac{1}{T} \sum_{\ell=1}^T \vip{O_j}{\hat{\rho}_{\ell}},
\end{equation}
which by linearity converge to $\tr(O_j \rho)$. The single-shot variance of $\hat{o}_j$ can be bounded in terms of the so-called shadow norm:
\begin{equation}
\begin{split}
    \V[\hat{o}_j] &\leq \max_{\text{states } \sigma} \E_{g \sim G, b \sim \mathcal{U}_g \kket{\sigma}} \vev{O_j}{\mathcal{M}^{-1} \mathcal{U}_g^\dagger}{b}^2\\
    &\eqqcolon \sns{O_j}{\mathrm{shadow}}^2.
\end{split}
\end{equation}
This variance controls the prediction error, rigorously established via probability tail bounds. In particular, taking a number of samples
\begin{equation}
    T = \O\l( \frac{\log(L/\delta)}{\epsilon^2} \max_{1 \leq j \leq L} \sns{O_j}{\mathrm{shadow}}^2 \r)
\end{equation}
ensures that, with probability at least $1 - \delta$, each estimate exhibits at most $\epsilon$ additive error:
\begin{equation}
    |\hat{o}_j(T) - \vip{O_j}{\rho}| \leq \epsilon.
\end{equation}
Note that for simplicity we employ the mean estimator throughout this paper, which suffices whenever the ensemble is either local Cliffords or matchgates and the observables are Pauli or Majorana operators~\cite[Supplemental Material, Theorem~12]{zhao2021fermionic}. In general, a median-of-means estimator can guarantee the advertised sample complexity regardless of ensemble.

Finally, we comment on the classical computation of $\hat{o}_j$. In order to evaluate Eq.~\eqref{eq:noiseless_estimator}, one may use Eqs.~\eqref{eq:proj_lambda} and \eqref{eq:inverse_channel} to express the $\ell$th-sample estimate as
\begin{equation}\label{eq:shadow_estimate_formula}
    \vip{O_j}{\hat{\rho}_\ell} = \sum_{\lambda \in R_G} f_\lambda^{-1} \sum_{k=1}^{\dim V_\lambda} \vip{O_j}{B_\lambda^k} \vev{B_\lambda^k}{\mathcal{U}_{g_\ell}^\dagger}{b_\ell}.
\end{equation}
Thus it suffices to be able to efficiently compute the expansion coefficients $\vip{O_j}{B_\lambda^k} = \tr(O_j B_\lambda^k)$ of the observable $O_j$ in a basis of $V_\lambda$, as well as the matrix elements $\vev{B_\lambda^k}{\mathcal{U}_g^\dagger}{b} = \ev{b}{U_g (B_\lambda^k)^\dagger U_g^\dagger}{b}$. Note that this does not require explicitly representing the classical shadow $\mathcal{M}^{-1} \mathcal{U}_g^\dagger \kket{b}$;~we only need to determine the diagonal entry of the rotated operator $U_g (B_\lambda^k)^\dagger U_g^\dagger$ for a given basis state $\ket{b}$.\\

\textbf{Robust shadow estimation.} We now summarize the robust shadow estimation protocol by Chen \emph{et al.}~\cite{chen2021robust};~we note that Refs.~\cite{karalekas2020quantum,van2022model,arrasmith2023development} describe analogous ideas in the case of random single-qubit measurements. The basic premise is the fact that Schur's lemma applies to the twirl of any channel, not just $\mathcal{M}_Z$. Suppose that instead of $ \mathcal{U}_g $, the quantum computer implements a noisy channel $ \widetilde{\mathcal{U}}_g $ which obeys the following assumptions:
\begin{assumptions}[{\cite[Simplifying noise assumption \textbf{A1}]{chen2021robust}}]
\label{assumption_1}
The noise in $ \widetilde{\mathcal{U}}_g $ is gate independent, time stationary, and Markovian. Hence there exists the decomposition $ \widetilde{\mathcal{U}}_g = \mathcal{E} \mathcal{U}_g $, where $ \mathcal{E} $ is a completely positive, trace-preserving map, independent of both the ideal unitary and the experimental time.
\end{assumptions}
They also assume the ability to prepare the state $\ket{0^n}$ with sufficiently high fidelity. Given these conditions, the noisy version of the shadow channel implemented in experiment becomes
\begin{equation}
    \widetilde{\mathcal{M}} \coloneqq \E_{g \sim G} \mathcal{U}_g^\dagger \mathcal{M}_Z \widetilde{\mathcal{U}}_g = \frac{1}{|G|} \sum_{g \in G} \mathcal{U}_g^\dagger \mathcal{M}_Z \mathcal{E} \mathcal{U}_g,
\end{equation}
which is now a twirl over the composite channel $ \mathcal{M}_Z \mathcal{E} $. Although $\mathcal{E}$ is unknown, Schur's lemma implies that the eigenbasis is preserved, as we now have
\begin{equation}\label{eq:noisy_channel_diag}
    \widetilde{\mathcal{M}} = \sum_{\lambda \in R_G} \widetilde{f}_\lambda \Pi_\lambda,
\end{equation}
where the eigenvalues depend on $\mathcal{E}$,
\begin{equation}\label{eq:noisy_channel_eigval}
    \widetilde{f}_\lambda = \frac{\tr(\mathcal{M}_Z \mathcal{E} \Pi_\lambda)}{\tr(\Pi_\lambda)}.
\end{equation}
Therefore if one knows $\widetilde{f}_\lambda$, then one can perform the correct linear inversion in the presence of noise, i.e., by replacing $f_\lambda^{-1}$ with $\widetilde{f}_\lambda^{-1}$ in Eq.~\eqref{eq:shadow_estimate_formula}.

Because $ \mathcal{E} $ depends on the details of the quantum hardware, it is not possible to determine $ \widetilde{f}_\lambda $ without an \emph{a priori} accurate characterization of the noise. Absent such information, a calibration protocol is proposed to experimentally estimate the value of $ \widetilde{f}_\lambda $. This proceeds by performing the classical shadows protocol on a fiducial state $\ket{0^n}$, rather than the unknown target state $\rho$. This enables the study of errors in the random circuits $U_g$. Because $\ket{0^n}$ is known exactly, one can compare its noiseless properties against the noisy experimental data to determine a calibration factor.

Specifically, Chen \emph{et al.}~\cite{chen2021robust} construct an estimator $\mathrm{NoiseEst}_G(\lambda, g, b)$ for each sample $(U_g, b)$ of the calibration experiment, which converges to $\widetilde{f}_\lambda$ in expectation over $g$ and $b$. Although they do not prescribe a generic expression for $\mathrm{NoiseEst}_G$ (instead considering particular choices of $G$), it is straightforward to derive one following their ideas. Let $D_\lambda \in V_\lambda$ be an observable supported exclusively by a single irrep such that $\ev{0^n}{D_\lambda}{0^n} \neq 0$. Then we have
\begin{equation}
    \vev{D_\lambda}{\widetilde{\mathcal{M}}}{0^n} = \widetilde{f}_\lambda \ev{0^n}{D_\lambda}{0^n}.
\end{equation}
On the other hand, using the fact that
\begin{equation}
\begin{split}
    \vev{D_\lambda}{\widetilde{\mathcal{M}}}{0^n} &= \bbra{D_\lambda} \E_{g \sim G, b \sim \mathcal{U}_g \kket{0^n}} \mathcal{U}_g^\dagger \kket{b}\\
    &= \E_{g \sim G, b \sim \mathcal{U}_g \kket{0^n}} \ev{b}{U_g D_\lambda U_g^\dagger}{b},
\end{split}
\end{equation}
it follows that the random variable
\begin{equation}\label{eq:noiseest}
    \mathrm{NoiseEst}_G(\lambda, g, b) = \frac{\ev{b}{U_g D_\lambda U_g^\dagger}{b}}{\ev{0^n}{D_\lambda}{0^n}}
\end{equation}
obeys $\E_{g,b}\l[\mathrm{NoiseEst}_G(\lambda, g, b)\r] = \widetilde{f}_\lambda$.

One can recover the definitions for $\mathrm{NoiseEst}_G$ introduced by Chen \emph{et al.}~\cite{chen2021robust} as follows. The global Clifford group $\Cl(n)$ has two irreps:~the span of the identity operator, $V_0 = \spn\{\I\}$ (which is trivial), and its orthogonal complement $V_1 = V_0^\perp$ (the set of all traceless operators). Choosing $D_1 = d \op{0^n}{0^n} - \I$ gives
\begin{equation}
    \mathrm{NoiseEst}_{\Cl(n)}(1, U, b) = \frac{d |\ev{b}{U}{0^n}|^2 - 1}{d - 1},
\end{equation}
where $U \in \Cl(n)$.

On the other hand, the local Clifford group $\Cl(1)^{\otimes n}$ has $2^n$ irreps, labeled by all subsets $I \subseteq [n]$. Each $I$ indexes a subsystem of qubits, and each subspace $V_I$ is the span of all $n$-qubit Pauli operators which act nontrivially on exactly that subsystem. Defining
\begin{equation}
    D_I \coloneqq \prod_{i \in I} Z_i,
\end{equation}
one obtains
\begin{equation}
\begin{split}
    \mathrm{NoiseEst}_{\Cl(1)^{\otimes n}}(I, U, b) &= \frac{\ev{b}{U D_I U^\dagger}{b}}{\ev{0^n}{D_I}{0^n}}\\
    &= \prod_{i \in I} \ev{b_i}{C_i Z C_i^\dagger}{b_i}
\end{split}
\end{equation}
where now $U = \bigotimes_{i \in [n]} C_i \in \Cl(1)^{\otimes n}$.

Any QEM strategy necessarily incurs a sampling overhead dependent on the amount of noise~\cite{takagi2022fundamental,takagi2022universal,tsubouchi2022universal,quek2022exponentially}. For global Clifford shadows, Chen \emph{et al.}~\cite{chen2021robust} show that the sample complexity is augmented by a factor of $\O(F_Z(\mathcal{E})^{-2})$ for estimating observables with constant Hilbert--Schmidt norm, where $F_Z(\mathcal{E}) = 2^{-n} \sum_{b \in \{0,1\}^n} \vev{b}{\mathcal{E}}{b}$ is the average $Z$-basis fidelity of $\mathcal{E}$. Meanwhile for local Clifford shadows, they prove that product noise of the form $\mathcal{E} = \bigotimes_{i\in[n]} \mathcal{E}_i$, satisfying $\min_{i \in [n]} F_Z(\mathcal{E}_i) \geq 1 - \xi$, exhibits an overhead factor of $e^{\O(k\xi)}$ for estimating $k$-local qubit observables.

\subsection{Symmetry-adjusted classical shadows}\label{subsec:EM_shadows_summary}

\begin{figure*}
\centering
\includegraphics[scale=0.5]{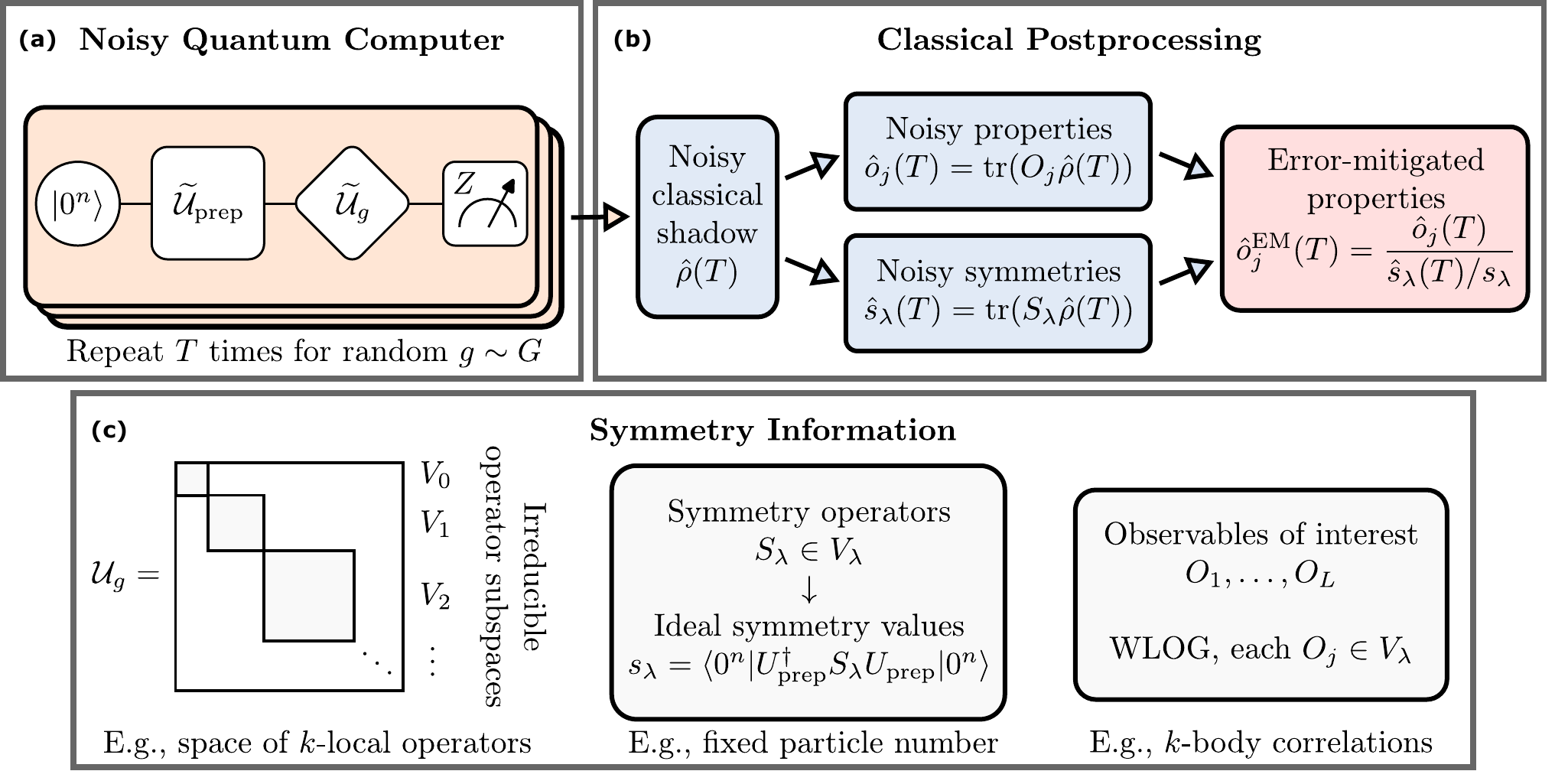}
\caption{\textbf{Schematic of the symmetry-adjusted classical shadows protocol.} \textbf{(a)} Given an ideal unitary $\mathcal{U}(\cdot) = U(\cdot)U^\dagger$, its noisy implementation is denoted by $\widetilde{\mathcal{U}}$. Assuming the target state $\rho = \mathcal{U}_{\mathrm{prep}}(\op{0^n}{0^n})$ obeys certain symmetries $S_\lambda$, \textbf{(b)} we can construct error-mitigated estimates using classical shadows produced by the noisy quantum computer. (While we depict the preparation of a pure state here, our formalism is equally valid if the target state is mixed.) In contrast to prior approaches that only address the noise in $\widetilde{\mathcal{U}}_g$, our protocol additionally incorporates the errors within $\widetilde{\mathcal{U}}_{\mathrm{prep}}$. \textbf{(c)} Our method is applicable whenever the symmetry is compatible with the irreps $V_\lambda$ of the group $G$ describing the classical shadows protocol.}
\label{fig:cartoon_flowchart}
\end{figure*}

The primary contribution of this paper, symmetry-adjusted classical shadows, is visualized in Figure~\ref{fig:cartoon_flowchart}. We describe it in detail now. Consider a classical shadows protocol over $G$ with target observables $O_1, \ldots, O_L$. Without loss of generality, let each $O_j \in V_\lambda$ for some subset of irreps $\lambda \in R' \subseteq R_G$. Suppose the experiment experiences an unknown noise channel $\mathcal{E}$ obeying Assumptions~\ref{assumption_1}.

We show that, if $\rho$ obeys symmetries which are ``compatible'' with the irreps in $R'$, then it is possible to construct an estimator which accurately predicts the ideal, noiseless observables. By compatible, we mean that there exist symmetry operators $S_\lambda \in V_\lambda$ for each $\lambda \in R'$ for which their ideal expectation values
\begin{equation}
    s_\lambda \coloneqq \tr(S_\lambda \rho)
\end{equation}
are known \emph{a priori}. In general, there is no reason to expect that a physical system has symmetries which exactly fit into the irreps of a classical-shadow measurement scheme. However, given a symmetry operator $S$, it is always possible to project it to $V_\lambda$ using the superoperator projector $\Pi_\lambda$, i.e., $S_\lambda = \Pi_\lambda(S)$.

Then, using noisy classical shadows $\hat{\rho}(T)$ of size $T$, we construct error-mitigated estimates as
\begin{equation}
    \hat{o}_j^{\mathrm{EM}}(T) \coloneqq \frac{\tr(O_j \hat{\rho}(T))}{\tr(S_\lambda \hat{\rho}(T)) / s_\lambda}.
\end{equation}
We find that the relevant noise characterization in this scenario is
\begin{equation}
    F_{Z, R'}(\mathcal{E}) \coloneqq \min_{\lambda \in R'} \frac{\tr(\mathcal{E} \mathcal{M}_Z \Pi_\lambda)}{\tr(\mathcal{M}_Z \Pi_\lambda)},
\end{equation}
which can be seen as a generalization of the noise fidelity $F_Z(\mathcal{E})$ described in Section~\ref{sec:background}. Here, $F_{Z, R'}(\mathcal{E})$ only considers how the noise channel acts within the irreducible subspaces of interest.

As two key applications, we study how symmetry-adjusted classical shadows perform in simulations of fermionic and qubit systems. For fermions, we consider $G$ corresponding to fermionic Gaussian unitaries~\cite{zhao2021fermionic} (also known as matchgate shadows~\cite{wan2023matchgate}). We establish the following performance bound for fermionic systems with particle-number symmetry, $N = \sum_{p \in [n]} a_p^\dagger a_p$.

\begin{theorem}[Fermions with particle-number symmetry, informal]\label{thm:informal_fermion_result}
    Let $\rho$ be an $n$-mode state with $\tr(N \rho) = \eta$ fermions. Under the noise model $\mathcal{E}$ satisfying Assumptions~\ref{assumption_1} and assuming $\eta = \O(n)$, matchgate shadows of size
    \begin{equation}
    T = \O( n^2 \log(n) \epsilon^{-2} F_{Z,\{2,4\}}(\mathcal{E})^{-2} )
    \end{equation}
    suffice to achieve prediction error
    \begin{equation}
    |\hat{o}_j(T) - \tr(O_j \rho)| \leq \epsilon + \O(\epsilon^2)
    \end{equation}
    with high probability, where the observables $O_j$ can be taken as all one- and two-body Majorana operators.
\end{theorem}
The dependence on system size $n$ and prediction error $\epsilon$ matches noiseless estimation with matchgate shadows~\cite{zhao2021fermionic,wan2023matchgate}. Meanwhile, the overhead of error mitigation is $\O(F_{Z,R'}(\mathcal{E})^{-2})$, analogous to prior related results~\cite{chen2021robust,koh2022classical}. The irreps $R' = \{2, 4\}$ correspond to the Majorana degree of the $k$-body observables.

For qubit systems, we consider $G$ essentially corresponding to the local Clifford group (i.e., random Pauli measurements)~\cite{huang2020predicting,paini2021estimating}. In order to make the irreducible structure compatible with commonly encountered symmetries, we introduce a technical modification that we call subsystem-symmetrized Pauli shadows (see Section~\ref{subsec:sym_pauli_summary} for a summary). The symmetry we consider here is generated by the total longitudinal magnetization, $M = \sum_{i \in [n]} Z_i$. For error-mitigated prediction of local qubit observables, we have the following result.

\begin{theorem}[Qubits with total magnetization symmetry, informal]\label{thm:informal_qubit_result}
    Let $\rho$ be an $n$-qubit state with a fixed magnetization, $\tr(M \rho) = m$. Under the noise model $\mathcal{E}$ satisfying Assumptions~\ref{assumption_1} and assuming $m = \Theta(1)$, subsystem-symmetrized Pauli shadows of size
    \begin{equation}
    T = \O( n \log(n) \epsilon^{-2} F_{Z,\{1,2\}}(\mathcal{E})^{-2} )
    \end{equation}
    suffices to achieve prediction error
    \begin{equation}
    |\hat{o}_j(T) - \tr(O_j \rho)| \leq \epsilon + \O(\epsilon^2)
    \end{equation}
    with high probability, where the observables $O_j$ can be taken as all one- and two-local Pauli operators.
\end{theorem}

Note that the irreps of subsystem-symmetrized Pauli shadows are labeled by Pauli weight. The variance bound we advertise here is linear in $n$, resulting from the extensive nature of the symmetry $M$. Specifically, we show that when $m = \Theta(1)$, $\sn{M}^2 = \O(n)$ dominates the asymptotic complexity over the $k$-local Pauli observables (for which our protocol exhibits the usual $\sn{O_j}^2 = 3^k$). This is consistent with standard Pauli shadows, wherein the shadow norm of arbitrary $k$-local observables scales at most linearly with spectral norm and exponentially in $k$~\cite{huang2020predicting,paini2021estimating}.

Besides these two examples, we describe symmetry-adjusted classical shadows for a more general class of groups $G$, and we establish accompanying bounds in Theorem~\ref{thm:main_theorem} in Section~\ref{sec:methods} (proven in Supplementary Section~\ref{sec:error_analysis}). This allows for applications to other systems and unitary distributions. See Section~\ref{sec:main_theory} for the general theory, and Sections~\ref{sec:applications_fermions} and \ref{sec:applications_qubits} for the details regarding Theorems~\ref{thm:informal_fermion_result} and \ref{thm:informal_qubit_result}, respectively.

Because our protocol always runs the full noisy quantum circuit, it has the potential to mitigate a wider range of errors than those covered by Assumptions~\ref{assumption_1}, albeit without the rigorous theoretical guarantees. This is a significant feature of the method, as the preparation of $\rho$ often dominates the total circuit complexity (i.e., $U_{\mathrm{prep}}$ in Figure~\ref{fig:cartoon_flowchart}). We explore this broader mitigation potential with a series of numerical experiments below, wherein we simulate noisy Trotter circuits for systems of interacting fermions and spin-$1/2$ particles, respectively.

\subsection{Subsystem-symmetrized Pauli shadows}\label{subsec:sym_pauli_summary}

While random Pauli measurements are efficient for predicting local qubit observables, the irreducible structure of the local Clifford group $\Cl(1)^{\otimes n}$ is difficult to reconcile with common symmetries under symmetry adjustment, such as the $\U(1)$ symmetry generated by $M = \sum_{i \in [n]} Z_i$. To remedy this issue, we modify the protocol by what we call \emph{subsystem symmetrization}:~define the group
\begin{equation}
    \SymCl{n} \coloneqq \Sym(n) \times \Cl(1)^{\otimes n},
\end{equation}
which has the unitary representation $U_{(\pi, C)} = S_\pi C$ where $S_\pi$ permutes the qubits according to $\pi \in \Sym(n)$ and $C \in \Cl(1)^{\otimes n}$. The circuit for $S_\pi$ can be obtained as a sequence of $\O(n^2)$ nearest-neighbor $\SWAP$ gates in $\O(n)$ depth via an odd--even decomposition of $\pi$~\cite{habermann1972parallel}. The following theorem summarizes its group-theoretic properties relevant to classical shadows.

\begin{theorem}[Irreducible representations of the subsystem-symmetrized local Clifford group]
    The representation $\mathcal{U} : \SymCl{n} \to \U(\mathcal{L}(\mathcal{H}))$, defined by $\mathcal{U}_{(\pi, C)}(\rho) = S_\pi C \rho C^\dagger S_\pi^\dagger$, decomposes into the irreps
    \begin{equation}
    V_k = \spn\{ P \in \Pauli(n) : |P| = k \}, \quad 0 \leq k \leq n.
    \end{equation}
    Under this group, the (noiseless) expressions for $\mathcal{M}$ and $\V[\hat{o}]$ coincide with those of standard Pauli shadows.
\end{theorem}

This modification therefore reduces the number of irreps from $2^n$ to $n + 1$, achieved by symmetrizing, for each $k$, over all $k$-qubit subsystems. Meanwhile, the desirable estimation properties from standard Pauli shadows are retained:~for instance, the shadow norm obeys $\sn{P}^2 = 3^k$ for $k$-local Pauli operators $P$.

The upshot is that the symmetry $M$ is now compatible with this group, thereby enabling results such as Theorem~\ref{thm:informal_qubit_result}. We describe this construction in Section~\ref{sec:applications_qubits}, with technical proofs in Supplementary Section~\ref{sec:pauli_shadows_appendix}.

\subsection{Spin-adapted matchgate shadows}\label{subsec:spin-adapt_summary}

Systems of spinful fermions often obey a spin symmetry, which allows for compressed block-diagonal representations according to the spin sectors. Such techniques are referred to as symmetry adaptation. We introduce such an adaptation of the matchgate shadows protocol wherein the random distribution is restricted to block-diagonal orthogonal transformations,
\begin{equation}
    Q = \begin{pmatrix}
    Q_\uparrow & 0\\
    0 & Q_\downarrow
    \end{pmatrix} \in \Orth(n) \oplus \Orth(n).
\end{equation}
We call this protocol \emph{spin-adapted} matchgate shadows. This restricted group remains informationally complete over operators which respect the spin sectors, thus sufficing for learning properties in systems with this symmetry. In fact, we show that the shadow norms for $k$-fermion operators under the spin-adapted protocol scale identically as in the unadapted setting. The main advantage of spin adaptation is that the block-diagonal transformation $Q = Q_\uparrow \oplus Q_\downarrow$ can be implemented as $U_{Q_\uparrow} \otimes P_{\downarrow}^s U_{Q_\downarrow}$, where $P_{\downarrow} = Z^{\otimes n/2}$ is the parity operator on the spin-down sector and $s = \delta_{-1, \det Q_{\uparrow}}$. This tensor-product unitary requires roughly half the number of gates and circuit depth compared to implementing a dense element of $\Orth(2n)$. We prove the necessary details in Supplementary Section~\ref{sec:spin_adaptation} and implement this modified protocol in our numerical experiments wherever applicable.

\subsection{Improved circuit design for fermionic Gaussian unitaries}\label{subsec:improved_circuit_summary}

Fermionic Gaussian unitaries are a broad class of free-fermion rotations, and they are ubiquitous primitives in algorithms for simulating (interacting) fermions. In the context of classical shadows, they form the basis for randomized measurements in matchgate shadows~\cite{zhao2021fermionic,wan2023matchgate,ogorman2022fermionic}. Such unitaries can be described by an orthogonal transformation $Q \in \Orth(2n)$ of the Majorana operators,
\begin{equation}
    \mathcal{U}_Q(\gamma_\mu) = U_Q \gamma_\mu U_Q^\dagger = \sum_{\nu \in [2n]} Q_{\nu\mu} \gamma_\nu
\end{equation}
for each $\mu \in [2n]$. The quantum circuits implementing these transformations take $\O(n^2)$ gates in $\O(n)$ depth~\cite{jiang2018quantum,oszmaniec2022fermion}. While this scaling is necessary in general by parameter counting, constant-factor savings can substantially improve performance in practice, especially on noisy quantum computers.

To this end, we introduce a more efficient compilation scheme for fermionic Gaussian unitaries, given an arbitrary $Q \in \Orth(2n)$. Our circuit design improves the parallelization of gates compared to prior art~\cite{jiang2018quantum,oszmaniec2022fermion}. The key idea is to observe that two Majorana modes essentially correspond to one qubit under the Jordan--Wigner transformation~\cite{jordanwigner}. Thus, the optimal approach to compiling $U_Q$ into single- and two-qubit gates involves decomposing the matrix $Q$ into elementary blocks of $4 \times 4$ transformations, rather than the $2 \times 2$ Givens rotations utilized in prior designs.

The details of this scheme are described in Supplementary Section~\ref{sec:FGU_circuit_design} and implemented in code at our open-source repository ({\small\url{https://github.com/zhao-andrew/symmetry-adjusted-classical-shadows}}). We make use of this improved design in our numerical simulations. We demonstrate the improvements in circuit size in Figure~\ref{fig:circuit_design_comparison}, with respect to a gate set native to superconducting platforms. From these results we numerically infer roughly $1/3$ reduction in depth and $1/2$ reduction in gate count over prior designs.

\begin{figure}
\centering
\includegraphics[scale=0.5]{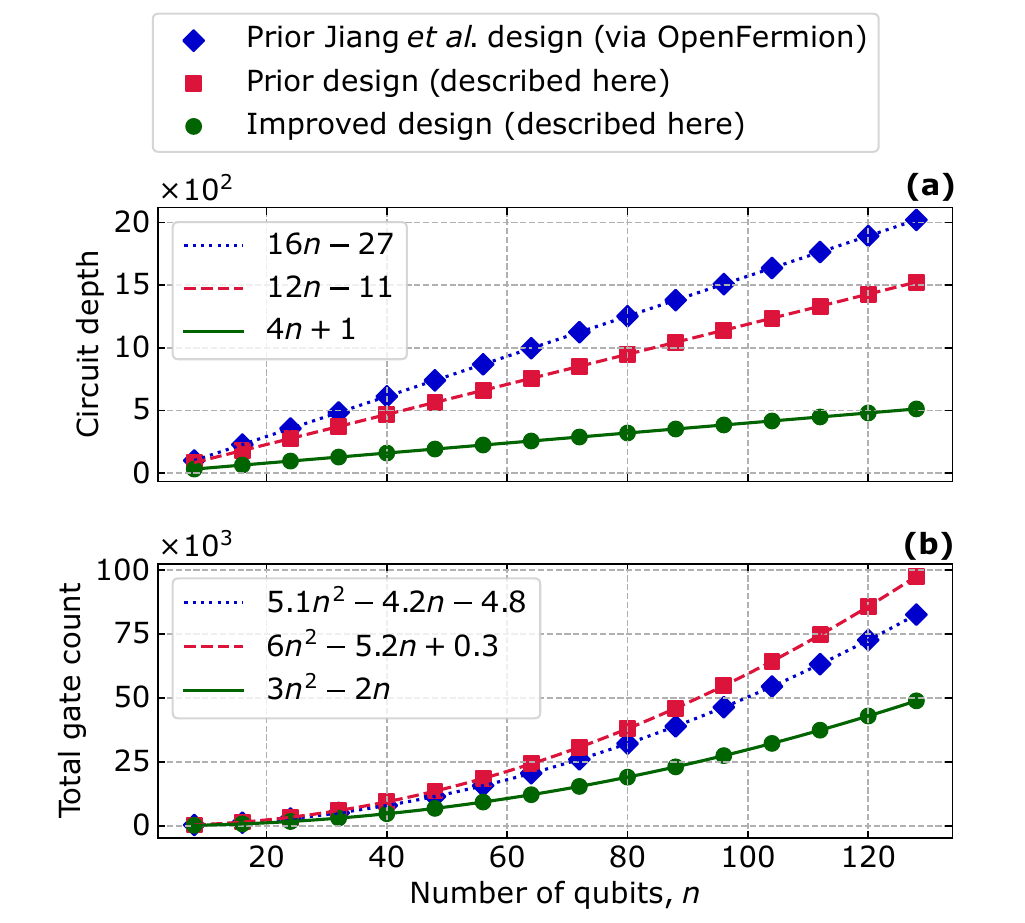}
\caption{\textbf{Resource comparison of our improved fermionic Gaussian circuit design versus prior designs.} Our circuit costs (green) are compared against those of Jiang \emph{et al.}~\cite{jiang2018quantum} (blue), implemented in OpenFermion~\cite{openfermion}, and a naive scheme described in Supplementary Section~\ref{sec:FGU_circuit_design} (red). All circuits were optimized to a superconducting gate set consisting of $\sqrt{\iSWAP}$ gates and single-qubit rotations. We compare both \textbf{(a)} circuit depth and \textbf{(b)} gate count on random inputs $Q \in \Orth(2n)$. Linear and quadratic fits are made, demonstrating a roughly $3\times$ and $2\times$ savings, respectively.}
\label{fig:circuit_design_comparison}
\end{figure}

\subsection{Numerical experiments}\label{sec:numerics}

We now demonstrate the error-mitigation capabilities of symmetry-adjusted classical shadows through numerical simulations. We focus on the task of estimating one- and two-body observables in both fermion and qubit systems which obey the global $\U(1)$ symmetries described in Section~\ref{sec:methods}.

For each type of system, we first present results when the noise models obey Assumptions~\ref{assumption_1} (readout errors). We demonstrate the successful mitigation at varying sample sizes, noise rates, and system sizes, confirming the correctness of our theory.

Next, we investigate how symmetry adjustment performs under a more comprehensive noise model based on superconducting-qubit platforms. These simulations were performed using the Quantum Virtual Machine (QVM) within the Cirq open-source software package~\cite{cirq,isakov2021simulations}. It uses existing hardware data on a native gate set (single-qubit rotations and two-qubit $\sqrt{\iSWAP}$ gates on a square lattice) to mimic the realistic performance of a noisy quantum computer. We use the calibration data provided of Google's 23-qubit Rainbow processor based on the Sycamore architecture, which was used in quantum experiments simulating quantum chemistry and strongly correlated materials~\cite{arute2020hartree,arute2020observation}. The noise model consists of depolarizing channels, two-qubit coherent errors, single-qubit idling noise, and readout errors. Error rates vary across the chip;~on the $2 \times 4$ grid that we simulated, the average single- and two-qubit Pauli error rates are ${\sim}0.15\%$ and ${\sim}1.5\%$, respectively. A precise description of the noise model can be found in Supplementary Section~\ref{subsec:noise_model_details}.

Throughout, we use the following conventions for figures. Noiseless data (blue squares) correspond to simulations of an ideal quantum computer, which experiences no noise channel and only exhibits the fundamental sampling error. Unmitigated data (black X's) are simulations of classical shadows on a noisy quantum computer, using standard postprocessing routines. The mitigated estimates (red diamonds) are instead postprocessed as symmetry-adjusted classical shadows, as described in Section~\ref{sec:main_theory}. In some experiments, we also compare against robust shadow estimation~\cite{chen2021robust} (RShadow, green crosses), which involves simulating the calibration protocol on $\ket{0^n}$ under the same noise model. Finally, the true values (teal curves) are the ground truth, against which we determine the prediction error.

Uncertainty bars represent one standard deviation of the combined sampling and postprocessing, computed by empirical bootstrapping~\cite{efron1992bootstrap}. To ease the computational load, we slightly modify the procedure by batching samples;~see Supplementary Section~\ref{subsec:bootstrap_error_bars} for details.\\

\textbf{Fermionic systems.} Our first set of numerical experiments consider the application to matchgate shadows to learn and mitigate noise in one- and two-body fermionic observables. The symmetry we consider is fixed particle number, $\tr(N \rho) = \eta$. As we show in Section~\ref{sec:applications_fermions}, this symmetry projects into the relevant irreps $R' = \{2, 4\}$ of the matchgate shadows as
\begin{align}
    S_2 &= \Pi_2(N) = -\frac{1}{2} \sum_{p \in [n]} Z_p,\\
    S_4 &= \Pi_4(N^2) = \frac{1}{2} \sum_{p < q} Z_p Z_q,
\end{align}
represented under the Jordan--Wigner transformation~\cite{jordanwigner} for simplicity. Their ideal values are
\begin{align}
    s_2 &= \tr\l( S_2 \rho \r) = \eta - \frac{n}{2},\\
    s_4 &= \tr\l( S_4 \rho \r) = \frac{1}{2} \binom{n}{2} - \eta(n - \eta).
\end{align}

\begin{figure*}
\centering
\includegraphics[scale=0.5]{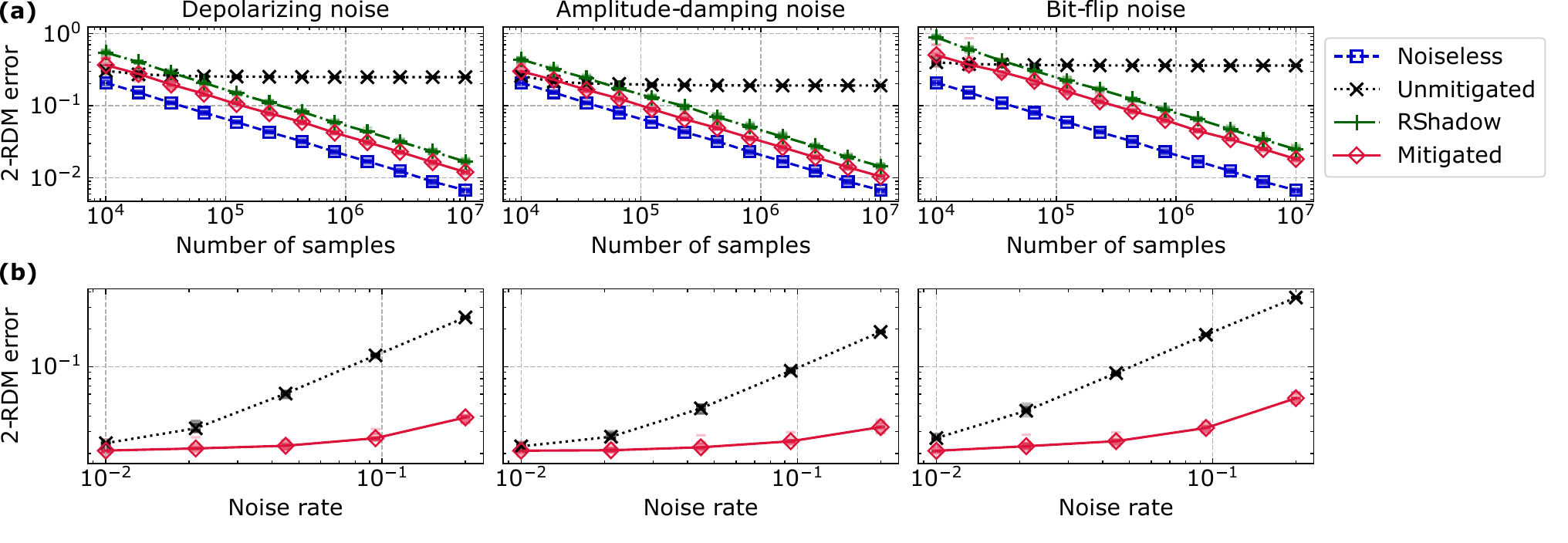}
\caption{\textbf{Estimation error of the fermionic 2-RDM reconstructed from matchgate shadows.} Error is quantified by the spectral-norm difference between the estimated and true 2-RDM. We simulate the protocol on a collection of 20 random Slater determinants on $n = 8$ modes with $\eta = 2$ fermions;~faint dashes are results for individual states, markers indicate the median error across the 20 states. \textbf{(a)} Scaling of estimation error with the total number of samples $T$, fixing the noise rate to $p = 0.2$ for all noise models. \textbf{(b)} Scaling of the estimation error with the noise rate $p$, fixing the total number of samples to $T = 10^6$.}
\label{fig:random_slater}
\end{figure*}

\textbf{Readout noise model (fermions).} First, we consider the reconstruction of the fermionic two-body reduced density matrix (2-RDM) from matchgate shadows. The 2-RDM elements of a state $\rho$ are given by
\begin{equation}
    {}^2 D_{rs}^{pq} = \tr(a_p^\dagger a_q^\dagger a_s a_r \rho), \quad p, q, r, s \in [n].
\end{equation}
In general, knowledge of the $k$-RDM allows one to calculate any $k$-body observable of the system. By anticommutation relations, there are only $\binom{n}{2}^2$ unique matrix elements, corresponding to the indices $p < q$ and $r < s$. We therefore represent ${}^2 D$ as an $\binom{n}{2} \times \binom{n}{2}$ Hermitian matrix, flattening along those index pairs. Estimates ${}^2 \hat{D}_{rs}^{pq}(T) = \tr(a_p^\dagger a_q^\dagger a_s a_r \hat{\rho}(T))$ are computed from $T$ matchgate-shadow samples. Here, our figure of merit for the prediction error is the spectral-norm difference between the reconstructed and the numerically exact 2-RDMs, $\epsilon = \| {}^2 \hat{D} - {}^2 D \|_\infty$.

We demonstrate 2-RDM reconstruction on an ensemble of 20 random Slater determinants (noninteracting-fermion states with fixed particle number). An $\eta$-fermion Slater determinant is specified by the first $\eta$ columns of an $n \times n$ unitary matrix, so we generate the random states by uniformly drawing elements of $\U(n)$. This $n \times n$ representation is then lifted to the $2n \times 2n$ fermionic Gaussian representation, which allows us to apply the random matchgate transformations $Q \in \B(2n)$ efficiently. This simulates the action of $\rho \mapsto U_Q \rho U_Q^\dagger$. The measurement is then simulated using the algorithm of Ref.~\cite[Sec.~5.1]{bravyi2012classical}. Finally, to simulate the readout noise we implement the effective noise channel on the sampled bit strings offline.

While the 2-RDM of free-fermion states can be computed from the 1-RDM using Wick's theorem, we do not employ any such tricks here;~we use Slater determinants simply to facilitate fast classical simulation. We also do not use any additional error-mitigation strategies, such as RDM positivity constraints~\cite{rubin2018application}, that could in principle be applied in tandem.

The results are presented in Figure~\ref{fig:random_slater}. We consider a small system size, $n = 8$ and $\eta = 2$, and simulate three types of single-qubit noise channels before readout:~depolarizing, amplitude damping, and bit flip. The noise rate $p$ represents the probability of such an error occurring, independently on each qubit (defined in Supplementary Section~\ref{subsec:readout_errors}). In the top row, we show how the prediction error varies with the total number of samples $T$. As expected, the noiseless estimates (corresponding to $p = 0$) converge as $\sim T^{-1/2}$, which is the standard shot-limited behavior. Then, setting $p = 0.2$, we see how the unmitigated data experiences an error floor beyond which taking additional samples does not improve the accuracy. On the other hand, the mitigated results clearly bypass this error floor and recover the shot-noise scaling with $T$, thus validating the theory of symmetry-adjusted classical shadows. Compared to the noiseless simulations, our mitigated data exhibit a constant factor increase in the sampling cost, corresponding to the $\O(F_{Z,R'}^{-2})$ overhead of error mitigation, as it appears in Theorem~\ref{thm:main_theorem}.

\begin{figure*}
\centering
\includegraphics[scale=0.5]{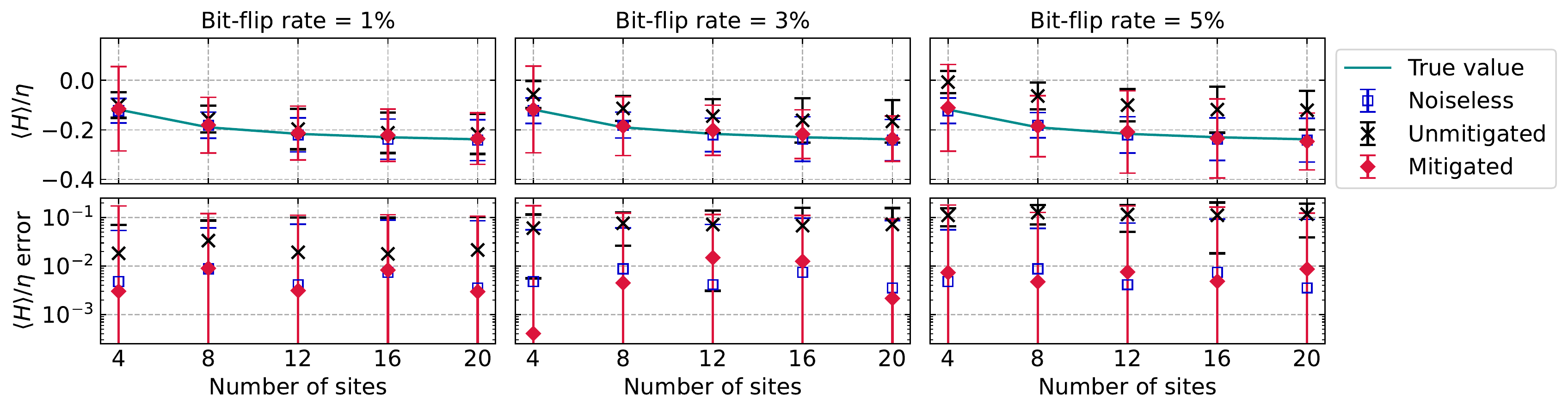}
\caption{\textbf{Estimation of the energy per electron of a 1D spinful Fermi--Hubbard model.} We consider an interaction strength $U/t = 4$ on $4 \leq L \leq 20$ sites, and we measure $\langle H \rangle / \eta$ using spin-adapted matchgate shadows with $T = 2 \times 10^6$ samples. We simulate the protocol on the ground state of the noninteracting component $J$ of the Hamiltonian, Eq.~\eqref{eq:hopping_term}, at half filling in each spin sector, which is a Slater determinant with $\eta_\sigma = L/2$ electrons per sector. These simulations involve $n' = 2L + 2$ noisy qubits experiencing readout bit-flip errors with probabilities $1\%$, $3\%$, and $5\%$. The additional qubit per spin sector is prepared in $\ket{0}$ to avoid division by zero (see Section~\ref{sec:applications_fermions}). Error bars denote one standard error of the mean.}
\label{fig:fh_energy}
\end{figure*}

For these experiments, we also compare to the performance of robust shadow estimation (RShadow) by Chen \emph{et al.}~\cite{chen2021robust}, which requires simulating the calibration procedure on $\ket{0^n}$. For a fair comparison, we allocate $T/2$ samples to the calibration step and $T/2$ samples to the estimation step, so that the total number of samples is the same. While Chen \emph{et al.}~\cite{chen2021robust} did not originally consider matchgate shadows, from our generalization in Eq.~\eqref{eq:noiseest} we can construct $\mathrm{NoiseEst}_{\B(2n)}$ by taking $D_{\lambda} = S_{2k}$, which obeys $\ev{0^n}{S_2}{0^n} = -n/2$ and $\ev{0^n}{S_4}{0^n} = n(n-1)/4$. The single-shot estimator is then
\begin{equation}\label{eq:fermion_RSE_estimator}
    \mathrm{NoiseEst}_{\B(2n)}(2k, Q, b) = \frac{\ev{b}{U_Q S_{2k} U_Q^\dagger}{b}}{\ev{0^n}{S_{2k}}{0^n}}.
\end{equation}
As expected, RShadow behaves similarly to symmetry-adjusted classical shadows in this scenario (wherein the noise obeys Assumptions~\ref{assumption_1}). However, even here we observe that our approach exhibits a constant-factor advantage in the sampling cost. We attribute the performance of RShadow to its calibration procedure, which our method avoids.

In the bottom row of Figure~\ref{fig:random_slater}, we simulate the same collection of random Slater determinants, but now varying the noise rate $p$ at a fixed shadow size $T = 10^6$. While the unmitigated errors quickly grow with increasing noise rate as expected, the mitigated estimates remain under control. Note that the mitigated errors still grow modestly because we have fixed the number of samples;~in order to achieve a constant prediction error, one would need to grow $T$ proportional to $F_{Z,R'}^{-2}$ (which is $p$-dependent). Our key takeaway is that the combination of both rows of plots indicates the ability to handle a range of common noise channels at fairly high error rates. Indeed, the growing errors seen in the bottom row can be suppressed by simply taking more samples, which is precisely what the top row demonstrates.

Next, we consider the simulation of a 1D spinful Fermi--Hubbard chain of $L = n/2$ sites (for a total of $n$ fermionic modes/qubits). Under open boundary conditions, the Hamiltonian for this model is
\begin{equation}
    H = J + V,
\end{equation}
where
\begin{align}
    J &= -t \sum_{i \in [L-1]} \sum_{\sigma \in \{\uparrow, \downarrow\}} a_{i,\sigma}^\dagger a_{i+1,\sigma} + \mathrm{h.c.}, \label{eq:hopping_term}\\
    V &= U \sum_{i \in [L]} N_{i,\uparrow} N_{i, \downarrow}, \label{eq:interaction_term}
\end{align}
are the hopping and interaction terms, respectively. The creation operators $a_{i,\sigma}^\dagger$ produce an electron at site $i$ with spin $\sigma$, and $N_{i,\sigma} = a_{i,\sigma}^\dagger a_{i,\sigma}$ is the associated occupation-number operator. We set units such that the hopping strength is $t = 1$.

For the target state, we use the ground state of the noninteracting term $J$, which is also a Slater determinant. This allows us to use the same simulation techniques as before to efficiently simulate up to $20$ sites. The number of electrons in each spin sector is $\eta_\sigma = L/2$, for a total of $\eta = \eta_{\uparrow} + \eta_{\downarrow} = n/2$ electrons. Thus the system is at half filling, which requires the use of ancilla qubits to avoid division by zero (see Section~\ref{sec:applications_fermions}). In fact, we simulate $n' = n + 2$ qubits because we append an ancilla qubit to each spin sector. This is because we furthermore employ spin-adapted matchgate shadows, as described previously in Section~\ref{subsec:spin-adapt_summary}. This modification essentially treats each spin sector independently when performing the randomized measurements, so each sector is at half filling.

The Fermi--Hubbard results are shown in Figure~\ref{fig:fh_energy}. We consider the estimation of energy per electron, $\langle H \rangle / \eta$. We set the interaction strength to $U/t = 4$ and the noise model to single-qubit bit-flip errors, with probabilities $p \in \{0.01, 0.03, 0.05\}$. The energy per electron (top) and absolute estimation error (bottom) are plotted as the system size grows, keeping the number of samples fixed to $T = 2 \times 10^6$. Again, these results serve to validate our theory, this time highlighting the performance as the system size grows. This also demonstrates the use of spin-adapted matchgate shadows and the successful use of ancillas to avoid division by zero in $\hat{o}_j^{\mathrm{EM}}$.

\begin{figure*}
\centering
\includegraphics[scale=0.5]{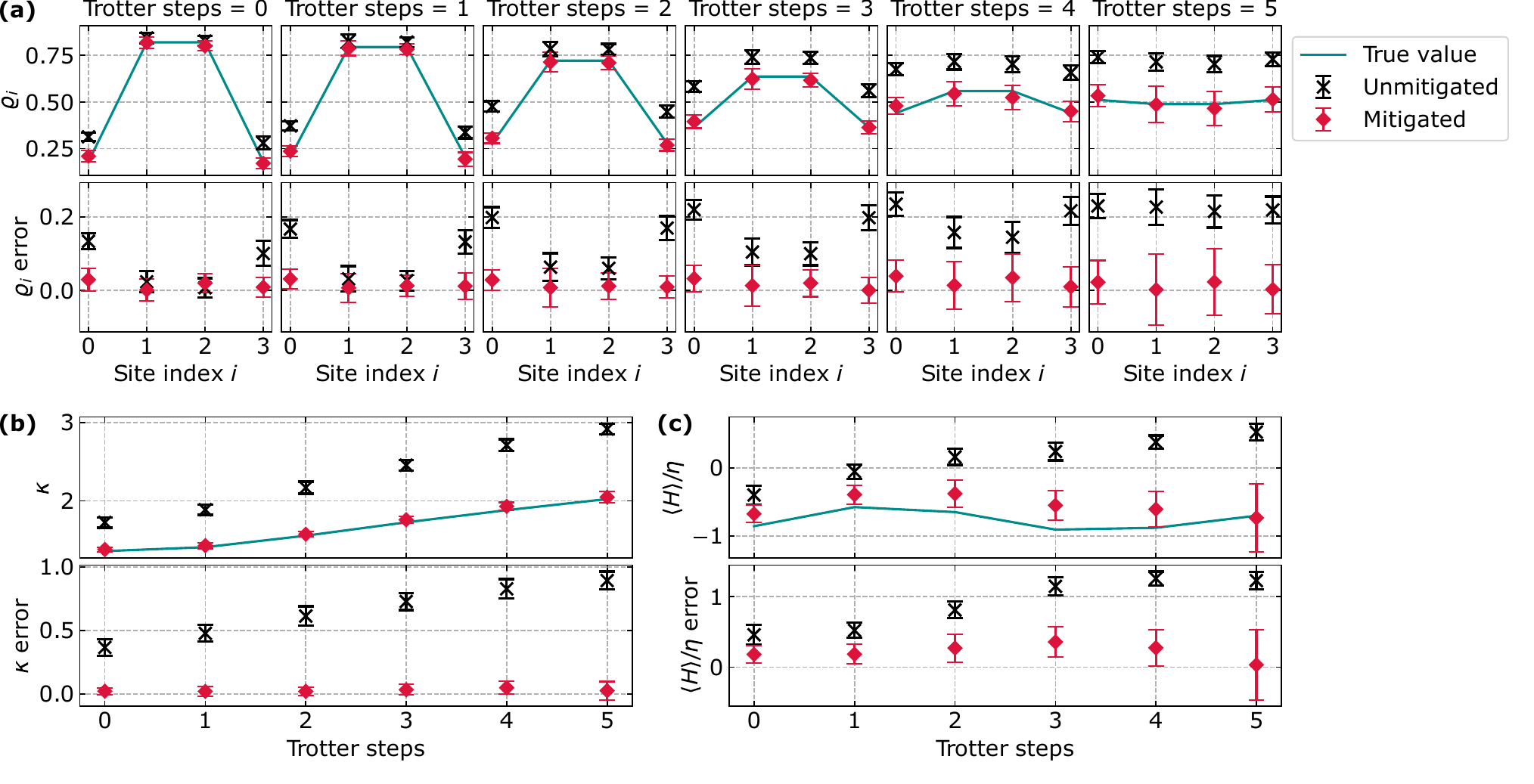}
\caption{\textbf{Prediction of local properties in the four-site 1D Fermi--Hubbard model undergoing Trotterized time evolution.} The interaction strength is set to $U/t = 4$ and the Trotter step size is $\delta t = 0.2$. The noise model is that of the Google Sycamore Rainbow processor~\cite{arute2020hartree,arute2020observation}, implemented within the Cirq QVM~\cite{isakov2021simulations,cirq}. The size of the noisy circuit grows systematically with the number of Trotter steps. We use spin-adapted matchgate shadows, taking $T = 9.6 \times 10^5$ samples. \textbf{(a)} Charge density $\varrho_i$ at each site $i \in [L]$. \textbf{(b)} Charge spread $\kappa$. \textbf{(c)} Energy per electron $\langle H \rangle / \eta$. Error bars denote one standard error of the mean.}
\label{fig:fh_qvm}
\end{figure*}

\textbf{QVM noise model (fermions).} Now we turn to the gate-level noise model simulated through the QVM~\cite{cirq,isakov2021simulations}. This model strongly violates Assumptions~\ref{assumption_1}, reflecting the fact that the state-preparation circuit $U_{\mathrm{prep}}$ is typically the dominant source of errors.

As our testbed fermionic system, we again consider the 1D spinful Fermi--Hubbard chain with open boundary conditions and interaction strength $U/t = 4$. Rather than the static problem, here we simulate Trotterized time evolution of the Hamiltonian. The number of Trotter steps provides a systematic way to increase the circuit depth (and hence the cumulative amount of noise) within the same model. Note that because we are studying the behavior of error mitigation, the ground truth of these simulations corresponds to the noiseless Trotter circuit with a finite step size (i.e., we are not comparing to the exact, non-Trotterized dynamics).

We closely follow the setup of the experiment performed in Ref.~\cite{arute2020observation} (which was in fact performed on the Sycamore processor that our noise model is based on), using code made available by the authors at Ref.~\cite{recirq}. Because simulating the full noisy circuit is exponentially expensive, we restrict to a four-site instance ($n = 2L = 8$). The initial state is the ground state in the $\eta_{\uparrow}, \eta_{\downarrow} = 1$ sector of the noninteracting Hamiltonian
\begin{equation}
\begin{split}
    H_0 = J + \sum_{i \in [L]} \sum_{\sigma \in \{\uparrow, \downarrow\}} \varepsilon_{i,\sigma} N_{i, \sigma},
\end{split}
\end{equation}
where $J$ is the hopping term defined in Eq.~\eqref{eq:hopping_term} and we set the on-site potentials to have a Gaussian form, $\varepsilon_{i,\sigma} = -\lambda_\sigma e^{-\frac{1}{2} (i + 1 - c)^2 / s^2}$. This generates a Slater determinant whose charge density
\begin{equation}
    \varrho_i = \langle N_{i,\uparrow} + N_{i,\downarrow} \rangle.
\end{equation}
has a Gaussian profile, centered around $c$ with width $s$ and magnitude $\lambda_\sigma$. We set the parameters to $c = L/2 + 1/2 = 2.5$, $s = 7/3$, and $\lambda_\sigma = 4\delta_{\sigma,\uparrow}$. This initial state is prepared by the appropriate single-particle basis rotations~\cite{wecker2015solving,kivlichan2018quantum,jiang2018quantum} on the state $\ket{1000}$ within each spin sector. Denote this unitary by $U(H_0)$. The system is then evolved by Trotterized dynamics according to $H$, with $R \in \{0, 1, \ldots, 5\}$ steps of size $\delta t = 0.2$. Let $J_{\mathrm{even}}$ (resp., $J_{\mathrm{odd}}$) be the terms in $J$ with $i$ even (resp., odd), and similarly for $V_{\mathrm{even}}, V_{\mathrm{odd}}$. One Trotter step is ordered as
\begin{equation}
    U_{\mathrm{Trot}} = e^{-\i J_{\mathrm{odd}} \delta t} e^{-\i V_{\mathrm{odd}} \delta t} e^{-\i V_{\mathrm{even}} \delta t} e^{-\i J_{\mathrm{even}} \delta t},
\end{equation}
which is then compiled into the native gate set. The full state-preparation circuit is then
\begin{equation}
    U_{\mathrm{prep}}(R) = U_{\mathrm{Trot}}^R U(H_0) X_{0,\downarrow} X_{0,\uparrow},
\end{equation}
where $X_{0,\sigma}$ places a spin-$\sigma$ electron on the first site from the vacuum (i.e., prepares $\ket{1000}$ in each spin sector). Note that $R = 0$ corresponds to only preparing the initial Slater determinant, which still has nontrivial circuit depth. Further details on the construction of these circuits can be found in Refs.~\cite{arute2020observation,recirq}.

\begin{figure*}
\includegraphics[scale=0.5]{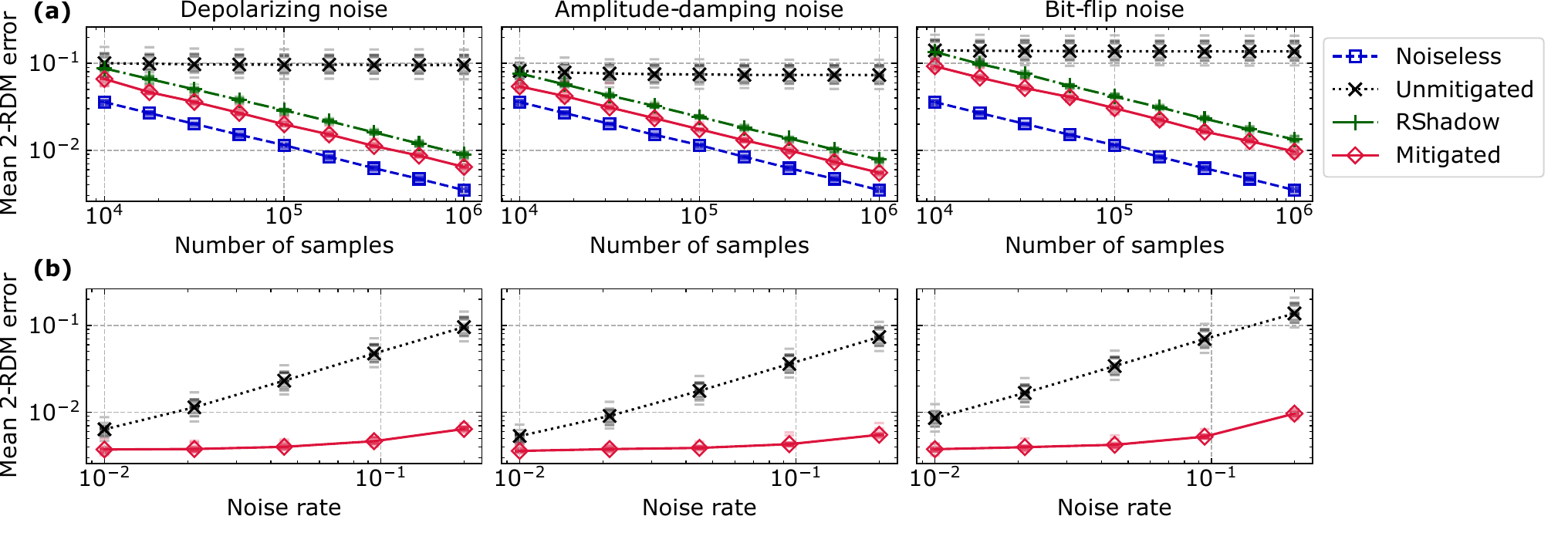}
\caption{\textbf{Average estimation error of the 2-RDM over all two-qubit subsystems, reconstructed from (subsystem-symmetrized) Pauli shadows.} Error is quantified by the spectral-norm difference between the estimated and true 2-RDM for each subsystem, taking the mean over all $\binom{n}{2}$ subsystems. We simulate the protocol on a collection of 20 random eight-qubit matrix product states, with maximum bond dimension $\chi \leq n$ and fixed $Z$ magnetization, $m = 0$. Because $m = 0$ implies $s_1 = 0$, we simulate a system of nine qubits with the ancilla prepared in $\ket{0}$. \textbf{(a)} Scaling of estimation error with the total number of samples $T$, fixing the noise rate to $p = 0.2$ for all noise models. \textbf{(b)} Scaling of the estimation error with the noise rate $p$, fixing the total number of samples to $T = 10^6$.}
\label{fig:random_mps}
\end{figure*}

One final detail of Ref.~\cite{arute2020observation} that we follow is their method of qubit assignment averaging (QAA). This technique is employed as a means of ameliorating inhomogeneities in error rates across the quantum device. QAA works by identifying a collection of different assignments for the physical qubit labels and uniformly averaging over them (keepng the Jordan--Wigner convention fixed). For example, one may vary qubit assignments by selecting a different portion of the chip, or rotating/flipping the layout. Here, we fix a $2 \times 4$ grid of qubits and perform QAA over four different orderings of those eight qubits;~see Supplementary Section~\ref{subsec:qaa} for the specific assignments chosen.

For each target state $U_{\mathrm{prep}}(R)\ket{0^n}$, we collect $T = 9.6 \times 10^5$ spin-adapted matchgate shadow samples. In Figure~\ref{fig:fh_qvm}, we plot the Trotterized time evolution of charge density throughout the chain, as well as the charge spread
\begin{equation}
    \kappa = \sum_{i \in [L]} \l| i - (L - 1)/2 \r| \varrho_i,
\end{equation}
which quantifies how the density spreads away from the center of the chain. These quantities are only one-body observables, so as an exemplary two-body observable we also estimate the energy per electron, $\langle H \rangle / \eta$.

Because Assumptions~\ref{assumption_1} no longer hold, we no longer have the guarantees of Theorem~\ref{thm:main_theorem} and we do not observe an arbitrary amount of error mitigation. We see that as the circuit size grows, so too do the prediction error and uncertainty. This behavior is a reflection of the noise assumptions being increasingly violated. Nonetheless, our results still show a substantial amount of noise reduction, and overall we maintain the qualitative features of the dynamics compared to the unmitigated protocol. In Supplementary Section~\ref{subsec:gate_dep}, we provide a quantitative estimate of how much the QVM noise model violates Assumptions~\ref{assumption_1}. There we observe a fundamental error floor which is roughly an order of magnitude below the mitigated errors actually achieved here, indicating the potential for further mitigation beyond what we have presently demonstrated.\\

\textbf{Qubit systems.} Next, we study the application of symmetry-adjusted classical shadows to subsystem-symmetrized Pauli shadows, to predict one- and two-body qubit observables in the presence of noise. We consider a fixed magnetization symmetry $\tr(M \rho) = m$, which, as we show in Section~\ref{sec:applications_qubits}, projects into the relevant irreps $R' = \{1, 2\}$ as
\begin{align}
    S_1 &= \Pi_1(M) = \sum_{i \in [n]} Z_i,\\
    S_2 &= \Pi_2(M^2) = 2 \sum_{i < j} Z_i Z_j.
\end{align}
The ideal symmetry values in this case are
\begin{align}
    s_1 &= m,\\
    s_2 &= m^2 - n.
\end{align}

\textbf{Readout noise model (qubits).} For our first demonstration, we simulate random matrix product states (MPS) with maximum bond dimension $\chi \leq n$, lying in the $m = 0$ symmetry sector of $M$. We use the definition of a random MPS from Refs.~\cite{garnerone2010typicality,garnerone2010statistical}. Numerically, we implement all MPS calculations using the open-source software ITensor~\cite{fishman2022itensor}, which can guarantee the correct symmetry sector using efficient tensor-network representations. Within such representations, it is straightforward to apply random local Clifford gates and $\SWAP$ gates, and to sample measurements in the computational basis.

\begin{figure*}
\centering
\includegraphics[scale=0.5]{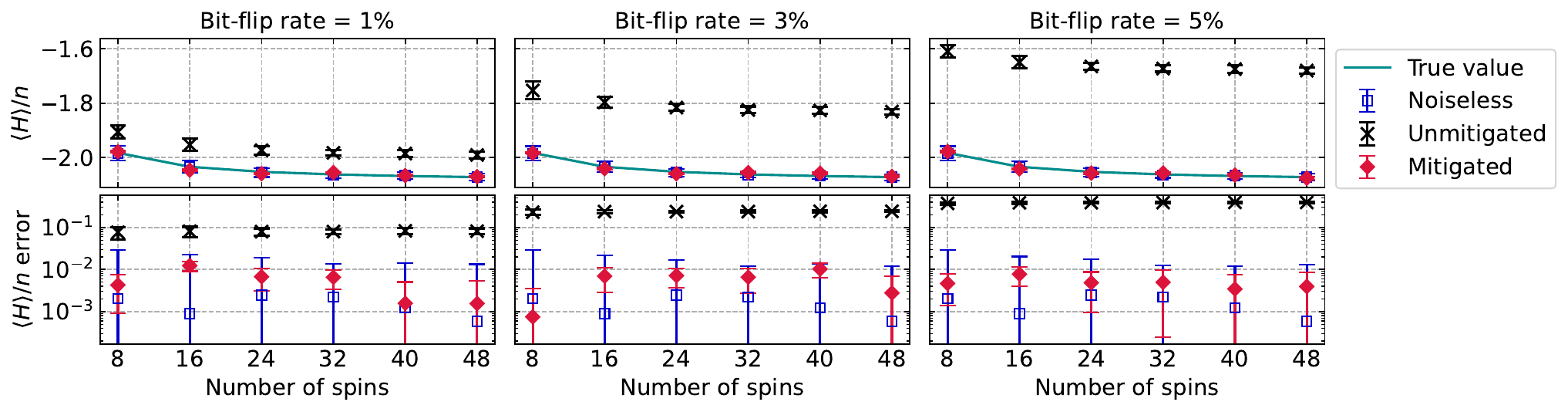}
\caption{\textbf{Estimation of energy per particle in the 1D XXZ ground state.} We measure $\langle H \rangle / n$ from $T = 10^6$ subsystem-symmetrized Pauli shadows. The anisotropy of the XXZ model is $\Delta = 1.5$. Note that, as an antiferromagnetic model, the ground state exhibits an $m = 0$ net magnetization symmetry. The noise model is single-qubit bit flip with probabilities $p = 1\%$, $3\%$, and $5\%$. Error bars denote one standard error of the mean.}
\label{fig:xxz_S2_and_energy}
\end{figure*}

Unlike fermions, qubits are not symmetrized, so their 2-RDMs
\begin{equation}
    {}^2 D_{ij} = \tr_{[n] \setminus \{ i, j \}} \rho
\end{equation}
are in general distinct between different two-qubit subsystems. Our accuracy metric here is therefore the mean 2-RDM error over all pairs of qubits:
\begin{equation}
    \epsilon = \frac{1}{\binom{n}{2}} \sum_{i<j} \| {{}^2 \hat{D}_{ij}(T)} - {{}^2 D_{ij}} \|_\infty.
\end{equation}
From subsystem-symmetrized Pauli shadows $\hat{\rho}(T)$ of size $T$, we reconstruct the qubit 2-RDMs by estimating all one- and two-local Pauli expectation values and forming the $4 \times 4$ matrices
\begin{equation}
    {}^2 \hat{D}_{ij}(T) = \frac{1}{4} \sum_{W, W' \in \Pauli(1)} \tr\l( W_i W'_j \hat{\rho}(T) \r) W \otimes W'.
\end{equation}

The results are shown in Figure~\ref{fig:random_mps}. Similar to the conclusions drawn from Figure~\ref{fig:random_slater} for the fermionic case, we observe that our theory is validated in two important parameters (number of samples and error rate). We note here that this simple demonstration also validates our subsystem-symmetrized Pauli shadows protocol and the use of ancillas in this scenario as well (recall that the random states we study here have vanishing symmetry value, $m = s_1 = 0$).

Our next set of numerical experiments are performed on the ground state of an antiferromagnetic XXZ Heisenberg chain with open boundary conditions:
\begin{equation}
    H = J \sum_{i \in [n-1]} \l( X_i X_{i+1} + Y_i Y_{i+1} + \Delta Z_i Z_{i+1} \r).
\end{equation}
Throughout, we set units such that $J = 1$ and consider an anisotropy of $\Delta = 1.5$. This Hamiltonian commutes with the symmetry operator $M$, and in particular the ground state obeys $m = 0$ (assuming the number of spins $n$ is even). We find the ground state via the density-matrix renormalization group (DMRG) algorithm~\cite{white1992density}, represented as an MPS;~therefore we can employ the same classical simulation algorithms as before. Although $m = 0$ implies a vanishing conserved quantity for the one-body subspace, $s_1 = 0$, we do not employ the ancilla technique for these simulations because we will only be interested in predicting strictly two-body observables (for which $s_2 = m^2 - n \neq 0$).

In Figure~\ref{fig:xxz_S2_and_energy} we show the mitigation of energy per spin $\langle H \rangle / n$ at different system sizes and bit-flip rates on each qubit. For these experiments, the number of samples taken is $T = 10^6$. Again, the results validate our theory for Pauli-shadow symmetry adjustment over a range of noise rates and system sizes. In particular, although we require estimating the symmetry operator $M$ which has variance $\O(n)$ (as opposed to $H/n$, which has constant variance), we see that in practice it suffices to take a number of samples constant in system size. This may indicate that our analysis of the worst-case sampling bounds for symmetry-adjusted classical shadows may be overly pessimistic in typical settings.\\

\begin{figure*}
\centering
\includegraphics[scale=0.5]{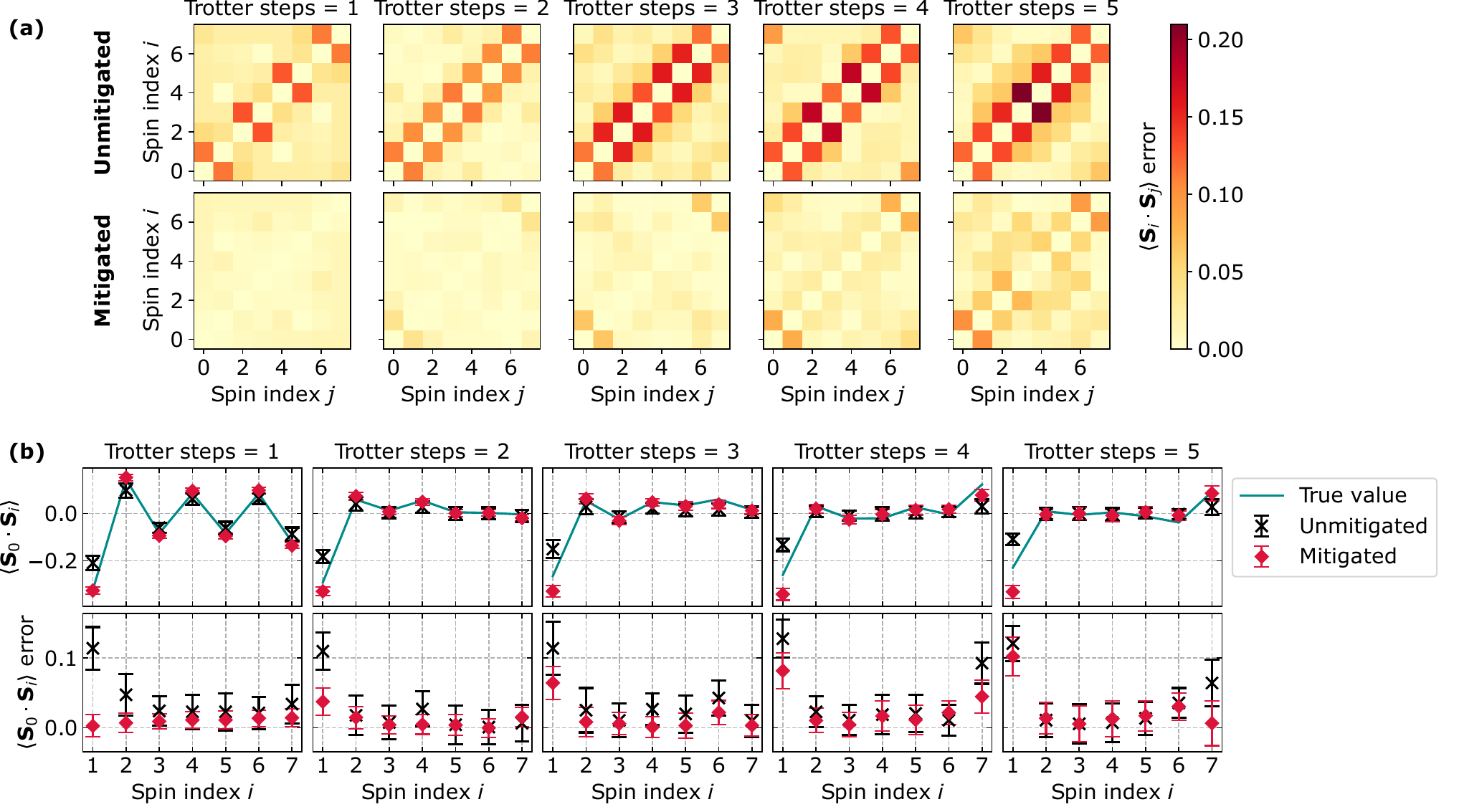}
\caption{\textbf{Prediction of spin--spin correlations between qubits in the XXZ chain undergoing Trotterized time evolution.} The anistropy in the model is $\Delta = 1.5$, and we take a Trotter step size of $\delta t = 0.2$. The noise model is the superconducting-hardware model implemented within the QVM. We take $T = 4.8 \times 10^5$ subsystem-symmetrized Pauli shadows to estimate the observables. \textbf{(a)} Heatmaps of estimation errors for $\langle \bm{S}_i \cdot \bm{S}_j \rangle$, under unmitigated versus our symmetry-enabled mitigated postprocessing. \textbf{(b)} Plots of the $\langle \bm{S}_0 \cdot \bm{S}_i \rangle$ correlation functions, to show further details of that particular row of the heatmaps. Error bars denote one standard error of the mean.}
\label{fig:xxz_qvm_corr}
\end{figure*}

\textbf{QVM noise model (qubits).} We now turn to simulations using the QVM noise model, taking the same XXZ Heisenberg spin chain ($\Delta = 1.5$ and $n = 8$) as our testbed system. Similar to our numerical experiments with the Fermi--Hubbard model, we simulate Trotter circuits of the XXZ model starting from a product state within the symmetry sector of $m = 0$. Again, we will only be interested in strictly two-local observables so we do not employ the ancilla trick here.

Our initial state is a N\'{e}el-ordered product state, $\ket{01010101} = \prod_{j \text{ odd}} X_j \ket{0^n}$. Defining $H_{\mathrm{even}}$ and $H_{\mathrm{odd}}$ as the terms in $H$ with $i$ even and odd, respectively, a single Trotter step is given by
\begin{equation}
    U_{\mathrm{Trot}} = e^{-\i H_{\mathrm{odd}} \delta t} e^{-\i H_{\mathrm{even}} \delta t},
\end{equation}
where we take the step size to be $\delta t = 0.2$. Hence, the full state-preparation circuit for $R$ steps is
\begin{equation}
    U_{\mathrm{prep}}(R) = U_{\mathrm{Trot}}^R \prod_{j \text{ odd}} X_j,
\end{equation}
which is then compiled into the native gate set. For each $R$, we collect $T = 4.8 \times 10^5$ samples using subsystem-symmetrized Pauli shadows. Because the initial state is a simple basis state, we only display results for $R \in \{1, \ldots, 5\}$ for these studies. In line with our Fermi--Hubbard simulations on the QVM, we perform QAA 
here as well, averaging over twelve different assignments of the same $2 \times 4$ qubits (see Supplementary Section~\ref{subsec:qaa} for details).

First, we consider the spin--spin correlations $\langle \bm{S}_i \cdot \bm{S}_j \rangle$, where
\begin{equation}
    \bm{S}_i = \frac{1}{2} \begin{pmatrix}
    X_i \\ Y_i \\ Z_i
    \end{pmatrix},
\end{equation}
for all qubit pairs $(i, j)$ throughout the chain. We plot the prediction errors of these correlation functions in Figure~\ref{fig:xxz_qvm_corr}, with the unmitigated data in the first row and mitigated data in the second row. We observe that, while the shallower Trotter circuits are well handled by symmetry-adjusted classical shadows, the mitigation power diminishes as the circuit grows deeper. To examine this effect closer, we plot in the bottom two rows of Figure~\ref{fig:xxz_qvm_corr} the correlation functions between the first spin and the rest of the chain. We see that the $\langle \bm{S}_0 \cdot \bm{S}_{1} \rangle$ errors are particularly dominant due to the magnitude of its true value. Although the absolute error is only marginally improved by symmetry adjustment for some of these pairs, the qualitative behavior is more faithfully recovered than in the unmitigated data (wherein the increasing circuit noise washes out the antiferromagnetic correlations).

\begin{figure*}
\centering
\includegraphics[scale=0.5]{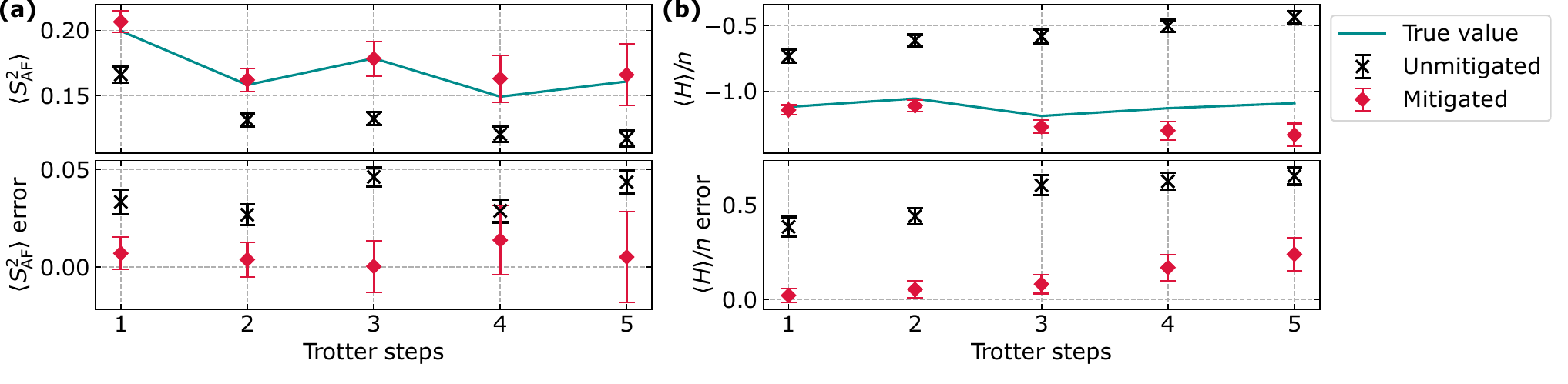}
\caption{\textbf{Prediction of macroscopic observables in the Trotterized XXZ model under the QVM noise model.} Observables are estimated with $T = 4.8 \times 10^5$ subsystem-symmetrized Pauli shadows. \textbf{(a)} N\'{e}el order parameter $\langle S_{\mathrm{AF}}^2 \rangle$. \textbf{(b)} Energy per spin $\langle H \rangle / n$. Error bars denote one standard error of the mean.}
\label{fig:xxz_qvm_props}
\end{figure*}

Next, we consider macroscopic observables in Figure~\ref{fig:xxz_qvm_props}:~the N\'{e}el order parameter
\begin{equation}
    \langle S_{\mathrm{AF}}^2 \rangle = \frac{1}{n^2} \sum_{i,j \in [n]} (-1)^{i+j} \langle \bm{S}_i \cdot \bm{S}_j \rangle,
\end{equation}
and the energy per spin $\langle H \rangle / n$. Again we see general trends similar to the other QVM simulations:~the mitigated results are in closer qualitative agreement with the true values than the unmitigated data, at the cost of larger uncertainty bars, and without arbitrary amounts of error suppression. Symmetry adjustment consistently reduces the absolute error compared to the unmitigated data, although we note that some of the energy estimates are still a few standard deviations away from the true value. This is attributed to the violation of Assumptions~\ref{assumption_1}, and we leave it an open problem of how to further ameliorate this property.

\section{Discussion}\label{sec:discussion}

In this paper, we have introduced symmetry-adjusted classical shadows, a QEM protocol applicable to quantum systems with known symmetries. Our approach builds on the highly successful classical-shadow tomography~\cite{huang2020predicting,paini2021estimating}, modifying the classically computed linear-inversion step according to symmetry information in the presence of noise. Because our strategy is performed in postprocessing on the noisy measurement data, it allows for straightforward combinations with other QEM strategies. As opposed to prior related works~\cite{karalekas2020quantum,chen2021robust,koh2022classical,van2022model,arrasmith2023development}, the main advantage of our approach is the use of the entire noisy circuit, thereby bypassing the need for calibration experiments and accounting for errors in state preparation. Meanwhile, in contrast with other symmetry-based strategies~\cite{bonet2018low,mcardle2019error,cai2021quantum,jnane2023quantum}, we require no additional quantum resources, utilize finer-grained symmetry information, and can easily take advantage of a wider range of symmetries (e.g., particle number as opposed to only parity conservation). We note that, while this work has focused on local observables (linear functions of $\rho$), classical shadows can seamlessly be used for nonlinear observable estimation as well~\cite{huang2020predicting}. Because symmetry adjustment works at the level of the the shadow channel inversion, our QEM strategy applies within that context just as well.

Overall, our findings reveal that as a low-cost scheme, symmetry-adjusted classical shadows by itself is already potent for practical error mitigation. Our analytical results guarantee the accuracy of prediction under readout noise assumptions. Even when these assumptions are violated in practice, we expect these results to still provide intuition regarding the mitigation behavior. Indeed, this expectation is validated by our numerical experiments with superconducting-qubit noise models on the Cirq QVM~\cite{cirq,isakov2021simulations}. From these simulations, we have observed substantial quantitative improvement when the cumulative circuit noise is sufficiently weak, and qualitative improvements across all experiments performed.

Along the way, we have developed a number of ancillary results that may also be of independent interest. Of note are (1)~the subsystem-symmetrized Pauli shadows, which uniformly symmetrizes the irreps of the local Clifford group among subsystems;~(2)~an improved circuit compilation scheme for fermionic Gaussian unitaries, which treats Majorana modes on a more natural footing to improve two-qubit gate parallelization;~and (3)~symmetry-adapted matchgate shadows, which uses block-diagonal transformations within spin sectors to reduce the size of the random matchgate circuits. We expect that these techniques will find broader applicability in quantum simulation beyond the scope of this paper.

A number of pertinent open questions and future directions remain. For simplicity of the protocol, and because of the examples that we focused on, we restricted attention to multiplicity-free groups. However, tools to generalize to non-multiplicity-free groups already exist, and in the context of character randomized benchmarking~\cite{helsen2019new} such an extension has been developed successfully~\cite{claes2021character}. It would therefore be useful to extend our ideas similarly, and investigate what effect (if any) multiplicities have on symmetry-adjusted classical shadows.

Regarding the protocols considered, we have focused on local observable estimation in systems with global $\U(1)$ symmetry. However, it is worth noting that the $n$-qubit Clifford group possesses only one nontrivial irrep, making it essentially compatible with any symmetry. Because its shadow norm is exponentially large for local observables, it is an unfavorable choice for typical quantum-simulation applications. One may wonder whether the desirable universality of this irrep can nonetheless be harnessed, analogous to how we constructed the subsystem-symmetrized Pauli shadows. A particularly interesting candidate for future studies would be the global $\SU(2)$/$\Cl(1)$ control, introduced in Ref.~\cite{van2022hardware}, which features both a group with high amounts of symmetry and low sample complexity for local observable estimation. Similarly, the single-fermion $\U(n)$ basis rotations used in Ref.~\cite{low2022classical} would also be a promising option to study. Alternatively, one may consider different classes of symmetries, such as local (rather than global) symmetries.

One key advantage of symmetry adjustment is its flexibility, allowing for easy integration with other error-mitigation strategies. Investigating this interplay is a clear target for future work. Particularly valuable would be other techniques to massage the circuit noise into approximately satisfying Assumptions~\ref{assumption_1}, for instance by randomized compiling~\cite{wallman2016noise}. From our usage of QAA~\cite{arute2020observation} in the numerical experiments, we have already shown heuristically that the mere choice of qubit assignments appears to have such an effect.

Indeed, the reliance on such assumptions for rigorous guarantees may be viewed as a limitation of this work. While our numerical results are encouraging, it behooves one to seek a more comprehensive error analysis applicable to a wider range of noise models. For example, while gate-dependent errors are particularly detrimental to our method, they have been closely studied in the context of randomized benchmarking~\cite{proctor2017randomized,wallman2018randomized,carignan2018randomized,merkel2021randomized}. The tools developed therein may be valuable to this setting as well. Establishing a better understanding here may also inspire extensions to surpass the limitations of the current theory. We leave such goals to future work.

\emph{Note added.}---Shortly after our manuscript appeared on the arXiv preprint server, two related works~\cite{wu2023error,brieger2023stability} subsequently appeared. The former develops a calibration estimator equivalent to our Eq.~\eqref{eq:fermion_RSE_estimator}, while the latter analytically studies the effects of gate-dependent noise on Clifford shadow protocols. The formulation and analyses of symmetry-adjusted classical shadows remain original to our manuscript.

\section{Methods}\label{sec:methods}

\subsection{Theory of symmetry-adjusted classical shadows}\label{sec:main_theory}

Here we describe the theory behind the symmetry-adjusted classical shadows estimator. This approach uses known symmetry information about the ideal, noiseless state $\rho$ that we wish to prepare (but are only able to produce a noisy version of). In this section we describe the idea for an arbitrary multiplicity-free group $G$;~in the subsequent subsections, we will provide concrete applications to the efficient estimation of local fermionic and qubit observables, respectively.

Suppose $ \rho $ is a quantum state obeying a known symmetry, corresponding to a collection of operators $S_\lambda \in V_\lambda$ for which the values $ s_\lambda \coloneqq \vip{S_\lambda}{\rho} $ are known \emph{a priori}. For example, if the system has a symmetry operator $S$ which spans multiple irreps, then we can construct $S_\lambda$ using the projectors $\Pi_\lambda$:
\begin{equation}
    \kket{S_\lambda} = \Pi_\lambda \kket{S}.
\end{equation}
By construction, $S_\lambda$ is an eigenoperator of both $\mathcal{M}$ and $\widetilde{\mathcal{M}}$:
\begin{equation}\label{eq:symmetry_eigenop}
\begin{split}
    \mathcal{M}\kket{S_\lambda} &= f_\lambda \kket{S_\lambda},\\
    \widetilde{\mathcal{M}}\kket{S_\lambda} &= \widetilde{f}_\lambda \kket{S_\lambda}.
\end{split}
\end{equation}
If one is interested in only a subset $R'
\subseteq R_G$ of the irreps, then it suffices to only know those symmetries $S_\lambda$ for which $\lambda \in R'$.

Because the ideal values of $ s_\lambda $ and $ f_\lambda $ are already known, we can use the estimated noisy expectation value of $ S_\lambda $ to build an estimate for $\widetilde{f}_\lambda$. We start with the standard postprocessing of classical shadows:~applying $\mathcal{M}^{-1}$ to the measurement outcomes of the noisy quantum experiments produces, in expectation, the effective state
\begin{equation}\label{eq:noisy_rho}
    \kket{\widetilde{\rho}} \coloneqq \mathcal{M}^{-1}\widetilde{\mathcal{M}} \kket{\rho} = \E_{g \sim G, b \sim \widetilde{\mathcal{U}}_g \kket{\rho}} \mathcal{M}^{-1} \mathcal{U}_{g}^\dagger \kket{b},
\end{equation}
which clearly differs from $\kket{\rho}$ when $\widetilde{\mathcal{M}} \neq \mathcal{M}$. Nonetheless, we can use this noisy data to estimate the value of $\vip{S_\lambda}{\widetilde{\rho}}$, which is equal to
\begin{equation}\label{eq:noisy_symmetry}
    \vip{S_\lambda}{\widetilde{\rho}} = \frac{\widetilde{f}_\lambda}{f_\lambda} s_\lambda
\end{equation}
by Eq.~\eqref{eq:symmetry_eigenop}. In fact, this relation applies to any $O \in V_\lambda$:
\begin{equation}\label{eq:EM_estimator_recovery}
    \vip{O}{\widetilde{\rho}} = \frac{\widetilde{f}_\lambda}{f_\lambda} \vip{O}{\rho}.
\end{equation}
Hence while we use Eq.~\eqref{eq:noisy_symmetry} to learn $\widetilde{f}_\lambda$ from the symmetry $S_\lambda$, this is in turn applicable to all other operators within the same irrep. This leads to the recovery of the ideal expectation values as
\begin{equation}\label{eq:corrected_observable}
    \vip{O}{\rho} = \frac{\vip{O}{\widetilde{\rho}}}{\vip{S_\lambda}{\widetilde{\rho}} / s_\lambda}.
\end{equation}

Having established the theory in expectation, we now analyze the implementation in practice. Let $T$ be the number of classical-shadow snapshots, $\kket{\hat{\rho}_\ell} = \mathcal{M}^{-1} \mathcal{U}_{g_\ell}^\dagger \kket{b_\ell}$ for $\ell = 1, \ldots, T$, obtained by sampling the noisy quantum computer. Recall that these snapshots converge to $\widetilde{\rho}$ rather than $\rho$. From their empirical average, $\hat{\rho}(T) = (1/T) \sum_{\ell=1}^T \hat{\rho}_{\ell}$, we can estimate the lefthand side of Eq.~\eqref{eq:noisy_symmetry} as
\begin{equation}
    \hat{s}_\lambda(T) \coloneqq \vip{S_\lambda}{\hat{\rho}(T)}.
\end{equation}
This in turn provides an estimate for $\widetilde{f}_\lambda$,
\begin{equation}\label{eq:f_hat_est}
    \hat{f}_\lambda(T) \coloneqq f_\lambda \frac{\hat{s}_\lambda(T)}{s_\lambda}.
\end{equation}
This can be understood as a generalization of $\mathrm{NoiseEst}_G(\lambda, g, b)$ from Eq.~\eqref{eq:noiseest}, making the replacements $D_\lambda \to S_\lambda$ and $\op{0^n}{0^n} \to \rho$. Indeed, one can view the calibration state $\ket{0^n}$ as obeying the symmetries given by its stabilizer group.

Consider the estimation of observables $O_1, \ldots, O_L$ with symmetry-adjusted classcial shadows. If any observable is supported over multiple irreps, then we can always decompose it as a linear combination of basis elements across those irreps. Thus without loss of generality we suppose that each $O_j \in V_\lambda$ for some $\lambda \in R'$. From the same noisy classical shadow $\hat{\rho}(T)$, we also have estimates for their noisy expectation values:~$\E\vip{O_j}{\hat{\rho}(T)} = \vip{O_j}{\widetilde{\rho}_j}$. Then, following Eq.~\eqref{eq:corrected_observable} we can directly construct error-mitigated estimators as
\begin{equation}\label{eq:EM_estimator}
    \hat{o}_j^{\mathrm{EM}}(T) \coloneqq \frac{\hat{o}_j(T)}{\hat{s}_\lambda(T) / s_\lambda},
\end{equation}
which converges to $\vip{O_j}{\rho}$ in the $T \to \infty$ limit (if Assumptions~\ref{assumption_1} hold). Because $\E[X/Y] \neq \E[X]/\E[Y]$ (for nontrivial random variables $X$ and $Y$), Eq.~\eqref{eq:EM_estimator} describes a biased estimator. In the following theorem, we quantify this bias by bounding the total prediction error of $\hat{o}_j^{\mathrm{EM}}(T)$. This in turn bounds the number of symmetry-adjusted classical-shadow samples $T$ required.

\begin{theorem}\label{thm:main_theorem}
Fix accuracy and confidence parameters $\epsilon, \delta \in (0, 1)$. Let $O_1, \ldots, O_L$ be a collection of observables, each supported on an irrep of $\mathcal{U} : G \to \U(\mathcal{L}(\mathcal{H}))$ as $O_j \in V_\lambda$ for $\lambda \in R' \subseteq R_G$. Let $S_\lambda \in V_\lambda$ be a symmetry operator for each $\lambda \in R'$, for which the ideal values $s_\lambda = \tr(S_\lambda \rho)$ of the target state $\rho$ are known \emph{a priori}. Suppose that each noisy unitary satisfies Assumptions~\ref{assumption_1}, $\widetilde{\mathcal{U}}_g = \mathcal{E} \mathcal{U}_g$, and define the quantities
\begin{align}
    F_{Z,R'}(\mathcal{E}) &\coloneqq \min_{\lambda \in R'} \frac{\tr(\mathcal{E} \mathcal{M}_Z \Pi_\lambda)}{\tr(\mathcal{M}_Z \Pi_\lambda)},\\
    \sigma^2 &\coloneqq \max_{1 \leq j \leq L, \lambda \in R'}\l\{ \V[\hat{o}_j], \V\l[\frac{\hat{s}_\lambda}{s_\lambda} \r] \r\}.
\end{align}
Then, a (noisy) classical shadow $\hat{\rho}(T)$ of size
\begin{equation}
    T = \O\l( \frac{\log((L + |R'|)/\delta)}{F_{Z,R'}(\mathcal{E})^2 \epsilon^2} \sigma^2 \r)
\end{equation}
can be used to construct error-mitigated estimates
\begin{equation}
    \hat{o}_j^{\mathrm{EM}}(T) \coloneqq \frac{\tr(O_j \hat{\rho}(T))}{\tr(S_\lambda \hat{\rho}(T)) / s_\lambda}
\end{equation}
which obey
\begin{equation}
    |\hat{o}_j^{\mathrm{EM}}(T) - \tr(O_j\rho)| \leq (\|O_j\|_\infty + 1) \epsilon + \O(\|O_j\|_\infty \epsilon^2)
\end{equation}
for all $1 \leq j \leq L$, with success probability at least $1 - \delta$.
\end{theorem}

The proof of this statement is provided in Supplementary Section~\ref{sec:error_analysis}. Note that $\| \cdot \|_\infty$ denotes the spectral (operator) norm. We phrase this result in terms of variances, rather than the state-independent shadow norm, because knowledge about $\rho$ (namely, its symmetries) can potentially provide tighter bounds. Note that the variance is with respect to the effective noisy state $\widetilde{\rho}$, which was defined in Eq.~\eqref{eq:noisy_rho}.

Let us make a few remarks on this result. First, although the symmetry operators appear in the denominator of Eq.~\eqref{eq:EM_estimator}, they affect the sample complexity as usual for classical shadows, albeit normalized by the value $s_\lambda$ of the symmetry sector. Thus the division by $\hat{s}_\lambda(T)$ does not hinder our control over the variance, except when $s_\lambda = 0$ (we furthermore show how to handle such pathological cases in the application examples below). For typical applications, the variance $\V[\hat{s}_\lambda / s_\lambda] \leq \sns{S_\lambda / s_\lambda}{\mathrm{shadow}}^2$ will be comparable to the baseline variance of estimation, $\V[\hat{o}_j] \leq \max_{j'} \sns{O_{j'}}{\mathrm{shadow}}^2$. Additionally, the number of irreps considered is typically $|R'| \ll L$ (for instance, in the concrete examples considered in this work, $|R'|$ is a constant). Thus, we expect that the inclusion of symmetry operators incurs negligible overheads for most applications.

Instead, the primary overhead arises from the fact that error-mitigated estimation necessarily comes at the cost of larger overall variances~\cite{takagi2022fundamental,takagi2022universal,tsubouchi2022universal,quek2022exponentially}. The quantity
\begin{equation}
    F_{Z,R'}(\mathcal{E}) = \min_{\lambda \in R'} \frac{\tr(\mathcal{E} \mathcal{M}_Z \Pi_\lambda)}{\tr(\mathcal{M}_Z \Pi_\lambda)}
\end{equation}
characterizes an effective noise strength, and it can be seen as a generalization of the average $Z$-basis fidelity of $\mathcal{E}$,
\begin{equation}
    F_Z(\mathcal{E}) = \frac{\tr(\mathcal{E} \mathcal{M}_Z)}{\tr(\mathcal{M}_Z)} = \frac{1}{2^n} \sum_{b \in \{0, 1\}^n} \vev{b}{\mathcal{E}}{b},
\end{equation}
which appears in prior works on noise-robust classical shadows~\cite{chen2021robust,koh2022classical}. In contrast to $F_Z(\mathcal{E})$, the quantity $F_{Z,R'}(\mathcal{E})$ is a more fine-grained characterization of the noise channel, averaged within the relevant subspaces $V_\lambda$. Similar to prior results~\cite{karalekas2020quantum,chen2021robust,koh2022classical,van2022model,arrasmith2023development}, the sampling overhead of our error-mitigated estimates also depends inverse quadratically on this noise fidelity.

Finally, the error bound we obtain is $\O(\|O_j\|_\infty \epsilon)$ when $\epsilon < 1$. Note that $\|O_j\|_\infty = 1$ for Pauli and Majorana operators. Our result also features error terms of order $\O(\|O_j\|_\infty \epsilon^2)$, which reflect the biased nature of $\hat{o}_j^{\mathrm{EM}}(T)$ as a ratio of two random variables. Nonetheless, our theorem establishes that this bias vanishes as $\epsilon^2 \sim 1/T$, so that for sufficiently large $T$ the prediction error is dominated by the standard shot-noise scaling of $\epsilon \sim 1/\sqrt{T}$.

\subsection{Application to fermionic (matchgate) shadows}\label{sec:applications_fermions}

The first application of symmetry-adjusted classical shadows that we consider is the estimation of local fermionic observables. This is achieved efficiently by fermionic classical shadows~\cite{zhao2021fermionic}, wherein the group $G$ corresponds fermionic Gaussian unitaries (also referred to as matchgate shadows~\cite{wan2023matchgate}). We will consider a commonly encountered symmetry in fermionic systems:~fixed particle number. However, it will be clear how the general idea can apply to other symmetries, such as spin.\\

\textbf{Background on matchgate shadows.} We begin with a review of matchgate shadows. Let $a_p^\dagger, a_p$ be creation and annihilation operators for a system of $n$ fermionic modes, $p \in [n]$. The associated Majorana operators are
\begin{equation}
    \gamma_{2p} = a_p + a_p^\dagger, \quad \gamma_{2p + 1} = -\i(a_p - a_p^\dagger).
\end{equation}
Under the Jordan--Wigner transformation~\cite{jordanwigner}, these are mapped to Pauli operators as
\begin{equation}\label{eq:jw_majorana}
    \gamma_{2p} = \l(\prod_{q<p} Z_q\r) X_p, \quad \gamma_{2p + 1} = \l(\prod_{q<p} Z_q\r) Y_p.
\end{equation}
Recall from Eq.~\eqref{eq:majorana_def} that all $d^2$ basis operators are generated by taking arbitrary products:
\begin{equation}
    \Gamma_{\bm{\mu}} = (-\i)^{\binom{m}{2}} \gamma_{\mu_1} \cdots \gamma_{\mu_m},
\end{equation}
where $\bm{\mu} = (\mu_1, \ldots, \mu_m) \subseteq [2n]$. By convention, we order $\mu_1 < \cdots < \mu_m$. We can group all the $m$-degree Majorana indices by defining the set
\begin{equation}
    \comb{2n}{m} \coloneqq \{ \bm{\mu} \subseteq [2n] : |\bm{\mu}| = m \}.
\end{equation}
Physical fermionic observables have even degree $m = 2k$. An important subset of such operators comprises those which are diagonal in the standard basis, corresponding to the index set
\begin{equation}
    \diags{2n}{2k} \coloneqq \{ (2p_1, 2p_1 + 1, \ldots, 2p_k, 2p_k + 1) : \bm{p} \in \comb{n}{k} \}.
\end{equation}
Using Eq.~\eqref{eq:jw_majorana}, each $\bm{\tau} \in \diags{2n}{2k}$ corresponds to the Pauli-$Z$ operator $\Gamma_{\bm{\tau}} = Z_{p_1} \cdots Z_{p_k}$ under the Jordan--Wigner mapping.

The group of fermionic Gaussian unitaries is the image of the homomorphism $U : \Orth(2n) \to \U(d)$ whose adjoint action obeys
\begin{equation}\label{eq:matchgate_def}
    U_Q \gamma_\mu U_Q^\dagger = \sum_{\nu \in [2n]} Q_{\nu\mu} \gamma_\nu, \quad Q \in \Orth(2n).
\end{equation}
These unitaries are equivalent to (generalized) matchgate circuits~\cite{helsen2022matchgate} and constitute a class of classically simulatable circuits~\cite{valiant2001quantum,knill2001fermionic,terhal2002classical,bravyi2004lagrangian,divincenzo2005fermionic,jozsa2008matchgates}. Fermionic (matchgate) shadows then randomize over certain subgroups 
$G \subseteq \Orth(2n)$ of these Gaussian unitaries. The measurement channel takes the form
\begin{equation}\label{eq:matchgate_channel}
    \mathcal{M} = \sum_{k=0}^n f_{2k} \Pi_{2k},
\end{equation}
where the eigenvalues are
\begin{equation}
    f_{2k} = \l. \binom{n}{k} \middle/ \binom{2n}{2k} \r.
\end{equation}
and each irrep is the image of
\begin{equation}\label{eq:matchgate_projectors}
    \Pi_{2k} = \frac{1}{d} \sum_{\bm{\mu} \in \comb{2n}{2k}} \vop{\Gamma_{\bm{\mu}}}{\Gamma_{\bm{\mu}}}.
\end{equation}
While $\mathcal{U}$ carries $2n + 1$ unique irreps (each labeled by a Majorana degree $m$)~\cite{claes2021character,helsen2022matchgate}, only the $n + 1$ irreps $\lambda = 2k$ have nonvanishing $f_\lambda$~\cite{zhao2021fermionic,wan2023matchgate}. Therefore $\mathcal{M}^{-1}$ is formally the pseudoinverse restricted to those subspaces. Finally, the shadow norm of $k$-body Majorana operators is~\cite{zhao2021fermionic}
\begin{equation}\label{eq:FGU_shadow_norm}
    \sn{\Gamma_{\bm{\mu}}}^2 = f_{2k}^{-1} = \O(n^k).
\end{equation}
Variance expressions for arbitrary observables can be found in Refs.~\cite{wan2023matchgate,ogorman2022fermionic}. For the postprocessing of $T$ shadows into estimates of all $k$-body Majorana observables, we describe an algorithm in Supplementary Section~\ref{subsec:matchgate_computation} which runs in time $\O(n^k T)$.

We now comment on the choice of $G \subseteq \Orth(2n)$. Fermionic classical shadows were introduced in Ref.~\cite{zhao2021fermionic}, which initially considered the intersection of proper matchgate circuits [the special orthogonal group $\SO(2n)$] with $n$-qubit Clifford unitaries $\Cl(n)$. The result is the group of all $2n \times 2n$ signed permutation matrices with determinant $1$, denoted by $\B^{+}(2n) \subset \SO(2n)$. They also showed that its unsigned subgroup, $\Alt(2n) \subset \B^{+}(2n)$, possesses the same irrep structure~\cite[Supplemental Material, Theorem~11]{zhao2021fermionic}. While the full, continuous group $\SO(2n)$ has not yet been analyzed for classical shadows, it was studied for character randomized benchmarking~\cite{helsen2019new} in Ref.~\cite[Sec.~VI]{claes2021character}, wherein they demonstrated the presence of multiplicities. These multiplicities can be avoided by enlarging to the generalized matchgate group, i.e., all of $\Orth(2n)$~\cite[Lemma~3]{helsen2022matchgate}. Ref.~\cite{wan2023matchgate} applied these generalized matchgates to fermionic classical shadows, and in particular they prove that the Clifford intersection in this setting (now yielding the subgroup $\B(2n) \subset \Orth(2n)$ of signed permutation matrices with either determinant $\pm 1$) is a $3$-design for $\Orth(2n)$. This implies that $\B(2n)$ is also multiplicity-free.

Due to the variety of options, for the rest of this paper we assume matchgate shadows under any $G$ with the desired irreps. We note that Ref.~\cite{ogorman2022fermionic} introduced a smaller subset of $\B(2n)$ based on perfect matchings, which has the same channel $\mathcal{M}$ and variances;~however its connection to representation theory was not explored.\\

\textbf{Utilizing particle-number symmetry.} Suppose the ideal state we wish to prepare lies in the $\eta$-particle sector of $\mathcal{H}$. This is a $\U(1)$ symmetry generated by the fermion-number operator, $N = \sum_{p\in[n]} a_p^\dagger a_p$. In particular, powers of $N$ obey
\begin{equation}\label{eq:fermion_number}
    \tr(N^k \rho) = \eta^k,
\end{equation}
which provides us a collection of conserved quantities with which to perform symmetry adjustment. Recall from Eq.~\eqref{eq:matchgate_projectors} that $\Pi_{m}$ projects onto the irrep
\begin{equation}
    V_m = \spn\{ \Gamma_{\bm{\mu}} : \bm{\mu} \in \comb{2n}{m} \}.
\end{equation}
Then, projecting $N^k$ onto $V_{2k}$ yields the symmetry operators $S_{2k}$, and solving the resulting linear system of equations recovers the ideal values for $s_{2k} = \tr(S_{2k} \rho)$. For ease of exposition we will consider only $k = 1, 2$, but one may generalize to higher $k$ using these ideas.

Concretely, we start with the fact that $a_p^\dagger a_p = (\I - \Gamma_{(2p, 2p+1)})/2$, and $\Gamma_{(2p, 2p+1)} \Gamma_{(2q, 2q+1)} = \Gamma_{(2p, 2p+1, 2q, 2q + 1)}$ for $p < q$. Then, expanding $N$ and $N^2$ into a linear combination of Majorana operators, one finds
\begin{align}
    S_2 &= \Pi_2(N) = -\frac{1}{2} \sum_{\bm{\mu} \in \diags{2n}{2}} \Gamma_{\bm{\mu}},\\
    S_4 &= \Pi_4(N^2) = \frac{1}{2} \sum_{\bm{\mu} \in \diags{2n}{4}} \Gamma_{\bm{\mu}}.
\end{align}
Using Eq.~\eqref{eq:fermion_number} and the relations between $S_2$ and $S_4$ to $N$ and $N^2$ (for example, $N = n\I/2 + S_2$), we arrive at:
\begin{align}
    s_2 &= \tr\l( S_2 \rho \r) = \eta - \frac{n}{2}, \label{eq:s2_val}\\
    s_4 &= \tr\l( S_4 \rho \r) = \frac{1}{2} \binom{n}{2} - \eta(n - \eta). \label{eq:s4_val}
\end{align}

For the sampling cost incurred by these symmetry operators, we argue that the typical shadow norms of these symmetries are $\sn{S_{2k}/s_{2k}}^2 = \O(n^k)$, which is the same as the base estimation. To see this, consider a triangle inequality on the shadow norm:
\begin{equation}
\begin{split}
    \sn{S_{2k}} &\leq \frac{1}{2} \sum_{\bm{\mu} \in \diags{2n}{2k}} \sn{\Gamma_{\bm{\mu}}}\\
    &= \frac{1}{2} \binom{n}{k} \sqrt{\l. \binom{2n}{2k} \middle/ \binom{n}{k} \r.}\\
    &= \O( n^{3k/2} ).
\end{split}
\end{equation}
Thus $\sn{S_{2}}^2 = \O(n^{3})$ and $\sn{S_{4}}^2 = \O(n^{6})$. Next, we need to examine how $s_{2k}^2$ scales with system size. Assuming that $s_2, s_4 \neq 0$ and that the number of electrons is $\eta = \O(n)$, then from Eqs.~\eqref{eq:s2_val} and \eqref{eq:s4_val} we see that $s_2^2 = \Theta(n^2)$ and $s_4^2 = \Theta(n^4)$. Thus
\begin{align}
    \sn{S_2/s_2}^2 &= \O(n),\\
    \sn{S_4/s_4}^2 &= \O(n^2).
\end{align}

\textbf{Avoiding division by zero.} One potential obstruction to symmetry adjustment is when some $s_{2k} = 0$. This can occur whenever the particle number takes a specific value:
\begin{align}
     s_2 = 0 &\text{ if } \eta = \frac{n}{2}, \label{eq:s2_zero}\\
     s_4 = 0 &\text{ if } \eta = \frac{n \pm \sqrt{n}}{2}. \label{eq:s4_zero}
\end{align}
Equation~\eqref{eq:s2_zero} occurs at half filling, which is fairly common. On the other hand, Eq.~\eqref{eq:s4_zero} occurs only when the number of modes $n$ is a perfect square and the number of particles $\eta$ is one of two specific values, so it is less likely to occur. Nonetheless, there is a straightforward way to circumvent both possibilities by introducing a single ancilla qubit.

To do so, append an additional fermion mode initialized in the unoccupied state $\ket{0}$, so that the ideal state is now the $(n + 1)$-mode state $\rho' = \rho \otimes \op{0}{0}$. Given that $\rho$ has $\eta$ particles on $n$ modes, $\rho'$ is an $\eta$-particle state on $n + 1$ modes. The new symmetry operators on the $(n + 1)$-mode Hilbert space are
\begin{align}
    S'_2 &= -\frac{1}{2} \sum_{\bm{\mu} \in \diags{2(n + 1)}{2}} \Gamma_{\bm{\mu}},\\
    S'_4 &= \frac{1}{2} \sum_{\bm{\mu} \in \diags{2(n + 1)}{4}} \Gamma_{\bm{\mu}},
\end{align}
which have ideal values
\begin{align}
    s'_2 &= \tr\l( S_2 \rho' \r) = \eta - \frac{n + 1}{2},\\
    s'_4 &= \tr\l( S_4 \rho' \r) = \frac{1}{2} \binom{n + 1}{2} - \eta(n + 1 - \eta).
\end{align}
It is straightforward to check that, if either condition Eq.~\eqref{eq:s2_zero} or Eq.~\eqref{eq:s4_zero} holds, then $s'_2$ and $s'_4$ are always nonzero for $n > 1$.

Under the Jordan--Wigner mapping, this modification is easily achieved by initializing a single ancilla qubit in $\ket{0}$. Recall that the terms in the symmetries $S_2, S_4$ are the diagonal operators $\Gamma_{(2p,2p+1)} = Z_p$ and $\Gamma_{(2p,2p+1,2q,2q+1)} = Z_p Z_q$. Note also that the ancilla qubit is acted on only during the random unitary $U_Q$ (where now $Q$ has dimension $2n + 2$) and otherwise does not interact with the $n$ system qubits.

\subsection{Application to qubit (Pauli) shadows}\label{sec:applications_qubits}

Now we turn to the application for local observable estimation in systems of spin-$1/2$ particles (qubits). Random Pauli measurements are efficient for this task;~however, for compatibility with the global $\U(1)$ symmetry considered in this work, we must slightly modify the protocol to accommodate its irreps. We begin with a review of the standard Pauli shadows protocol, followed by our modification.\\

\textbf{Background on standard Pauli shadows.} The local Clifford group $\Cl(1)^{\otimes n}$ is implemented by uniformly drawing a single-qubit Clifford gate for each qubit independently. It has $2^n$ irreducible representations, corresponding to all $k$-qubit subsystems $I \subseteq [n]$, where $|I| = k \in \{0, 1, \ldots, n\}$~\cite{gambetta2012characterization}. Twirling $\mathcal{M}_Z$ by this group yields
\begin{equation}
    \mathcal{M} = \sum_{I \subseteq [n]} f_I \Pi_I,
\end{equation}
where $f_I = 3^{-|I|}$ and $\Pi_I$ projects onto the subspace of operators which act nontrivially on precisely the subsystem $I$. The squared shadow norm for $k$-local Pauli operators $P$ is~\cite{huang2020predicting}
\begin{equation}\label{eq:pauli_shadow_norm}
    \sn{P}^2 = 3^k.
\end{equation}
A more general variance bound was derived in Ref.~\cite{paini2021estimating}:~a simple loose bound of their result can be stated as $\V[\hat{o}] \leq 3^k R \| O \|_{\infty}^2$, where $O$ is an arbitrary $k$-local traceless observable and $R$ is the number of terms in its Pauli decomposition. However, they argue that a tighter expression, essentially $3^k \| O \|_{\infty}^2$, is typically a good approximation to the variance.\\

\textbf{Subsystem symmetrization of Pauli shadows.} The irreps of $\Cl(1)^{\otimes n}$ are difficult to reconcile with commonly encountered symmetries. For example, consider a conserved total magnetization $M = \sum_{i \in [n]} Z_i$. In terms of qubits, this is equivalent to the different Hamming-weight sectors. Each term $Z_i$ lies in a different irrep $I = \{i\}$, so $M$ spans multiple irreps rather than having a single conserved quantity per irrep.

To remedy this conflict, we introduce what we call subsystem-symmetrized Pauli shadows, which randomizes over a group whose irreps are labeled only by the qubit locality $k$, rather than any specific subsystem $I$ of $k$ qubits. (This is analogous to how the matchgate irreps depend only on fermionic locality, due to the inherent antisymmetry of fermions.) We formalize the group as follows.

\begin{definition}
\label{def:ss-pauli}
    The subsystem-symmetrized local Clifford group is defined as $\SymCl{n} \coloneqq \Sym(n) \times \Cl(1)^{\otimes n}$, where $\Sym(n)$ is the symmetric group and $\Cl(1)$ is the single-qubit Clifford group. Its unitary action on $\mathcal{H}$ is given by
    \begin{equation}
    U_{(\pi, C)} = S_\pi C,
    \end{equation}
    where $C = \bigotimes_{i \in [n]} C_i \in \Cl(1)^{\otimes n}$ and $\pi \in \Sym(n)$ is represented by a permutation of the $n$ qubits:
    \begin{equation}
    S_\pi \ket{b_0} \cdots \ket{b_{n-1}} = \ket{b_{\pi^{-1}(0)}} \cdots \ket{b_{\pi^{-1}(n-1)}}
    \end{equation}
    for all $b_i \in \{0, 1\}$, $i \in [n]$.
\end{definition}

The unitaries $S_\pi$ can be implemented with $\O(n^2)$ gates and depth $\O(n)$, for example by constructing a parallelized network of nearest-neighbor $\SWAP$ gates according to an odd--even sorting algorithm~\cite{habermann1972parallel} applied to $\pi$. Representing $\pi$ as an array of the permuted elements of $[n]$, the sorting algorithm returns a sequence of adjacent transpositions $i \leftrightarrow i + 1$ which maps $\pi$ to $(0, 1, \ldots, n-1)$. This sequence therefore implements $\pi^{-1}$ as desired. Each such transposition then maps to a $\SWAP_{i,i+1}$ gate to construct the quantum circuit. For the postprocessing of $T$ shadows into $k$-local Pauli estimates, we review in Supplementary Section~\ref{subsec:pauli_computation} the algorithm which runs in time $\O(n^k T)$.

We prove the relevant properties of subsystem-symmetrized Pauli shadows in Supplementary Section~\ref{sec:pauli_shadows_appendix}, namely its irreps and the shadow norm of local observables. We summarize the results here:~each irrep is the space of all $k$-local operators,
\begin{equation}
    V_k = \spn(\mathcal{B}_k) \text{ where } \mathcal{B}_k \coloneqq \{ P \in \Pauli(n) : |P| = k \},
\end{equation}
for each $k \in \{0, 1, \ldots, n\}$. Hence the (noisy) measurement channel is
\begin{equation}
    \widetilde{\mathcal{M}} = \sum_{k=0}^n \widetilde{f}_k \Pi_k,
\end{equation}
where $\widetilde{f}_k = \tr(\mathcal{M}_Z \mathcal{E} \Pi_k) / (3^k \binom{n}{k})$ and
\begin{equation}
    \Pi_k = \frac{1}{d} \sum_{P \in \mathcal{B}_k} \vop{P}{P}.
\end{equation}
When $\mathcal{E}$ is the identity channel, we recover $f_k = 3^{-k}$. Also in the absence of noise, the variance formulas are exactly the same as in standard Pauli shadows.

% \subsubsection{Utilizing total magnetization symmetry}\label{subsec:qubit_symmetry}

The canonical example we have considered in this paper is a $\U(1)$ symmetry generated by a total magnetization $M = \sum_{i \in [n]} Z_i$. Suppose the ideal state has a known value of $m = \tr(M \rho)$ (equivalently, $\rho$ lives in a sector of fixed Hamming weight $(n - m)/2$). The symmetries projected into the irreps of $\SymCl{n}$ are then
\begin{align}
    S_1 &= \Pi_1(M) = \sum_{i \in [n]} Z_i,\\
    S_2 &= \Pi_2(M^2) = 2 \sum_{i < j} Z_i Z_j,
\end{align}
whose ideal values are
\begin{align}
    s_1 &= m,\\
    s_2 &= m^2 - n.
\end{align}
As in the fermionic setting, we encounter issues if $s_1$ or $s_2$ vanish (i.e., $m = 0$ or $m = \pm\sqrt{n}$, respectively). In this case, we can perform the same ancilla trick, appending a qubit in $\ket{0}$ and modifying the conserved quantities to
\begin{align}
    s'_1 &= m + 1,\\
    s'_2 &= (m + 1)^2 - (n + 1).
\end{align}

The variances of the symmetry operators are
\begin{align}
    \V[\hat{s}_1/s_1] &= \O(n),\\
    \V[\hat{s}_2/s_2] &= \O(1),
\end{align}
whenever the ideal state lives in a symmetry sector of constant $m = \Theta(1)$. We show this in Supplementary Section~\ref{subsec:ss-pauli_variance}, along with general $m$-dependent expressions. This $n$-dependent variance bound reflects the fact that the symmetries are extensive properties. While local Pauli operators have variances bounded by a constant, we point out that many local observables of interest are linear combinations of an extensive number of Pauli terms. As such, their shadow norms typically grow with system size as well (which can also be seen by the fact that the shadow norm scales with operator norm).

\section*{Data availability}

The data used in this work are available from the corresponding author upon reasonable request.

\section*{Code availability}

The code for running the numerical experiments is available at this link ({\small\url{https://github.com/zhao-andrew/symmetry-adjusted-classical-shadows}}).

\begin{acknowledgments}
    This work was supported by the National Science Foundation STAQ Project (PHY-1818914, PHY-2325080) and CHE-2037832. Support is also acknowledged from the U.S. Department of Energy, Office of Science National Quantum Information Science Research Center, Quantum Systems Accelerator. The authors thank the UNM Center for Advanced Research Computing, supported in part by the National Science Foundation, for providing the high performance computing and large-scale storage resources used in this work.
\end{acknowledgments}

\section*{Author contributions}

Project design and conceptualization were envisioned by A.Z.~and A.M. The analyses and numerical experiments were led by A.Z.~and discussed with A.M. The manuscript was written by A.Z.~and A.M.

\section*{Competing interests}

The authors declare no competing interests.

\bibliography{references}

\begin{thebibliography}{144}%
\makeatletter
\providecommand \@ifxundefined [1]{%
 \@ifx{#1\undefined}
}%
\providecommand \@ifnum [1]{%
 \ifnum #1\expandafter \@firstoftwo
 \else \expandafter \@secondoftwo
 \fi
}%
\providecommand \@ifx [1]{%
 \ifx #1\expandafter \@firstoftwo
 \else \expandafter \@secondoftwo
 \fi
}%
\providecommand \natexlab [1]{#1}%
\providecommand \enquote  [1]{``#1''}%
\providecommand \bibnamefont  [1]{#1}%
\providecommand \bibfnamefont [1]{#1}%
\providecommand \citenamefont [1]{#1}%
\providecommand \href@noop [0]{\@secondoftwo}%
\providecommand \href [0]{\begingroup \@sanitize@url \@href}%
\providecommand \@href[1]{\@@startlink{#1}\@@href}%
\providecommand \@@href[1]{\endgroup#1\@@endlink}%
\providecommand \@sanitize@url [0]{\catcode `\\12\catcode `\$12\catcode `\&12\catcode `\#12\catcode `\^12\catcode `\_12\catcode `\%12\relax}%
\providecommand \@@startlink[1]{}%
\providecommand \@@endlink[0]{}%
\providecommand \url  [0]{\begingroup\@sanitize@url \@url }%
\providecommand \@url [1]{\endgroup\@href {#1}{\urlprefix }}%
\providecommand \urlprefix  [0]{URL }%
\providecommand \Eprint [0]{\href }%
\providecommand \doibase [0]{https://doi.org/}%
\providecommand \selectlanguage [0]{\@gobble}%
\providecommand \bibinfo  [0]{\@secondoftwo}%
\providecommand \bibfield  [0]{\@secondoftwo}%
\providecommand \translation [1]{[#1]}%
\providecommand \BibitemOpen [0]{}%
\providecommand \bibitemStop [0]{}%
\providecommand \bibitemNoStop [0]{.\EOS\space}%
\providecommand \EOS [0]{\spacefactor3000\relax}%
\providecommand \BibitemShut  [1]{\csname bibitem#1\endcsname}%
\let\auto@bib@innerbib\@empty
%</preamble>
\bibitem [{\citenamefont {Preskill}(2018)}]{preskill2018quantum}%
  \BibitemOpen
  \bibfield  {author} {\bibinfo {author} {\bibfnamefont {J.}~\bibnamefont {Preskill}},\ }\bibfield  {title} {\bibinfo {title} {Quantum computing in the {NISQ} era and beyond},\ }\href {https://doi.org/10.22331/q-2018-08-06-79} {\bibfield  {journal} {\bibinfo  {journal} {Quantum}\ }\textbf {\bibinfo {volume} {2}},\ \bibinfo {pages} {79} (\bibinfo {year} {2018})}\BibitemShut {NoStop}%
\bibitem [{\citenamefont {Bharti}\ \emph {et~al.}(2022)\citenamefont {Bharti} \emph {et~al.}}]{bharti2022noisy}%
  \BibitemOpen
  \bibfield  {author} {\bibinfo {author} {\bibfnamefont {K.}~\bibnamefont {Bharti}} \emph {et~al.},\ }\bibfield  {title} {\bibinfo {title} {Noisy intermediate-scale quantum algorithms},\ }\href {https://doi.org/10.1103/RevModPhys.94.015004} {\bibfield  {journal} {\bibinfo  {journal} {Rev. Mod. Phys.}\ }\textbf {\bibinfo {volume} {94}},\ \bibinfo {pages} {015004} (\bibinfo {year} {2022})}\BibitemShut {NoStop}%
\bibitem [{\citenamefont {Feynman}(1982)}]{feynman1982simulating}%
  \BibitemOpen
  \bibfield  {author} {\bibinfo {author} {\bibfnamefont {R.~P.}\ \bibnamefont {Feynman}},\ }\bibfield  {title} {\bibinfo {title} {Simulating physics with computers},\ }\href {https://doi.org/10.1007/BF02650179} {\bibfield  {journal} {\bibinfo  {journal} {Int. J. Theor. Phys.}\ }\textbf {\bibinfo {volume} {21}},\ \bibinfo {pages} {467} (\bibinfo {year} {1982})}\BibitemShut {NoStop}%
\bibitem [{\citenamefont {Georgescu}\ \emph {et~al.}(2014)\citenamefont {Georgescu}, \citenamefont {Ashhab},\ and\ \citenamefont {Nori}}]{georgescu2014quantum}%
  \BibitemOpen
  \bibfield  {author} {\bibinfo {author} {\bibfnamefont {I.~M.}\ \bibnamefont {Georgescu}}, \bibinfo {author} {\bibfnamefont {S.}~\bibnamefont {Ashhab}},\ and\ \bibinfo {author} {\bibfnamefont {F.}~\bibnamefont {Nori}},\ }\bibfield  {title} {\bibinfo {title} {Quantum simulation},\ }\href {https://doi.org/10.1103/RevModPhys.86.153} {\bibfield  {journal} {\bibinfo  {journal} {Rev. Mod. Phys.}\ }\textbf {\bibinfo {volume} {86}},\ \bibinfo {pages} {153} (\bibinfo {year} {2014})}\BibitemShut {NoStop}%
\bibitem [{\citenamefont {McArdle}\ \emph {et~al.}(2020)\citenamefont {McArdle}, \citenamefont {Endo}, \citenamefont {Aspuru-Guzik}, \citenamefont {Benjamin},\ and\ \citenamefont {Yuan}}]{mcardle2020quantum}%
  \BibitemOpen
  \bibfield  {author} {\bibinfo {author} {\bibfnamefont {S.}~\bibnamefont {McArdle}}, \bibinfo {author} {\bibfnamefont {S.}~\bibnamefont {Endo}}, \bibinfo {author} {\bibfnamefont {A.}~\bibnamefont {Aspuru-Guzik}}, \bibinfo {author} {\bibfnamefont {S.~C.}\ \bibnamefont {Benjamin}},\ and\ \bibinfo {author} {\bibfnamefont {X.}~\bibnamefont {Yuan}},\ }\bibfield  {title} {\bibinfo {title} {Quantum computational chemistry},\ }\href {https://doi.org/10.1103/RevModPhys.92.015003} {\bibfield  {journal} {\bibinfo  {journal} {Rev. Mod. Phys.}\ }\textbf {\bibinfo {volume} {92}},\ \bibinfo {pages} {015003} (\bibinfo {year} {2020})}\BibitemShut {NoStop}%
\bibitem [{\citenamefont {Bauer}\ \emph {et~al.}(2020)\citenamefont {Bauer}, \citenamefont {Bravyi}, \citenamefont {Motta},\ and\ \citenamefont {Chan}}]{bauer2020quantum}%
  \BibitemOpen
  \bibfield  {author} {\bibinfo {author} {\bibfnamefont {B.}~\bibnamefont {Bauer}}, \bibinfo {author} {\bibfnamefont {S.}~\bibnamefont {Bravyi}}, \bibinfo {author} {\bibfnamefont {M.}~\bibnamefont {Motta}},\ and\ \bibinfo {author} {\bibfnamefont {G.~K.-L.}\ \bibnamefont {Chan}},\ }\bibfield  {title} {\bibinfo {title} {Quantum algorithms for quantum chemistry and quantum materials science},\ }\href {https://doi.org/10.1021/acs.chemrev.9b00829} {\bibfield  {journal} {\bibinfo  {journal} {Chem. Rev.}\ }\textbf {\bibinfo {volume} {120}},\ \bibinfo {pages} {12685} (\bibinfo {year} {2020})}\BibitemShut {NoStop}%
\bibitem [{\citenamefont {Peruzzo}\ \emph {et~al.}(2014)\citenamefont {Peruzzo} \emph {et~al.}}]{peruzzo2014variational}%
  \BibitemOpen
  \bibfield  {author} {\bibinfo {author} {\bibfnamefont {A.}~\bibnamefont {Peruzzo}} \emph {et~al.},\ }\bibfield  {title} {\bibinfo {title} {A variational eigenvalue solver on a photonic quantum processor},\ }\href {https://doi.org/10.1038/ncomms5213} {\bibfield  {journal} {\bibinfo  {journal} {Nat. Commun.}\ }\textbf {\bibinfo {volume} {5}},\ \bibinfo {pages} {4213} (\bibinfo {year} {2014})}\BibitemShut {NoStop}%
\bibitem [{\citenamefont {McClean}\ \emph {et~al.}(2016)\citenamefont {McClean}, \citenamefont {Romero}, \citenamefont {Babbush},\ and\ \citenamefont {Aspuru-Guzik}}]{mcclean2016theory}%
  \BibitemOpen
  \bibfield  {author} {\bibinfo {author} {\bibfnamefont {J.~R.}\ \bibnamefont {McClean}}, \bibinfo {author} {\bibfnamefont {J.}~\bibnamefont {Romero}}, \bibinfo {author} {\bibfnamefont {R.}~\bibnamefont {Babbush}},\ and\ \bibinfo {author} {\bibfnamefont {A.}~\bibnamefont {Aspuru-Guzik}},\ }\bibfield  {title} {\bibinfo {title} {The theory of variational hybrid quantum-classical algorithms},\ }\href {https://doi.org/10.1088/1367-2630/18/2/023023} {\bibfield  {journal} {\bibinfo  {journal} {New J. Phys.}\ }\textbf {\bibinfo {volume} {18}},\ \bibinfo {pages} {023023} (\bibinfo {year} {2016})}\BibitemShut {NoStop}%
\bibitem [{\citenamefont {Yuan}\ \emph {et~al.}(2019)\citenamefont {Yuan}, \citenamefont {Endo}, \citenamefont {Zhao}, \citenamefont {Li},\ and\ \citenamefont {Benjamin}}]{yuan2019theory}%
  \BibitemOpen
  \bibfield  {author} {\bibinfo {author} {\bibfnamefont {X.}~\bibnamefont {Yuan}}, \bibinfo {author} {\bibfnamefont {S.}~\bibnamefont {Endo}}, \bibinfo {author} {\bibfnamefont {Q.}~\bibnamefont {Zhao}}, \bibinfo {author} {\bibfnamefont {Y.}~\bibnamefont {Li}},\ and\ \bibinfo {author} {\bibfnamefont {S.~C.}\ \bibnamefont {Benjamin}},\ }\bibfield  {title} {\bibinfo {title} {Theory of variational quantum simulation},\ }\href {https://doi.org/10.22331/q-2019-10-07-191} {\bibfield  {journal} {\bibinfo  {journal} {Quantum}\ }\textbf {\bibinfo {volume} {3}},\ \bibinfo {pages} {191} (\bibinfo {year} {2019})}\BibitemShut {NoStop}%
\bibitem [{\citenamefont {Cerezo}\ \emph {et~al.}(2021)\citenamefont {Cerezo} \emph {et~al.}}]{cerezo2021variational}%
  \BibitemOpen
  \bibfield  {author} {\bibinfo {author} {\bibfnamefont {M.}~\bibnamefont {Cerezo}} \emph {et~al.},\ }\bibfield  {title} {\bibinfo {title} {Variational quantum algorithms},\ }\href {https://doi.org/10.1038/s42254-021-00348-9} {\bibfield  {journal} {\bibinfo  {journal} {Nat. Rev. Phys.}\ }\textbf {\bibinfo {volume} {3}},\ \bibinfo {pages} {625} (\bibinfo {year} {2021})}\BibitemShut {NoStop}%
\bibitem [{\citenamefont {Osborne}(2006)}]{osborne2006efficient}%
  \BibitemOpen
  \bibfield  {author} {\bibinfo {author} {\bibfnamefont {T.~J.}\ \bibnamefont {Osborne}},\ }\bibfield  {title} {\bibinfo {title} {Efficient approximation of the dynamics of one-dimensional quantum spin systems},\ }\href {https://doi.org/10.1103/PhysRevLett.97.157202} {\bibfield  {journal} {\bibinfo  {journal} {Phys. Rev. Lett.}\ }\textbf {\bibinfo {volume} {97}},\ \bibinfo {pages} {157202} (\bibinfo {year} {2006})}\BibitemShut {NoStop}%
\bibitem [{\citenamefont {Bravyi}\ \emph {et~al.}(2021)\citenamefont {Bravyi}, \citenamefont {Gosset},\ and\ \citenamefont {Movassagh}}]{bravyi2021classical}%
  \BibitemOpen
  \bibfield  {author} {\bibinfo {author} {\bibfnamefont {S.}~\bibnamefont {Bravyi}}, \bibinfo {author} {\bibfnamefont {D.}~\bibnamefont {Gosset}},\ and\ \bibinfo {author} {\bibfnamefont {R.}~\bibnamefont {Movassagh}},\ }\bibfield  {title} {\bibinfo {title} {Classical algorithms for quantum mean values},\ }\href {https://doi.org/10.1038/s41567-020-01109-8} {\bibfield  {journal} {\bibinfo  {journal} {Nat. Phys.}\ }\textbf {\bibinfo {volume} {17}},\ \bibinfo {pages} {337} (\bibinfo {year} {2021})}\BibitemShut {NoStop}%
\bibitem [{\citenamefont {Napp}\ \emph {et~al.}(2022)\citenamefont {Napp}, \citenamefont {La~Placa}, \citenamefont {Dalzell}, \citenamefont {Brand\~ao},\ and\ \citenamefont {Harrow}}]{napp2022efficient}%
  \BibitemOpen
  \bibfield  {author} {\bibinfo {author} {\bibfnamefont {J.~C.}\ \bibnamefont {Napp}}, \bibinfo {author} {\bibfnamefont {R.~L.}\ \bibnamefont {La~Placa}}, \bibinfo {author} {\bibfnamefont {A.~M.}\ \bibnamefont {Dalzell}}, \bibinfo {author} {\bibfnamefont {F.~G. S.~L.}\ \bibnamefont {Brand\~ao}},\ and\ \bibinfo {author} {\bibfnamefont {A.~W.}\ \bibnamefont {Harrow}},\ }\bibfield  {title} {\bibinfo {title} {Efficient classical simulation of random shallow {2D} quantum circuits},\ }\href {https://doi.org/10.1103/PhysRevX.12.021021} {\bibfield  {journal} {\bibinfo  {journal} {Phys. Rev. X}\ }\textbf {\bibinfo {volume} {12}},\ \bibinfo {pages} {021021} (\bibinfo {year} {2022})}\BibitemShut {NoStop}%
\bibitem [{\citenamefont {Wild}\ and\ \citenamefont {Alhambra}(2023)}]{wild2023classical}%
  \BibitemOpen
  \bibfield  {author} {\bibinfo {author} {\bibfnamefont {D.~S.}\ \bibnamefont {Wild}}\ and\ \bibinfo {author} {\bibfnamefont {A.~M.}\ \bibnamefont {Alhambra}},\ }\bibfield  {title} {\bibinfo {title} {Classical simulation of short-time quantum dynamics},\ }\href {https://doi.org/10.1103/PRXQuantum.4.020340} {\bibfield  {journal} {\bibinfo  {journal} {PRX Quantum}\ }\textbf {\bibinfo {volume} {4}},\ \bibinfo {pages} {020340} (\bibinfo {year} {2023})}\BibitemShut {NoStop}%
\bibitem [{\citenamefont {Fowler}\ \emph {et~al.}(2012)\citenamefont {Fowler}, \citenamefont {Mariantoni}, \citenamefont {Martinis},\ and\ \citenamefont {Cleland}}]{fowler2012surface}%
  \BibitemOpen
  \bibfield  {author} {\bibinfo {author} {\bibfnamefont {A.~G.}\ \bibnamefont {Fowler}}, \bibinfo {author} {\bibfnamefont {M.}~\bibnamefont {Mariantoni}}, \bibinfo {author} {\bibfnamefont {J.~M.}\ \bibnamefont {Martinis}},\ and\ \bibinfo {author} {\bibfnamefont {A.~N.}\ \bibnamefont {Cleland}},\ }\bibfield  {title} {\bibinfo {title} {Surface codes: {T}owards practical large-scale quantum computation},\ }\href {https://doi.org/10.1103/PhysRevA.86.032324} {\bibfield  {journal} {\bibinfo  {journal} {Phys. Rev. A}\ }\textbf {\bibinfo {volume} {86}},\ \bibinfo {pages} {032324} (\bibinfo {year} {2012})}\BibitemShut {NoStop}%
\bibitem [{\citenamefont {Kelly}\ \emph {et~al.}(2015)\citenamefont {Kelly} \emph {et~al.}}]{kelly2015state}%
  \BibitemOpen
  \bibfield  {author} {\bibinfo {author} {\bibfnamefont {J.}~\bibnamefont {Kelly}} \emph {et~al.},\ }\bibfield  {title} {\bibinfo {title} {State preservation by repetitive error detection in a superconducting quantum circuit},\ }\href {https://doi.org/10.1038/nature14270} {\bibfield  {journal} {\bibinfo  {journal} {Nature}\ }\textbf {\bibinfo {volume} {519}},\ \bibinfo {pages} {66} (\bibinfo {year} {2015})}\BibitemShut {NoStop}%
\bibitem [{\citenamefont {Egan}\ \emph {et~al.}(2021)\citenamefont {Egan} \emph {et~al.}}]{egan2021fault}%
  \BibitemOpen
  \bibfield  {author} {\bibinfo {author} {\bibfnamefont {L.}~\bibnamefont {Egan}} \emph {et~al.},\ }\bibfield  {title} {\bibinfo {title} {Fault-tolerant control of an error-corrected qubit},\ }\href {https://doi.org/10.1038/s41586-021-03928-y} {\bibfield  {journal} {\bibinfo  {journal} {Nature}\ }\textbf {\bibinfo {volume} {598}},\ \bibinfo {pages} {281} (\bibinfo {year} {2021})}\BibitemShut {NoStop}%
\bibitem [{\citenamefont {Postler}\ \emph {et~al.}(2022)\citenamefont {Postler} \emph {et~al.}}]{postler2022demonstration}%
  \BibitemOpen
  \bibfield  {author} {\bibinfo {author} {\bibfnamefont {L.}~\bibnamefont {Postler}} \emph {et~al.},\ }\bibfield  {title} {\bibinfo {title} {Demonstration of fault-tolerant universal quantum gate operations},\ }\href {https://doi.org/10.1038/s41586-022-04721-1} {\bibfield  {journal} {\bibinfo  {journal} {Nature}\ }\textbf {\bibinfo {volume} {605}},\ \bibinfo {pages} {675} (\bibinfo {year} {2022})}\BibitemShut {NoStop}%
\bibitem [{\citenamefont {Zhao}\ \emph {et~al.}(2022)\citenamefont {Zhao} \emph {et~al.}}]{zhao2022realization}%
  \BibitemOpen
  \bibfield  {author} {\bibinfo {author} {\bibfnamefont {Y.}~\bibnamefont {Zhao}} \emph {et~al.},\ }\bibfield  {title} {\bibinfo {title} {Realization of an error-correcting surface code with superconducting qubits},\ }\href {https://doi.org/10.1103/PhysRevLett.129.030501} {\bibfield  {journal} {\bibinfo  {journal} {Phys. Rev. Lett.}\ }\textbf {\bibinfo {volume} {129}},\ \bibinfo {pages} {030501} (\bibinfo {year} {2022})}\BibitemShut {NoStop}%
\bibitem [{\citenamefont {Sundaresan}\ \emph {et~al.}(2023)\citenamefont {Sundaresan} \emph {et~al.}}]{sundaresan2023demonstrating}%
  \BibitemOpen
  \bibfield  {author} {\bibinfo {author} {\bibfnamefont {N.}~\bibnamefont {Sundaresan}} \emph {et~al.},\ }\bibfield  {title} {\bibinfo {title} {Demonstrating multi-round subsystem quantum error correction using matching and maximum likelihood decoders},\ }\href {https://doi.org/10.1038/s41467-023-38247-5} {\bibfield  {journal} {\bibinfo  {journal} {Nat. Commun.}\ }\textbf {\bibinfo {volume} {14}},\ \bibinfo {pages} {2852} (\bibinfo {year} {2023})}\BibitemShut {NoStop}%
\bibitem [{\citenamefont {{Google Quantum AI}}(2023)}]{google2023suppressing}%
  \BibitemOpen
  \bibfield  {author} {\bibinfo {author} {\bibnamefont {{Google Quantum AI}}},\ }\bibfield  {title} {\bibinfo {title} {Suppressing quantum errors by scaling a surface code logical qubit},\ }\href {https://doi.org/10.1038/s41586-022-05434-1} {\bibfield  {journal} {\bibinfo  {journal} {Nature}\ }\textbf {\bibinfo {volume} {614}},\ \bibinfo {pages} {676} (\bibinfo {year} {2023})}\BibitemShut {NoStop}%
\bibitem [{\citenamefont {Sivak}\ \emph {et~al.}(2023)\citenamefont {Sivak} \emph {et~al.}}]{sivak2023real}%
  \BibitemOpen
  \bibfield  {author} {\bibinfo {author} {\bibfnamefont {V.~V.}\ \bibnamefont {Sivak}} \emph {et~al.},\ }\bibfield  {title} {\bibinfo {title} {Real-time quantum error correction beyond break-even},\ }\href {https://doi.org/10.1038/s41586-023-05782-6} {\bibfield  {journal} {\bibinfo  {journal} {Nature}\ }\textbf {\bibinfo {volume} {616}},\ \bibinfo {pages} {50} (\bibinfo {year} {2023})}\BibitemShut {NoStop}%
\bibitem [{\citenamefont {Ni}\ \emph {et~al.}(2023)\citenamefont {Ni} \emph {et~al.}}]{ni2023beating}%
  \BibitemOpen
  \bibfield  {author} {\bibinfo {author} {\bibfnamefont {Z.}~\bibnamefont {Ni}} \emph {et~al.},\ }\bibfield  {title} {\bibinfo {title} {Beating the break-even point with a discrete-variable-encoded logical qubit},\ }\href {https://doi.org/10.1038/s41586-023-05784-4} {\bibfield  {journal} {\bibinfo  {journal} {Nature}\ }\textbf {\bibinfo {volume} {616}},\ \bibinfo {pages} {56} (\bibinfo {year} {2023})}\BibitemShut {NoStop}%
\bibitem [{\citenamefont {O'Malley}\ \emph {et~al.}(2016)\citenamefont {O'Malley} \emph {et~al.}}]{omalley2016scalable}%
  \BibitemOpen
  \bibfield  {author} {\bibinfo {author} {\bibfnamefont {P.~J.~J.}\ \bibnamefont {O'Malley}} \emph {et~al.},\ }\bibfield  {title} {\bibinfo {title} {Scalable quantum simulation of molecular energies},\ }\href {https://doi.org/10.1103/PhysRevX.6.031007} {\bibfield  {journal} {\bibinfo  {journal} {Phys. Rev. X}\ }\textbf {\bibinfo {volume} {6}},\ \bibinfo {pages} {031007} (\bibinfo {year} {2016})}\BibitemShut {NoStop}%
\bibitem [{\citenamefont {Kandala}\ \emph {et~al.}(2017)\citenamefont {Kandala} \emph {et~al.}}]{kandala2017hardware}%
  \BibitemOpen
  \bibfield  {author} {\bibinfo {author} {\bibfnamefont {A.}~\bibnamefont {Kandala}} \emph {et~al.},\ }\bibfield  {title} {\bibinfo {title} {Hardware-efficient variational quantum eigensolver for small molecules and quantum magnets},\ }\href {https://doi.org/10.1038/nature23879} {\bibfield  {journal} {\bibinfo  {journal} {Nature}\ }\textbf {\bibinfo {volume} {549}},\ \bibinfo {pages} {242} (\bibinfo {year} {2017})}\BibitemShut {NoStop}%
\bibitem [{\citenamefont {Colless}\ \emph {et~al.}(2018)\citenamefont {Colless} \emph {et~al.}}]{colless2018computation}%
  \BibitemOpen
  \bibfield  {author} {\bibinfo {author} {\bibfnamefont {J.~I.}\ \bibnamefont {Colless}} \emph {et~al.},\ }\bibfield  {title} {\bibinfo {title} {Computation of molecular spectra on a quantum processor with an error-resilient algorithm},\ }\href {https://doi.org/10.1103/PhysRevX.8.011021} {\bibfield  {journal} {\bibinfo  {journal} {Phys. Rev. X}\ }\textbf {\bibinfo {volume} {8}},\ \bibinfo {pages} {011021} (\bibinfo {year} {2018})}\BibitemShut {NoStop}%
\bibitem [{\citenamefont {Dumitrescu}\ \emph {et~al.}(2018)\citenamefont {Dumitrescu} \emph {et~al.}}]{dumitrescu2018cloud}%
  \BibitemOpen
  \bibfield  {author} {\bibinfo {author} {\bibfnamefont {E.~F.}\ \bibnamefont {Dumitrescu}} \emph {et~al.},\ }\bibfield  {title} {\bibinfo {title} {Cloud quantum computing of an atomic nucleus},\ }\href {https://doi.org/10.1103/PhysRevLett.120.210501} {\bibfield  {journal} {\bibinfo  {journal} {Phys. Rev. Lett.}\ }\textbf {\bibinfo {volume} {120}},\ \bibinfo {pages} {210501} (\bibinfo {year} {2018})}\BibitemShut {NoStop}%
\bibitem [{\citenamefont {Hempel}\ \emph {et~al.}(2018)\citenamefont {Hempel} \emph {et~al.}}]{hempel2018quantum}%
  \BibitemOpen
  \bibfield  {author} {\bibinfo {author} {\bibfnamefont {C.}~\bibnamefont {Hempel}} \emph {et~al.},\ }\bibfield  {title} {\bibinfo {title} {Quantum chemistry calculations on a trapped-ion quantum simulator},\ }\href {https://doi.org/10.1103/PhysRevX.8.031022} {\bibfield  {journal} {\bibinfo  {journal} {Phys. Rev. X}\ }\textbf {\bibinfo {volume} {8}},\ \bibinfo {pages} {031022} (\bibinfo {year} {2018})}\BibitemShut {NoStop}%
\bibitem [{\citenamefont {Kandala}\ \emph {et~al.}(2019)\citenamefont {Kandala}, \citenamefont {Temme}, \citenamefont {C{\'o}rcoles}, \citenamefont {Mezzacapo}, \citenamefont {Chow},\ and\ \citenamefont {Gambetta}}]{kandala2019error}%
  \BibitemOpen
  \bibfield  {author} {\bibinfo {author} {\bibfnamefont {A.}~\bibnamefont {Kandala}}, \bibinfo {author} {\bibfnamefont {K.}~\bibnamefont {Temme}}, \bibinfo {author} {\bibfnamefont {A.~D.}\ \bibnamefont {C{\'o}rcoles}}, \bibinfo {author} {\bibfnamefont {A.}~\bibnamefont {Mezzacapo}}, \bibinfo {author} {\bibfnamefont {J.~M.}\ \bibnamefont {Chow}},\ and\ \bibinfo {author} {\bibfnamefont {J.~M.}\ \bibnamefont {Gambetta}},\ }\bibfield  {title} {\bibinfo {title} {Error mitigation extends the computational reach of a noisy quantum processor},\ }\href {https://doi.org/10.1038/s41586-019-1040-7} {\bibfield  {journal} {\bibinfo  {journal} {Nature}\ }\textbf {\bibinfo {volume} {567}},\ \bibinfo {pages} {491} (\bibinfo {year} {2019})}\BibitemShut {NoStop}%
\bibitem [{\citenamefont {Kokail}\ \emph {et~al.}(2019)\citenamefont {Kokail} \emph {et~al.}}]{kokail2019self}%
  \BibitemOpen
  \bibfield  {author} {\bibinfo {author} {\bibfnamefont {C.}~\bibnamefont {Kokail}} \emph {et~al.},\ }\bibfield  {title} {\bibinfo {title} {Self-verifying variational quantum simulation of lattice models},\ }\href {https://doi.org/10.1038/s41586-019-1177-4} {\bibfield  {journal} {\bibinfo  {journal} {Nature}\ }\textbf {\bibinfo {volume} {569}},\ \bibinfo {pages} {355} (\bibinfo {year} {2019})}\BibitemShut {NoStop}%
\bibitem [{\citenamefont {Nam}\ \emph {et~al.}(2020)\citenamefont {Nam} \emph {et~al.}}]{nam2020ground}%
  \BibitemOpen
  \bibfield  {author} {\bibinfo {author} {\bibfnamefont {Y.}~\bibnamefont {Nam}} \emph {et~al.},\ }\bibfield  {title} {\bibinfo {title} {Ground-state energy estimation of the water molecule on a trapped-ion quantum computer},\ }\href {https://doi.org/10.1038/s41534-020-0259-3} {\bibfield  {journal} {\bibinfo  {journal} {npj Quantum Inf.}\ }\textbf {\bibinfo {volume} {6}},\ \bibinfo {pages} {33} (\bibinfo {year} {2020})}\BibitemShut {NoStop}%
\bibitem [{\citenamefont {Arute}\ \emph {et~al.}(2019)\citenamefont {Arute} \emph {et~al.}}]{arute2019quantum}%
  \BibitemOpen
  \bibfield  {author} {\bibinfo {author} {\bibfnamefont {F.}~\bibnamefont {Arute}} \emph {et~al.} (\bibinfo {collaboration} {{Google AI Quantum and collaborators}}),\ }\bibfield  {title} {\bibinfo {title} {Quantum supremacy using a programmable superconducting processor},\ }\href {https://doi.org/10.1038/s41586-019-1666-5} {\bibfield  {journal} {\bibinfo  {journal} {Nature}\ }\textbf {\bibinfo {volume} {574}},\ \bibinfo {pages} {505} (\bibinfo {year} {2019})}\BibitemShut {NoStop}%
\bibitem [{\citenamefont {Rubin}\ \emph {et~al.}(2020)\citenamefont {Rubin} \emph {et~al.}}]{arute2020hartree}%
  \BibitemOpen
  \bibfield  {author} {\bibinfo {author} {\bibfnamefont {N.~C.}\ \bibnamefont {Rubin}} \emph {et~al.} (\bibinfo {collaboration} {{Google AI Quantum and collaborators}}),\ }\bibfield  {title} {\bibinfo {title} {Hartree-{F}ock on a superconducting qubit quantum computer},\ }\href {https://doi.org/10.1126/science.abb9811} {\bibfield  {journal} {\bibinfo  {journal} {Science}\ }\textbf {\bibinfo {volume} {369}},\ \bibinfo {pages} {1084} (\bibinfo {year} {2020})}\BibitemShut {NoStop}%
\bibitem [{\citenamefont {Harrigan}\ \emph {et~al.}(2021)\citenamefont {Harrigan} \emph {et~al.}}]{harrigan2021quantum}%
  \BibitemOpen
  \bibfield  {author} {\bibinfo {author} {\bibfnamefont {M.~P.}\ \bibnamefont {Harrigan}} \emph {et~al.} (\bibinfo {collaboration} {{Google AI Quantum and collaborators}}),\ }\bibfield  {title} {\bibinfo {title} {Quantum approximate optimization of non-planar graph problems on a planar superconducting processor},\ }\href {https://doi.org/10.1038/s41567-020-01105-y} {\bibfield  {journal} {\bibinfo  {journal} {Nat. Phys.}\ }\textbf {\bibinfo {volume} {17}},\ \bibinfo {pages} {332} (\bibinfo {year} {2021})}\BibitemShut {NoStop}%
\bibitem [{\citenamefont {Jiang}\ \emph {et~al.}(2020)\citenamefont {Jiang} \emph {et~al.}}]{arute2020observation}%
  \BibitemOpen
  \bibfield  {author} {\bibinfo {author} {\bibfnamefont {Z.}~\bibnamefont {Jiang}} \emph {et~al.} (\bibinfo {collaboration} {{Google AI Quantum and collaborators}}),\ }\bibfield  {title} {\bibinfo {title} {Observation of separated dynamics of charge and spin in the {F}ermi-{H}ubbard model},\ }\href {https://arxiv.org/abs/2010.07965} {\bibfield  {journal} {\bibinfo  {journal} {arXiv:2010.07965}\ } (\bibinfo {year} {2020})}\BibitemShut {NoStop}%
\bibitem [{\citenamefont {Zhong}\ \emph {et~al.}(2020)\citenamefont {Zhong} \emph {et~al.}}]{zhong2020quantum}%
  \BibitemOpen
  \bibfield  {author} {\bibinfo {author} {\bibfnamefont {H.-S.}\ \bibnamefont {Zhong}} \emph {et~al.},\ }\bibfield  {title} {\bibinfo {title} {Quantum computational advantage using photons},\ }\href {https://doi.org/10.1126/science.abe8770} {\bibfield  {journal} {\bibinfo  {journal} {Science}\ }\textbf {\bibinfo {volume} {370}},\ \bibinfo {pages} {1460} (\bibinfo {year} {2020})}\BibitemShut {NoStop}%
\bibitem [{\citenamefont {Huggins}\ \emph {et~al.}(2022)\citenamefont {Huggins}, \citenamefont {O’Gorman}, \citenamefont {Rubin}, \citenamefont {Reichman}, \citenamefont {Babbush},\ and\ \citenamefont {Lee}}]{huggins2022unbiasing}%
  \BibitemOpen
  \bibfield  {author} {\bibinfo {author} {\bibfnamefont {W.~J.}\ \bibnamefont {Huggins}}, \bibinfo {author} {\bibfnamefont {B.~A.}\ \bibnamefont {O’Gorman}}, \bibinfo {author} {\bibfnamefont {N.~C.}\ \bibnamefont {Rubin}}, \bibinfo {author} {\bibfnamefont {D.~R.}\ \bibnamefont {Reichman}}, \bibinfo {author} {\bibfnamefont {R.}~\bibnamefont {Babbush}},\ and\ \bibinfo {author} {\bibfnamefont {J.}~\bibnamefont {Lee}},\ }\bibfield  {title} {\bibinfo {title} {Unbiasing fermionic quantum {M}onte {C}arlo with a quantum computer},\ }\href {https://doi.org/10.1038/s41586-021-04351-z} {\bibfield  {journal} {\bibinfo  {journal} {Nature}\ }\textbf {\bibinfo {volume} {603}},\ \bibinfo {pages} {416} (\bibinfo {year} {2022})}\BibitemShut {NoStop}%
\bibitem [{\citenamefont {Kim}\ \emph {et~al.}(2023{\natexlab{a}})\citenamefont {Kim} \emph {et~al.}}]{kim2023scalable}%
  \BibitemOpen
  \bibfield  {author} {\bibinfo {author} {\bibfnamefont {Y.}~\bibnamefont {Kim}} \emph {et~al.},\ }\bibfield  {title} {\bibinfo {title} {Scalable error mitigation for noisy quantum circuits produces competitive expectation values},\ }\href {https://doi.org/10.1038/s41567-022-01914-3} {\bibfield  {journal} {\bibinfo  {journal} {Nat. Phys.}\ }\textbf {\bibinfo {volume} {19}},\ \bibinfo {pages} {752} (\bibinfo {year} {2023}{\natexlab{a}})}\BibitemShut {NoStop}%
\bibitem [{\citenamefont {Huang}\ \emph {et~al.}(2022)\citenamefont {Huang} \emph {et~al.}}]{huang2022quantum}%
  \BibitemOpen
  \bibfield  {author} {\bibinfo {author} {\bibfnamefont {H.-Y.}\ \bibnamefont {Huang}} \emph {et~al.},\ }\bibfield  {title} {\bibinfo {title} {Quantum advantage in learning from experiments},\ }\href {https://doi.org/10.1126/science.abn7293} {\bibfield  {journal} {\bibinfo  {journal} {Science}\ }\textbf {\bibinfo {volume} {376}},\ \bibinfo {pages} {1182} (\bibinfo {year} {2022})}\BibitemShut {NoStop}%
\bibitem [{\citenamefont {Stanisic}\ \emph {et~al.}(2022)\citenamefont {Stanisic} \emph {et~al.}}]{stanisic2022observing}%
  \BibitemOpen
  \bibfield  {author} {\bibinfo {author} {\bibfnamefont {S.}~\bibnamefont {Stanisic}} \emph {et~al.},\ }\bibfield  {title} {\bibinfo {title} {Observing ground-state properties of the {F}ermi-{H}ubbard model using a scalable algorithm on a quantum computer},\ }\href {https://doi.org/10.1038/s41467-022-33335-4} {\bibfield  {journal} {\bibinfo  {journal} {Nat. Commun.}\ }\textbf {\bibinfo {volume} {13}},\ \bibinfo {pages} {5743} (\bibinfo {year} {2022})}\BibitemShut {NoStop}%
\bibitem [{\citenamefont {Tazhigulov}\ \emph {et~al.}(2022)\citenamefont {Tazhigulov} \emph {et~al.}}]{tazhigulov2022simulating}%
  \BibitemOpen
  \bibfield  {author} {\bibinfo {author} {\bibfnamefont {R.~N.}\ \bibnamefont {Tazhigulov}} \emph {et~al.},\ }\bibfield  {title} {\bibinfo {title} {Simulating models of challenging correlated molecules and materials on the {S}ycamore quantum processor},\ }\href {https://doi.org/10.1103/PRXQuantum.3.040318} {\bibfield  {journal} {\bibinfo  {journal} {PRX Quantum}\ }\textbf {\bibinfo {volume} {3}},\ \bibinfo {pages} {040318} (\bibinfo {year} {2022})}\BibitemShut {NoStop}%
\bibitem [{\citenamefont {Madsen}\ \emph {et~al.}(2022)\citenamefont {Madsen} \emph {et~al.}}]{madsen2022quantum}%
  \BibitemOpen
  \bibfield  {author} {\bibinfo {author} {\bibfnamefont {L.~S.}\ \bibnamefont {Madsen}} \emph {et~al.},\ }\bibfield  {title} {\bibinfo {title} {Quantum computational advantage with a programmable photonic processor},\ }\href {https://doi.org/10.1038/s41586-022-04725-x} {\bibfield  {journal} {\bibinfo  {journal} {Nature}\ }\textbf {\bibinfo {volume} {606}},\ \bibinfo {pages} {75} (\bibinfo {year} {2022})}\BibitemShut {NoStop}%
\bibitem [{\citenamefont {Motta}\ \emph {et~al.}(2023)\citenamefont {Motta} \emph {et~al.}}]{motta2023quantum}%
  \BibitemOpen
  \bibfield  {author} {\bibinfo {author} {\bibfnamefont {M.}~\bibnamefont {Motta}} \emph {et~al.},\ }\bibfield  {title} {\bibinfo {title} {Quantum chemistry simulation of ground-and excited-state properties of the sulfonium cation on a superconducting quantum processor},\ }\href {https://doi.org/10.1039/D2SC06019A} {\bibfield  {journal} {\bibinfo  {journal} {Chem. Sci.}\ }\textbf {\bibinfo {volume} {14}},\ \bibinfo {pages} {2915} (\bibinfo {year} {2023})}\BibitemShut {NoStop}%
\bibitem [{\citenamefont {O'Brien}\ \emph {et~al.}(2023)\citenamefont {O'Brien} \emph {et~al.}}]{obrien2023purification}%
  \BibitemOpen
  \bibfield  {author} {\bibinfo {author} {\bibfnamefont {T.~E.}\ \bibnamefont {O'Brien}} \emph {et~al.} (\bibinfo {collaboration} {{Google Quantum AI and collaborators}}),\ }\bibfield  {title} {\bibinfo {title} {Purification-based quantum error mitigation of pair-correlated electron simulations},\ }\href {https://doi.org/10.1038/s41567-023-02240-y} {\bibfield  {journal} {\bibinfo  {journal} {Nat. Phys.}\ } (\bibinfo {year} {2023})}\BibitemShut {NoStop}%
\bibitem [{\citenamefont {Morvan}\ \emph {et~al.}(2023)\citenamefont {Morvan} \emph {et~al.}}]{morvan2023phase}%
  \BibitemOpen
  \bibfield  {author} {\bibinfo {author} {\bibfnamefont {A.}~\bibnamefont {Morvan}} \emph {et~al.} (\bibinfo {collaboration} {Google Quantum AI and collaborators}),\ }\bibfield  {title} {\bibinfo {title} {Phase transition in random circuit sampling},\ }\href {https://arxiv.org/abs/2304.11119} {\bibfield  {journal} {\bibinfo  {journal} {arXiv:2304.11119}\ } (\bibinfo {year} {2023})}\BibitemShut {NoStop}%
\bibitem [{\citenamefont {Kim}\ \emph {et~al.}(2023{\natexlab{b}})\citenamefont {Kim} \emph {et~al.}}]{kim2023evidence}%
  \BibitemOpen
  \bibfield  {author} {\bibinfo {author} {\bibfnamefont {Y.}~\bibnamefont {Kim}} \emph {et~al.},\ }\bibfield  {title} {\bibinfo {title} {Evidence for the utility of quantum computing before fault tolerance},\ }\href {https://doi.org/10.1038/s41586-023-06096-3} {\bibfield  {journal} {\bibinfo  {journal} {Nature}\ }\textbf {\bibinfo {volume} {618}},\ \bibinfo {pages} {500} (\bibinfo {year} {2023}{\natexlab{b}})}\BibitemShut {NoStop}%
\bibitem [{\citenamefont {Endo}\ \emph {et~al.}(2021)\citenamefont {Endo}, \citenamefont {Cai}, \citenamefont {Benjamin},\ and\ \citenamefont {Yuan}}]{endo2021hybrid}%
  \BibitemOpen
  \bibfield  {author} {\bibinfo {author} {\bibfnamefont {S.}~\bibnamefont {Endo}}, \bibinfo {author} {\bibfnamefont {Z.}~\bibnamefont {Cai}}, \bibinfo {author} {\bibfnamefont {S.~C.}\ \bibnamefont {Benjamin}},\ and\ \bibinfo {author} {\bibfnamefont {X.}~\bibnamefont {Yuan}},\ }\bibfield  {title} {\bibinfo {title} {Hybrid quantum-classical algorithms and quantum error mitigation},\ }\href {https://doi.org/10.7566/JPSJ.90.032001} {\bibfield  {journal} {\bibinfo  {journal} {J. Phys. Soc. Jpn.}\ }\textbf {\bibinfo {volume} {90}},\ \bibinfo {pages} {032001} (\bibinfo {year} {2021})}\BibitemShut {NoStop}%
\bibitem [{\citenamefont {Cai}\ \emph {et~al.}(2023)\citenamefont {Cai} \emph {et~al.}}]{cai2022quantum}%
  \BibitemOpen
  \bibfield  {author} {\bibinfo {author} {\bibfnamefont {Z.}~\bibnamefont {Cai}} \emph {et~al.},\ }\bibfield  {title} {\bibinfo {title} {Quantum error mitigation},\ }\href {https://doi.org/10.1103/RevModPhys.95.045005} {\bibfield  {journal} {\bibinfo  {journal} {Rev. Mod. Phys.}\ }\textbf {\bibinfo {volume} {95}},\ \bibinfo {pages} {045005} (\bibinfo {year} {2023})}\BibitemShut {NoStop}%
\bibitem [{\citenamefont {Wecker}\ \emph {et~al.}(2015{\natexlab{a}})\citenamefont {Wecker}, \citenamefont {Hastings},\ and\ \citenamefont {Troyer}}]{wecker2015progress}%
  \BibitemOpen
  \bibfield  {author} {\bibinfo {author} {\bibfnamefont {D.}~\bibnamefont {Wecker}}, \bibinfo {author} {\bibfnamefont {M.~B.}\ \bibnamefont {Hastings}},\ and\ \bibinfo {author} {\bibfnamefont {M.}~\bibnamefont {Troyer}},\ }\bibfield  {title} {\bibinfo {title} {Progress towards practical quantum variational algorithms},\ }\href {https://doi.org/10.1103/PhysRevA.92.042303} {\bibfield  {journal} {\bibinfo  {journal} {Phys. Rev. A}\ }\textbf {\bibinfo {volume} {92}},\ \bibinfo {pages} {042303} (\bibinfo {year} {2015}{\natexlab{a}})}\BibitemShut {NoStop}%
\bibitem [{\citenamefont {Gonthier}\ \emph {et~al.}(2022)\citenamefont {Gonthier}, \citenamefont {Radin}, \citenamefont {Buda}, \citenamefont {Doskocil}, \citenamefont {Abuan},\ and\ \citenamefont {Romero}}]{gonthier2022measurements}%
  \BibitemOpen
  \bibfield  {author} {\bibinfo {author} {\bibfnamefont {J.~F.}\ \bibnamefont {Gonthier}}, \bibinfo {author} {\bibfnamefont {M.~D.}\ \bibnamefont {Radin}}, \bibinfo {author} {\bibfnamefont {C.}~\bibnamefont {Buda}}, \bibinfo {author} {\bibfnamefont {E.~J.}\ \bibnamefont {Doskocil}}, \bibinfo {author} {\bibfnamefont {C.~M.}\ \bibnamefont {Abuan}},\ and\ \bibinfo {author} {\bibfnamefont {J.}~\bibnamefont {Romero}},\ }\bibfield  {title} {\bibinfo {title} {Measurements as a roadblock to near-term practical quantum advantage in chemistry: resource analysis},\ }\href {https://doi.org/10.1103/PhysRevResearch.4.033154} {\bibfield  {journal} {\bibinfo  {journal} {Phys. Rev. Res.}\ }\textbf {\bibinfo {volume} {4}},\ \bibinfo {pages} {033154} (\bibinfo {year} {2022})}\BibitemShut {NoStop}%
\bibitem [{\citenamefont {Huang}\ \emph {et~al.}(2020)\citenamefont {Huang}, \citenamefont {Kueng},\ and\ \citenamefont {Preskill}}]{huang2020predicting}%
  \BibitemOpen
  \bibfield  {author} {\bibinfo {author} {\bibfnamefont {H.-Y.}\ \bibnamefont {Huang}}, \bibinfo {author} {\bibfnamefont {R.}~\bibnamefont {Kueng}},\ and\ \bibinfo {author} {\bibfnamefont {J.}~\bibnamefont {Preskill}},\ }\bibfield  {title} {\bibinfo {title} {Predicting many properties of a quantum system from very few measurements},\ }\href {https://doi.org/10.1038/s41567-020-0932-7} {\bibfield  {journal} {\bibinfo  {journal} {Nat. Phys.}\ }\textbf {\bibinfo {volume} {16}},\ \bibinfo {pages} {1050} (\bibinfo {year} {2020})}\BibitemShut {NoStop}%
\bibitem [{\citenamefont {Paini}\ \emph {et~al.}(2021)\citenamefont {Paini}, \citenamefont {Kalev}, \citenamefont {Padilha},\ and\ \citenamefont {Ruck}}]{paini2021estimating}%
  \BibitemOpen
  \bibfield  {author} {\bibinfo {author} {\bibfnamefont {M.}~\bibnamefont {Paini}}, \bibinfo {author} {\bibfnamefont {A.}~\bibnamefont {Kalev}}, \bibinfo {author} {\bibfnamefont {D.}~\bibnamefont {Padilha}},\ and\ \bibinfo {author} {\bibfnamefont {B.}~\bibnamefont {Ruck}},\ }\bibfield  {title} {\bibinfo {title} {Estimating expectation values using approximate quantum states},\ }\href {https://doi.org/10.22331/q-2021-03-16-413} {\bibfield  {journal} {\bibinfo  {journal} {Quantum}\ }\textbf {\bibinfo {volume} {5}},\ \bibinfo {pages} {413} (\bibinfo {year} {2021})}\BibitemShut {NoStop}%
\bibitem [{\citenamefont {Cotler}\ and\ \citenamefont {Wilczek}(2020)}]{cotler2020quantum}%
  \BibitemOpen
  \bibfield  {author} {\bibinfo {author} {\bibfnamefont {J.}~\bibnamefont {Cotler}}\ and\ \bibinfo {author} {\bibfnamefont {F.}~\bibnamefont {Wilczek}},\ }\bibfield  {title} {\bibinfo {title} {Quantum overlapping tomography},\ }\href {https://doi.org/10.1103/PhysRevLett.124.100401} {\bibfield  {journal} {\bibinfo  {journal} {Phys. Rev. Lett.}\ }\textbf {\bibinfo {volume} {124}},\ \bibinfo {pages} {100401} (\bibinfo {year} {2020})}\BibitemShut {NoStop}%
\bibitem [{\citenamefont {Bonet-Monroig}\ \emph {et~al.}(2020)\citenamefont {Bonet-Monroig}, \citenamefont {Babbush},\ and\ \citenamefont {O'Brien}}]{bonet2019nearly}%
  \BibitemOpen
  \bibfield  {author} {\bibinfo {author} {\bibfnamefont {X.}~\bibnamefont {Bonet-Monroig}}, \bibinfo {author} {\bibfnamefont {R.}~\bibnamefont {Babbush}},\ and\ \bibinfo {author} {\bibfnamefont {T.~E.}\ \bibnamefont {O'Brien}},\ }\bibfield  {title} {\bibinfo {title} {Nearly optimal measurement scheduling for partial tomography of quantum states},\ }\href {https://doi.org/10.1103/PhysRevX.10.031064} {\bibfield  {journal} {\bibinfo  {journal} {Phys. Rev. X}\ }\textbf {\bibinfo {volume} {10}},\ \bibinfo {pages} {031064} (\bibinfo {year} {2020})}\BibitemShut {NoStop}%
\bibitem [{\citenamefont {Tilly}\ \emph {et~al.}(2022)\citenamefont {Tilly} \emph {et~al.}}]{tilly2022variational}%
  \BibitemOpen
  \bibfield  {author} {\bibinfo {author} {\bibfnamefont {J.}~\bibnamefont {Tilly}} \emph {et~al.},\ }\bibfield  {title} {\bibinfo {title} {The variational quantum eigensolver: a review of methods and best practices},\ }\href {https://doi.org/10.1016/j.physrep.2022.08.003} {\bibfield  {journal} {\bibinfo  {journal} {Physics Reports}\ }\textbf {\bibinfo {volume} {986}},\ \bibinfo {pages} {1} (\bibinfo {year} {2022})}\BibitemShut {NoStop}%
\bibitem [{\citenamefont {Zhao}(2023)}]{zhao2023learning}%
  \BibitemOpen
  \bibfield  {author} {\bibinfo {author} {\bibfnamefont {A.}~\bibnamefont {Zhao}},\ }\emph {\bibinfo {title} {Learning, Optimizing, and Simulating Fermions with Quantum Computers}},\ \href {https://arxiv.org/abs/2312.10399} {Ph.D. thesis},\ \bibinfo  {school} {University of New Mexico} (\bibinfo {year} {2023})\BibitemShut {NoStop}%
\bibitem [{\citenamefont {Sugiyama}\ \emph {et~al.}(2013)\citenamefont {Sugiyama}, \citenamefont {Turner},\ and\ \citenamefont {Murao}}]{sugiyama2013precision}%
  \BibitemOpen
  \bibfield  {author} {\bibinfo {author} {\bibfnamefont {T.}~\bibnamefont {Sugiyama}}, \bibinfo {author} {\bibfnamefont {P.~S.}\ \bibnamefont {Turner}},\ and\ \bibinfo {author} {\bibfnamefont {M.}~\bibnamefont {Murao}},\ }\bibfield  {title} {\bibinfo {title} {Precision-guaranteed quantum tomography},\ }\href {https://doi.org/10.1103/PhysRevLett.111.160406} {\bibfield  {journal} {\bibinfo  {journal} {Phys. Rev. Lett.}\ }\textbf {\bibinfo {volume} {111}},\ \bibinfo {pages} {160406} (\bibinfo {year} {2013})}\BibitemShut {NoStop}%
\bibitem [{\citenamefont {Gu{\c{t}}{\u{a}}}\ \emph {et~al.}(2020)\citenamefont {Gu{\c{t}}{\u{a}}}, \citenamefont {Kahn}, \citenamefont {Kueng},\ and\ \citenamefont {Tropp}}]{guta2020fast}%
  \BibitemOpen
  \bibfield  {author} {\bibinfo {author} {\bibfnamefont {M.}~\bibnamefont {Gu{\c{t}}{\u{a}}}}, \bibinfo {author} {\bibfnamefont {J.}~\bibnamefont {Kahn}}, \bibinfo {author} {\bibfnamefont {R.}~\bibnamefont {Kueng}},\ and\ \bibinfo {author} {\bibfnamefont {J.~A.}\ \bibnamefont {Tropp}},\ }\bibfield  {title} {\bibinfo {title} {Fast state tomography with optimal error bounds},\ }\href {https://doi.org/10.1088/1751-8121/ab8111} {\bibfield  {journal} {\bibinfo  {journal} {J. Phys. A: Math. Theor.}\ }\textbf {\bibinfo {volume} {53}},\ \bibinfo {pages} {204001} (\bibinfo {year} {2020})}\BibitemShut {NoStop}%
\bibitem [{\citenamefont {Aaronson}(2020)}]{aaronson2020shadow}%
  \BibitemOpen
  \bibfield  {author} {\bibinfo {author} {\bibfnamefont {S.}~\bibnamefont {Aaronson}},\ }\bibfield  {title} {\bibinfo {title} {Shadow tomography of quantum states},\ }\href {https://doi.org/10.1137/18M120275X} {\bibfield  {journal} {\bibinfo  {journal} {SIAM J. Comput.}\ }\textbf {\bibinfo {volume} {49}},\ \bibinfo {eid} {STOC18-368} (\bibinfo {year} {2020})}\BibitemShut {NoStop}%
\bibitem [{\citenamefont {Aaronson}\ and\ \citenamefont {Rothblum}(2019)}]{aaronson2019gentle}%
  \BibitemOpen
  \bibfield  {author} {\bibinfo {author} {\bibfnamefont {S.}~\bibnamefont {Aaronson}}\ and\ \bibinfo {author} {\bibfnamefont {G.~N.}\ \bibnamefont {Rothblum}},\ }\bibfield  {title} {\bibinfo {title} {Gentle measurement of quantum states and differential privacy},\ }in\ \href {https://doi.org/10.1145/3313276.3316378} {\emph {\bibinfo {booktitle} {Proceedings of the 51st Annual ACM SIGACT Symposium on Theory of Computing}}}\ (\bibinfo  {publisher} {Association for Computing Machinery},\ \bibinfo {address} {New York},\ \bibinfo {year} {2019})\ pp.\ \bibinfo {pages} {322--333}\BibitemShut {NoStop}%
\bibitem [{\citenamefont {Elben}\ \emph {et~al.}(2020)\citenamefont {Elben} \emph {et~al.}}]{elben2020mixed}%
  \BibitemOpen
  \bibfield  {author} {\bibinfo {author} {\bibfnamefont {A.}~\bibnamefont {Elben}} \emph {et~al.},\ }\bibfield  {title} {\bibinfo {title} {Mixed-state entanglement from local randomized measurements},\ }\href {https://doi.org/10.1103/PhysRevLett.125.200501} {\bibfield  {journal} {\bibinfo  {journal} {Phys. Rev. Lett.}\ }\textbf {\bibinfo {volume} {125}},\ \bibinfo {pages} {200501} (\bibinfo {year} {2020})}\BibitemShut {NoStop}%
\bibitem [{\citenamefont {Rath}\ \emph {et~al.}(2021)\citenamefont {Rath}, \citenamefont {Branciard}, \citenamefont {Minguzzi},\ and\ \citenamefont {Vermersch}}]{rath2021quantum}%
  \BibitemOpen
  \bibfield  {author} {\bibinfo {author} {\bibfnamefont {A.}~\bibnamefont {Rath}}, \bibinfo {author} {\bibfnamefont {C.}~\bibnamefont {Branciard}}, \bibinfo {author} {\bibfnamefont {A.}~\bibnamefont {Minguzzi}},\ and\ \bibinfo {author} {\bibfnamefont {B.}~\bibnamefont {Vermersch}},\ }\bibfield  {title} {\bibinfo {title} {Quantum {F}isher information from randomized measurements},\ }\href {https://doi.org/10.1103/PhysRevLett.127.260501} {\bibfield  {journal} {\bibinfo  {journal} {Phys. Rev. Lett.}\ }\textbf {\bibinfo {volume} {127}},\ \bibinfo {pages} {260501} (\bibinfo {year} {2021})}\BibitemShut {NoStop}%
\bibitem [{\citenamefont {Vitale}\ \emph {et~al.}(2023)\citenamefont {Vitale}, \citenamefont {Rath}, \citenamefont {Jurcevic}, \citenamefont {Elben}, \citenamefont {Branciard},\ and\ \citenamefont {Vermersch}}]{vitale2023estimation}%
  \BibitemOpen
  \bibfield  {author} {\bibinfo {author} {\bibfnamefont {V.}~\bibnamefont {Vitale}}, \bibinfo {author} {\bibfnamefont {A.}~\bibnamefont {Rath}}, \bibinfo {author} {\bibfnamefont {P.}~\bibnamefont {Jurcevic}}, \bibinfo {author} {\bibfnamefont {A.}~\bibnamefont {Elben}}, \bibinfo {author} {\bibfnamefont {C.}~\bibnamefont {Branciard}},\ and\ \bibinfo {author} {\bibfnamefont {B.}~\bibnamefont {Vermersch}},\ }\bibfield  {title} {\bibinfo {title} {Estimation of the quantum {F}isher information on a quantum processor},\ }\href {https://arxiv.org/abs/2307.16882} {\bibfield  {journal} {\bibinfo  {journal} {arXiv:2307.16882}\ } (\bibinfo {year} {2023})}\BibitemShut {NoStop}%
\bibitem [{\citenamefont {Levy}\ \emph {et~al.}(2021)\citenamefont {Levy}, \citenamefont {Luo},\ and\ \citenamefont {Clark}}]{levy2021classical}%
  \BibitemOpen
  \bibfield  {author} {\bibinfo {author} {\bibfnamefont {R.}~\bibnamefont {Levy}}, \bibinfo {author} {\bibfnamefont {D.}~\bibnamefont {Luo}},\ and\ \bibinfo {author} {\bibfnamefont {B.~K.}\ \bibnamefont {Clark}},\ }\bibfield  {title} {\bibinfo {title} {Classical shadows for quantum process tomography on near-term quantum computers},\ }\href {https://arxiv.org/abs/2110.02965} {\bibfield  {journal} {\bibinfo  {journal} {arXiv:2110.02965}\ } (\bibinfo {year} {2021})}\BibitemShut {NoStop}%
\bibitem [{\citenamefont {Kunjummen}\ \emph {et~al.}(2023)\citenamefont {Kunjummen}, \citenamefont {Tran}, \citenamefont {Carney},\ and\ \citenamefont {Taylor}}]{kunjummen2023shadow}%
  \BibitemOpen
  \bibfield  {author} {\bibinfo {author} {\bibfnamefont {J.}~\bibnamefont {Kunjummen}}, \bibinfo {author} {\bibfnamefont {M.~C.}\ \bibnamefont {Tran}}, \bibinfo {author} {\bibfnamefont {D.}~\bibnamefont {Carney}},\ and\ \bibinfo {author} {\bibfnamefont {J.~M.}\ \bibnamefont {Taylor}},\ }\bibfield  {title} {\bibinfo {title} {Shadow process tomography of quantum channels},\ }\href {https://doi.org/10.1103/PhysRevA.107.042403} {\bibfield  {journal} {\bibinfo  {journal} {Phys. Rev. A}\ }\textbf {\bibinfo {volume} {107}},\ \bibinfo {pages} {042403} (\bibinfo {year} {2023})}\BibitemShut {NoStop}%
\bibitem [{\citenamefont {Sack}\ \emph {et~al.}(2022)\citenamefont {Sack}, \citenamefont {Medina}, \citenamefont {Michailidis}, \citenamefont {Kueng},\ and\ \citenamefont {Serbyn}}]{sack2022avoiding}%
  \BibitemOpen
  \bibfield  {author} {\bibinfo {author} {\bibfnamefont {S.~H.}\ \bibnamefont {Sack}}, \bibinfo {author} {\bibfnamefont {R.~A.}\ \bibnamefont {Medina}}, \bibinfo {author} {\bibfnamefont {A.~A.}\ \bibnamefont {Michailidis}}, \bibinfo {author} {\bibfnamefont {R.}~\bibnamefont {Kueng}},\ and\ \bibinfo {author} {\bibfnamefont {M.}~\bibnamefont {Serbyn}},\ }\bibfield  {title} {\bibinfo {title} {Avoiding barren plateaus using classical shadows},\ }\href {https://doi.org/10.1103/PRXQuantum.3.020365} {\bibfield  {journal} {\bibinfo  {journal} {PRX Quantum}\ }\textbf {\bibinfo {volume} {3}},\ \bibinfo {pages} {020365} (\bibinfo {year} {2022})}\BibitemShut {NoStop}%
\bibitem [{\citenamefont {Boyd}\ and\ \citenamefont {Koczor}(2022)}]{boyd2022training}%
  \BibitemOpen
  \bibfield  {author} {\bibinfo {author} {\bibfnamefont {G.}~\bibnamefont {Boyd}}\ and\ \bibinfo {author} {\bibfnamefont {B.}~\bibnamefont {Koczor}},\ }\bibfield  {title} {\bibinfo {title} {Training variational quantum circuits with {CoVaR}: {C}ovariance root finding with classical shadows},\ }\href {https://doi.org/10.1103/PhysRevX.12.041022} {\bibfield  {journal} {\bibinfo  {journal} {Phys. Rev. X}\ }\textbf {\bibinfo {volume} {12}},\ \bibinfo {pages} {041022} (\bibinfo {year} {2022})}\BibitemShut {NoStop}%
\bibitem [{\citenamefont {Chan}\ \emph {et~al.}(2022)\citenamefont {Chan}, \citenamefont {Meister}, \citenamefont {Goh},\ and\ \citenamefont {Koczor}}]{chan2022algorithmic}%
  \BibitemOpen
  \bibfield  {author} {\bibinfo {author} {\bibfnamefont {H.~H.~S.}\ \bibnamefont {Chan}}, \bibinfo {author} {\bibfnamefont {R.}~\bibnamefont {Meister}}, \bibinfo {author} {\bibfnamefont {M.~L.}\ \bibnamefont {Goh}},\ and\ \bibinfo {author} {\bibfnamefont {B.}~\bibnamefont {Koczor}},\ }\bibfield  {title} {\bibinfo {title} {Algorithmic shadow spectroscopy},\ }\href {https://arxiv.org/abs/2212.11036} {\bibfield  {journal} {\bibinfo  {journal} {arXiv:2212.11036}\ } (\bibinfo {year} {2022})}\BibitemShut {NoStop}%
\bibitem [{\citenamefont {Zhao}\ \emph {et~al.}(2021)\citenamefont {Zhao}, \citenamefont {Rubin},\ and\ \citenamefont {Miyake}}]{zhao2021fermionic}%
  \BibitemOpen
  \bibfield  {author} {\bibinfo {author} {\bibfnamefont {A.}~\bibnamefont {Zhao}}, \bibinfo {author} {\bibfnamefont {N.~C.}\ \bibnamefont {Rubin}},\ and\ \bibinfo {author} {\bibfnamefont {A.}~\bibnamefont {Miyake}},\ }\bibfield  {title} {\bibinfo {title} {Fermionic partial tomography via classical shadows},\ }\href {https://doi.org/10.1103/PhysRevLett.127.110504} {\bibfield  {journal} {\bibinfo  {journal} {Phys. Rev. Lett.}\ }\textbf {\bibinfo {volume} {127}},\ \bibinfo {pages} {110504} (\bibinfo {year} {2021})}\BibitemShut {NoStop}%
\bibitem [{\citenamefont {Wan}\ \emph {et~al.}(2023)\citenamefont {Wan}, \citenamefont {Huggins}, \citenamefont {Lee},\ and\ \citenamefont {Babbush}}]{wan2023matchgate}%
  \BibitemOpen
  \bibfield  {author} {\bibinfo {author} {\bibfnamefont {K.}~\bibnamefont {Wan}}, \bibinfo {author} {\bibfnamefont {W.~J.}\ \bibnamefont {Huggins}}, \bibinfo {author} {\bibfnamefont {J.}~\bibnamefont {Lee}},\ and\ \bibinfo {author} {\bibfnamefont {R.}~\bibnamefont {Babbush}},\ }\bibfield  {title} {\bibinfo {title} {Matchgate shadows for fermionic quantum simulation},\ }\href {https://doi.org/10.1007/s00220-023-04844-0} {\bibfield  {journal} {\bibinfo  {journal} {Commun. Math. Phys.}\ }\textbf {\bibinfo {volume} {404}},\ \bibinfo {pages} {629} (\bibinfo {year} {2023})}\BibitemShut {NoStop}%
\bibitem [{\citenamefont {O'Gorman}(2022)}]{ogorman2022fermionic}%
  \BibitemOpen
  \bibfield  {author} {\bibinfo {author} {\bibfnamefont {B.}~\bibnamefont {O'Gorman}},\ }\bibfield  {title} {\bibinfo {title} {Fermionic tomography and learning},\ }\href {https://arxiv.org/abs/2207.14787} {\bibfield  {journal} {\bibinfo  {journal} {arXiv:2207.14787}\ } (\bibinfo {year} {2022})}\BibitemShut {NoStop}%
\bibitem [{\citenamefont {Low}(2022)}]{low2022classical}%
  \BibitemOpen
  \bibfield  {author} {\bibinfo {author} {\bibfnamefont {G.~H.}\ \bibnamefont {Low}},\ }\bibfield  {title} {\bibinfo {title} {Classical shadows of fermions with particle number symmetry},\ }\href {https://arxiv.org/abs/2208.08964} {\bibfield  {journal} {\bibinfo  {journal} {arXiv:2208.08964}\ } (\bibinfo {year} {2022})}\BibitemShut {NoStop}%
\bibitem [{\citenamefont {Babbush}\ \emph {et~al.}(2023)\citenamefont {Babbush} \emph {et~al.}}]{babbush2023quantum}%
  \BibitemOpen
  \bibfield  {author} {\bibinfo {author} {\bibfnamefont {R.}~\bibnamefont {Babbush}} \emph {et~al.},\ }\bibfield  {title} {\bibinfo {title} {Quantum simulation of exact electron dynamics can be more efficient than classical mean-field methods},\ }\href {https://doi.org/10.1038/s41467-023-39024-0} {\bibfield  {journal} {\bibinfo  {journal} {Nat. Commun.}\ }\textbf {\bibinfo {volume} {14}},\ \bibinfo {pages} {4058} (\bibinfo {year} {2023})}\BibitemShut {NoStop}%
\bibitem [{\citenamefont {Denzler}\ \emph {et~al.}(2023)\citenamefont {Denzler}, \citenamefont {Mele}, \citenamefont {Derbyshire}, \citenamefont {Guaita},\ and\ \citenamefont {Eisert}}]{denzler2023learning}%
  \BibitemOpen
  \bibfield  {author} {\bibinfo {author} {\bibfnamefont {J.}~\bibnamefont {Denzler}}, \bibinfo {author} {\bibfnamefont {A.~A.}\ \bibnamefont {Mele}}, \bibinfo {author} {\bibfnamefont {E.}~\bibnamefont {Derbyshire}}, \bibinfo {author} {\bibfnamefont {T.}~\bibnamefont {Guaita}},\ and\ \bibinfo {author} {\bibfnamefont {J.}~\bibnamefont {Eisert}},\ }\bibfield  {title} {\bibinfo {title} {Learning fermionic correlations by evolving with random translationally invariant {H}amiltonians},\ }\href {https://arxiv.org/abs/2309.12933} {\bibfield  {journal} {\bibinfo  {journal} {arXiv:2309.12933}\ } (\bibinfo {year} {2023})}\BibitemShut {NoStop}%
\bibitem [{\citenamefont {Gu}\ \emph {et~al.}(2023)\citenamefont {Gu}, \citenamefont {Yuan},\ and\ \citenamefont {Wu}}]{gu2023efficient}%
  \BibitemOpen
  \bibfield  {author} {\bibinfo {author} {\bibfnamefont {T.}~\bibnamefont {Gu}}, \bibinfo {author} {\bibfnamefont {X.}~\bibnamefont {Yuan}},\ and\ \bibinfo {author} {\bibfnamefont {B.}~\bibnamefont {Wu}},\ }\bibfield  {title} {\bibinfo {title} {Efficient measurement schemes for bosonic systems},\ }\href {https://doi.org/10.1088/2058-9565/ace6cd} {\bibfield  {journal} {\bibinfo  {journal} {Quantum Sci. Technol.}\ }\textbf {\bibinfo {volume} {8}},\ \bibinfo {pages} {045008} (\bibinfo {year} {2023})}\BibitemShut {NoStop}%
\bibitem [{\citenamefont {Becker}\ \emph {et~al.}(2024)\citenamefont {Becker}, \citenamefont {Datta}, \citenamefont {Lami},\ and\ \citenamefont {Rouz{\'e}}}]{becker2022classical}%
  \BibitemOpen
  \bibfield  {author} {\bibinfo {author} {\bibfnamefont {S.}~\bibnamefont {Becker}}, \bibinfo {author} {\bibfnamefont {N.}~\bibnamefont {Datta}}, \bibinfo {author} {\bibfnamefont {L.}~\bibnamefont {Lami}},\ and\ \bibinfo {author} {\bibfnamefont {C.}~\bibnamefont {Rouz{\'e}}},\ }\bibfield  {title} {\bibinfo {title} {Classical shadow tomography for continuous variables quantum systems},\ }\href {https://doi.org/10.1109/TIT.2024.3357972} {\bibfield  {journal} {\bibinfo  {journal} {IEEE Trans. Inf. Theory}\ }\textbf {\bibinfo {volume} {70}},\ \bibinfo {pages} {3427} (\bibinfo {year} {2024})}\BibitemShut {NoStop}%
\bibitem [{\citenamefont {Elben}\ \emph {et~al.}(2023)\citenamefont {Elben} \emph {et~al.}}]{elben2023randomized}%
  \BibitemOpen
  \bibfield  {author} {\bibinfo {author} {\bibfnamefont {A.}~\bibnamefont {Elben}} \emph {et~al.},\ }\bibfield  {title} {\bibinfo {title} {The randomized measurement toolbox},\ }\href {https://doi.org/10.1038/s42254-022-00535-2} {\bibfield  {journal} {\bibinfo  {journal} {Nat. Rev. Phys.}\ }\textbf {\bibinfo {volume} {5}},\ \bibinfo {pages} {9} (\bibinfo {year} {2023})}\BibitemShut {NoStop}%
\bibitem [{\citenamefont {Seif}\ \emph {et~al.}(2023)\citenamefont {Seif}, \citenamefont {Cian}, \citenamefont {Zhou}, \citenamefont {Chen},\ and\ \citenamefont {Jiang}}]{seif2023shadow}%
  \BibitemOpen
  \bibfield  {author} {\bibinfo {author} {\bibfnamefont {A.}~\bibnamefont {Seif}}, \bibinfo {author} {\bibfnamefont {Z.-P.}\ \bibnamefont {Cian}}, \bibinfo {author} {\bibfnamefont {S.}~\bibnamefont {Zhou}}, \bibinfo {author} {\bibfnamefont {S.}~\bibnamefont {Chen}},\ and\ \bibinfo {author} {\bibfnamefont {L.}~\bibnamefont {Jiang}},\ }\bibfield  {title} {\bibinfo {title} {Shadow distillation: {Q}uantum error mitigation with classical shadows for near-term quantum processors},\ }\href {https://doi.org/10.1103/PRXQuantum.4.010303} {\bibfield  {journal} {\bibinfo  {journal} {PRX Quantum}\ }\textbf {\bibinfo {volume} {4}},\ \bibinfo {pages} {010303} (\bibinfo {year} {2023})}\BibitemShut {NoStop}%
\bibitem [{\citenamefont {Hu}\ \emph {et~al.}(2022)\citenamefont {Hu}, \citenamefont {LaRose}, \citenamefont {You}, \citenamefont {Rieffel},\ and\ \citenamefont {Wang}}]{hu2022logical}%
  \BibitemOpen
  \bibfield  {author} {\bibinfo {author} {\bibfnamefont {H.-Y.}\ \bibnamefont {Hu}}, \bibinfo {author} {\bibfnamefont {R.}~\bibnamefont {LaRose}}, \bibinfo {author} {\bibfnamefont {Y.-Z.}\ \bibnamefont {You}}, \bibinfo {author} {\bibfnamefont {E.}~\bibnamefont {Rieffel}},\ and\ \bibinfo {author} {\bibfnamefont {Z.}~\bibnamefont {Wang}},\ }\bibfield  {title} {\bibinfo {title} {Logical shadow tomography: {E}fficient estimation of error-mitigated observables},\ }\href {https://arxiv.org/abs/2203.07263} {\bibfield  {journal} {\bibinfo  {journal} {arXiv:2203.07263}\ } (\bibinfo {year} {2022})}\BibitemShut {NoStop}%
\bibitem [{\citenamefont {McClean}\ \emph {et~al.}(2020{\natexlab{a}})\citenamefont {McClean}, \citenamefont {Jiang}, \citenamefont {Rubin}, \citenamefont {Babbush},\ and\ \citenamefont {Neven}}]{mcclean2020decoding}%
  \BibitemOpen
  \bibfield  {author} {\bibinfo {author} {\bibfnamefont {J.~R.}\ \bibnamefont {McClean}}, \bibinfo {author} {\bibfnamefont {Z.}~\bibnamefont {Jiang}}, \bibinfo {author} {\bibfnamefont {N.~C.}\ \bibnamefont {Rubin}}, \bibinfo {author} {\bibfnamefont {R.}~\bibnamefont {Babbush}},\ and\ \bibinfo {author} {\bibfnamefont {H.}~\bibnamefont {Neven}},\ }\bibfield  {title} {\bibinfo {title} {Decoding quantum errors with subspace expansions},\ }\href {https://doi.org/10.1038/s41467-020-14341-w} {\bibfield  {journal} {\bibinfo  {journal} {Nat. Commun.}\ }\textbf {\bibinfo {volume} {11}},\ \bibinfo {pages} {636} (\bibinfo {year} {2020}{\natexlab{a}})}\BibitemShut {NoStop}%
\bibitem [{\citenamefont {Koczor}(2021)}]{koczor2021exponential}%
  \BibitemOpen
  \bibfield  {author} {\bibinfo {author} {\bibfnamefont {B.}~\bibnamefont {Koczor}},\ }\bibfield  {title} {\bibinfo {title} {Exponential error suppression for near-term quantum devices},\ }\href {https://doi.org/10.1103/PhysRevX.11.031057} {\bibfield  {journal} {\bibinfo  {journal} {Phys. Rev. X}\ }\textbf {\bibinfo {volume} {11}},\ \bibinfo {pages} {031057} (\bibinfo {year} {2021})}\BibitemShut {NoStop}%
\bibitem [{\citenamefont {Huggins}\ \emph {et~al.}(2021{\natexlab{a}})\citenamefont {Huggins} \emph {et~al.}}]{huggins2021virtual}%
  \BibitemOpen
  \bibfield  {author} {\bibinfo {author} {\bibfnamefont {W.~J.}\ \bibnamefont {Huggins}} \emph {et~al.},\ }\bibfield  {title} {\bibinfo {title} {Virtual distillation for quantum error mitigation},\ }\href {https://doi.org/10.1103/PhysRevX.11.041036} {\bibfield  {journal} {\bibinfo  {journal} {Phys. Rev. X}\ }\textbf {\bibinfo {volume} {11}},\ \bibinfo {pages} {041036} (\bibinfo {year} {2021}{\natexlab{a}})}\BibitemShut {NoStop}%
\bibitem [{\citenamefont {Jnane}\ \emph {et~al.}(2024)\citenamefont {Jnane}, \citenamefont {Steinberg}, \citenamefont {Cai}, \citenamefont {Nguyen},\ and\ \citenamefont {Koczor}}]{jnane2023quantum}%
  \BibitemOpen
  \bibfield  {author} {\bibinfo {author} {\bibfnamefont {H.}~\bibnamefont {Jnane}}, \bibinfo {author} {\bibfnamefont {J.}~\bibnamefont {Steinberg}}, \bibinfo {author} {\bibfnamefont {Z.}~\bibnamefont {Cai}}, \bibinfo {author} {\bibfnamefont {H.~C.}\ \bibnamefont {Nguyen}},\ and\ \bibinfo {author} {\bibfnamefont {B.}~\bibnamefont {Koczor}},\ }\bibfield  {title} {\bibinfo {title} {Quantum error mitigated classical shadows},\ }\href {https://doi.org/10.1103/PRXQuantum.5.010324} {\bibfield  {journal} {\bibinfo  {journal} {PRX Quantum}\ }\textbf {\bibinfo {volume} {5}},\ \bibinfo {pages} {010324} (\bibinfo {year} {2024})}\BibitemShut {NoStop}%
\bibitem [{\citenamefont {Temme}\ \emph {et~al.}(2017)\citenamefont {Temme}, \citenamefont {Bravyi},\ and\ \citenamefont {Gambetta}}]{temme2017error}%
  \BibitemOpen
  \bibfield  {author} {\bibinfo {author} {\bibfnamefont {K.}~\bibnamefont {Temme}}, \bibinfo {author} {\bibfnamefont {S.}~\bibnamefont {Bravyi}},\ and\ \bibinfo {author} {\bibfnamefont {J.~M.}\ \bibnamefont {Gambetta}},\ }\bibfield  {title} {\bibinfo {title} {Error mitigation for short-depth quantum circuits},\ }\href {https://doi.org/10.1103/PhysRevLett.119.180509} {\bibfield  {journal} {\bibinfo  {journal} {Phys. Rev. Lett.}\ }\textbf {\bibinfo {volume} {119}},\ \bibinfo {pages} {180509} (\bibinfo {year} {2017})}\BibitemShut {NoStop}%
\bibitem [{\citenamefont {Chen}\ \emph {et~al.}(2021)\citenamefont {Chen}, \citenamefont {Yu}, \citenamefont {Zeng},\ and\ \citenamefont {Flammia}}]{chen2021robust}%
  \BibitemOpen
  \bibfield  {author} {\bibinfo {author} {\bibfnamefont {S.}~\bibnamefont {Chen}}, \bibinfo {author} {\bibfnamefont {W.}~\bibnamefont {Yu}}, \bibinfo {author} {\bibfnamefont {P.}~\bibnamefont {Zeng}},\ and\ \bibinfo {author} {\bibfnamefont {S.~T.}\ \bibnamefont {Flammia}},\ }\bibfield  {title} {\bibinfo {title} {Robust shadow estimation},\ }\href {https://doi.org/10.1103/PRXQuantum.2.030348} {\bibfield  {journal} {\bibinfo  {journal} {PRX Quantum}\ }\textbf {\bibinfo {volume} {2}},\ \bibinfo {pages} {030348} (\bibinfo {year} {2021})}\BibitemShut {NoStop}%
\bibitem [{\citenamefont {Koh}\ and\ \citenamefont {Grewal}(2022)}]{koh2022classical}%
  \BibitemOpen
  \bibfield  {author} {\bibinfo {author} {\bibfnamefont {D.~E.}\ \bibnamefont {Koh}}\ and\ \bibinfo {author} {\bibfnamefont {S.}~\bibnamefont {Grewal}},\ }\bibfield  {title} {\bibinfo {title} {Classical shadows with noise},\ }\href {https://doi.org/10.22331/q-2022-08-16-776} {\bibfield  {journal} {\bibinfo  {journal} {Quantum}\ }\textbf {\bibinfo {volume} {6}},\ \bibinfo {pages} {776} (\bibinfo {year} {2022})}\BibitemShut {NoStop}%
\bibitem [{\citenamefont {Karalekas}\ \emph {et~al.}(2020)\citenamefont {Karalekas}, \citenamefont {Tezak}, \citenamefont {Peterson}, \citenamefont {Ryan}, \citenamefont {Da~Silva},\ and\ \citenamefont {Smith}}]{karalekas2020quantum}%
  \BibitemOpen
  \bibfield  {author} {\bibinfo {author} {\bibfnamefont {P.~J.}\ \bibnamefont {Karalekas}}, \bibinfo {author} {\bibfnamefont {N.~A.}\ \bibnamefont {Tezak}}, \bibinfo {author} {\bibfnamefont {E.~C.}\ \bibnamefont {Peterson}}, \bibinfo {author} {\bibfnamefont {C.~A.}\ \bibnamefont {Ryan}}, \bibinfo {author} {\bibfnamefont {M.~P.}\ \bibnamefont {Da~Silva}},\ and\ \bibinfo {author} {\bibfnamefont {R.~S.}\ \bibnamefont {Smith}},\ }\bibfield  {title} {\bibinfo {title} {A quantum-classical cloud platform optimized for variational hybrid algorithms},\ }\href {https://doi.org/10.1088/2058-9565/ab7559} {\bibfield  {journal} {\bibinfo  {journal} {Quantum Sci. Technol.}\ }\textbf {\bibinfo {volume} {5}},\ \bibinfo {pages} {024003} (\bibinfo {year} {2020})}\BibitemShut {NoStop}%
\bibitem [{\citenamefont {Van Den~Berg}\ \emph {et~al.}(2022)\citenamefont {Van Den~Berg}, \citenamefont {Minev},\ and\ \citenamefont {Temme}}]{van2022model}%
  \BibitemOpen
  \bibfield  {author} {\bibinfo {author} {\bibfnamefont {E.}~\bibnamefont {Van Den~Berg}}, \bibinfo {author} {\bibfnamefont {Z.~K.}\ \bibnamefont {Minev}},\ and\ \bibinfo {author} {\bibfnamefont {K.}~\bibnamefont {Temme}},\ }\bibfield  {title} {\bibinfo {title} {Model-free readout-error mitigation for quantum expectation values},\ }\href {https://doi.org/10.1103/PhysRevA.105.032620} {\bibfield  {journal} {\bibinfo  {journal} {Phys. Rev. A}\ }\textbf {\bibinfo {volume} {105}},\ \bibinfo {pages} {032620} (\bibinfo {year} {2022})}\BibitemShut {NoStop}%
\bibitem [{\citenamefont {Arrasmith}\ \emph {et~al.}(2023)\citenamefont {Arrasmith}, \citenamefont {Patterson}, \citenamefont {Boughton},\ and\ \citenamefont {Paini}}]{arrasmith2023development}%
  \BibitemOpen
  \bibfield  {author} {\bibinfo {author} {\bibfnamefont {A.}~\bibnamefont {Arrasmith}}, \bibinfo {author} {\bibfnamefont {A.}~\bibnamefont {Patterson}}, \bibinfo {author} {\bibfnamefont {A.}~\bibnamefont {Boughton}},\ and\ \bibinfo {author} {\bibfnamefont {M.}~\bibnamefont {Paini}},\ }\bibfield  {title} {\bibinfo {title} {Development and demonstration of an efficient readout error mitigation technique for use in {NISQ} algorithms},\ }\href {https://arxiv.org/abs/2303.17741} {\bibfield  {journal} {\bibinfo  {journal} {arXiv:2303.17741}\ } (\bibinfo {year} {2023})}\BibitemShut {NoStop}%
\bibitem [{\citenamefont {Bonet-Monroig}\ \emph {et~al.}(2018)\citenamefont {Bonet-Monroig}, \citenamefont {Sagastizabal}, \citenamefont {Singh},\ and\ \citenamefont {O'Brien}}]{bonet2018low}%
  \BibitemOpen
  \bibfield  {author} {\bibinfo {author} {\bibfnamefont {X.}~\bibnamefont {Bonet-Monroig}}, \bibinfo {author} {\bibfnamefont {R.}~\bibnamefont {Sagastizabal}}, \bibinfo {author} {\bibfnamefont {M.}~\bibnamefont {Singh}},\ and\ \bibinfo {author} {\bibfnamefont {T.~E.}\ \bibnamefont {O'Brien}},\ }\bibfield  {title} {\bibinfo {title} {Low-cost error mitigation by symmetry verification},\ }\href {https://doi.org/10.1103/PhysRevA.98.062339} {\bibfield  {journal} {\bibinfo  {journal} {Phys. Rev. A}\ }\textbf {\bibinfo {volume} {98}},\ \bibinfo {pages} {062339} (\bibinfo {year} {2018})}\BibitemShut {NoStop}%
\bibitem [{\citenamefont {McArdle}\ \emph {et~al.}(2019)\citenamefont {McArdle}, \citenamefont {Yuan},\ and\ \citenamefont {Benjamin}}]{mcardle2019error}%
  \BibitemOpen
  \bibfield  {author} {\bibinfo {author} {\bibfnamefont {S.}~\bibnamefont {McArdle}}, \bibinfo {author} {\bibfnamefont {X.}~\bibnamefont {Yuan}},\ and\ \bibinfo {author} {\bibfnamefont {S.}~\bibnamefont {Benjamin}},\ }\bibfield  {title} {\bibinfo {title} {Error-mitigated digital quantum simulation},\ }\href {https://doi.org/10.1103/PhysRevLett.122.180501} {\bibfield  {journal} {\bibinfo  {journal} {Phys. Rev. Lett.}\ }\textbf {\bibinfo {volume} {122}},\ \bibinfo {pages} {180501} (\bibinfo {year} {2019})}\BibitemShut {NoStop}%
\bibitem [{\citenamefont {Cai}(2021)}]{cai2021quantum}%
  \BibitemOpen
  \bibfield  {author} {\bibinfo {author} {\bibfnamefont {Z.}~\bibnamefont {Cai}},\ }\bibfield  {title} {\bibinfo {title} {Quantum error mitigation using symmetry expansion},\ }\href {https://doi.org/10.22331/q-2021-09-21-548} {\bibfield  {journal} {\bibinfo  {journal} {Quantum}\ }\textbf {\bibinfo {volume} {5}},\ \bibinfo {pages} {548} (\bibinfo {year} {2021})}\BibitemShut {NoStop}%
\bibitem [{\citenamefont {Isakov}\ \emph {et~al.}(2021)\citenamefont {Isakov} \emph {et~al.}}]{isakov2021simulations}%
  \BibitemOpen
  \bibfield  {author} {\bibinfo {author} {\bibfnamefont {S.~V.}\ \bibnamefont {Isakov}} \emph {et~al.},\ }\bibfield  {title} {\bibinfo {title} {Simulations of quantum circuits with approximate noise using qsim and {C}irq},\ }\href {https://arxiv.org/abs/2111.02396} {\bibfield  {journal} {\bibinfo  {journal} {arXiv:2111.02396}\ } (\bibinfo {year} {2021})}\BibitemShut {NoStop}%
\bibitem [{\citenamefont {Fulton}\ and\ \citenamefont {Harris}(2004)}]{fulton2004representation}%
  \BibitemOpen
  \bibfield  {author} {\bibinfo {author} {\bibfnamefont {W.}~\bibnamefont {Fulton}}\ and\ \bibinfo {author} {\bibfnamefont {J.}~\bibnamefont {Harris}},\ }\href@noop {} {\emph {\bibinfo {title} {Representation Theory:~A First Course}}}\ (\bibinfo  {publisher} {Springer-Verlag},\ \bibinfo {address} {New York},\ \bibinfo {year} {2004})\BibitemShut {NoStop}%
\bibitem [{\citenamefont {Takagi}\ \emph {et~al.}(2022)\citenamefont {Takagi}, \citenamefont {Endo}, \citenamefont {Minagawa},\ and\ \citenamefont {Gu}}]{takagi2022fundamental}%
  \BibitemOpen
  \bibfield  {author} {\bibinfo {author} {\bibfnamefont {R.}~\bibnamefont {Takagi}}, \bibinfo {author} {\bibfnamefont {S.}~\bibnamefont {Endo}}, \bibinfo {author} {\bibfnamefont {S.}~\bibnamefont {Minagawa}},\ and\ \bibinfo {author} {\bibfnamefont {M.}~\bibnamefont {Gu}},\ }\bibfield  {title} {\bibinfo {title} {Fundamental limits of quantum error mitigation},\ }\href {https://doi.org/10.1038/s41534-022-00618-z} {\bibfield  {journal} {\bibinfo  {journal} {npj Quantum Inf.}\ }\textbf {\bibinfo {volume} {8}},\ \bibinfo {pages} {114} (\bibinfo {year} {2022})}\BibitemShut {NoStop}%
\bibitem [{\citenamefont {Takagi}\ \emph {et~al.}(2023)\citenamefont {Takagi}, \citenamefont {Tajima},\ and\ \citenamefont {Gu}}]{takagi2022universal}%
  \BibitemOpen
  \bibfield  {author} {\bibinfo {author} {\bibfnamefont {R.}~\bibnamefont {Takagi}}, \bibinfo {author} {\bibfnamefont {H.}~\bibnamefont {Tajima}},\ and\ \bibinfo {author} {\bibfnamefont {M.}~\bibnamefont {Gu}},\ }\bibfield  {title} {\bibinfo {title} {Universal sampling lower bounds for quantum error mitigation},\ }\href {https://doi.org/10.1103/PhysRevLett.131.210602} {\bibfield  {journal} {\bibinfo  {journal} {Phys. Rev. Lett.}\ }\textbf {\bibinfo {volume} {131}},\ \bibinfo {pages} {210602} (\bibinfo {year} {2023})}\BibitemShut {NoStop}%
\bibitem [{\citenamefont {Tsubouchi}\ \emph {et~al.}(2023)\citenamefont {Tsubouchi}, \citenamefont {Sagawa},\ and\ \citenamefont {Yoshioka}}]{tsubouchi2022universal}%
  \BibitemOpen
  \bibfield  {author} {\bibinfo {author} {\bibfnamefont {K.}~\bibnamefont {Tsubouchi}}, \bibinfo {author} {\bibfnamefont {T.}~\bibnamefont {Sagawa}},\ and\ \bibinfo {author} {\bibfnamefont {N.}~\bibnamefont {Yoshioka}},\ }\bibfield  {title} {\bibinfo {title} {Universal cost bound of quantum error mitigation based on quantum estimation theory},\ }\href {https://doi.org/10.1103/PhysRevLett.131.210601} {\bibfield  {journal} {\bibinfo  {journal} {Phys. Rev. Lett.}\ }\textbf {\bibinfo {volume} {131}},\ \bibinfo {pages} {210601} (\bibinfo {year} {2023})}\BibitemShut {NoStop}%
\bibitem [{\citenamefont {Quek}\ \emph {et~al.}(2022)\citenamefont {Quek}, \citenamefont {Fran{\c{c}}a}, \citenamefont {Khatri}, \citenamefont {Meyer},\ and\ \citenamefont {Eisert}}]{quek2022exponentially}%
  \BibitemOpen
  \bibfield  {author} {\bibinfo {author} {\bibfnamefont {Y.}~\bibnamefont {Quek}}, \bibinfo {author} {\bibfnamefont {D.~S.}\ \bibnamefont {Fran{\c{c}}a}}, \bibinfo {author} {\bibfnamefont {S.}~\bibnamefont {Khatri}}, \bibinfo {author} {\bibfnamefont {J.~J.}\ \bibnamefont {Meyer}},\ and\ \bibinfo {author} {\bibfnamefont {J.}~\bibnamefont {Eisert}},\ }\bibfield  {title} {\bibinfo {title} {Exponentially tighter bounds on limitations of quantum error mitigation},\ }\href {https://arxiv.org/abs/2210.11505} {\bibfield  {journal} {\bibinfo  {journal} {arXiv:2210.11505}\ } (\bibinfo {year} {2022})}\BibitemShut {NoStop}%
\bibitem [{\citenamefont {Habermann}(1972)}]{habermann1972parallel}%
  \BibitemOpen
  \bibfield  {author} {\bibinfo {author} {\bibfnamefont {A.~N.}\ \bibnamefont {Habermann}},\ }\bibfield  {title} {\bibinfo {title} {Parallel neighbor-sort (or the glory of the induction principle)},\ }\href {https://kilthub.cmu.edu/articles/journal_contribution/Parallel_neighbor-sort_or_the_glory_of_the_induction_principle_/6608258/files/12099395.pdf} {\bibfield  {journal} {\bibinfo  {journal} {Carnegie Mellon University Technical Report No. AD-759-248}\ } (\bibinfo {year} {1972})}\BibitemShut {NoStop}%
\bibitem [{\citenamefont {Jiang}\ \emph {et~al.}(2018)\citenamefont {Jiang}, \citenamefont {Sung}, \citenamefont {Kechedzhi}, \citenamefont {Smelyanskiy},\ and\ \citenamefont {Boixo}}]{jiang2018quantum}%
  \BibitemOpen
  \bibfield  {author} {\bibinfo {author} {\bibfnamefont {Z.}~\bibnamefont {Jiang}}, \bibinfo {author} {\bibfnamefont {K.~J.}\ \bibnamefont {Sung}}, \bibinfo {author} {\bibfnamefont {K.}~\bibnamefont {Kechedzhi}}, \bibinfo {author} {\bibfnamefont {V.~N.}\ \bibnamefont {Smelyanskiy}},\ and\ \bibinfo {author} {\bibfnamefont {S.}~\bibnamefont {Boixo}},\ }\bibfield  {title} {\bibinfo {title} {Quantum algorithms to simulate many-body physics of correlated fermions},\ }\href {https://doi.org/10.1103/PhysRevApplied.9.044036} {\bibfield  {journal} {\bibinfo  {journal} {Phys. Rev. Applied}\ }\textbf {\bibinfo {volume} {9}},\ \bibinfo {pages} {044036} (\bibinfo {year} {2018})}\BibitemShut {NoStop}%
\bibitem [{\citenamefont {Oszmaniec}\ \emph {et~al.}(2022)\citenamefont {Oszmaniec}, \citenamefont {Dangniam}, \citenamefont {Morales},\ and\ \citenamefont {Zimbor{\'a}s}}]{oszmaniec2022fermion}%
  \BibitemOpen
  \bibfield  {author} {\bibinfo {author} {\bibfnamefont {M.}~\bibnamefont {Oszmaniec}}, \bibinfo {author} {\bibfnamefont {N.}~\bibnamefont {Dangniam}}, \bibinfo {author} {\bibfnamefont {M.~E.~S.}\ \bibnamefont {Morales}},\ and\ \bibinfo {author} {\bibfnamefont {Z.}~\bibnamefont {Zimbor{\'a}s}},\ }\bibfield  {title} {\bibinfo {title} {Fermion sampling: a robust quantum computational advantage scheme using fermionic linear optics and magic input states},\ }\href {https://doi.org/10.1103/PRXQuantum.3.020328} {\bibfield  {journal} {\bibinfo  {journal} {PRX Quantum}\ }\textbf {\bibinfo {volume} {3}},\ \bibinfo {pages} {020328} (\bibinfo {year} {2022})}\BibitemShut {NoStop}%
\bibitem [{\citenamefont {Jordan}\ and\ \citenamefont {Wigner}(1928)}]{jordanwigner}%
  \BibitemOpen
  \bibfield  {author} {\bibinfo {author} {\bibfnamefont {P.}~\bibnamefont {Jordan}}\ and\ \bibinfo {author} {\bibfnamefont {E.}~\bibnamefont {Wigner}},\ }\bibfield  {title} {\bibinfo {title} {\"{U}ber das {P}aulische \"{A}quivalenzverbot},\ }\href {https://doi.org/10.1007/BF01331938} {\bibfield  {journal} {\bibinfo  {journal} {Z. Phys.}\ }\textbf {\bibinfo {volume} {47}},\ \bibinfo {pages} {631} (\bibinfo {year} {1928})}\BibitemShut {NoStop}%
\bibitem [{\citenamefont {McClean}\ \emph {et~al.}(2020{\natexlab{b}})\citenamefont {McClean} \emph {et~al.}}]{openfermion}%
  \BibitemOpen
  \bibfield  {author} {\bibinfo {author} {\bibfnamefont {J.~R.}\ \bibnamefont {McClean}} \emph {et~al.},\ }\bibfield  {title} {\bibinfo {title} {Open{F}ermon:~the electronic structure package for quantum computers},\ }\href {https://doi.org/10.1088/2058-9565/ab8ebc} {\bibfield  {journal} {\bibinfo  {journal} {Quantum Sci. Technol.}\ }\textbf {\bibinfo {volume} {5}},\ \bibinfo {pages} {034014} (\bibinfo {year} {2020}{\natexlab{b}})}\BibitemShut {NoStop}%
\bibitem [{\citenamefont {{Cirq Developers}}(2023)}]{cirq}%
  \BibitemOpen
  \bibfield  {author} {\bibinfo {author} {\bibnamefont {{Cirq Developers}}},\ }\href {https://doi.org/10.5281/zenodo.8161252} {\bibinfo {title} {Cirq}} (\bibinfo {year} {2023})\BibitemShut {NoStop}%
\bibitem [{\citenamefont {Efron}(1992)}]{efron1992bootstrap}%
  \BibitemOpen
  \bibfield  {author} {\bibinfo {author} {\bibfnamefont {B.}~\bibnamefont {Efron}},\ }\bibfield  {title} {\bibinfo {title} {Bootstrap methods: Another look at the jackknife},\ }in\ \href {https://doi.org/10.1007/978-1-4612-4380-9_41} {\emph {\bibinfo {booktitle} {Breakthroughs in Statistics}}},\ \bibinfo {editor} {edited by\ \bibinfo {editor} {\bibfnamefont {S.}~\bibnamefont {Kotz}}\ and\ \bibinfo {editor} {\bibfnamefont {N.~L.}\ \bibnamefont {Johnson}}}\ (\bibinfo  {publisher} {Springer},\ \bibinfo {address} {New York},\ \bibinfo {year} {1992})\ pp.\ \bibinfo {pages} {569--593}\BibitemShut {NoStop}%
\bibitem [{\citenamefont {Bravyi}\ and\ \citenamefont {K\"{o}nig}(2012)}]{bravyi2012classical}%
  \BibitemOpen
  \bibfield  {author} {\bibinfo {author} {\bibfnamefont {S.}~\bibnamefont {Bravyi}}\ and\ \bibinfo {author} {\bibfnamefont {R.}~\bibnamefont {K\"{o}nig}},\ }\bibfield  {title} {\bibinfo {title} {Classical simulation of dissipative fermionic linear optics},\ }\href {https://doi.org/10.26421/QIC12.11-12-2} {\bibfield  {journal} {\bibinfo  {journal} {Quantum Inf. Comput.}\ }\textbf {\bibinfo {volume} {12}},\ \bibinfo {pages} {925–943} (\bibinfo {year} {2012})}\BibitemShut {NoStop}%
\bibitem [{\citenamefont {Rubin}\ \emph {et~al.}(2018)\citenamefont {Rubin}, \citenamefont {Babbush},\ and\ \citenamefont {McClean}}]{rubin2018application}%
  \BibitemOpen
  \bibfield  {author} {\bibinfo {author} {\bibfnamefont {N.~C.}\ \bibnamefont {Rubin}}, \bibinfo {author} {\bibfnamefont {R.}~\bibnamefont {Babbush}},\ and\ \bibinfo {author} {\bibfnamefont {J.}~\bibnamefont {McClean}},\ }\bibfield  {title} {\bibinfo {title} {Application of fermionic marginal constraints to hybrid quantum algorithms},\ }\href {https://doi.org/10.1088/1367-2630/aab919} {\bibfield  {journal} {\bibinfo  {journal} {New J. Phys.}\ }\textbf {\bibinfo {volume} {20}},\ \bibinfo {pages} {053020} (\bibinfo {year} {2018})}\BibitemShut {NoStop}%
\bibitem [{\citenamefont {{Quantum AI team and collaborators}}(2020)}]{recirq}%
  \BibitemOpen
  \bibfield  {author} {\bibinfo {author} {\bibnamefont {{Quantum AI team and collaborators}}},\ }\href {https://doi.org/10.5281/zenodo.4091471} {\bibinfo {title} {{ReCirq}}} (\bibinfo {year} {2020})\BibitemShut {NoStop}%
\bibitem [{\citenamefont {Wecker}\ \emph {et~al.}(2015{\natexlab{b}})\citenamefont {Wecker}, \citenamefont {Hastings}, \citenamefont {Wiebe}, \citenamefont {Clark}, \citenamefont {Nayak},\ and\ \citenamefont {Troyer}}]{wecker2015solving}%
  \BibitemOpen
  \bibfield  {author} {\bibinfo {author} {\bibfnamefont {D.}~\bibnamefont {Wecker}}, \bibinfo {author} {\bibfnamefont {M.~B.}\ \bibnamefont {Hastings}}, \bibinfo {author} {\bibfnamefont {N.}~\bibnamefont {Wiebe}}, \bibinfo {author} {\bibfnamefont {B.~K.}\ \bibnamefont {Clark}}, \bibinfo {author} {\bibfnamefont {C.}~\bibnamefont {Nayak}},\ and\ \bibinfo {author} {\bibfnamefont {M.}~\bibnamefont {Troyer}},\ }\bibfield  {title} {\bibinfo {title} {Solving strongly correlated electron models on a quantum computer},\ }\href {https://doi.org/10.1103/PhysRevA.92.062318} {\bibfield  {journal} {\bibinfo  {journal} {Phys. Rev. A}\ }\textbf {\bibinfo {volume} {92}},\ \bibinfo {pages} {062318} (\bibinfo {year} {2015}{\natexlab{b}})}\BibitemShut {NoStop}%
\bibitem [{\citenamefont {Kivlichan}\ \emph {et~al.}(2018)\citenamefont {Kivlichan} \emph {et~al.}}]{kivlichan2018quantum}%
  \BibitemOpen
  \bibfield  {author} {\bibinfo {author} {\bibfnamefont {I.~D.}\ \bibnamefont {Kivlichan}} \emph {et~al.},\ }\bibfield  {title} {\bibinfo {title} {Quantum simulation of electronic structure with linear depth and connectivity},\ }\href {https://doi.org/10.1103/PhysRevLett.120.110501} {\bibfield  {journal} {\bibinfo  {journal} {Phys. Rev. Lett.}\ }\textbf {\bibinfo {volume} {120}},\ \bibinfo {pages} {110501} (\bibinfo {year} {2018})}\BibitemShut {NoStop}%
\bibitem [{\citenamefont {Garnerone}\ \emph {et~al.}(2010{\natexlab{a}})\citenamefont {Garnerone}, \citenamefont {de~Oliveira},\ and\ \citenamefont {Zanardi}}]{garnerone2010typicality}%
  \BibitemOpen
  \bibfield  {author} {\bibinfo {author} {\bibfnamefont {S.}~\bibnamefont {Garnerone}}, \bibinfo {author} {\bibfnamefont {T.~R.}\ \bibnamefont {de~Oliveira}},\ and\ \bibinfo {author} {\bibfnamefont {P.}~\bibnamefont {Zanardi}},\ }\bibfield  {title} {\bibinfo {title} {Typicality in random matrix product states},\ }\href {https://doi.org/10.1103/PhysRevA.81.032336} {\bibfield  {journal} {\bibinfo  {journal} {Phys. Rev. A}\ }\textbf {\bibinfo {volume} {81}},\ \bibinfo {pages} {032336} (\bibinfo {year} {2010}{\natexlab{a}})}\BibitemShut {NoStop}%
\bibitem [{\citenamefont {Garnerone}\ \emph {et~al.}(2010{\natexlab{b}})\citenamefont {Garnerone}, \citenamefont {de~Oliveira}, \citenamefont {Haas},\ and\ \citenamefont {Zanardi}}]{garnerone2010statistical}%
  \BibitemOpen
  \bibfield  {author} {\bibinfo {author} {\bibfnamefont {S.}~\bibnamefont {Garnerone}}, \bibinfo {author} {\bibfnamefont {T.~R.}\ \bibnamefont {de~Oliveira}}, \bibinfo {author} {\bibfnamefont {S.}~\bibnamefont {Haas}},\ and\ \bibinfo {author} {\bibfnamefont {P.}~\bibnamefont {Zanardi}},\ }\bibfield  {title} {\bibinfo {title} {Statistical properties of random matrix product states},\ }\href {https://doi.org/10.1103/PhysRevA.82.052312} {\bibfield  {journal} {\bibinfo  {journal} {Phys. Rev. A}\ }\textbf {\bibinfo {volume} {82}},\ \bibinfo {pages} {052312} (\bibinfo {year} {2010}{\natexlab{b}})}\BibitemShut {NoStop}%
\bibitem [{\citenamefont {Fishman}\ \emph {et~al.}(2022)\citenamefont {Fishman}, \citenamefont {White},\ and\ \citenamefont {Stoudenmire}}]{fishman2022itensor}%
  \BibitemOpen
  \bibfield  {author} {\bibinfo {author} {\bibfnamefont {M.}~\bibnamefont {Fishman}}, \bibinfo {author} {\bibfnamefont {S.~R.}\ \bibnamefont {White}},\ and\ \bibinfo {author} {\bibfnamefont {E.~M.}\ \bibnamefont {Stoudenmire}},\ }\bibfield  {title} {\bibinfo {title} {The {ITensor} software library for tensor network calculations},\ }\href {https://doi.org/10.21468/SciPostPhysCodeb.4} {\bibfield  {journal} {\bibinfo  {journal} {SciPost Phys. Codebases}\ }\textbf {\bibinfo {volume} {4}} (\bibinfo {year} {2022})}\BibitemShut {NoStop}%
\bibitem [{\citenamefont {White}(1992)}]{white1992density}%
  \BibitemOpen
  \bibfield  {author} {\bibinfo {author} {\bibfnamefont {S.~R.}\ \bibnamefont {White}},\ }\bibfield  {title} {\bibinfo {title} {Density matrix formulation for quantum renormalization groups},\ }\href {https://doi.org/10.1103/PhysRevLett.69.2863} {\bibfield  {journal} {\bibinfo  {journal} {Phys. Rev. Lett.}\ }\textbf {\bibinfo {volume} {69}},\ \bibinfo {pages} {2863} (\bibinfo {year} {1992})}\BibitemShut {NoStop}%
\bibitem [{\citenamefont {Helsen}\ \emph {et~al.}(2019)\citenamefont {Helsen}, \citenamefont {Xue}, \citenamefont {Vandersypen},\ and\ \citenamefont {Wehner}}]{helsen2019new}%
  \BibitemOpen
  \bibfield  {author} {\bibinfo {author} {\bibfnamefont {J.}~\bibnamefont {Helsen}}, \bibinfo {author} {\bibfnamefont {X.}~\bibnamefont {Xue}}, \bibinfo {author} {\bibfnamefont {L.~M.~K.}\ \bibnamefont {Vandersypen}},\ and\ \bibinfo {author} {\bibfnamefont {S.}~\bibnamefont {Wehner}},\ }\bibfield  {title} {\bibinfo {title} {A new class of efficient randomized benchmarking protocols},\ }\href {https://doi.org/10.1038/s41534-019-0182-7} {\bibfield  {journal} {\bibinfo  {journal} {npj Quantum Inf.}\ }\textbf {\bibinfo {volume} {5}},\ \bibinfo {pages} {71} (\bibinfo {year} {2019})}\BibitemShut {NoStop}%
\bibitem [{\citenamefont {Claes}\ \emph {et~al.}(2021)\citenamefont {Claes}, \citenamefont {Rieffel},\ and\ \citenamefont {Wang}}]{claes2021character}%
  \BibitemOpen
  \bibfield  {author} {\bibinfo {author} {\bibfnamefont {J.}~\bibnamefont {Claes}}, \bibinfo {author} {\bibfnamefont {E.}~\bibnamefont {Rieffel}},\ and\ \bibinfo {author} {\bibfnamefont {Z.}~\bibnamefont {Wang}},\ }\bibfield  {title} {\bibinfo {title} {Character randomized benchmarking for non-multiplicity-free groups with applications to subspace, leakage, and matchgate randomized benchmarking},\ }\href {https://doi.org/10.1103/PRXQuantum.2.010351} {\bibfield  {journal} {\bibinfo  {journal} {PRX Quantum}\ }\textbf {\bibinfo {volume} {2}},\ \bibinfo {pages} {010351} (\bibinfo {year} {2021})}\BibitemShut {NoStop}%
\bibitem [{\citenamefont {Van~Kirk}\ \emph {et~al.}(2022)\citenamefont {Van~Kirk}, \citenamefont {Cotler}, \citenamefont {Huang},\ and\ \citenamefont {Lukin}}]{van2022hardware}%
  \BibitemOpen
  \bibfield  {author} {\bibinfo {author} {\bibfnamefont {K.}~\bibnamefont {Van~Kirk}}, \bibinfo {author} {\bibfnamefont {J.}~\bibnamefont {Cotler}}, \bibinfo {author} {\bibfnamefont {H.-Y.}\ \bibnamefont {Huang}},\ and\ \bibinfo {author} {\bibfnamefont {M.~D.}\ \bibnamefont {Lukin}},\ }\bibfield  {title} {\bibinfo {title} {Hardware-efficient learning of quantum many-body states},\ }\href {https://arxiv.org/abs/2212.06084} {\bibfield  {journal} {\bibinfo  {journal} {arXiv:2212.06084}\ } (\bibinfo {year} {2022})}\BibitemShut {NoStop}%
\bibitem [{\citenamefont {Wallman}\ and\ \citenamefont {Emerson}(2016)}]{wallman2016noise}%
  \BibitemOpen
  \bibfield  {author} {\bibinfo {author} {\bibfnamefont {J.~J.}\ \bibnamefont {Wallman}}\ and\ \bibinfo {author} {\bibfnamefont {J.}~\bibnamefont {Emerson}},\ }\bibfield  {title} {\bibinfo {title} {Noise tailoring for scalable quantum computation via randomized compiling},\ }\href {https://doi.org/10.1103/PhysRevA.94.052325} {\bibfield  {journal} {\bibinfo  {journal} {Phys. Rev. A}\ }\textbf {\bibinfo {volume} {94}},\ \bibinfo {pages} {052325} (\bibinfo {year} {2016})}\BibitemShut {NoStop}%
\bibitem [{\citenamefont {Proctor}\ \emph {et~al.}(2017)\citenamefont {Proctor}, \citenamefont {Rudinger}, \citenamefont {Young}, \citenamefont {Sarovar},\ and\ \citenamefont {Blume-Kohout}}]{proctor2017randomized}%
  \BibitemOpen
  \bibfield  {author} {\bibinfo {author} {\bibfnamefont {T.}~\bibnamefont {Proctor}}, \bibinfo {author} {\bibfnamefont {K.}~\bibnamefont {Rudinger}}, \bibinfo {author} {\bibfnamefont {K.}~\bibnamefont {Young}}, \bibinfo {author} {\bibfnamefont {M.}~\bibnamefont {Sarovar}},\ and\ \bibinfo {author} {\bibfnamefont {R.}~\bibnamefont {Blume-Kohout}},\ }\bibfield  {title} {\bibinfo {title} {What randomized benchmarking actually measures},\ }\href {https://doi.org/10.1103/PhysRevLett.119.130502} {\bibfield  {journal} {\bibinfo  {journal} {Phys. Rev. Lett.}\ }\textbf {\bibinfo {volume} {119}},\ \bibinfo {pages} {130502} (\bibinfo {year} {2017})}\BibitemShut {NoStop}%
\bibitem [{\citenamefont {Wallman}(2018)}]{wallman2018randomized}%
  \BibitemOpen
  \bibfield  {author} {\bibinfo {author} {\bibfnamefont {J.~J.}\ \bibnamefont {Wallman}},\ }\bibfield  {title} {\bibinfo {title} {Randomized benchmarking with gate-dependent noise},\ }\href {https://doi.org/10.22331/q-2018-01-29-47} {\bibfield  {journal} {\bibinfo  {journal} {Quantum}\ }\textbf {\bibinfo {volume} {2}},\ \bibinfo {pages} {47} (\bibinfo {year} {2018})}\BibitemShut {NoStop}%
\bibitem [{\citenamefont {Carignan-Dugas}\ \emph {et~al.}(2018)\citenamefont {Carignan-Dugas}, \citenamefont {Boone}, \citenamefont {Wallman},\ and\ \citenamefont {Emerson}}]{carignan2018randomized}%
  \BibitemOpen
  \bibfield  {author} {\bibinfo {author} {\bibfnamefont {A.}~\bibnamefont {Carignan-Dugas}}, \bibinfo {author} {\bibfnamefont {K.}~\bibnamefont {Boone}}, \bibinfo {author} {\bibfnamefont {J.~J.}\ \bibnamefont {Wallman}},\ and\ \bibinfo {author} {\bibfnamefont {J.}~\bibnamefont {Emerson}},\ }\bibfield  {title} {\bibinfo {title} {From randomized benchmarking experiments to gate-set circuit fidelity: how to interpret randomized benchmarking decay parameters},\ }\href {https://doi.org/10.1088/1367-2630/aadcc7} {\bibfield  {journal} {\bibinfo  {journal} {New J. Phys.}\ }\textbf {\bibinfo {volume} {20}},\ \bibinfo {pages} {092001} (\bibinfo {year} {2018})}\BibitemShut {NoStop}%
\bibitem [{\citenamefont {Merkel}\ \emph {et~al.}(2021)\citenamefont {Merkel}, \citenamefont {Pritchett},\ and\ \citenamefont {Fong}}]{merkel2021randomized}%
  \BibitemOpen
  \bibfield  {author} {\bibinfo {author} {\bibfnamefont {S.~T.}\ \bibnamefont {Merkel}}, \bibinfo {author} {\bibfnamefont {E.~J.}\ \bibnamefont {Pritchett}},\ and\ \bibinfo {author} {\bibfnamefont {B.~H.}\ \bibnamefont {Fong}},\ }\bibfield  {title} {\bibinfo {title} {Randomized benchmarking as convolution: {F}ourier analysis of gate dependent errors},\ }\href {https://doi.org/10.22331/q-2021-11-16-581} {\bibfield  {journal} {\bibinfo  {journal} {Quantum}\ }\textbf {\bibinfo {volume} {5}},\ \bibinfo {pages} {581} (\bibinfo {year} {2021})}\BibitemShut {NoStop}%
\bibitem [{\citenamefont {Wu}\ and\ \citenamefont {Koh}(2024)}]{wu2023error}%
  \BibitemOpen
  \bibfield  {author} {\bibinfo {author} {\bibfnamefont {B.}~\bibnamefont {Wu}}\ and\ \bibinfo {author} {\bibfnamefont {D.~E.}\ \bibnamefont {Koh}},\ }\bibfield  {title} {\bibinfo {title} {Error-mitigated fermionic classical shadows on noisy quantum devices},\ }\href {https://doi.org/10.1038/s41534-024-00836-7} {\bibfield  {journal} {\bibinfo  {journal} {npj Quantum Inf.}\ }\textbf {\bibinfo {volume} {10}},\ \bibinfo {pages} {39} (\bibinfo {year} {2024})}\BibitemShut {NoStop}%
\bibitem [{\citenamefont {Brieger}\ \emph {et~al.}(2023)\citenamefont {Brieger}, \citenamefont {Heinrich}, \citenamefont {Roth},\ and\ \citenamefont {Kliesch}}]{brieger2023stability}%
  \BibitemOpen
  \bibfield  {author} {\bibinfo {author} {\bibfnamefont {R.}~\bibnamefont {Brieger}}, \bibinfo {author} {\bibfnamefont {M.}~\bibnamefont {Heinrich}}, \bibinfo {author} {\bibfnamefont {I.}~\bibnamefont {Roth}},\ and\ \bibinfo {author} {\bibfnamefont {M.}~\bibnamefont {Kliesch}},\ }\bibfield  {title} {\bibinfo {title} {Stability of classical shadows under gate-dependent noise},\ }\href {https://arxiv.org/abs/2310.19947} {\bibfield  {journal} {\bibinfo  {journal} {arXiv:2310.19947}\ } (\bibinfo {year} {2023})}\BibitemShut {NoStop}%
\bibitem [{\citenamefont {Helsen}\ \emph {et~al.}(2022)\citenamefont {Helsen}, \citenamefont {Nezami}, \citenamefont {Reagor},\ and\ \citenamefont {Walter}}]{helsen2022matchgate}%
  \BibitemOpen
  \bibfield  {author} {\bibinfo {author} {\bibfnamefont {J.}~\bibnamefont {Helsen}}, \bibinfo {author} {\bibfnamefont {S.}~\bibnamefont {Nezami}}, \bibinfo {author} {\bibfnamefont {M.}~\bibnamefont {Reagor}},\ and\ \bibinfo {author} {\bibfnamefont {M.}~\bibnamefont {Walter}},\ }\bibfield  {title} {\bibinfo {title} {Matchgate benchmarking: Scalable benchmarking of a continuous family of many-qubit gates},\ }\href {https://doi.org/10.22331/q-2022-02-21-657} {\bibfield  {journal} {\bibinfo  {journal} {Quantum}\ }\textbf {\bibinfo {volume} {6}},\ \bibinfo {pages} {657} (\bibinfo {year} {2022})}\BibitemShut {NoStop}%
\bibitem [{\citenamefont {Valiant}(2001)}]{valiant2001quantum}%
  \BibitemOpen
  \bibfield  {author} {\bibinfo {author} {\bibfnamefont {L.~G.}\ \bibnamefont {Valiant}},\ }\bibfield  {title} {\bibinfo {title} {Quantum computers that can be simulated classically in polynomial time},\ }in\ \href {https://doi.org/10.1145/380752.380785} {\emph {\bibinfo {booktitle} {Proceedings of the 33rd Annual ACM Symposium on Theory of Computing}}}\ (\bibinfo {year} {2001})\ pp.\ \bibinfo {pages} {114--123}\BibitemShut {NoStop}%
\bibitem [{\citenamefont {Knill}(2001)}]{knill2001fermionic}%
  \BibitemOpen
  \bibfield  {author} {\bibinfo {author} {\bibfnamefont {E.}~\bibnamefont {Knill}},\ }\bibfield  {title} {\bibinfo {title} {Fermionic linear optics and matchgates},\ }\href {https://arxiv.org/abs/quant-ph/0108033} {\bibfield  {journal} {\bibinfo  {journal} {arXiv:quant-ph/0108033}\ } (\bibinfo {year} {2001})}\BibitemShut {NoStop}%
\bibitem [{\citenamefont {Terhal}\ and\ \citenamefont {DiVincenzo}(2002)}]{terhal2002classical}%
  \BibitemOpen
  \bibfield  {author} {\bibinfo {author} {\bibfnamefont {B.~M.}\ \bibnamefont {Terhal}}\ and\ \bibinfo {author} {\bibfnamefont {D.~P.}\ \bibnamefont {DiVincenzo}},\ }\bibfield  {title} {\bibinfo {title} {Classical simulation of noninteracting-fermion quantum circuits},\ }\href {https://doi.org/10.1103/PhysRevA.65.032325} {\bibfield  {journal} {\bibinfo  {journal} {Phys. Rev. A}\ }\textbf {\bibinfo {volume} {65}},\ \bibinfo {pages} {032325} (\bibinfo {year} {2002})}\BibitemShut {NoStop}%
\bibitem [{\citenamefont {Bravyi}(2005)}]{bravyi2004lagrangian}%
  \BibitemOpen
  \bibfield  {author} {\bibinfo {author} {\bibfnamefont {S.}~\bibnamefont {Bravyi}},\ }\bibfield  {title} {\bibinfo {title} {Lagrangian representation for fermionic linear optics},\ }\href {https://doi.org/10.26421/qic5.3} {\bibfield  {journal} {\bibinfo  {journal} {Quantum Inf. Comput.}\ }\textbf {\bibinfo {volume} {5}},\ \bibinfo {pages} {216} (\bibinfo {year} {2005})}\BibitemShut {NoStop}%
\bibitem [{\citenamefont {DiVincenzo}\ and\ \citenamefont {Terhal}(2005)}]{divincenzo2005fermionic}%
  \BibitemOpen
  \bibfield  {author} {\bibinfo {author} {\bibfnamefont {D.~P.}\ \bibnamefont {DiVincenzo}}\ and\ \bibinfo {author} {\bibfnamefont {B.~M.}\ \bibnamefont {Terhal}},\ }\bibfield  {title} {\bibinfo {title} {Fermionic linear optics revisited},\ }\href {https://doi.org/10.1007/s10701-005-8657-0} {\bibfield  {journal} {\bibinfo  {journal} {Found. Phys.}\ }\textbf {\bibinfo {volume} {35}},\ \bibinfo {pages} {1967} (\bibinfo {year} {2005})}\BibitemShut {NoStop}%
\bibitem [{\citenamefont {Jozsa}\ and\ \citenamefont {Miyake}(2008)}]{jozsa2008matchgates}%
  \BibitemOpen
  \bibfield  {author} {\bibinfo {author} {\bibfnamefont {R.}~\bibnamefont {Jozsa}}\ and\ \bibinfo {author} {\bibfnamefont {A.}~\bibnamefont {Miyake}},\ }\bibfield  {title} {\bibinfo {title} {Matchgates and classical simulation of quantum circuits},\ }\href {https://doi.org/10.1098/rspa.2008.0189} {\bibfield  {journal} {\bibinfo  {journal} {Proc. R. Soc. A}\ }\textbf {\bibinfo {volume} {464}},\ \bibinfo {pages} {3089} (\bibinfo {year} {2008})}\BibitemShut {NoStop}%
\bibitem [{\citenamefont {Gambetta}\ \emph {et~al.}(2012)\citenamefont {Gambetta} \emph {et~al.}}]{gambetta2012characterization}%
  \BibitemOpen
  \bibfield  {author} {\bibinfo {author} {\bibfnamefont {J.~M.}\ \bibnamefont {Gambetta}} \emph {et~al.},\ }\bibfield  {title} {\bibinfo {title} {Characterization of addressability by simultaneous randomized benchmarking},\ }\href {https://doi.org/10.1103/PhysRevLett.109.240504} {\bibfield  {journal} {\bibinfo  {journal} {Phys. Rev. Lett.}\ }\textbf {\bibinfo {volume} {109}},\ \bibinfo {pages} {240504} (\bibinfo {year} {2012})}\BibitemShut {NoStop}%
\bibitem [{\citenamefont {Reck}\ \emph {et~al.}(1994)\citenamefont {Reck}, \citenamefont {Zeilinger}, \citenamefont {Bernstein},\ and\ \citenamefont {Bertani}}]{reck1994experimental}%
  \BibitemOpen
  \bibfield  {author} {\bibinfo {author} {\bibfnamefont {M.}~\bibnamefont {Reck}}, \bibinfo {author} {\bibfnamefont {A.}~\bibnamefont {Zeilinger}}, \bibinfo {author} {\bibfnamefont {H.~J.}\ \bibnamefont {Bernstein}},\ and\ \bibinfo {author} {\bibfnamefont {P.}~\bibnamefont {Bertani}},\ }\bibfield  {title} {\bibinfo {title} {Experimental realization of any discrete unitary operator},\ }\href {https://doi.org/10.1103/PhysRevLett.73.58} {\bibfield  {journal} {\bibinfo  {journal} {Phys. Rev. Lett.}\ }\textbf {\bibinfo {volume} {73}},\ \bibinfo {pages} {58} (\bibinfo {year} {1994})}\BibitemShut {NoStop}%
\bibitem [{\citenamefont {Clements}\ \emph {et~al.}(2016)\citenamefont {Clements}, \citenamefont {Humphreys}, \citenamefont {Metcalf}, \citenamefont {Kolthammer},\ and\ \citenamefont {Walmsley}}]{clements2016optimal}%
  \BibitemOpen
  \bibfield  {author} {\bibinfo {author} {\bibfnamefont {W.~R.}\ \bibnamefont {Clements}}, \bibinfo {author} {\bibfnamefont {P.~C.}\ \bibnamefont {Humphreys}}, \bibinfo {author} {\bibfnamefont {B.~J.}\ \bibnamefont {Metcalf}}, \bibinfo {author} {\bibfnamefont {W.~S.}\ \bibnamefont {Kolthammer}},\ and\ \bibinfo {author} {\bibfnamefont {I.~A.}\ \bibnamefont {Walmsley}},\ }\bibfield  {title} {\bibinfo {title} {Optimal design for universal multiport interferometers},\ }\href {https://doi.org/10.1364/OPTICA.3.001460} {\bibfield  {journal} {\bibinfo  {journal} {Optica}\ }\textbf {\bibinfo {volume} {3}},\ \bibinfo {pages} {1460} (\bibinfo {year} {2016})}\BibitemShut {NoStop}%
\bibitem [{\citenamefont {Huggins}\ \emph {et~al.}(2021{\natexlab{b}})\citenamefont {Huggins} \emph {et~al.}}]{huggins2021efficient}%
  \BibitemOpen
  \bibfield  {author} {\bibinfo {author} {\bibfnamefont {W.~J.}\ \bibnamefont {Huggins}} \emph {et~al.},\ }\bibfield  {title} {\bibinfo {title} {Efficient and noise resilient measurements for quantum chemistry on near-term quantum computers},\ }\href {https://doi.org/10.1038/s41534-020-00341-7} {\bibfield  {journal} {\bibinfo  {journal} {npj Quantum Inf.}\ }\textbf {\bibinfo {volume} {7}},\ \bibinfo {pages} {23} (\bibinfo {year} {2021}{\natexlab{b}})}\BibitemShut {NoStop}%
\bibitem [{\citenamefont {Chapman}\ and\ \citenamefont {Miyake}(2018)}]{chapman2018classical}%
  \BibitemOpen
  \bibfield  {author} {\bibinfo {author} {\bibfnamefont {A.}~\bibnamefont {Chapman}}\ and\ \bibinfo {author} {\bibfnamefont {A.}~\bibnamefont {Miyake}},\ }\bibfield  {title} {\bibinfo {title} {Classical simulation of quantum circuits by dynamical localization:~{A}nalytic results for {P}auli-observable scrambling in time-dependent disorder},\ }\href {https://doi.org/10.1103/PhysRevA.98.012309} {\bibfield  {journal} {\bibinfo  {journal} {Phys. Rev. A}\ }\textbf {\bibinfo {volume} {98}},\ \bibinfo {pages} {012309} (\bibinfo {year} {2018})}\BibitemShut {NoStop}%
\bibitem [{\citenamefont {Wallman}\ \emph {et~al.}(2015)\citenamefont {Wallman}, \citenamefont {Granade}, \citenamefont {Harper},\ and\ \citenamefont {Flammia}}]{wallman2015estimating}%
  \BibitemOpen
  \bibfield  {author} {\bibinfo {author} {\bibfnamefont {J.}~\bibnamefont {Wallman}}, \bibinfo {author} {\bibfnamefont {C.}~\bibnamefont {Granade}}, \bibinfo {author} {\bibfnamefont {R.}~\bibnamefont {Harper}},\ and\ \bibinfo {author} {\bibfnamefont {S.~T.}\ \bibnamefont {Flammia}},\ }\bibfield  {title} {\bibinfo {title} {Estimating the coherence of noise},\ }\href {https://doi.org/10.1088/1367-2630/17/11/113020} {\bibfield  {journal} {\bibinfo  {journal} {New J. Phys.}\ }\textbf {\bibinfo {volume} {17}},\ \bibinfo {pages} {113020} (\bibinfo {year} {2015})}\BibitemShut {NoStop}%
\bibitem [{\citenamefont {Feng}\ \emph {et~al.}(2016)\citenamefont {Feng} \emph {et~al.}}]{feng2016estimating}%
  \BibitemOpen
  \bibfield  {author} {\bibinfo {author} {\bibfnamefont {G.}~\bibnamefont {Feng}} \emph {et~al.},\ }\bibfield  {title} {\bibinfo {title} {Estimating the coherence of noise in quantum control of a solid-state qubit},\ }\href {https://doi.org/10.1103/PhysRevLett.117.260501} {\bibfield  {journal} {\bibinfo  {journal} {Phys. Rev. Lett.}\ }\textbf {\bibinfo {volume} {117}},\ \bibinfo {pages} {260501} (\bibinfo {year} {2016})}\BibitemShut {NoStop}%
\bibitem [{\citenamefont {Magesan}\ \emph {et~al.}(2011)\citenamefont {Magesan}, \citenamefont {Gambetta},\ and\ \citenamefont {Emerson}}]{magesan2011scalable}%
  \BibitemOpen
  \bibfield  {author} {\bibinfo {author} {\bibfnamefont {E.}~\bibnamefont {Magesan}}, \bibinfo {author} {\bibfnamefont {J.~M.}\ \bibnamefont {Gambetta}},\ and\ \bibinfo {author} {\bibfnamefont {J.}~\bibnamefont {Emerson}},\ }\bibfield  {title} {\bibinfo {title} {Scalable and robust randomized benchmarking of quantum processes},\ }\href {https://doi.org/10.1103/PhysRevLett.106.180504} {\bibfield  {journal} {\bibinfo  {journal} {Phys. Rev. Lett.}\ }\textbf {\bibinfo {volume} {106}},\ \bibinfo {pages} {180504} (\bibinfo {year} {2011})}\BibitemShut {NoStop}%
\bibitem [{\citenamefont {Magesan}\ \emph {et~al.}(2012)\citenamefont {Magesan}, \citenamefont {Gambetta},\ and\ \citenamefont {Emerson}}]{magesan2012characterizing}%
  \BibitemOpen
  \bibfield  {author} {\bibinfo {author} {\bibfnamefont {E.}~\bibnamefont {Magesan}}, \bibinfo {author} {\bibfnamefont {J.~M.}\ \bibnamefont {Gambetta}},\ and\ \bibinfo {author} {\bibfnamefont {J.}~\bibnamefont {Emerson}},\ }\bibfield  {title} {\bibinfo {title} {Characterizing quantum gates via randomized benchmarking},\ }\href {https://doi.org/10.1103/PhysRevA.85.042311} {\bibfield  {journal} {\bibinfo  {journal} {Phys. Rev. A}\ }\textbf {\bibinfo {volume} {85}},\ \bibinfo {pages} {042311} (\bibinfo {year} {2012})}\BibitemShut {NoStop}%
\bibitem [{\citenamefont {Boixo}\ \emph {et~al.}(2018)\citenamefont {Boixo} \emph {et~al.}}]{boixo2018characterizing}%
  \BibitemOpen
  \bibfield  {author} {\bibinfo {author} {\bibfnamefont {S.}~\bibnamefont {Boixo}} \emph {et~al.},\ }\bibfield  {title} {\bibinfo {title} {Characterizing quantum supremacy in near-term devices},\ }\href {https://doi.org/10.1038/s41567-018-0124-x} {\bibfield  {journal} {\bibinfo  {journal} {Nat. Phys.}\ }\textbf {\bibinfo {volume} {14}},\ \bibinfo {pages} {595} (\bibinfo {year} {2018})}\BibitemShut {NoStop}%
\bibitem [{\citenamefont {Neill}\ \emph {et~al.}(2018)\citenamefont {Neill} \emph {et~al.}}]{neill2018blueprint}%
  \BibitemOpen
  \bibfield  {author} {\bibinfo {author} {\bibfnamefont {C.}~\bibnamefont {Neill}} \emph {et~al.},\ }\bibfield  {title} {\bibinfo {title} {A blueprint for demonstrating quantum supremacy with superconducting qubits},\ }\href {https://doi.org/10.1126/science.aao4309} {\bibfield  {journal} {\bibinfo  {journal} {Science}\ }\textbf {\bibinfo {volume} {360}},\ \bibinfo {pages} {195} (\bibinfo {year} {2018})}\BibitemShut {NoStop}%
\bibitem [{\citenamefont {Huang}\ \emph {et~al.}(2023)\citenamefont {Huang} \emph {et~al.}}]{huang2023quantum}%
  \BibitemOpen
  \bibfield  {author} {\bibinfo {author} {\bibfnamefont {C.}~\bibnamefont {Huang}} \emph {et~al.},\ }\bibfield  {title} {\bibinfo {title} {Quantum instruction set design for performance},\ }\href {https://doi.org/10.1103/PhysRevLett.130.070601} {\bibfield  {journal} {\bibinfo  {journal} {Phys. Rev. Lett.}\ }\textbf {\bibinfo {volume} {130}},\ \bibinfo {pages} {070601} (\bibinfo {year} {2023})}\BibitemShut {NoStop}%
\bibitem [{\citenamefont {McClean}\ \emph {et~al.}(2017)\citenamefont {McClean}, \citenamefont {Kimchi-Schwartz}, \citenamefont {Carter},\ and\ \citenamefont {de~Jong}}]{mcclean2017hybrid}%
  \BibitemOpen
  \bibfield  {author} {\bibinfo {author} {\bibfnamefont {J.~R.}\ \bibnamefont {McClean}}, \bibinfo {author} {\bibfnamefont {M.~E.}\ \bibnamefont {Kimchi-Schwartz}}, \bibinfo {author} {\bibfnamefont {J.}~\bibnamefont {Carter}},\ and\ \bibinfo {author} {\bibfnamefont {W.~A.}\ \bibnamefont {de~Jong}},\ }\bibfield  {title} {\bibinfo {title} {Hybrid quantum-classical hierarchy for mitigation of decoherence and determination of excited states},\ }\href {https://doi.org/10.1103/PhysRevA.95.042308} {\bibfield  {journal} {\bibinfo  {journal} {Phys. Rev. A}\ }\textbf {\bibinfo {volume} {95}},\ \bibinfo {pages} {042308} (\bibinfo {year} {2017})}\BibitemShut {NoStop}%
\end{thebibliography}%

\onecolumngrid
\appendix

\renewcommand{\thefigure}{S\arabic{figure}}
\setcounter{figure}{0}

\begin{center}
\hypertarget{supplementary_material}{\large{\textbf{Supplementary Information}}}
\end{center}

\tableofcontents

\let\addcontentsline\oldaddcontentsline

\section{Error analysis}\label{sec:error_analysis}

Here we provide the proof for Theorem~\ref{thm:main_theorem} from the main text, restated below for convenience.

\newtheorem*{theorem4}{Theorem 4}
\begin{theorem4}[Restated from main text]%\label{thm:main_theorem}
Fix accuracy and confidence parameters $\epsilon, \delta \in (0, 1)$. Let $O_1, \ldots, O_L$ be a collection of observables, each supported on an irrep of $\mathcal{U} : G \to \U(\mathcal{L}(\mathcal{H}))$ as $O_j \in V_\lambda$ for $\lambda \in R' \subseteq R_G$. Let $S_\lambda \in V_\lambda$ be a symmetry operator for each $\lambda \in R'$, for which the ideal values $s_\lambda = \tr(S_\lambda \rho)$ of the target state $\rho$ are known \emph{a priori}. Suppose that each noisy unitary satisfies Assumptions~1, $\widetilde{\mathcal{U}}_g = \mathcal{E} \mathcal{U}_g$, and define the quantities
\begin{align}
    F_{Z,R'}(\mathcal{E}) &\coloneqq \min_{\lambda \in R'} \frac{\tr(\mathcal{E} \mathcal{M}_Z \Pi_\lambda)}{\tr(\mathcal{M}_Z \Pi_\lambda)},\\
    \sigma^2 &\coloneqq \max_{1 \leq j \leq L, \lambda \in R'}\l\{ \V[\hat{o}_j], \V\l[\frac{\hat{s}_\lambda}{s_\lambda} \r] \r\}.
\end{align}
Then, a (noisy) classical shadow $\hat{\rho}(T)$ of size
\begin{equation}
    T = \O\l( \frac{\log((L + |R'|)/\delta)}{F_{Z,R'}(\mathcal{E})^2 \epsilon^2} \sigma^2 \r)
\end{equation}
can be used to construct error-mitigated estimates
\begin{equation}
    \hat{o}_j^{\mathrm{EM}}(T) \coloneqq \frac{\tr(O_j \hat{\rho}(T))}{\tr(S_\lambda \hat{\rho}(T)) / s_\lambda}
\end{equation}
which obey
\begin{equation}
    |\hat{o}_j^{\mathrm{EM}}(T) - \tr(O_j\rho)| \leq (\|O_j\|_\infty + 1) \epsilon + \O(\|O_j\|_\infty \epsilon^2)
\end{equation}
for all $1 \leq j \leq L$, with success probability at least $1 - \delta$.
\end{theorem4}

\begin{proof}
Let $ \hat{\rho}_1, \ldots, \hat{\rho}_T $ be the $T$ noisy classical shadows. Construct the mean of these snapshots,
\begin{equation}
    \hat{\rho}(T) = \frac{1}{T} \sum_{\ell=1}^T \hat{\rho}_\ell.
\end{equation}
(It is straightforward to replace this by a median-of-means estimator if necessary.) In expectation we have $\E[\hat{\rho}(T)] = \widetilde{\rho}$, where the effective noisy state can be described as
\begin{equation}
    \widetilde{\rho} = \mathcal{M}^{-1} \l(\widetilde{\mathcal{M}} (\rho)\r).
\end{equation}
Let $O \in V_\lambda$, with symmetry $s_\lambda = \tr(S_\lambda \rho)$ in the same irrep. Define estimates of the noisy expectation values using $\hat{\rho}(T)$:
\begin{align}
    \bar{X} &= \tr(O \hat{\rho}(T))\\
    \bar{Y} &= \tr\l(\frac{S_\lambda}{s_\lambda} \hat{\rho}(T)\r).
\end{align}
In expectation, these random variables obey $\E[\bar{X}] = \tr(O \widetilde{\rho})$ and $\E[\bar{Y}] = \tr(S_\lambda \widetilde{\rho}) / s_\lambda = \widetilde{f}_\lambda / f_\lambda$. Therefore as established from the main text, we have
\begin{equation}
	r(\E[\bar{X}], \E[\bar{Y}]) = \frac{\tr(O \widetilde{\rho})}{\tr(S_\lambda \widetilde{\rho}) / s_\lambda} = \tr(O \rho),
\end{equation}
where we have defined the function $r(x, y) \coloneqq x/y$. From a finite number of samples, however, we can only construct $\hat{o}^{\mathrm{EM}}(T) \coloneqq r(\bar{X}, \bar{Y})$, which is generally a biased estimator since $ \E[r(\bar{X}, \bar{Y})] \neq r(\E[\bar{X}], \E[\bar{Y}]) $.

To quantify the estimation error, we employ Taylor's remainder theorem:~expanding $r(x, y)$ to first order about a point $ (x_0, y_0) $, we have
% \begin{widetext}
\begin{equation}\label{eq:taylor_thm}
\begin{split}
	r(x, y) &= r(x_0, y_0) + \partial_x r(x_0, y_0) (x - x_0) + \partial_y r(x_0, y_0) (y - y_0) + h_1(x, y),
\end{split}
\end{equation}
where the remainder term is
\begin{equation}\label{eq:taylor_remainder}
\begin{split}
    h_1(x, y) &= \frac{1}{2!} \l[ \partial_x^2 r(a, b) (x - x_0)^2 + \partial_y^2 r(a, b) (y - y_0)^2 + 2 \partial_{xy} r(a, b) (x - x_0) (y - y_0) \r]
\end{split}
\end{equation}
% \end{widetext}
for some points $a \in [\min(x, x_0), \max(x, x_0)]$ and $b \in [\min(y, y_0), \max(y, y_0)]$. The relevant partial derivatives of $r(x, y)$ are enumerated below:
\begin{align}
    \partial_x r(x,y) &= 1/y,\\
    \partial_y r(x,y) &= -x/y^2,\\
    \partial_x^2 r(x,y) &= 0,\\
    \partial_{xy} r(x,y) &= -1/y^2,\\
    \partial_y^2 r(x,y) &= 2x/y^3.
\end{align}

Suppose $T$ is large enough such that (with high probability) the estimation error of all \emph{noisy} observables are uniformly bounded by some $\tilde{\epsilon} \in (0, 1)$:
\begin{align}
    \l|\bar{X} - \E[\bar{X}]\r| &\leq \tilde{\epsilon}, \label{eq:X_eps}\\
    |\bar{Y} - \E[\bar{Y}]| &\leq \tilde{\epsilon}. \label{eq:Y_eps}
\end{align}
This is achieved by standard classical shadow arguments, which we will elaborate on later. For now, assuming these error bounds hold, we rearrange Eq.~\eqref{eq:taylor_thm}, set $(x, y) = (\bar{X}, \bar{Y})$ and $(x_0, y_0) = (\E[\bar{X}], \E[\bar{Y}])$, and apply a triangle inequality to obtain
\begin{equation}\label{eq:1st_error_bound}
% \begin{split}
    \l| r(\bar{X}, \bar{Y}) - r(\E[\bar{X}], \E[\bar{Y}]) \r| \leq \frac{1}{|{\E[\bar{Y}]}|} \tilde{\epsilon} + \frac{|{\E[\bar{X}]}|}{\E[\bar{Y}]^2} \tilde{\epsilon} + + |h_1(\bar{X}, \bar{Y})|.
%     &\quad + |h_1(\bar{X}, \bar{Y})|.
% \end{split}
\end{equation}

To proceed with this error bound, we make the following observations. First, note that
\begin{equation}
    \E[\bar{Y}] = \frac{\widetilde{f}_\lambda}{f_\lambda} = \frac{\tr(\mathcal{M}_Z \mathcal{E} \Pi_\lambda)}{\tr(\mathcal{M}_Z \Pi_\lambda)} \in [0, 1],
\end{equation}
which we will denote by $\xi_\lambda$. We assume that that noise channel $\mathcal{E}$ is such that $\xi_\lambda > 0$, as otherwise the quantity $r(\E[\bar{X}], \E[\bar{Y}])$ diverges. Next, because $\E[\bar{X}]/\E[\bar{Y}] = \tr(O \rho)$, we have the bound
\begin{equation}
    \l| \frac{\E[\bar{X}]}{\E[\bar{Y}]} \r| \leq \| O \|_\infty.
\end{equation}
Thus Eq.~\eqref{eq:1st_error_bound} becomes
\begin{equation}\label{eq:2nd_error_bound}
% \begin{split}
    \l| r(\bar{X}, \bar{Y}) - r(\E[\bar{X}], \E[\bar{Y}]) \r| \leq \frac{1}{\xi_\lambda} \tilde{\epsilon} + \frac{\|O\|_\infty}{\xi_\lambda} \tilde{\epsilon} + |h_1(\bar{X}, \bar{Y})|.
%     &\quad + |h_1(\bar{X}, \bar{Y})|.
% \end{split}
\end{equation}

We can bound the remainder term $|h_1(\bar{X}, \bar{Y})|$ as follows. Applying a triangle inequality to Eq.~\eqref{eq:taylor_remainder} yields
\begin{equation}
    | h_1(\bar{X}, \bar{Y}) | \leq \frac{|a|}{|b|^3} \tilde{\epsilon}^2 + \frac{1}{b^2} \tilde{\epsilon}^2.
\end{equation}
Taylor's remainder theorem tells us that the value of $a$ (resp., $b$) lies between $\bar{X}$ and $\E[\bar{X}]$ (resp., $\bar{Y}$ and $\E[\bar{Y}]$), which we know are at most $\tilde{\epsilon}$ apart. We can therefore bound
\begin{equation}
\begin{split}
    |a| &\leq \max\{ |{\E[\bar{X}]}|, |\bar{X}| \}\\
    &\leq \max\{ |{\E[\bar{X}]}|, |{\E[\bar{X}]} + \tilde{\epsilon}|, |{\E[\bar{X}]} - \tilde{\epsilon}| \}\\
    &\leq |{\E[\bar{X}]}| + \tilde{\epsilon}\\
    &\leq \xi_\lambda \|O\|_\infty + \tilde{\epsilon}.
\end{split}
\end{equation}
Similarly for $b$, using the fact that $\E[\bar{Y}] > 0$,
\begin{equation}
\begin{split}
    |b| &\geq \min\{ |{\E[\bar{Y}]}|, |\bar{Y}| \}\\
    &\geq \min\{ \E[\bar{Y}], \E[\bar{Y}] + \tilde{\epsilon}, |{\E[\bar{Y}]} - \tilde{\epsilon}| \}\\
    &= \min\{ \xi_\lambda, |\xi_\lambda - \tilde{\epsilon}| \}.
\end{split}
\end{equation}
If $\tilde{\epsilon} < \xi_\lambda$, then $|b| \geq \xi_\lambda - \tilde{\epsilon} > 0$ always holds. We will see later that this condition is always justified;~for now, we will just suppose that this lower bound on $|b|$ holds. Then the remainder obeys
\begin{equation}\label{eq:remainder_bound}
    | h_1(\bar{X}, \bar{Y}) | \leq \frac{1}{(\xi_\lambda - \tilde{\epsilon})^2} \l( \frac{\xi_\lambda \|O\|_\infty + \tilde{\epsilon}}{\xi_\lambda - \tilde{\epsilon}} + 1 \r) \tilde{\epsilon}^2.
\end{equation}

Combining Eqs.~\eqref{eq:2nd_error_bound} and \eqref{eq:remainder_bound}, we arrive at
\begin{equation}
% \begin{align}
    \l| \hat{o}^{\mathrm{EM}}(T) - \tr(O \rho) \r| \leq \frac{1}{\xi_\lambda} (\|O\|_\infty + 1) \tilde{\epsilon} + \frac{1}{(\xi_\lambda - \tilde{\epsilon})^2} \l( \frac{\xi_\lambda \|O\|_\infty  + \tilde{\epsilon}}{\xi_\lambda - \tilde{\epsilon}} + 1 \r) \tilde{\epsilon}^2.
    % &\quad + \frac{1}{(\xi_\lambda - \tilde{\epsilon})^2} \l( \frac{\xi_\lambda \|O\|_\infty  + \tilde{\epsilon}}{\xi_\lambda - \tilde{\epsilon}} + 1 \r) \tilde{\epsilon}^2. \notag
% \end{align}
\end{equation}
In order to bound this error by $\O(\|O\|_\infty\epsilon)$ for some desired $\epsilon \in (0, 1)$, we can choose $\tilde{\epsilon} = \xi_\lambda \epsilon$, yielding
\begin{equation}\label{eq:error_bound_full}
% \begin{align}\label{eq:error_bound_full}
    \l| \hat{o}^{\mathrm{EM}}(T) - \tr(O \rho) \r| \leq (\|O\|_\infty + 1) \epsilon + \frac{1}{(1 - \epsilon)^2} \l( \frac{\|O\|_\infty  + \epsilon}{1 - \epsilon} + 1 \r) \epsilon^2.
    % &\quad + \frac{1}{(1 - \epsilon)^2} \l( \frac{\|O\|_\infty  + \epsilon}{1 - \epsilon} + 1 \r) \epsilon^2. \notag
% \end{align}
\end{equation}
Thus, by demanding $\epsilon < 1$ we ensure that the required technical condition $\tilde{\epsilon} < \xi_\lambda$ is met. Now we need to verify that the remainder term is bounded by $\O(\|O\|_\infty\epsilon^2)$, so that the $\O(\|O\|_\infty \epsilon)$ term dominates asymptotically as $\epsilon \to 0$. Indeed, as long as $\epsilon$ is bounded away from 1 then
\begin{equation}
    \frac{1}{(1 - \epsilon)^2} \l( \frac{\|O\|_\infty  + \epsilon}{1 - \epsilon} + 1 \r) = \O(\|O\|_\infty).
\end{equation}

Finally, we analyze the sample complexity required to achieve the error bound of Eq.~\eqref{eq:error_bound_full}. In Eqs.~\eqref{eq:X_eps} and \eqref{eq:Y_eps} we required that the number of samples $T$ be such that the shot noise of $\bar{X}$ and $\bar{Y}$ are at most $\tilde{\epsilon}$. These random variables correspond to the observables $\{O_j\}_{j=1}^L \cup \{S_\lambda/s_\lambda\}_{\lambda \in R'}$. Standard classical-shadows theory informs us that
\begin{equation}
    T = \O\l( \frac{\log((L + |R'|)/\delta)}{\tilde{\epsilon}^2} \max_{\substack{1 \leq j \leq L\\\lambda \in R'}}\l\{ \V[\hat{o}_j], \V\l[\frac{\hat{s}_\lambda}{s_\lambda} \r] \r\} \r)
    \end{equation}
suffices to accomplish this task (with probability at least $1 - \delta$)~\cite{huang2020predicting}. Then, setting $\tilde{\epsilon} = \min_{\lambda \in R'} \xi_\lambda \epsilon$ ensures that $\tilde{\epsilon}$ is small enough for Eq.~\eqref{eq:error_bound_full} to apply to all target observables.
\end{proof}

\section{Subsystem-symmetrized Pauli shadows}\label{sec:pauli_shadows_appendix}

Here we prove the properties of the subsystem-symmetrized Pauli shadows introduced in the main text. In Appendix~\ref{subsec:ss-pauli_irreps} we identify the irreps, and in Appendix~\ref{subsec:ss-pauli_variance} we bound the variance of observables under this protocol, particularly the symmetry operators $S_k$ obtained from $M = \sum_{i \in [n]} Z_i$.

\subsection{Irreducible representations}\label{subsec:ss-pauli_irreps}

Recall that the subsystem-symmetrized local Clifford group is the direct product
\begin{equation}
    \SymCl{n} \coloneqq \Sym(n) \times \Cl(1)^{\otimes n},
\end{equation}
where $\pi \in \Sym(n)$ acts on $(\C^2)^{\otimes n}$ as
\begin{equation}
    S_\pi \ket{b} = \ket{\pi^{-1}(b)}.
\end{equation}
For shorthand, we write $\pi(b)$ for the $n$-bit string $b_{\pi(0)} \cdots b_{\pi(n-1)}$. It is clear that the adjoint representation $\mathcal{U}_{(\pi, C)}$ block diagonalizes into subspaces spanned by $k$-local Pauli operators:
\begin{equation}
    V_k \coloneqq \spn\{ P \in \Pauli(n) : |P| = k \}.
\end{equation}
This can be seen from the fact that neither single-qubit nor $\SWAP$ gates can change the operator locality;~however, $\SWAP$ gates \emph{can} map between equally sized subsystems on which the operator nontrivially acts. What remains is to show that each of these subspaces is irreducible.

First, we define the twirling map.

\begin{definition}\label{def:twirl}
    Let $\phi : G \to \U(V)$ be a unitary representation of a compact group $G$ on a vector space $V$, and let $\Phi : G \to \U(\mathcal{L}(V))$ be its adjoint action, i.e., $\Phi_g(\cdot) = \phi_g (\cdot) \phi_g^\dagger$. The $t$-fold twirl by $\Phi$ is defined as
    \begin{equation}
    \mathcal{T}_{t,\Phi} \coloneqq \E_{g \sim G} \Phi_g^{\otimes t},
    \end{equation}
    which is a linear map on $\mathcal{L}(V)^{\otimes t}$.
\end{definition}

Twirls have a number of convenient properties, mostly arising from the fact that $\Phi$ is a group homomorphism. For example, they are $G$-invariant from the left and right:
\begin{equation}
    \Phi_h^{\otimes t} \circ \mathcal{T}_{t,\Phi} = \mathcal{T}_{t,\Phi} = \mathcal{T}_{t,\Phi} \circ \Phi_h^{\otimes t}
\end{equation}
for all $h \in G$. This furthermore implies that they are in fact projectors:
\begin{equation}
    \mathcal{T}_{t,\Phi}^2 = \mathcal{T}_{t,\Phi}.
\end{equation}
The study of twirls also allows us to determine the irreducible representations of a group. This can be seen by the following well-known result for multiplicity-free groups, which for completeness we provide a self-contained proof of at the end of this subsection.

\begin{proposition}\label{prop:twirl_irreps}
    Let $G$, $V$, $\phi$, and $\Phi$ be as in Definition~\ref{def:twirl}. For any $X \in \mathcal{L}(V)$, the $1$-twirl of $X$ by $\Phi$ takes the form
    \begin{equation}
    \mathcal{T}_{1,\Phi}(X) = \sum_{\lambda \in R_G} \frac{\tr(X \Pi_\lambda)}{\tr(\Pi_\lambda)} \Pi_\lambda
    \end{equation}
    if and only if $\phi$ decomposes irreducibly as $V = \bigoplus_{\lambda \in R_G} V_\lambda$, where $\Pi_\lambda$ is the orthogonal projector onto $V_\lambda$.
\end{proposition}

Our strategy for determining the irreps of $G = \SymCl{n}$ is therefore to directly compute $\mathcal{T}_{1,\Phi}$, from which we can infer the irreps from its block-diagonal structure. To use Proposition~\ref{prop:twirl_irreps}, we will take $\phi$ as the unitary channel $\mathcal{U}$, so that $V = \mathcal{L}(\mathcal{H})$ and $\Phi(\cdot) = \mathcal{U}(\cdot)\mathcal{U}^\dagger$ (note that this is a superchannel). For technical reasons, it will be easier to first compute $\mathcal{T}_{2, \mathcal{U}}$, from which the desired twirl $\mathcal{T}_{1,\Phi}$ can be evaluated. The relation between these two twirls is given by the following lemma.

\begin{lemma}\label{lem:1twirl_to_2twirl}
    Let $\mathcal{U} : G \to \U(\mathcal{L}(\mathcal{H}))$ be a unitary representation and $\Phi : G \to \U(\mathcal{L}(\mathcal{L}(\mathcal{H})))$ its adjoint representation, i.e., $\Phi_g(\mathcal{A}) = \mathcal{U}_g^\dagger \mathcal{A} \mathcal{U}_g$ for any superoperator $\mathcal{A}$. The $1$-twirl by $\Phi$ can be computed from the $2$-twirl by $\mathcal{U}$ as
    \begin{equation}\label{eq:1-twirl_from_2-twirl}
    \mathcal{T}_{1,\Phi}(\mathcal{A})(X) = \tr_2\l[ \mathcal{T}_{2,\mathcal{U}} (\mathcal{A}) (\I \otimes X) \r],
    \end{equation}
    for all $\mathcal{A} \in \mathcal{L}(\mathcal{L}(\mathcal{H}))$ and $X \in \mathcal{L}(\mathcal{H})$. Here, the domain of $\mathcal{T}_{2,\mathcal{U}}$ is understood with respect to the isomorphism $\mathcal{L}(\mathcal{L}(\mathcal{H})) \cong \mathcal{L}(\mathcal{H})^{\otimes 2}$, given by
    \begin{equation}\label{eq:2-twirl_isomorphism}
    \vop{A}{B} \cong A \otimes B^\dagger.
    \end{equation}
\end{lemma}

\begin{proof}
    Write $\mathcal{A} = \sum_{i,j \in [d^2]} \mathcal{A}_{ij} \vop{B_i}{B_j}$, where $\mathcal{A}_{ij} \in \C$ and $\{B_i\}_{i \in [d^2]}$ is an orthonormal operator basis. By a direct calculation:
    \begin{align}
    \mathcal{T}_{1,\Phi}(\mathcal{A})(X) &= \E_{g \sim G} \Phi_g(\mathcal{A})(X) \notag\\
    &= \E_{g \sim G} \mathcal{U}_g^\dagger \mathcal{A} \mathcal{U}_g(X) = \E_{g \sim G} \mathcal{U}_g^\dagger \mathcal{A}(U_g X U_g^\dagger) \notag\\
    &= \E_{g \sim G} \mathcal{U}_g^\dagger \sum_{i,j \in [d^2]} \mathcal{A}_{ij} \kket{B_i} \vip{B_j}{U_g X U_g^\dagger} \notag\\
    &= \E_{g \sim G} \sum_{i,j \in [d^2]} \mathcal{A}_{ij} U_g^\dagger B_i U_g \tr(B_j^\dagger U_g X U_g^\dagger) \notag\\
    &= \E_{g \sim G} \tr_2\l[ \sum_{i,j \in [d^2]} \mathcal{A}_{ij} ( U_g^\dagger B_i U_g )   \otimes ( U_g^\dagger B_j^\dagger U_g X ) \r] \notag\\
    &= \tr_2\l[ \sum_{i,j \in [d^2]} \mathcal{A}_{ij} \mathcal{T}_{2,\mathcal{U}}(B_i \otimes B_j^\dagger)(\I \otimes X) \r] \notag\\
    &= \tr_2\l[ \mathcal{T}_{2,\mathcal{U}} \l( \sum_{i,j \in [d^2]} \mathcal{A}_{ij} \vop{B_i}{B_j} \r) (\I \otimes X) \r] \notag\\
    &= \tr_2\l[ \mathcal{T}_{2,\mathcal{U}} (\mathcal{A}) (\I \otimes X) \r].
    \end{align}
\end{proof}

Before we can compute $\mathcal{T}_{2, \mathcal{U}}$ for the subsystem-symmetrized local Clifford group, we will need a small result about the group orbit of a $k$-local Pauli operator $P$ under the action of $\mathcal{U}$. The orbit is defined as
\begin{equation}
    G \cdot P \coloneqq \{ \mathcal{U}_g(P) : g \in G \}.
\end{equation}
This will help us determine how the twirl acts on Pauli operators, which as an basis is used to compute the matrix elements of $\mathcal{T}_{2, \mathcal{U}}$. To this end, we define an orthonormal basis of $k$-local Pauli operators,
\begin{equation}
    \mathcal{B}_k \coloneqq \{ P \in \Pauli(n) / \sqrt{d} : |P| = k \},
\end{equation}
which contains $|\mathcal{B}_k| = 3^k \binom{n}{k}$ elements.

\begin{lemma}\label{lem:sspauli_orbit}
   Let $G = \SymCl{n}$. The orbit $G \cdot P$ of any $P \in \mathcal{B}_k$ is equal to $\pm \mathcal{B}_k$, i.e., the set of all signed $k$-local Pauli operators.
\end{lemma}

\begin{proof}
    Let the nontrivial support of $P$ be $I \subseteq [n]$, $|I| = k$. For each (normalized) Pauli matrix acting on subsystem $I$, its orbit by all single-qubit Clifford gates is $\pm \{X, Y, Z\} / \sqrt{2}$. Meanwhile, the trivial factors $\I/\sqrt{2}$ acting on $[n] \setminus I$ are invariant to any unitary transformation. Therefore $\Cl(1)^{\otimes n} \cdot P$ is the set of all normalized Pauli operators acting nontrivially only on the qubits in $I$ (with both signs $\pm 1$).
    
    Then, conjugation by $S_\pi$ for arbitrary $\pi \in \Sym(n)$ permutes the $k$ nontrivial factors of $P$ among the $n$ qubits. The orbit over all permutations yields all possible $\binom{n}{k}$ supports. Taking the direct product of both these Clifford- and symmetric-group actions therefore yields all $k$-local Pauli operators, with prefactors $\pm 1$.
\end{proof}

We are now ready to compute the $2$-fold twirl by $\mathcal{U}$. We comment that the high-level proof structure of this lemma is inspired by that of Ref.~\cite[Section~IV A 1]{wan2023matchgate}.

\begin{lemma}\label{lem:2-twirl_pauli}
    Let $\mathcal{U} : G \to \U(\mathcal{L}(\mathcal{H}))$ be the unitary representation of $G = \SymCl{n}$, defined by $\mathcal{U}_{(\pi, C)}(\rho) = S_\pi C \rho C^\dagger S_\pi^\dagger$. Its $2$-fold twirl is the projector
    \begin{equation}\label{eq:2-twirl-sspauli}
    \mathcal{T}_{2,\mathcal{U}} = \sum_{k=0}^n \vop{\Sigma_k^{(2)}}{\Sigma_k^{(2)}},
    \end{equation}
    where $\kket{\Sigma_k^{(2)}} \in \mathcal{L}(\mathcal{H})^{\otimes 2}$ is defined as
    \begin{equation}\label{eq:2-twirl-basis}
    \kket{\Sigma_k^{(2)}} = \frac{1}{\sqrt{3^k \binom{n}{k}}} \sum_{P \in \mathcal{B}_k} \kket{P} \kket{P}.
    \end{equation}
\end{lemma}

\begin{proof}
    First, we will establish that for any two basis Pauli operators $P \neq Q$, we have $\mathcal{T}_{2,\mathcal{U}} \kket{P}\kket{Q} = 0$. Thus we only need to consider basis elements of $\mathcal{L}(\mathcal{H})^{\otimes 2}$ of the form $\kket{P}\kket{P}$. Next, we will show that $\mathcal{T}_{2,\mathcal{U}} \kket{P}\kket{P} = \mathcal{T}_{2,\mathcal{U}} \kket{P'}\kket{P'}$ whenever $|P| = |P'|$. Finally, using these two properties we can derive Eq.~\eqref{eq:2-twirl-sspauli}.

    Fix the basis of Pauli operators such that $P, Q \in \bigcup_{0 \leq k \leq n} \mathcal{B}_k$. If $P \neq Q$, then there exists at least one qubit $i \in [n]$ on which $P$ and $Q$ act as a different Pauli matrix. Hence there always exists some $W \in \Pauli(1)$ which anticommutes with one and commutes with the other, e.g., $\mathcal{W}_i \kket{P} = W_i P W_i^\dagger = -P$ and $\mathcal{W}_i \kket{Q} = Q$. Note that $\mathcal{W}_i$ is equal to $\mathcal{U}_{(e, W_i)}$ where $e \in \Sym(n)$ is the identity permutation. Thus using the property that $\mathcal{T}_{2,\mathcal{U}} \circ \mathcal{W}_i^{\otimes 2} = \mathcal{T}_{2,\mathcal{U}}$, we have
    \begin{equation}\label{eq:PQ_vanish}
    \begin{split}
    \mathcal{T}_{2,\mathcal{U}} \kket{P}\kket{Q} &= \mathcal{T}_{2,\mathcal{U}} \mathcal{W}_i^{\otimes 2} \kket{P}\kket{Q}\\
    &= -\mathcal{T}_{2,\mathcal{U}} \kket{P}\kket{Q},
    \end{split}
    \end{equation}
    implying that $\mathcal{T}_{2,\mathcal{U}} \kket{P}\kket{Q} = 0$.

    Now let $P, P'$ be $k$-local Pauli operators for any $k$. If they act nontrivially on different subsets $I, I' \subseteq [n]$ of qubits, then let $\pi \in \Sym(n)$ be a permutation that maps $I$ to $I'$. Given this permutation, if they act as different Pauli matrices on their new shared support $I'$, then furthermore let $C_i \in \Cl(1)$ for $i \in I'$ be Clifford gates that map each one to the other. Writing $C = \bigotimes_{i \in I'} C_i \otimes \I^{\otimes(n-k)}$, this transformation acts as $\mathcal{U}_{(\pi, C)}^{\otimes 2} \kket{P}\kket{P} = \kket{P'}\kket{P'}$, which implies that
    \begin{equation}\label{eq:PP_same}
    \begin{split}
    \mathcal{T}_{2,\mathcal{U}} \kket{P}\kket{P} &= \mathcal{T}_{2,\mathcal{U}} \mathcal{U}_{(\pi, C)}^{\otimes 2} \kket{P}\kket{P}\\
    &= \mathcal{T}_{2,\mathcal{U}} \kket{P'}\kket{P'}.
    \end{split}
    \end{equation}

    We are now ready to derive Eq.~\eqref{eq:2-twirl-sspauli}. As established by Eq.~\eqref{eq:PQ_vanish}, we only need to expand the 2-fold twirl in the basis of $\kket{P}\kket{P}$:
    \begin{equation}
    \begin{split}
    \mathcal{T}_{2,\mathcal{U}} &= \sum_{P, P' \in \Pauli(n)} \bbra{P}\bbra{P} \mathcal{T}_{2,\mathcal{U}} \kket{P'} \kket{P'} \kket{P} \vop{P}{P'} \bbra{P'}\\
    &= \sum_{k=0}^n \sum_{P, P' \in \mathcal{B}_k} \bbra{P}\bbra{P} \mathcal{T}_{2,\mathcal{U}} \kket{P'} \kket{P'} \kket{P} \vop{P}{P'} \bbra{P'},
    \end{split}
    \end{equation}
    where the second simplification is due to the fact that $\mathcal{U}_{(\pi, C)}$ preserves Pauli locality, hence $\bbra{P}\bbra{P} \mathcal{T}_{2,\mathcal{U}} \kket{P'} \kket{P'} = 0$ whenever $|P| \neq |P'|$. Now we invoke Eq.~\eqref{eq:PP_same}, which implies that $\bbra{P}\bbra{P} \mathcal{T}_{2,\mathcal{U}} \kket{P'} \kket{P'} = c_k'$ for all $P, P' \in \mathcal{B}_k$ (i.e., the matrix element does not depend on the particular choice of $P, P'$). Hence
    \begin{equation}
    \begin{split}
    \mathcal{T}_{2,\mathcal{U}} &= \sum_{k=0}^n c_k' \sum_{P, P' \in \mathcal{B}_k} \kket{P}\kket{P} \bbra{P'}\bbra{P'}\\
    &= \sum_{k=0}^n c_k \vop{\Sigma_k^{(2)}}{\Sigma_k^{(2)}},
    \end{split}
    \end{equation}
    where we have rescaled $c_k = c_k' 3^k \binom{n}{k}$ to account for the normalization of $\kket{\Sigma_k^{(2)}}$.

    Finally, we show that all $c_k = 1$ by proving that $\mathcal{T}_{2,\mathcal{U}} \kket{\Sigma_k^{(2)}} = \kket{\Sigma_k^{(2)}}$.  Expand the expression:
    \begin{equation}\label{eq:T2_basis}
    \mathcal{T}_{2,\mathcal{U}} \kket{\Sigma_k^{(2)}} = \frac{1}{|G|} \sum_{g \in G} \frac{1}{\sqrt{|\mathcal{B}_k|}} \sum_{P \in \mathcal{B}_k} \mathcal{U}_{g}\kket{P} \otimes \mathcal{U}_{g}\kket{P}.
    \end{equation}
    We first compute the average over the group for some fixed $P$. By Lemma~\ref{lem:sspauli_orbit}, we know that the orbit $G \cdot P = \{ \mathcal{U}_g\kket{P} : g \in G \} = \pm\mathcal{B}_k$. Thus
    \begin{equation}
    \sum_{g \in G} \mathcal{U}_{g}\kket{P} \otimes \mathcal{U}_{g}\kket{P} = 2 \frac{|G|}{|G \cdot P|} \sum_{Q \in \mathcal{B}_k} (\pm 1)^2 \kket{Q} \kket{Q},
    \end{equation}
    where the factor of 2 is due to the fact that for each $Q \in \mathcal{B}_k$, both $\pm Q \in G \cdot P$, and the factor of $|G|/|G \cdot P|$ takes care of double counting when summing over all elements of $G$. Noting that $|G \cdot P| = 2|\mathcal{B}_k|$, we can plug this result into Eq.~\eqref{eq:T2_basis} to find that
    \begin{equation}
    \begin{split}
    \mathcal{T}_{2,\mathcal{U}} \kket{\Sigma_k^{(2)}} &= \frac{1}{|\mathcal{B}_k|^{3/2}} \sum_{P \in \mathcal{B}_k} \sum_{Q \in \mathcal{B}_k} \kket{Q}\kket{Q}\\
    &= \frac{1}{\sqrt{|\mathcal{B}_k|}} \sum_{Q \in \mathcal{B}_k} \kket{Q}\kket{Q}\\
    &= \kket{\Sigma_k^{(2)}},
    \end{split}
    \end{equation}
    as desired.
\end{proof}

We are now ready to prove the main result of this section:~the irreps of $\SymCl{n}$ are labeled by the Pauli weights $k \in \{0, 1, \ldots, n\}$. The proof structure is as follows:~from the expression for $\mathcal{T}_{2, \mathcal{U}}$ from Lemma~\ref{lem:2-twirl_pauli}, we can compute $\mathcal{T}_{1, \Phi}$ by using Lemma~\ref{lem:1twirl_to_2twirl}. Then by examining $\mathcal{T}_{1, \Phi}$, we use Proposition~\ref{prop:twirl_irreps} to infer the irreps.

\begin{theorem}
    The representation $\mathcal{U} : \SymCl{n} \to \U(\mathcal{L}(\mathcal{H}))$, defined by $\mathcal{U}_{(\pi, C)}(\rho) = S_\pi C \rho C^\dagger S_\pi^\dagger$, decomposes into the irreps
    \begin{equation}
    V_k = \spn(\mathcal{B}_k), \quad k \in \{0, 1, \ldots, n\}.
    \end{equation}
\end{theorem}

\begin{proof}
    From Lemma~\ref{lem:2-twirl_pauli}, we have
    \begin{equation}
    \mathcal{T}_{2,\mathcal{U}} = \sum_{k=0}^n \vop{\Sigma_k^{(2)}}{\Sigma_k^{(2)}},
    \end{equation}
    where $\kket{\Sigma_k^{(2)}}$ is defined in Eq.~\eqref{eq:2-twirl-basis}. Using Lemma~\ref{lem:1twirl_to_2twirl}, we compute $\mathcal{T}_{1,\Phi}(\mathcal{A})$ by evaluating $\mathcal{T}_{2,\mathcal{U}}(\mathcal{A})$ for arbitrary superoperators $\mathcal{A}$. Let us express $\mathcal{A}$ in the Pauli basis:
    \begin{equation}
    \mathcal{A} = \sum_{k,\ell=0}^n \sum_{P \in \mathcal{B}_k} \sum_{Q \in \mathcal{B}_\ell} \mathcal{A}_{PQ} \vop{P}{Q}.
    \end{equation}
    Recall from Eq.~\eqref{eq:2-twirl_isomorphism} that in order to evaluate $\mathcal{T}_{2,\mathcal{U}}(\mathcal{A})$, we need $\mathcal{T}_{2,\mathcal{U}}(Q \otimes P)$ for every $P, Q$. But because $\mathcal{T}_{2,\mathcal{U}}$ projects onto symmetrized basis elements $P \otimes P$, we only have to consider the case where $P = Q$:
    \begin{align}
    \mathcal{T}_{2,\mathcal{U}}\kket{P \otimes P} &= \kket{\Sigma_k^{(2)}} \vip{\Sigma_k^{(2)}}{P \otimes P} \notag\\
    &= \kket{\Sigma_k^{(2)}} \frac{1}{\sqrt{|\mathcal{B}_k|}} \sum_{P' \in \mathcal{B}_k} \vip{P'}{P} \vip{P'}{P} \notag\\
    &= \frac{1}{\sqrt{|\mathcal{B}_k|}} \kket{\Sigma_k^{(2)}},
    \end{align}
    where $|\mathcal{B}_k| = 3^k \binom{n}{k}$.

    Inserting this result into Eq.~\eqref{eq:1-twirl_from_2-twirl} yields
    \begin{align}
    \mathcal{T}_{1,\Phi}(\mathcal{A})(X) &= \tr_1\l[ \mathcal{T}_{2,\mathcal{U}}(\mathcal{A}) (X \otimes \I) \r]\\
    &= \tr_1\l[ \sum_{k=0}^n \sum_{P \in \mathcal{B}_k} \mathcal{A}_{PP} \frac{1}{\sqrt{|\mathcal{B}_k|}} \Sigma_k^{(2)} (X \otimes \I) \r] \notag\\
    &= \tr_1\l[ \sum_{k=0}^n \sum_{P \in \mathcal{B}_k} \mathcal{A}_{PP} \frac{1}{\sqrt{|\mathcal{B}_k|}} \r. \notag\\
    &\quad \times \l. \frac{1}{\sqrt{|\mathcal{B}_k|}} \sum_{P' \in \mathcal{B}_k} (P' \otimes P') (X \otimes \I) \r] \notag\\
    &= \sum_{k=0}^n \frac{\sum_{P \in \mathcal{B}_k} \mathcal{A}_{PP}}{|\mathcal{B}_k|} \sum_{P' \in \mathcal{B}_k} \tr(P' X) P'. \notag
    \end{align}
    We make a number of observations here. First, note that $\sum_{P \in \mathcal{B}_k} \mathcal{A}_{PP} = \tr(\mathcal{A} \Pi_k)$, where $\Pi_k = \sum_{P \in \mathcal{B}_k} \vop{P}{P}$.  Also, $|\mathcal{B}_k| = \tr(\Pi_k)$. Finally, the sum over $P'$ can be represented as
    \begin{equation}
    \begin{split}
    \sum_{P' \in \mathcal{B}_k} \tr(P' X) P' &= \sum_{P' \in \mathcal{B}_k} \kket{P'} \vip{P'}{X}\\
    &= \Pi_k \kket{X}.
    \end{split}
    \end{equation}
    Because this holds for all $X \in \mathcal{L}(\mathcal{H})$, we can say that
    \begin{equation}
    \mathcal{T}_{1,\Phi}(\mathcal{A}) = \sum_{k=0}^n \frac{\tr(\mathcal{A} \Pi_k)}{\tr(\Pi_k)} \Pi_k.
    \end{equation}
    By Proposition~\ref{prop:twirl_irreps}, we know that the twirl has this expression if and only if the irreducible subspaces of $\mathcal{U}$ are $V_k = \spn(\mathcal{B}_k)$.
\end{proof}

Finally, we close this subsection with the deferred proof of the well-known result Proposition~\ref{prop:twirl_irreps}, for completeness.

\begin{proof}[Proof (of Proposition~\ref{prop:twirl_irreps})]
    For the forward direction, suppose $\phi = \bigoplus_{\lambda \in R_G} \phi^{(\lambda)}$ where each $\phi^{(\lambda)} : G \to \U(V_\lambda)$ is irreducible. (This is guaranteed by Maschke's theorem, and generalizes to the Peter--Weyl theorem for compact groups~\cite{fulton2004representation}.) Because $\mathcal{T}_{1,\Phi}(X)$ commutes with all $\phi_g$, they are simultaneously block diagonal, so $\mathcal{T}_{1,\Phi}(X) = \bigoplus_{\lambda \in R_G} \mathcal{T}_{1,\Phi^{(\lambda)}}(X)$ where $\Phi^{(\lambda)}(\cdot) = \phi^{(\lambda)}(\cdot)\phi^{(\lambda) \dagger}$. Because $\phi^{(\lambda)}$ is irreducible, by Schur's lemma $\mathcal{T}_{1,\Phi^{(\lambda)}}(X)$ must be a multiple of the identity on $V_\lambda$. Therefore
    \begin{equation}
    \mathcal{T}_{1,\Phi}(X) = \bigoplus_{\lambda \in R_G} c_\lambda(X) \I_{V_\lambda} = \sum_{\lambda \in R_G} c_\lambda(X) \Pi_\lambda.
    \end{equation}
    From the orthogonality of projectors $\Pi_\lambda \Pi_{\lambda'} = \delta_{\lambda \lambda'} \Pi_\lambda$, the scalar $c_\lambda(X)$ is determined by
    \begin{equation}
    \begin{split}
    c_\lambda(X) \tr(\Pi_\lambda) &= \tr(\Pi_\lambda \mathcal{T}_{1,\Phi}(X))\\
    &= \E_{g \sim G} \tr\l( \phi_g^\dagger \Pi_\lambda \phi_g X \r)\\
    &= \E_{g \sim G} \tr(\Pi_\lambda X) = \tr(\Pi_\lambda X).
    \end{split}
    \end{equation}

    For the reverse direction, suppose the twirl takes the form
    \begin{equation}
    \mathcal{T}_{1,\Phi}(X) = \bigoplus_{\lambda} \frac{\tr(X \Pi_\lambda)}{\tr(\Pi_\lambda)} \I_{V_\lambda},
    \end{equation}
    where we denote each block by $\mathcal{T}_{1,\Phi^{(\lambda)}}(X)$. Again because $\mathcal{T}_{1,\Phi}(X)$ and $\phi_g$ commute, the matrix $\phi_g$ is block diagonal in the subspaces $V_\lambda$ for all $g \in G$. We need to show that each block $\phi^{(\lambda)}$ is irreducible.

    Recall that $\phi^{(\lambda)}$ is irreducible if the only subspaces $W \subseteq V_\lambda$ for which $\phi^{(\lambda)}_G(W) \subseteq W$ are $W = \{0\}$ or $W = V_\lambda$. Indeed, let $W \subseteq V_\lambda$ be a subspace such that for any $\ket{v} \in W$ and $g \in G$, $\phi^{(\lambda)}_g\ket{v} \in W$. Suppose there exists a vector $\ket{x} \in V_\lambda$ that is orthogonal to $W$ and set $X = \op{x}{x}$. Then
    \begin{equation}\label{eq:project_W_1}
    \mathcal{T}_{1,\Phi^{(\lambda)}}(\op{x}{x}) \ket{v} = \E_{g \sim G} \phi^{(\lambda)}_g \ket{x} \ev{x}{(\phi^{(\lambda)}_g)^\dagger}{v} = 0
    \end{equation}
    because all $(\phi^{(\lambda)}_g)^\dagger\ket{v} \in W$. However, from $\mathcal{T}_{1,\Phi^{(\lambda)}}(X) = c_\lambda(X) \I_{V_\lambda}$ we see that also
    \begin{equation}\label{eq:project_W_2}
    \mathcal{T}_{1,\Phi}(\op{x}{x})_\lambda \ket{v} = \frac{\ip{x}{x}}{\dim V_\lambda} \ket{v}.
    \end{equation}
    Supposing $\ket{x} \neq 0$, we see that $\ket{v} = 0$ is the only possible element of $W$ to satisfy Eqs.~\eqref{eq:project_W_1} and \eqref{eq:project_W_2} simultaneously. Hence $W = \{0\}$. Otherwise, $\ket{x} = 0$ is the only element of $V_\lambda$ orthogonal to $W$, implying that there is in fact no nontrivial subspace orthogonal to $W$. Thus $W = V_\lambda$ in this case.
\end{proof}

\subsection{Variance of symmetry operators}\label{subsec:ss-pauli_variance}

In this section, we analyze the variance associated with the symmetry operators,
\begin{align}
    \frac{S_1}{s_1} &= \frac{1}{m} \sum_{i \in [n]} Z_i,\\
    \frac{S_2}{s_2} &= \frac{2}{m^2 - n} \sum_{i < j} Z_i Z_j.
\end{align}
Because our error analysis of symmetry-adjusted classical shadows bootstraps from the variance of unmitigated estimation, we only need to compute quantities related to the noiseless protocol. In this case, the subsystem-symmetrized local Clifford group yields the same channel and variances as the standard local Clifford group because the random permutations have no effect on the twirling on computational basis states.

Specifically, using the two- and three-fold twirls we can express
\begin{equation}
    \mathcal{M}(\rho) = \tr_1\l[ \sum_{b \in \{0,1\}^n} \mathcal{T}_{2,G}\l( \op{b}{b}^{\otimes 2} \r) (\rho \otimes \I) \r]
\end{equation}
and
\begin{align}
    \V_\rho[\hat{o}] &= \tr\l[ \sum_{b \in \{0,1\}^n} \mathcal{T}_{3,G}\l( \op{b}{b}^{\otimes 3} \r) \l( \rho \otimes \mathcal{M}^{-1}(O)^{\otimes 2} \r) \r] \notag\\
    &\quad - \tr(O \rho)^2,
\end{align}
where the $t$-fold twirl by $U : G \to \U(\mathcal{H})$ is defined as $\mathcal{T}_{t,G} \coloneqq \E_{g \sim G} \mathcal{U}_g^{\otimes t}$. These expressions are the same whether we take $G = \SymCl{n}$ or $\Cl(1)^{\otimes n}$, due to the following equivalence:
% \begin{widetext}
\begin{equation}
\begin{split}
    \sum_{b \in \{0,1\}^n} \mathcal{T}_{t,\SymCl{n}}\l( \op{b}{b}^{\otimes t} \r) &= \sum_{b \in \{0,1\}^n} \E_{(\pi, C) \sim \SymCl{n}} \l[ \l( C^\dagger S_\pi^\dagger \op{b}{b} S_\pi C \r{)^{\otimes t}} \r]\\
    &= \E_{C \sim \Cl(1)^{\otimes n}} \l[ (C^\dagger)^{\otimes t} \E_{\pi \sim \Sym(n)} \l[ \sum_{b \in \{0,1\}^n} \op{\pi(b)}{\pi(b)}^{\otimes t} \r] C^{\otimes t} \r]\\
    &= \E_{C \sim \Cl(1)^{\otimes n}} \l[ (C^\dagger)^{\otimes t} \sum_{b \in \{0,1\}^n} \op{b}{b}^{\otimes t} C^{\otimes t} \r]\\
    &= \sum_{b \in \{0,1\}^n} \mathcal{T}_{t,\Cl(1)^{\otimes n}}\l( \op{b}{b}^{\otimes t} \r).
\end{split}
\end{equation}
The third equality follows due to the fact that permutations are bijections, hence each $\sum_{b \in \{0,1\}^n} \op{\pi(b)}{\pi(b)}^{\otimes t}$ is just a reordering of the terms in $\sum_{b \in \{0,1\}^n} \op{b}{b}^{\otimes t}$.

As an immediate consequence, we see that the variance of observables under subsystem symmetrization are exactly the same as with standard Pauli shadows. For the rest of this section, we will explicitly compute the variance of $S_k$ using known Haar-averaging formulas over the Clifford group~\cite[Eqs.~(S35) and (S36)]{huang2020predicting}:
\begin{align}
    \E_{U \sim \Cl(1)} U^\dagger \op{x}{x} U \ev{x}{U A U^\dagger}{x} &= \frac{A + \tr(A) \I}{6}\\
    \E_{U \sim \Cl(1)} U^\dagger \op{x}{x} U \ev{x}{U B_0 U^\dagger}{x} \ev{x}{U C_0 U^\dagger}{x} &= \frac{\tr(B_0 C_0) \I + B_0 C_0 + C_0 B_0}{24},
\end{align}
for all unit vectors $\ket{x} \in \C^{2}$ and Hermitian matrices $A, B_0, C_0 \in \C^{2 \times 2}$, with $\tr B_0  = \tr C_0  = 0$. The extension to $\Cl(1)^{\otimes n}$ follows by linearity and statistical independence. We first apply these formulas to $S_1$ to compute its variance. Writing $C = \bigotimes_{i \in [n]} C_i$ and $\ket{b} = \bigotimes_{i \in [n]} \ket{b_i}$, we have
\begin{equation}\label{eq:Es1^2}
\begin{split}
    \E[\hat{s}_1^2] &= \tr\l( \rho \sum_{b \in \{0, 1\}^n} \E_{C \sim \Cl(1)^{\otimes n}} C^\dagger \op{b}{b} C \ev{b}{C \mathcal{M}^{-1}(S_1) C^\dagger}{b}^2 \r)\\
    &= \tr\l( \rho \sum_{b \in \{0, 1\}^n} \E_{C \sim \Cl(1)^{\otimes n}} C^\dagger \op{b}{b} C \times 3^2 \sum_{i, j \in [n]} \ev{b}{C Z_i C^\dagger}{b} \ev{b}{C Z_j C^\dagger}{b} \r)\\
    &= 9 \tr\l( \rho \sum_{i \in [n]} \sum_{b_i \in \{0, 1\}} \E_{C_i \sim \Cl(1)} C_i^\dagger \op{b_i}{b_i} C_i \ev{b_i}{C_i Z_i C_i^\dagger}{b_i}^2 \r)\\
    &\quad + 9 \times 2 \tr\l( \rho \sum_{i < j} \sum_{b_i, b_j \in \{0, 1\}} \E_{C_i, C_j \sim \Cl(1)} C_i^\dagger \op{b_i}{b_i} C_i \otimes C_j^\dagger \op{b_j}{b_j} C_j \ev{b_i}{C_i Z_i C_i^\dagger}{b_i} \ev{b_j}{C_j Z_j C_j^\dagger}{b_j} \r)\\
    &= 9 \tr\l( \rho \sum_{i \in [n]} \frac{\I}{3} \r) + 18 \tr\l( \rho \sum_{i < j} \frac{Z_i}{3} \frac{Z_j}{3} \r)\\
    &= 3n + \tr(S_2 \rho).
\end{split}
\end{equation}
Thus the variance is
\begin{equation}
\begin{split}
    \V_\rho\l[\frac{\hat{s}_1}{s_1}\r] &= \frac{1}{s_1^2} \l( \E[\hat{s}_1^2] - \E[\hat{s}_1]^2 \r)\\
    &= \frac{1}{m^2} \l( 3n + \tr(S_2 \rho) - \tr(S_1 \rho)^2 \r).
\end{split}
\end{equation}
If $\rho$ is the ideal state with symmetries $\tr(S_1 \rho) = m$ and $\tr(S_2 \rho) = m^2 - n$, then $\V_\rho[\hat{s}_1/s_1] = 2n/m^2$. On the other hand, if we make the noisy replacement $\rho \to \widetilde{\rho} = \mathcal{M}^{-1} \widetilde{\mathcal{M}}(\rho)$, then we can obtain a bound
\begin{equation}\label{eq:s1_noisy_var}
\begin{split}
    \V_{\widetilde{\rho}}\l[ \frac{\hat{s}_1}{s_1} \r] &= \frac{1}{m^2} \l( 3n + F_{Z,2}(m^2 - n) - (F_{Z,1} m)^2 \r)\\
    &\leq \frac{2n}{m^2} + 1.
\end{split}
\end{equation}

Next we compute the variance of estimating $S_2$. Analogous to the calculation presented in Eq.~\eqref{eq:Es1^2}, we expand $\ev{b}{C \mathcal{M}^{-1}(S_2) C^\dagger}{b}^2$ and group terms based on the overlapping of indices:
\begin{equation}\label{eq:S_2_squared}
\begin{split}
    \ev{b}{C \mathcal{M}^{-1}(S_2) C^\dagger}{b}^2 &= (3^2 \times 2)^2 \sum_{i < j} \sum_{k < l} \ev{b}{C Z_i Z_j C^\dagger}{b} \ev{b}{C Z_k Z_l C^\dagger}{b}\\
    &= (3^2 \times 2)^2 \l( \sum_{(i = k) < (j = l)} + \sum_{\substack{(i = k) < j, l \\ j \neq l}} + \sum_{k < (i = l) < j} \r.\\
    &\quad \l. + \sum_{i < (j = k) < l} + \sum_{\substack{i, k < (j = l) \\ i \neq k}} + \sum_{\substack{i < j; k < l \\ k \neq i \neq l; k \neq j \neq l}} \r) \ev{b}{C Z_i Z_j C^\dagger}{b} \ev{b}{C Z_k Z_l C^\dagger}{b}.
\end{split}
\end{equation}
We now go through each summation and evaluate the expectations:
\begin{equation}
\begin{split}
    &\quad \sum_{(i = k) < (j = l)} \sum_{b_i, b_j \in \{0, 1\}} \E_{C_i, C_j \sim \Cl(1)} C_i^\dagger \op{b_i}{b_i} C_i \otimes C_j^\dagger \op{b_j}{b_j} C_j \ev{b_i}{C_i Z_i C_i^\dagger}{b_i}^2 \ev{b_j}{C_j Z_j C_j^\dagger}{b_j}^2\\
    &= \sum_{i < j} \frac{\I}{3^2} = \frac{1}{3^2} \binom{n}{2} \I,
\end{split}
\end{equation}
\begin{equation}\label{eq:i=k<j,l}
\begin{split}
    &\quad \sum_{\substack{(i = k) < j, l \\ j \neq l}} \sum_{b_i, b_j, b_l \in \{0, 1\}} \E_{C_i, C_j, C_l \sim \Cl(1)} C_i^\dagger \op{b_i}{b_i} C_i \otimes C_j^\dagger \op{b_j}{b_j} C_j \otimes C_l^\dagger \op{b_l}{b_l} C_l\\
    &\qquad\qquad\qquad\qquad\qquad\qquad\qquad\quad \times \ev{b_i}{C_i Z_i C_i^\dagger}{b_i}^2 \ev{b_j}{C_j Z_j C_j^\dagger}{b_j} \ev{b_l}{C_l Z_l C_l^\dagger}{b_l}\\
    &= 2\sum_{i < j < l} \frac{\I}{3} \frac{Z_j}{3} \frac{Z_l}{3} = \frac{2}{3^3} \sum_{i <
    j < l} Z_j Z_l,
\end{split}
\end{equation}
\begin{equation}
\begin{split}
    &\quad \sum_{k < (i = l) < j} \sum_{b_i, b_j, b_k \in \{0, 1\}} \E_{C_i, C_j, C_k \sim \Cl(1)} C_i^\dagger \op{b_i}{b_i} C_i \otimes C_j^\dagger \op{b_j}{b_j} C_j \otimes C_k^\dagger \op{b_k}{b_k} C_k\\
    &\qquad\qquad\qquad\qquad\qquad\qquad\qquad\quad \times \ev{b_i}{C_i Z_i C_i^\dagger}{b_i}^2 \ev{b_j}{C_j Z_j C_j^\dagger}{b_j} \ev{b_k}{C_k Z_k C_k^\dagger}{b_k}\\
    &= \frac{1}{3^3} \sum_{k < i < j} Z_k Z_j,
\end{split}
\end{equation}
\begin{equation}
\begin{split}
    &\quad \sum_{i < (j = k) < l} \sum_{b_i, b_j, b_l \in \{0, 1\}} \E_{C_i, C_j, C_l \sim \Cl(1)} C_i^\dagger \op{b_i}{b_i} C_i \otimes C_j^\dagger \op{b_j}{b_j} C_j \otimes C_l^\dagger \op{b_l}{b_l} C_l\\
    &\qquad\qquad\qquad\qquad\qquad\qquad\qquad\quad \times \ev{b_i}{C_i Z_i C_i^\dagger}{b_i} \ev{b_j}{C_j Z_j C_j^\dagger}{b_j}^2 \ev{b_l}{C_l Z_l C_l^\dagger}{b_l}\\
    &= \frac{1}{3^3} \sum_{i < j < l} Z_i Z_l,
\end{split}
\end{equation}
\begin{equation}\label{eq:i,k<j=l}
\begin{split}
    &\quad \sum_{\substack{i, k < (j = l) \\ i \neq k}} \sum_{b_i, b_j, b_k \in \{0, 1\}} \E_{C_i, C_j, C_k \sim \Cl(1)} C_i^\dagger \op{b_i}{b_i} C_i \otimes C_j^\dagger \op{b_j}{b_j} C_j \otimes C_k^\dagger \op{b_k}{b_k} C_k\\
    &\qquad\qquad\qquad\qquad\qquad\qquad\qquad\quad \times \ev{b_i}{C_i Z_i C_i^\dagger}{b_i} \ev{b_j}{C_j Z_j C_j^\dagger}{b_j}^2 \ev{b_k}{C_k Z_k C_k^\dagger}{b_k}\\
    &= \frac{2}{3^3} \sum_{i < k < j} Z_i Z_k,
\end{split}
\end{equation}
\begin{equation}
\begin{split}
    \sum_{\substack{i < j; k < l \\ k \neq i \neq l; k \neq j \neq l}} \bigotimes_{q \in \{i, j, k, l\}} \sum_{b_q \in \{0, 1\}} \E_{C_q \sim \Cl(1)} C_q^\dagger \op{b_q}{b_q} C_q \ev{b_q}{C_q Z_q C_q^\dagger}{b_q} &= \frac{1}{3^4} \sum_{\substack{i < j; k < l \\ k \neq i \neq l; k \neq j \neq l}} Z_i Z_j Z_k Z_l\\
    &= \frac{6}{3^4} \sum_{i < j < k < l} Z_i Z_j Z_k Z_l.
\end{split}
\end{equation}
Eqs.~\eqref{eq:i=k<j,l} to \eqref{eq:i,k<j=l} can be combined by relabeling the indices and recognizing that the resulting three-index summation has $3\binom{n}{3} = (n - 2)\binom{n}{2}$ terms of the form $Z_p Z_q$ (symmetric across the index pairs $p, q \in \{i, j, k\}$ with $p \neq q$). Hence there are only $\binom{n}{2}$ unique terms, all of which are repeated $n - 2$ times:
\begin{equation}
\begin{split}
    \frac{2}{3^3} \sum_{i < j < k} ( Z_i Z_j + Z_i Z_k + Z_j Z_k ) &= \frac{2}{3^3} (n - 2) \sum_{i < j} Z_i Z_j.
\end{split}
\end{equation}
Combining these expressions, we obtain
% \begin{align}
%     \sum_{b \in \{0, 1\}^n} \E_{C \sim \Cl(1)^{\otimes n}} C^\dagger \op{b}{b} C \ev{b}{C \mathcal{M}^{-1}(S_2) C^\dagger}{b}^2 &= 3^4 \l[ \frac{1}{3^2} \binom{n}{2} \I + \frac{2}{3^3} \sum_{i < j < k} \l( Z_i Z_j + Z_i Z_k + Z_j Z_k \r) + \frac{6}{3^4} \sum_{i < j < k < l} Z_i Z_j Z_k Z_l \r] \notag\\
%     &= 9 \binom{n}{2} \I + 6 \sum_{i < j < k} \l( Z_i Z_j + Z_i Z_k + Z_j Z_k \r) + 6 \sum_{i < j < k < l} Z_i Z_j Z_k Z_l.
% \end{align}
\begin{equation}\label{eq:S2_2nd_moment}
    \sum_{b \in \{0, 1\}^n} \E_{C \sim \Cl(1)^{\otimes n}} C^\dagger \op{b}{b} C \ev{b}{C \mathcal{M}^{-1}(S_2) C^\dagger}{b}^2 = 36 \binom{n}{2} \I + 24 (n - 2) \sum_{i < j} Z_i Z_j + 24 \sum_{i < j < k < l} Z_i Z_j Z_k Z_l.
\end{equation}

Due to the presence of the four-body term, we will need the conserved quantity associated with $M^4$:
\begin{equation}\label{eq:M4}
\begin{split}
    M^4 &= (M^2)^2 = \l( n \I + 2 \sum_{i < j} Z_i Z_j \r{)^2}\\
    &= n^2 \I + 4n \sum_{i < j} Z_i Z_j + 4 \sum_{i < j} \sum_{k < l} Z_i Z_j Z_k Z_l.
\end{split}
\end{equation}
The four-index sum here can be grouped as we did in Eq.~\eqref{eq:S_2_squared}, and the conditions can be simplified as before. Along with the fact that $Z_i^2 = \I$, a straightforward calculation reveals
\begin{equation}
    \sum_{i < j} \sum_{k < l} Z_i Z_j Z_k Z_l = \binom{n}{2} \I + 2 (n - 2) \sum_{i < j} Z_i Z_j + 6 \sum_{i < j < k < l} Z_i Z_j Z_k Z_l.
\end{equation}
Plugging this into Eq.~\eqref{eq:M4}, combined with $\tr(M^4 \rho) = m^4$ and $\tr(S_2 \rho) = m^2 - n$, we arrive at
\begin{equation}
    \tr(S_4 \rho) = 24 \tr\l( \rho \sum_{i < j < k < l} Z_i Z_j Z_k Z_l \r) = m^4 - n^2 - 4\binom{n}{2} - (6n - 8)(m^2 - n),
\end{equation}
where $S_4 \coloneqq \Pi_4(M^4) = 4! \sum_{i < j < k < l} Z_i Z_j Z_k Z_l$. Finally, applying this result to Eqs.~\eqref{eq:S2_2nd_moment}, we can compute the variance with respect to an ideal state lying in the symmetry sector:
\begin{equation}
\begin{split}
    \V_\rho\l[\frac{\hat{s}_2}{s_2}\r] &= \frac{1}{(m^2 - n)^2}\l[ 36\binom{n}{2} + 12(n - 2) \tr(S_2 \rho) + \tr(S_4 \rho) - \tr(S_2 \rho)^2 \r]\\
    &= \frac{1}{(m^2 - n)^2} \l[ 32\binom{n}{2} + (6n - 16)(m^2 - n) + m^4 - n^2 \r] - 1.
\end{split}
\end{equation}
Again making the replacement $\rho \to \widetilde{\rho}$, we instead have the following bound:
\begin{equation}\label{eq:s2_noisy_var}
\begin{split}
    \V_{\widetilde{\rho}}\l[ \frac{\hat{s}_2}{s_2} \r] &= \frac{1}{(m^2 - n)^2}\l[ 36\binom{n}{2} + 12(n - 2) F_{Z,2} \tr(S_2 \rho) + F_{Z,4} \tr(S_4 \rho) \r] - F_{Z,2}^2\\
    &\leq \frac{1}{(m^2 - n)^2} \l[ 18n(n - 1) + 12(n - 2)m^2 + m^4 +6n^2 +8m^2 \r]\\
    &= \frac{1}{(m^2 - n)^2} \l[ 6n(4n - 3) + 12nm^2 + m^2(m^2 - 16) \r].
\end{split}
\end{equation}
% \end{widetext}

If $m = \Theta(1)$, which is the case in our numerical experiments of antiferromagnetic spin systems, then $\V_{\widetilde{\rho}}[ \hat{s}_2/s_2 ] = \O(1)$. Recall from Eq.~\eqref{eq:s1_noisy_var} that $\V_{\widetilde{\rho}}[ \hat{s}_1/s_1 ] = \O(n/m^2)$. Therefore the variance overhead of estimating these symmetry operators is, asymptotically,
\begin{equation}
    \max\l\{ \V_{\widetilde{\rho}}\l[ \frac{\hat{s}_1}{s_1}\r], \V_{\widetilde{\rho}}\l[ \frac{\hat{s}_2}{s_2} \r] \r\} = \O(n)
\end{equation}
when $\tr(M \rho) = m = \Theta(1)$.

Systems wherein $m$ depends on $n$ will require a case-by-case analysis, which we leave to the reader. As a pathological example, consider two different functions which are both $\Theta(\sqrt{n})$:~if $m = \alpha\sqrt{n}$ for some constant $\alpha \neq 1$, then $\V_{\widetilde{\rho}}[ \hat{s}_2/s_2 ] = \O(1)$. However, if instead $m = \sqrt{n + c}$ for some constant $c \neq 0$, then $\V_{\widetilde{\rho}}[ \hat{s}_2/s_2 ] = \O(n^2)$. In both cases, $\V_{\widetilde{\rho}}[ \hat{s}_1/s_1 ] = \O(1)$. Thus the specific form of $m = f(n)$ can drastically affect the asymptotic bounds here.

\section{Spin-adapted fermionic shadows}\label{sec:spin_adaptation}

It is well known that number-conserving fermion basis rotations which preserve spin symmetries can be block diagonalized according to the spin sectors, leading to savings in both classical and quantum resources. Here we show how to leverage spin symmetries for the broader class of fermionic Gaussian transformations, and in particular we construct a spin-adapted matchgate shadows protocol. As such, this scheme will be informationally complete only over spin-conserving observables.

Let $n_\uparrow$ and $n_\downarrow$ be the number of spin-up and spin-down fermionic modes, respectively. The total number of modes is $n = n_\uparrow + n_\downarrow$, and we order the labels such that all spin-up modes come first. Gaussian transformations which do not mix between different spin types are block diagonal,
\begin{equation}\label{eq:spin_preserving_transformation}
    Q = \begin{pmatrix}
    Q_\uparrow & 0\\
    0 & Q_\downarrow
    \end{pmatrix} \in \Orth(2n),
\end{equation}
where $Q_\sigma \in \Orth(2n_\sigma) \eqqcolon G_\sigma$ for each $\sigma \in \{\uparrow, \downarrow\}$. Let $G_{\mathrm{spin}} \coloneqq G_\uparrow \times G_\downarrow$ be the spin-adapted group, i.e., the set of all elements of the form of Eq.~\eqref{eq:spin_preserving_transformation}. In order to calculate properties of this ensemble for classical shadows, we shall use the fact that $U_Q$ is nearly equivalent to the tensor product $U_{Q_\uparrow} \otimes U_{Q_\downarrow}$, up to a factor that depends on the determinant of $Q_\uparrow$. More precisely, we have the following.

\begin{lemma}\label{lem:spin_tensor_product}
    For all block-diagonal $Q$ of the form of Eq.~\eqref{eq:spin_preserving_transformation}, the unitary can be written as
    \begin{equation}\label{eq:spin_preserving_tensor_product}
    U_Q = U_{Q_\uparrow} \otimes (P_\downarrow^{s(Q_\uparrow)} U_{Q_\downarrow})
    \end{equation}
    where
    \begin{equation}
    s(Q_\uparrow) = \begin{cases}
    1 & \text{if } \det(Q_\uparrow) = -1,\\
    0 & \text{else},
    \end{cases}
    \end{equation}
    and $P_\sigma = (-i)^{n_\sigma} \prod_{\mu=0}^{2n_\sigma - 1} c_{\mu,\sigma}$ is the parity operator on the $\sigma$-spin sector.
\end{lemma}

For technical reasons, we have introduced two new sets of Majorana operators $\{ c_{\mu,\sigma} : \mu \in [2n_\sigma] \}$ on each $n_\sigma$-mode Hilbert space, in order to talk about the different spin sectors in terms of the standard tensor product. Because the tensor product does not respect the antisymmetry of fermions, these new Majorana operators are related to the usual Majorana operators (acting on the full $n$-mode Hilbert space) via
\begin{align}
    \gamma_{\mu,\uparrow} &= c_{\mu,\uparrow} \otimes \I,\\
    \gamma_{\mu,\downarrow} &= P_\uparrow \otimes c_{\mu,\downarrow}.
\end{align}
Lemma~\ref{lem:spin_tensor_product} therefore addresses this technicality of maintaining the anticommutation relations when expressing $U_Q$ as a tensor product.

\begin{proof}[Proof (of Lemma~\ref{lem:spin_tensor_product})]
    It is clear that conjugation by $U_{Q_\sigma}$ transforms each $c_{\mu,\sigma}$ as desired. What we need to ensure is that the Majorana operators $\gamma_{\mu,\sigma}$ on the full Hilbert space transform properly. Indeed, Eq.~\eqref{eq:spin_preserving_tensor_product} performs the desired transformation;~for the spin-up sector, we simply have
    \begin{equation}
    U_Q \gamma_{\mu,\uparrow} U_Q^\dagger = U_{Q_\uparrow} c_{\mu,\uparrow} U_{Q_\uparrow}^\dagger \otimes \I.
    \end{equation}
    
    For the spin-down sector, we will make use of following commutation relations:
    \begin{align}
    P_\sigma U_{Q_\sigma} &= \det(Q_\sigma) U_{Q_\sigma} P_\sigma\\
    P_\sigma c_{\mu,\sigma} &= -c_{\mu,\sigma} P_\sigma.
    \end{align}
    Consider the case $\det(Q_\uparrow) = 1$. Then $s(Q_\uparrow) = 0$, and so $U_Q = U_{Q_\uparrow} \otimes U_{Q_\downarrow}$:
    \begin{equation}
    \begin{split}
    U_Q \gamma_{\mu,\downarrow} U_Q^\dagger &= U_{Q_\uparrow} P_{\uparrow} U_{Q_\uparrow}^\dagger \otimes U_{Q_\downarrow} c_{\mu,\downarrow} U_{Q_\downarrow}^\dagger\\
    &= P_\uparrow \otimes U_{Q_\downarrow} c_{\mu,\downarrow} U_{Q_\downarrow}^\dagger,
    \end{split}
    \end{equation}
    as desired. If instead $\det(Q_\uparrow) = -1$, then $U_Q = U_{Q_\uparrow} \otimes (P_\downarrow U_{Q_\downarrow})$ and so
    \begin{align}
    U_Q \gamma_{\mu,\downarrow} U_Q^\dagger &= U_{Q_\uparrow} P_{\uparrow} U_{Q_\uparrow}^\dagger \otimes P_\downarrow U_{Q_\downarrow} c_{\mu,\downarrow} U_{Q_\downarrow}^\dagger P_\downarrow^\dagger \notag\\
    &= (-P_\uparrow) \otimes (-\det(Q_\downarrow)^2 P_\downarrow P_\downarrow^\dagger U_{Q_\downarrow} c_{\mu,\downarrow} U_{Q_\downarrow}^\dagger)\notag\\
    &= P_\uparrow \otimes U_{Q_\downarrow} c_{\mu,\downarrow} U_{Q_\downarrow}^\dagger,
    \end{align}
    where we have used the fact that $P_\downarrow P_\downarrow^\dagger = P_\downarrow^2 = \I$.
\end{proof}

In the context of classical shadows, the appearance of $P_\downarrow$ is inconsequential because it merely acts as a phasing operator which appears directly before measurement. To see this, first observe that $\mathcal{M}_Z = \mathcal{M}_{Z,\uparrow} \otimes \mathcal{M}_{Z, \downarrow}$, where $\mathcal{M}_{Z, \sigma} = \sum_{b \in \{0, 1\}^{n_\sigma}} \vop{b}{b}$. Using the fact that $\mathcal{P}_\downarrow \kket{b} = P_\downarrow \op{b}{b} P_\downarrow^\dagger = \op{b}{b} = \kket{b}$, we have that that shadow channel of the spin-adapted ensemble is
\begin{align}
    \mathcal{M}_{G_\mathrm{spin}} &= \E_{Q_\uparrow, Q_\downarrow}\l[ \mathcal{U}_{Q_\uparrow}^\dagger \otimes (\mathcal{U}_{Q_\downarrow}^\dagger \mathcal{P}_\downarrow^{s(Q_\uparrow)}) (\mathcal{M}_{Z,\uparrow} \otimes \mathcal{M}_{Z, \downarrow}) \mathcal{U}_{Q_\uparrow} \otimes (\mathcal{P}_\downarrow^{s(Q_\uparrow)} \mathcal{U}_{Q_\downarrow}) \r]\\ % \r. \notag\\
    % &\qquad \times \l. \mathcal{U}_{Q_\uparrow} \otimes (\mathcal{P}_\downarrow^{s(Q_\uparrow)} \mathcal{U}_{Q_\downarrow}) \r] \notag\\
    &= \E_{Q_\uparrow, Q_\downarrow}\l[ (\mathcal{U}_{Q_\uparrow}^\dagger \mathcal{M}_{Z,\uparrow} \mathcal{U}_{Q_\uparrow}) \otimes (\mathcal{U}_{Q_\downarrow}^\dagger \mathcal{M}_{Z,\downarrow} \mathcal{U}_{Q_\downarrow}) \r] \notag\\
    &= \mathcal{M}_\uparrow \otimes \mathcal{M}_\downarrow.
\end{align}
Therefore the spin-adapted matchgate shadows behaves as two independent instances on each spin sector. The estimators and variance bounds also follow straightforwardly;~first, the shadow channel is
\begin{equation}
\begin{split}
    \mathcal{M}_{G_\mathrm{spin}} &= \mathcal{M}_\uparrow \otimes \mathcal{M}_\downarrow\\
    &= \sum_{j=0}^{n_\uparrow} \sum_{\ell=0}^{n_\downarrow} f_{n_\uparrow,j} f_{n_\downarrow,\ell} \Pi_j \otimes \Pi_\ell,
\end{split}
\end{equation}
where for ease of notation in this section, we define
\begin{equation}
    f_{n_\sigma, j} \coloneqq \l. \binom{n_\sigma}{j} \middle/ \binom{2n_\sigma}{2j} \r..
\end{equation}

For the variance, we use the property that the shadow norm of a tensor-product distribution is the product of shadow norms on each subsystem. This can be seen from the fact that shadow norm of an operator $A$ is the spectral norm of a related operator $A^G \coloneqq \E_{U \sim G} \sum_{b \in \{0, 1\}^n} U^\dagger \op{b}{b} U \ev{b}{U \mathcal{M}_G^{-1}(A) U^\dagger}{b}^2$:
\begin{equation}
\begin{split}
    \sns{A}{G}^2 &= \max_{\text{states } \rho} \tr(\rho A^G)\\
    &= \| A^G \|_\infty,
\end{split}
\end{equation}
which holds because $A^G$ is positive semidefinite. For clarity, in this section we use the notation $\sns{\cdot}{G}$ for the shadow norm associated with the group $G$. Thus for any shadow channel formed as a tensor product $\mathcal{M}_{G_1 \oplus G_2} = \mathcal{M}_{G_1} \otimes \mathcal{M}_{G_2}$, we have
\begin{equation}
\begin{split}
    \sns{A_1 \otimes A_2}{G_1 \oplus G_2}^2 &= \max_{\text{states } \rho} \tr\l( \rho \l( A_1^{G_1} \otimes A_2^{G_2} \r) \r)\\
    &= \| A_1^{G_1} \otimes A_2^{G_2} \|_\infty\\
    &= \| A_1^{G_1} \|_\infty \| A_2^{G_2} \|_\infty\\
    &= \sns{A_1}{G_1}^2 \sns{A_2}{G_2}^2.
\end{split}
\end{equation}
(This argument generalizes to multiple tensor products.) Within the context of our spin-adapted ensemble, this implies that any spin-respecting Majorana operator
\begin{equation}
\begin{split}
    \Gamma_{\bm{p}, \bm{q}} &= (-i)^{j + \ell} \gamma_{p_1,\uparrow} \cdots \gamma_{p_{2j},\uparrow} \gamma_{q_1,\downarrow} \cdots \gamma_{q_{2\ell},\downarrow}\\
    &= (-i)^j c_{p_1,\uparrow} \cdots c_{p_{2j},\uparrow} \otimes (-i)^{\ell} c_{q_1,\downarrow} \cdots c_{q_{2\ell},\downarrow},
\end{split}
\end{equation}
has a squared shadow norm of
\begin{equation}
    \sns{\Gamma_{\bm{p}, \bm{q}}}{G_\mathrm{spin}}^2 = f_{n_\uparrow,j}^{-1} f_{n_\downarrow,\ell}^{-1}.
\end{equation}
Thus, for Majorana operators of constant degree $2(j + \ell) \leq 2k$, the variance scales as $\O(n_\uparrow^j n_\downarrow^\ell) = \O(n^k)$, just as in the unadapted setting. Note that spin-respecting here means that the operator factorizes into an even-degree Majorana operator on each spin sector.

The advantage of this ensemble is that the required circuit depth and gate count are roughly halved, since we only need to implement two independent matchgate circuits on $n_\sigma$ qubits each. Furthermore, one can also check that the shadow norm constant factors in the spin-adapted setting are also slightly smaller (for example, for $k = 2$ and $n_\sigma = n/2$, the ratio of spin-adapted to unadapted shadow norms is asymptotically $\lim_{n\to\infty} f_{n/2,1}^{-2}/f_{n,2}^{-1} = 3/4$).

\section{Improved compilation of fermionic Gaussian unitaries}\label{sec:FGU_circuit_design}

In this section we describe a new scheme for compiling the matchgate circuit $U_Q$ for arbitrary $Q \in \Orth(2n)$, under the Jordan--Wigner mapping. This approach improves upon the circuit depth of prior art~\cite{jiang2018quantum} by optimizing the parallelization of nearest-neighbor single- and two-qubit gates. We accomplish this by modifying previously established ideas to better respect the mapping of Majorana modes to qubits. Our improved design is implemented in code at our open-source repository ({\small\url{https://github.com/zhao-andrew/symmetry-adjusted-classical-shadows}}).

\subsection{A previous circuit design}\label{subsec:prior_circuit_design}

First we will review a prior circuit design to encode the action of $U_Q$ into a sequence of single- and two-qubit gates, from which it will become clear where there is room for improved parallelization. While the precise scheme that we describe here has not previously appeared in the literature, the high-level ideas follow from a combination of already developed results~\cite{reck1994experimental,wecker2015solving,kivlichan2018quantum,jiang2018quantum,clements2016optimal,oszmaniec2022fermion}.

Recall that our convention for the Jordan--Wigner mapping is
\begin{align}
    \gamma_{2p} &= Z_0 \cdots Z_{p-1} X_p,\\
    \gamma_{2p + 1} &= Z_0 \cdots Z_{p-1} Y_p
\end{align}
for $p \in [n]$, and our convention for the Gaussian transformation is
\begin{equation}\label{eq:gaussian_transformation_def}
    U_Q \gamma_\mu U_Q^\dagger = \sum_{\nu \in [2n]} Q_{\nu\mu} \gamma_\nu.
\end{equation}
It is straightforward to check that $\mathcal{U} : \Orth(2n) \to \U(\mathcal{L}(\mathcal{H}))$ is a group homomorphism:~$\mathcal{U}_Q \mathcal{U}_{Q'} = \mathcal{U}_{QQ'}$ for any $Q, Q' \in \Orth(2n)$. From this property, a circuit for arbitrary $U_Q$ can be constructed by a QR decomposition of $Q$. Such a decomposition yields a sequence of nearest-neighbor Givens rotations, which we then map to single- and adjacent two-qubit gates.

One possible QR decomposition is
\begin{equation}
    Q = G_1 \cdots G_L D,
\end{equation}
where each $G_j$ is a Givens rotation among adjacent rows and columns, and $D$ is the upper-right triangular matrix from the QR decomposition. Because $Q$ is an orthogonal matrix, $D$ is guaranteed to be a diagonal matrix with $\pm 1$ entries along the diagonal. This is equivalent to the Reck \emph{et al}.~\cite{reck1994experimental} design, and the number of Givens rotations is $L = \O(n^2)$ in depth $\O(n)$. By the homomorphism property of $\mathcal{U}$, this matrix decomposition yields a sequence of circuit elements that implements the desired unitary:
\begin{equation}
    \mathcal{U}_Q = \mathcal{U}_{G_1 \cdots G_L D} = \mathcal{U}_{G_1} \cdots \mathcal{U}_{G_L} \mathcal{U}_D.
\end{equation}
Alternatively, the Clements \emph{et al}.~\cite{clements2016optimal} design computes a decomposition of the form\footnote{The use of the Clements \emph{et al}.~\cite{clements2016optimal} design was first pointed out in Ref.~\cite{huggins2021efficient} by Dominic Berry, in the context of number-preserving matchgate circuits.}
\begin{equation}
    Q = G_{R + 1} \cdots G_{R + L} D G_{R} \cdots G_{1}.
\end{equation}
The total number of Givens rotations here is the same, $R + L = \O(n^2)$. However, by utilizing rotations that act from both left and right, it optimizes parallelization to reduce the depth by a constant factor (roughly $1/2$).

Refs.~\cite{wecker2015solving,kivlichan2018quantum,jiang2018quantum} showed how to convert these Givens rotations into number-preserving quantum gates;~here we seek to generalize to fermionic Gaussian unitaries which do not necessarily conserve particle number. While Ref.~\cite{jiang2018quantum} also considered this scenario, they maintained the representation of Givens rotations as number-preserving gates. Their circuit design breaks particle-number symmetry by interspersing particle--hole transformations throughout the decomposition.

Instead, we will use a representation that inherently features non-number-preserving rotations. Suppose that the Givens rotation $G_j$ acts nontrivially on the axes $(\mu, \mu + 1)$ as
\begin{equation}
    G_j = \begin{pmatrix}
    1 & \cdots & 0 & 0 & \cdots & 0\\
    \vdots & \ddots & \vdots & \vdots &  & \vdots\\
    0 & \cdots & \cos\theta_j & -\sin\theta_j & \cdots & 0\\
    0 & \cdots & \sin\theta_j & \cos\theta_j & \cdots & 0\\
    \vdots &  & \vdots & \vdots & \ddots & \vdots\\
    0 & \cdots & 0 & 0 & \cdots & 1
    \end{pmatrix} \in \SO(2n).
\end{equation}
The quantum gate which achieves this transformation is a single- or two-qubit Pauli rotation, given by
\begin{equation}\label{eq:givens_rotation_gates}
\begin{split}
    U_{G_j} &= e^{-(\theta_j/2) \gamma_{\mu} \gamma_{\mu + 1}}\\
    &= \begin{cases}
    e^{-\i(\theta_j/2) Z_p} & \text{if $\mu = 2p$ is even},\\
    e^{-\i(\theta_j/2) X_p X_{p + 1}} & \text{if $\mu = 2p + 1$ is odd}.
    \end{cases}
\end{split}
\end{equation}
Indeed, one may check that $\mathcal{U}_{G_j}(\gamma_p) = \sum_{q\in[2n]} [G_j]_{qp} \gamma_q$, as desired.

To implement the diagonal matrix $D$ of signs, we require a different scheme. In particular, we can construct $U_D$ as a single layer of Pauli gates. Consider the $2 \times 2$ block along the diagonal
\begin{equation}
	D^{(p)} = \begin{pmatrix}
		D_{2p} & 0 \\ 0 & D_{2p+1}
	\end{pmatrix},
\end{equation}
which describes the transformation
\begin{align}
    \gamma_{2p} &\mapsto D_{2p} \gamma_{2p},\\
    \gamma_{2p+1} &\mapsto D_{2p+1} \gamma_{2p+1}.
\end{align}
If $ D_{2p} = D_{2p+1} = 1 $, then clearly no operations are required. If instead $ D_{2p} = D_{2p+1} = -1 $, then conjugation by $ Z_p $ applies the desired signs on $\gamma_{2p}$ and $\gamma_{2p+1}$ while leaving all other Majorana operators invariant.

The remaining cases, $ D_{2p} = -D_{2p+1} $, can be handled as follows. First, suppose $ D_{2p} = 1 $. We wish to find the gates which perform the transformation
\begin{align}
	\gamma_{2p} &\mapsto \gamma_{2p},\\
	\gamma_{2p+1} &\mapsto -\gamma_{2p+1},
\end{align}
while leaving all other Majorana operators invariant. We can almost accomplish this with $ X_p $, since it will map $ X_p $ to itself and $ Y_p $ to $ -Y_p $. It also commutes with all Majorana operators $ \gamma_{2q}, \gamma_{2q+1} $ for $ q < p $. However, for $ q > p $ this will accrue unwanted signs:
\begin{align}
    X_p \gamma_{2q} X_p &= X_p (Z_0 \cdots Z_p \cdots Z_{q-1} X_q) X_p \notag\\
    &= -Z_0 \cdots Z_p \cdots Z_{q-1} X_q \notag\\
    &= -\gamma_{2q},\\
	X_p \gamma_{2q+1} X_p &= X_p (Z_0 \cdots Z_p \cdots Z_{q-1} Y_q) X_p \notag\\
    &= -Z_0 \cdots Z_p \cdots Z_{q-1} Y_q \notag\\
    &= -\gamma_{2q+1}.
\end{align}
To correct these signs, we introduce a Pauli-$ Z $ string running in the opposite direction of the Jordan--Wigner convention. That is, define
\begin{equation}
	P_{p}^{(X)} \coloneqq X_p Z_{p+1} \cdots Z_{n-1}.
\end{equation}
This unitary has the correct action on $ \gamma_{2p},\gamma_{2p+1} $ and continues to commute with the Majorana operators $ \gamma_{2q}, \gamma_{2q+1} $ with $q < p$. For $q > p$, however, we now have
% \begin{widetext}
\begin{align}
	P_{p}^{(X)} \gamma_{2q} P_{p}^{(X)} &= (X_p Z_{p+1} \cdots Z_{n-1}) (Z_0 \cdots Z_p \cdots Z_{q-1} X_q) (X_p Z_{p+1} \cdots Z_{n-1}) \notag\\
	&= Z_0 \cdots (-Z_p) \cdots Z_{q-1} (-X_q) \notag\\
    &= \gamma_{2q},\\
    P_{p}^{(X)} \gamma_{2q+1} P_{p}^{(X)} &= (X_p Z_{p+1} \cdots Z_{n-1}) (Z_0 \cdots Z_p \cdots Z_{q-1} Y_q) (X_p Z_{p+1} \cdots Z_{n-1}) \notag\\
    &= Z_0 \cdots (-Z_p) \cdots Z_{q-1} (-Y_q) \notag\\
    &= \gamma_{2q+1}.
\end{align}
% \end{widetext}
Thus $ P_p^{(X)} $ implements the desired transformation by $ D^{(p)} = \mathrm{diag}(1,-1) $. For $ D^{(p)} = \mathrm{diag}(-1,1) $, we simply replace $ P_p^{(X)} $ by an analogously defined $ P_p^{(Y)} $. This causes the sign of $ \gamma_{2p} $, rather than $ \gamma_{2p+1} $, to flip, while retaining all other properties.

Altogether, we determine these transformations for all $ 2 \times 2 $ diagonal blocks of $D$, resulting in $n$ Pauli strings of the form
\begin{equation}
    W_p = \begin{cases}
    \I & \text{if } D^{(p)} = \mathrm{diag}(1, 1),\\
     Z_p & \text{if } D^{(p)} = \mathrm{diag}(-1, -1),\\
     P_p^{(X)} & \text{if } D^{(p)} = \mathrm{diag}(1, -1),\\
     P_p^{(Y)} & \text{if } D^{(p)} = \mathrm{diag}(-1, 1).
    \end{cases}
\end{equation}
The overall transformation is then simply the product of these Pauli strings, which can be concatenated into a single layer of Pauli gates:
\begin{equation}\label{eq:pauli_gates_D}
    U_D = 
    \prod_{p \in [n]} W_p.
\end{equation}
Note that the order of this product does not matter, since Pauli gates commute up to an unobservable global phase.

\subsection{Discussion on suboptimality}

Now we observe that, depending on the parity of $\mu \in \{0, \ldots, 2n-2\}$, $U_{G_j}$ is either a single- or two-qubit gate. However, the decomposition of $Q$ described above is implicitly optimized under the assumption that only two-qubit gates are present:~each Givens rotation acts on two axes at a time, and it is assumed that this corresponds to physically acting on two wires at a time. This results in underutilized space in the quantum circuit whenever a single-qubit $Z$ rotation occurs, as it leaves a qubit wire needlessly idle. This is true for both the Reck \emph{et al.}~\cite{reck1994experimental} and Clements \emph{et al.}~\cite{clements2016optimal} designs. Ultimately, this suboptimality is due to the fact that $Q$ is a $2n \times 2n$ matrix, so there is a two-to-one correspondence between axes and qubits:~the rows/columns labeled by $(2p, 2p + 1)$ correspond to two Majorana operators, both of which are in turn associated with a single qubit $p$. Note that this discrepancy is not present in circuit designs for the class of number-conserving rotations~\cite{wecker2015solving,kivlichan2018quantum,jiang2018quantum}, which are instead more compactly represented by an $n \times n$ unitary matrix already.

\subsection{Circuit design with improved parallelization}

Now we introduce a circuit design which explicitly accounts for this two-to-one correspondence. The basic idea is to generalize the notion of Givens rotations, which act on a two-dimensional subspace to zero out a single matrix element, to a four-dimensional orthogonal transformation which zeroes out blocks of $2 \times 2$ at a time. Each $4 \times 4$ orthogonal transformation acts on the axes $(2p, 2p + 1, 2p + 2, 2p + 3)$, which corresponds to qubits $p$ and $p + 1$. By performing this process according to the scheme of Clements \emph{et al.}~\cite{clements2016optimal} (but now treating each $2 \times 2$ block of $Q$ as a ``single'' element), we obtain a decomposition wherein the optimal parallelization of the scheme is fully preserved in terms of interactions between nearest-neighbor qubits. Finally, each $4 \times 4$ orthogonal transformation is ultimately decomposed into six rotations of the form of Eq.~\eqref{eq:givens_rotation_gates} and a layer of Pauli gates, achieved by the standard decomposition that we described in Section~\ref{subsec:prior_circuit_design} (i.e., by bootstrapping off the prior scheme within blocks of $2n = 4$). Note that in principle one may instead implement the $4 \times 4$ orthogonal transformations using any gate set of one's choice, rather than $XX$ and $Z$ rotations.

We now describe the algorithm in detail. First we compute a decomposition analogous to the Clements \emph{et al.}~\cite{clements2016optimal} design, 
\begin{equation}\label{eq:Q_block_decomposition}
    Q = B_{R + 1} \cdots B_{R + L} D G B_{R} \cdots B_{1},
\end{equation}
but instead of Givens rotations, each $B_k$ acts nontrivially on a $4 \times 4$ block. (Note that there is a single $2 \times 2$ Givens rotation $G$ as well, which serves to zero out a final matrix element that we will elaborate on later.) We accomplish this by treating $Q$ as an $n \times n$ matrix of $2 \times 2$ blocks,
\begin{equation}\label{eq:left_mult_zero}
    Q_{\bm{p},\bm{q}} = \begin{pmatrix}
    Q_{2p,2q} & Q_{2p,2q+1}\\
    Q_{2p+1,2q} & Q_{2p+1,2q+1}
    \end{pmatrix},
\end{equation}
for each $p, q \in [n]$. Just as Givens rotations are chosen to zero a specific matrix element, each $B_k$ acts to zero out a particular $2 \times 2$ block $Q_{\bm{p},\bm{q}}$.

Suppose we want to find a $B_{R+i}$ which acts from the left ($i = 1, \ldots, L$) to zero out the block $Q_{\bm{p},\bm{q}}$. Then we perform a QR decomposition on the $4 \times 2$ submatrix which includes the target block and the block directly above it:
\begin{equation}
    \l(\begin{array}{c}
    Q_{\bm{p}-1,\bm{q}}\\
    \hline
    Q_{\bm{p},\bm{q}}
    \end{array}\r) = B_{R+i}' \l(\begin{array}{c c}
    * & *\\
    0 & *\\
    \hline
    0 & 0\\
    0 & 0
    \end{array}\r),
\end{equation}
hence zeroing out the lower block as desired. Here, $B_{R+i}' \in \Orth(4)$ is computed from the QR decomposition, and so the orthogonal matrix $B_{R+i} \in \Orth(2n)$ appearing in Eq.~\eqref{eq:Q_block_decomposition} is defined as $B_{R+i}'$ along the axes $(2p-2, 2p-1, 2p, 2p + 1)$ and the identity elsewhere.

Similarly, if we want a $B_{j}$ which acts from the right ($j = 1, \ldots, R$), then we consider instead a $2 \times 4$ submatrix with the target block on the left:
\begin{equation}\label{eq:2x4_block}
    \l(\begin{array}{c | c}
    Q_{\bm{p},\bm{q}} & Q_{\bm{p},\bm{q}+1}
    \end{array}\r).
\end{equation}
This can be zeroed out by performing an LQ decomposition (which is essentially just the transpose of the QR decomposition). For notation in this section, let tildes denote the flipping of rows in a matrix, for example
\begin{equation}
    M = \begin{pmatrix}
    M_{11} & M_{12}\\
    M_{21} & M_{22}\\
    M_{31} & M_{32}\\
    M_{41} & M_{42}
    \end{pmatrix} \mapsto \tilde{M} = \begin{pmatrix}
    M_{41} & M_{42}\\
    M_{31} & M_{32}\\
    M_{21} & M_{22}\\
    M_{11} & M_{12}
    \end{pmatrix}.
\end{equation}
Then performing an LQ decomposition on the row-flipped version of Eq.~\eqref{eq:2x4_block}, we have
\begin{equation}\label{eq:2x4_block_transpose_qr}
\begin{split}
    \l(\begin{array}{c | c}
    \tilde{Q}_{\bm{p},\bm{q}} & \tilde{Q}_{\bm{p},\bm{q}+1}
    \end{array}\r) &= \l(\begin{array}{c c | c c}
    * & 0 & 0 & 0\\
    * & * & 0 & 0
    \end{array}\r) B_{j}'\\
    &= \l(\begin{array}{c c | c c}
    0 & 0 & 0 & *\\
    0 & 0 & * & *
    \end{array}\r) \tilde{B}_{j}'.
\end{split}
\end{equation}
Flipping the rows back to normal on the lefthand side, we get
\begin{equation}
    \l(\begin{array}{c | c}
    Q_{\bm{p},\bm{q}} & Q_{\bm{p},\bm{q}+1}
    \end{array}\r) = \l(\begin{array}{c c | c c}
    0 & 0 & * & *\\
    0 & 0 & 0 & *
    \end{array}\r) \tilde{B}_{j}'
\end{equation}
as desired. Then we define $B_{j} \in \Orth(2n)$ acting as $\tilde{B}_{j}'$ on the axes $(2q, 2q+1, 2q+2, 2q+3)$ and trivially elsewhere.

Now we address the need for the sole Givens rotation $G$ appearing in Eq.~\eqref{eq:Q_block_decomposition}. As the zeroing-out procedure described above progresses, the nonzero blocks get ``pushed'' towards the diagonal until the final matrix is ($2 \times 2$)-block diagonal. These nonzero blocks must be triangular because they are produced by QR/LQ decompositions;~but since $Q$ is orthogonal, this implies that the final triangular blocks along the diagonal must be diagonal themselves. The exception to this is either the leftmost or rightmost block, depending on whether $n$ is even or odd. This is because the decomposition procedure inevitably leaves one of those blocks untouched, so it was never made triangular/diagonal.

This can be visualized as follows:~if $n$ is odd, then we have
% \begin{widetext}
\begin{equation}
    Q \to \l(\begin{array}{c c | c c | c c}
    * & * & * & * & * & *\\
    * & * & * & * & * & *\\
    \hline
    * & * & * & * & * & *\\
    * & * & * & * & * & *\\
    \hline
    \bm{0} & \bm{0} & * & * & * & *\\
    \bm{0} & \bm{0} & \bm{0} & * & * & *
    \end{array}\r)
    \to
    \l(\begin{array}{c c | c c | c c}
    * & * & * & * & * & *\\
    \bm{0} & * & * & * & * & *\\
    \hline
    \bm{0} & \bm{0} & * & * & * & *\\
    \bm{0} & \bm{0} & * & * & * & *\\
    \hline
    0 & 0 & * & * & * & *\\
    0 & 0 & 0 & * & * & *
    \end{array}\r)
    \to
    \l(\begin{array}{c c | c c | c c}
    * & * & * & * & * & *\\
    0 & * & * & * & * & *\\
    \hline
    0 & 0 & * & * & * & *\\
    0 & 0 & \bm{0} & * & * & *\\
    \hline
    0 & 0 & \bm{0} & \bm{0} & * & *\\
    0 & 0 & 0 & \bm{0} & * & *
    \end{array}\r)
    =
    \l(\begin{array}{c c | c c | c c}
    \pm 1 & 0 & 0 & 0 & 0 & 0\\
    0 & \pm 1 & 0 & 0 & 0 & 0\\
    \hline
    0 & 0 & \pm 1 & 0 & 0 & 0\\
    0 & 0 & 0 & \pm 1 & 0 & 0\\
    \hline
    0 & 0 & 0 & 0 & * & *\\
    0 & 0 & 0 & 0 & * & *
    \end{array}\r).
\end{equation}
We use boldface to clarify which matrix elements are newly zeroed at each step. The condition that $Q$ is an orthogonal matrix implies the final equality. It also enforces the remaining $2 \times 2$ block to be orthogonal, so that we can diagonalize it by computing the appropriate Givens rotation acting on axes $(2n-2, 2n-1)$. This elucidates the appearance of $G$ in Eq.~\eqref{eq:Q_block_decomposition}. On the other hand, if $n$ is even, then the top-left block remains instead:
\begin{align}
    Q &\to \l(\begin{array}{c c | c c | c c | c c}
    * & * & * & * & * & * & * & *\\
    * & * & * & * & * & * & * & *\\
    \hline
    * & * & * & * & * & * & * & *\\
    * & * & * & * & * & * & * & *\\
    \hline
    * & * & * & * & * & * & * & *\\
    * & * & * & * & * & * & * & *\\
    \hline
    \bm{0} & \bm{0} & * & * & * & * & * & *\\
    \bm{0} & \bm{0} & \bm{0} & * & * & * & * & *
    \end{array}\r)
    \to
    \l(\begin{array}{c c | c c | c c | c c}
    * & * & * & * & * & * & * & *\\
    * & * & * & * & * & * & * & *\\
    \hline
    * & * & * & * & * & * & * & *\\
    \bm{0} & * & * & * & * & * & * & *\\
    \hline
    \bm{0} & \bm{0} & * & * & * & * & * & *\\
    \bm{0} & \bm{0} & * & * & * & * & * & *\\
    \hline
    0 & 0 & * & * & * & * & * & *\\
    0 & 0 & 0 & * & * & * & * & *
    \end{array}\r)
    \to
    \l(\begin{array}{c c | c c | c c | c c}
    * & * & * & * & * & * & * & *\\
    * & * & * & * & * & * & * & *\\
    \hline
    * & * & * & * & * & * & * & *\\
    0 & * & * & * & * & * & * & *\\
    \hline
    0 & 0 & * & * & * & * & * & *\\
    0 & 0 & \bm{0} & * & * & * & * & *\\
    \hline
    0 & 0 & \bm{0} & \bm{0} & * & * & * & *\\
    0 & 0 & 0 & \bm{0} & * & * & * & *
    \end{array}\r)
    \to
    \l(\begin{array}{c c | c c | c c | c c}
    * & * & * & * & * & * & * & *\\
    * & * & * & * & * & * & * & *\\
    \hline
    * & * & * & * & * & * & * & *\\
    0 & * & * & * & * & * & * & *\\
    \hline
    0 & 0 & * & * & * & * & * & *\\
    0 & 0 & 0 & * & * & * & * & *\\
    \hline
    0 & 0 & 0 & 0 & \bm{0} & \bm{0} & * & *\\
    0 & 0 & 0 & 0 & \bm{0} & \bm{0} & \bm{0} & *
    \end{array}\r) \notag\\
    &\to
    \l(\begin{array}{c c | c c | c c | c c}
    * & * & * & * & * & * & * & *\\
    * & * & * & * & * & * & * & *\\
    \hline
    * & * & * & * & * & * & * & *\\
    0 & * & * & * & * & * & * & *\\
    \hline
    0 & 0 & \bm{0} & \bm{0} & * & * & * & *\\
    0 & 0 & 0 & \bm{0} & \bm{0} & * & * & *\\
    \hline
    0 & 0 & 0 & 0 & 0 & 0 & * & *\\
    0 & 0 & 0 & 0 & 0 & 0 & 0 & *
    \end{array}\r)
    \to
    \l(\begin{array}{c c | c c | c c | c c}
    * & * & * & * & * & * & * & *\\
    * & * & * & * & * & * & * & *\\
    \hline
    \bm{0} & \bm{0} & * & * & * & * & * & *\\
    0 & \bm{0} & \bm{0} & * & * & * & * & *\\
    \hline
    0 & 0 & 0 & 0 & * & * & * & *\\
    0 & 0 & 0 & 0 & 0 & * & * & *\\
    \hline
    0 & 0 & 0 & 0 & 0 & 0 & * & *\\
    0 & 0 & 0 & 0 & 0 & 0 & 0 & *
    \end{array}\r)
    = \l(\begin{array}{c c | c c | c c | c c}
    * & * & 0 & 0 & 0 & 0 & 0 & 0\\
    * & * & 0 & 0 & 0 & 0 & 0 & 0\\
    \hline
    0 & 0 & \pm 1 & 0 & 0 & 0 & 0 & 0\\
    0 & 0 & 0 & \pm 1 & 0 & 0 & 0 & 0\\
    \hline
    0 & 0 & 0 & 0 & \pm 1 & 0 & 0 & 0\\
    0 & 0 & 0 & 0 & 0 & \pm 1 & 0 & 0\\
    \hline
    0 & 0 & 0 & 0 & 0 & 0 & \pm 1 & 0\\
    0 & 0 & 0 & 0 & 0 & 0 & 0 & \pm 1
    \end{array}\r).
\end{align}
% \end{widetext}
In this case, $G$ needs to act on axes $(0, 1)$.

Thus we have obtained the decomposition of Eq.~\eqref{eq:Q_block_decomposition} as desired. The implementation of each component then follows from bootstrapping the prior techniques:~the diagonal matrix $D$ becomes a layer of Pauli gates, described by Eq.~\eqref{eq:pauli_gates_D};~and the four-dimensional orthogonal transformations $B_j$ are further decomposed into Givens rotations, described in Section~\ref{subsec:prior_circuit_design} (wherein $n = 2$).

\section{Classical shadows postprocessing details}

In this section we provide details for the classical postprocesisng of local observable estimators from classical shadows. We include this for a self-contained and explicit presentation, and also to address the modified shadows protocols (subsystem symmetrization and spin adaptation) introduced in this paper. These algorithms are implemented at our open-source repository ({\small\url{https://github.com/zhao-andrew/symmetry-adjusted-classical-shadows}}).

\subsection{Matchgate shadows}\label{subsec:matchgate_computation}

For any orthogonal matrix $Q \in \Orth(2n)$, Ref.~\cite{wan2023matchgate} derived formulas involving the multiplication of $2n \times 2n$ matrices and the computation of Pfaffians of $2k \times 2k$ submatrices for estimating $k$-body Majorana observables. However when restricting $Q \in \B(2n)$, there exists a significantly cheaper method that does not involve such numerical linear algebra routines. This algorithm was implicitly described in Ref.~\cite{zhao2021fermionic}, but not explicitly outlined. We do so here;~for $k = \O(1)$, it runs in time $\O(n^k T)$ to return estimates for all $2j$-degree Majorana operators, $1 \leq j \leq k$, from $T$ samples. Note that the number of operators is $\O(n^{2k})$, so our approach has significant savings over a naive iteration. Furthermore, it largely involves integer storage and manipulations rather than floating-point operations.

Any $k$-body fermionic observable can be decomposed into a linear combination of polynomially many $(\leq k)$-body Majorana operators. Thus it suffices to consider $\Gamma_{\bm{\mu}}$, for all $\bm{\mu} \in \bigcup_{j \leq k} \comb{2n}{2j}$. Each matchgate-shadow sample $\hat{\rho}_{Q,b}$ is classically stored as $(Q, b)$, where $b \in \{0, 1\}^n$ and $Q$ is represented as an array $\pi$ of the permuted elements of $[2n]$ along with signs $s \in \{-1, +1\}^{2n}$. Specifically, the matrix elements of $Q \in \B(2n)$ are related to $(s, \pi)$ by $Q_{\mu\nu} = s_\mu \delta_{\pi(\mu), \nu}$.

The estimator for $\tr(\Gamma_{\bm{\mu}} \rho)$ can be written as $\tr(\Gamma_{\bm{\mu}} \hat{\rho}_{Q,b}) = f_{2j}^{-1} \ev{b}{U_Q \Gamma_{\bm{\mu}} U_Q^\dagger}{b}$, where $U_Q \Gamma_{\bm{\mu}} U_Q^\dagger$ can be expanded in terms of subdeterminants of $Q$ according to Ref.~\cite[Appendix~A]{chapman2018classical}. However, a simplified derivation is possible here by using the fact that $Q$ implements a signed permutation:
\begin{equation}
\begin{split}
    U_Q \gamma_\mu U_Q^\dagger &= \sum_{\nu \in [2n]} Q_{\nu\mu} \gamma_\nu\\
    &= \sum_{\nu \in [2n]} s_\nu \delta_{\pi(\nu),\mu} \gamma_\nu\\
    &= s_{\pi^{-1}(\mu)} \gamma_{\pi^{-1}(\mu)}.
\end{split}
\end{equation}
Hence for operators of degree $2j$,
\begin{equation}
\begin{split}
    U_Q \Gamma_{\bm{\mu}} U_Q^\dagger &= (-\i)^j U_Q \gamma_{\mu_1} \cdots \gamma_{\mu_{2j}} U_Q^\dagger\\
    &= (-\i)^j s_{\pi^{-1}(\mu_1)} \cdots s_{\pi^{-1}(\mu_{2j})} \gamma_{\pi^{-1}(\mu_1)} \cdots \gamma_{\pi^{-1}(\mu_{2j})}.
    % &\quad \times \gamma_{\pi^{-1}(\mu_1)} \cdots \gamma_{\pi^{-1}(\mu_{2j})}.
\end{split}
\end{equation}
We would like to retain the ordering of indices when working with the multidegree Majorana operators;~therefore we introduce a further a permutation as $\tilde{\pi}^{-1}(\mu_i)$, which is defined to satisfy $\tilde{\pi}^{-1}(\mu_1) < \cdots < \tilde{\pi}^{-1}(\mu_{2j})$. This incurs another sign factor $(-1)^p$ where $p \in \{0, 1\}$ is the parity of the permutation which sends $\pi^{-1}(\bm{\mu}) \mapsto \tilde{\pi}^{-1}(\bm{\mu})$. Collecting all signs as $\sgn_Q(\bm{\mu}) = (-1)^p s_{\pi^{-1}(\mu_1)} \cdots s_{\pi^{-1}(\mu_{2j})}$, we arrive at
\begin{equation}\label{eq:matchgate_shadow_estimator_permutation}
    \tr(\Gamma_{\bm{\mu}} \hat{\rho}_{Q,b}) = f_{2j}^{-1} \sgn_Q(\bm{\mu}) \ev{b}{\Gamma_{\tilde{\pi}^{-1}(\bm{\mu})}}{b}.
\end{equation}
The matrix element $\ev{b}{\Gamma_{\tilde{\pi}^{-1}(\bm{\mu})}}{b}$ is nonzero if and only if $\tilde{\pi}^{-1}(\bm{\mu}) \in \diags{2n}{2j}$, from which its value of $\pm 1$ is straightforward to determine (e.g., by mapping to Pauli-$Z$ operators). In total, evaluating Eq.~\eqref{eq:matchgate_shadow_estimator_permutation} takes time $\O(n^2 + j^2 + j)$, corresponding respectively to the inversion of $\pi \in \Sym(2n)$, the calculation of $\tilde{\pi}^{-1}(\bm{\mu})$ and its parity on $2j$ indices, and evaluating the product of $2j + 1$ signs and $\ev{b}{\Gamma_{\tilde{\pi}^{-1}(\bm{\mu})}}{b}$, the latter requiring only checking $2j$ indices and $j$ bits of $b$. Assuming $j \leq k = \O(1)$, this implies a computational complexity of $\O(n^2)$ per operator per sample.

To compute this estimator for all $\bm{\mu} \in \comb{2n}{2j}$, a naive approach iterates through each $\bm{\mu}$, of which there are $\binom{2n}{2j} = \O(n^{2j})$ many. Repeating this for each of the $T$ samples would therefore cost $\O(T (n^{2j} + n^2)) = \O(T n^{2j})$ time. Noting that $T = \Ot(n^{j} \epsilon^{-2})$ suffices for $\epsilon$-accurate estimation,\footnote{The notation $\Ot(\cdot)$ suppresses polylogarithmic factors in the complexity.} the total complexity of $\Ot(n^{3j} \epsilon^{-2}) \leq \Ot(n^{3k} \epsilon^{-2})$ would be unacceptably large.

We can speed up the computation over all operators per sample to $\O(n^j)$ by using the fact that many $\ev{b}{\Gamma_{\tilde{\pi}^{-1}(\bm{\mu})}}{b}$ vanish. That is, rather than compute $\tilde{\pi}^{-1}(\bm{\mu})$ for all $\bm{\mu} \in \comb{2n}{2j}$ and checking whether each is an element of $\diags{2n}{2j}$, we work backwards by looping over all target elements $\bm{\tau} \in \diags{2n}{2j}$ and computing $\pi(\bm{\tau})$ to find its preimage. As before, let $\tilde{\pi}(\bm{\tau}) \in \comb{2n}{2j}$ be the reordering of $\pi(\bm{\tau})$ with associated sign $(-1)^p$. Then for each $\bm{\tau} \in \diags{2n}{2j}$, we compute the estimator for $\Gamma_{\tilde{\pi}(\bm{\tau})}$,
\begin{equation}\label{eq:reversed_majorana_estimate}
    \tr(\Gamma_{\tilde{\pi}(\bm{\tau})} \hat{\rho}_{Q,b}) = f_{2j}^{-1} \sgn_Q(\pi(\bm{\tau})) \ev{b}{\Gamma_{\bm{\tau}}}{b},
\end{equation}
where the cumulative sign is $\sgn_Q(\pi(\bm{\tau})) = (-1)^p s_{\tau_1} \cdots s_{\tau_{2j}}$. All other Majorana operators not in the preimage are implicitly assigned an estimate of $0$. Hence we only iterate over the $\binom{n}{j} = \O(n^j)$ elements of $\diags{2n}{2j}$, with each evaluation of Eq.~\eqref{eq:reversed_majorana_estimate} taking $\O(j^2) = \O(1)$ time. Note that this approach also avoids the need to find the inverse permutation $\pi^{-1}$.

Performing this procedure over all $T$ samples results in a time complexity of $\O(n^j T) \leq \O(n^k T)$, running over all $j \in \{1, \ldots, k\}$. We can also include an additive $\O(n^{2k})$ cost to preallocate storage for $\bigcup_{j \leq k} \comb{2n}{2j}$. While not strictly necessary, this is convenient in practice, and besides when $T = \Ot(n^k \epsilon^{-2})$ the total complexity is $\Ot(n^{2k} \epsilon^{-2})$ whether or not we preallocate memory.

For the spin-adapted shadows, because the protocol factorizes across the spin sectors, we perform this algorithm on each sector independently. The estimator for operators of the form $\Gamma_{\bm{\mu}} \otimes \Gamma_{\bm{\nu}}$ is then the product of the independent estimates. Note that if either $|\bm{\mu}|$ or $|\bm{\nu}|$ are odd, then the estimator always vanishes;~this reflects the fact that the spin-adapted ensemble is not informationally complete over such operators.

\subsection{Pauli shadows}\label{subsec:pauli_computation}

Because single-qubit measurements factorize, we consider each qubit $i \in [n]$ independently. Given the random Clifford $C_i \in \Cl(1)$ and measurement outcome $b_i \in \{0, 1\}$, the estimator for $\sigma^{(i)} \in \{\I, X, Y, Z\}$ is~\cite{huang2020predicting}
\begin{equation}\label{eq:shadow_qubit_est}
    \tr(\sigma^{(i)} \hat{\rho}_{C_i, b_i}) = 3 \ev{b_i}{C_i \sigma^{(i)} C_i^\dagger}{b_i} - \tr(\sigma^{(i)}).
\end{equation}
Each Pauli-shadow sample is stored as $(W_i, b_i)$, where $W_i = C_i^\dagger Z C_i \in \pm\{X, Y, Z\}$. Evaluating Eq.~\eqref{eq:shadow_qubit_est} reduces to checking the conditions
\begin{equation}\label{eq:shadow_qubit_cases}
    \tr(\sigma^{(i)} \hat{\rho}_{C_i, b_i}) = \begin{cases}
    \pm 3 \ev{b_i}{Z}{b_i} & \text{if } \sigma^{(i)} = \pm W_i,\\
    1 & \text{if } \sigma^{(i)} = \I,\\
    0 & \text{else}.
    \end{cases}
\end{equation}
The product over $i \in [n]$ then estimates $P = \bigotimes_{i \in [n]} \sigma^{(i)}$ for the full $n$-qubit system. This suffices to estimate any $k$-local observable, which can be decomposed into a linear combination of polynomially many $(\leq k)$-local Pauli operators.

The total time complexity of estimating all $k$-local Pauli operators with $T$ snapshots is $\O(n^k T)$. The algorithm is as follows. For each $W = (W_0, \ldots, W_{n-1})$, we take, for each $j \leq k$, all $\binom{n}{j}$ combinations $W_{i_1}, \ldots, W_{i_j}$ and compute Eq.~\eqref{eq:shadow_qubit_cases} for each $\sigma^{(i)} = \pm W_i$. We assign the result as an estimate for the $j$-local operator $P = W_{i_1} \otimes \cdots \otimes W_{i_j} \otimes \I^{\otimes (n-j)}$, and implicitly assign $0$ to all other Pauli operators. Note that there are a total of $\sum_{j \leq k} 3^j \binom{n}{j} = \O(n^k)$ local Pauli operators, so preallocating storage here is asymptotically negligible.

For the subsystem-symmetrized protocol, the $n$-qubit estimator now takes the form
\begin{equation}
    \tr(P \hat{\rho}_{(\pi, C), b}) = 3^{|P|} \ev{b}{S_\pi C P C^\dagger S_\pi^\dagger}{b}.
\end{equation}
Using the fact that $S_\pi^\dagger \ket{b} = \bigotimes_{i \in [n]} \ket{\pi(b_i)}$, we can simply apply the standard scheme described above, but with the replacement $b_i \to \pi(b_i)$. For each sample this is only an additive $\O(n)$ cost.

\section{Additional details on numerical experiments}\label{sec:additional_numerics}

Here we supply further information regarding the numerical simulations, to both provide additional insight into our results and facilitate easier replication of our results by the motivated reader.

\subsection{Readout noise models}\label{subsec:readout_errors}

In the main text, we demonstrated our mitigation strategy under single-qubit readout errors. The noise channels occur immediately before measurement and are implemented probabilistically:~independently and identically (i.i.d.) on each qubit per circuit repetition. We consider depolarizing, amplitude-damping, and bit-flip errors occurring with probability $p$, which are respectively
\begin{align}
    \mathcal{E}_{\mathrm{dep}}(\rho) &= (1 - p) \rho + p \frac{\I}{2},\\
    \mathcal{E}_{\mathrm{AD}}(\rho) &= E_0 \rho E_0^\dagger + E_1 \rho E_1^\dagger,\\
    E_0 =& \begin{pmatrix}
    1 & 0\\
    0 & \sqrt{1 - p}
    \end{pmatrix}, E_1 = \begin{pmatrix}
    0 & \sqrt{p}\\
    0 & 0
    \end{pmatrix} \notag\\
    \mathcal{E}_{\mathrm{BF}}(\rho) &= (1 - p) \rho + p X \rho X.
\end{align}
These models obey Assumptions~1, although we comment that more complicated noise channels can also satisfy the assumptions, such as non-i.i.d.~errors, correlated multiqubit errors, and even coherent gate errors~\cite{chen2021robust}.

\subsection{QVM gate set and noise model}\label{subsec:noise_model_details}

The noise model we implement on the Cirq Quantum Virtual Machine is based on the Google Sycamore processor ``Rainbow,'' a 2D grid of 23 superconducting qubits. We use the calibration data obtained from November 16, 2021, which can be found in the Cirq open-source repository~\cite{cirq}. The native gate set that we compile our circuits to include single-qubit rotations in the form of phased XZ gates,
\begin{equation}
\begin{split}
    &\PhXZ(x, z, a)\\
    &= \begin{pmatrix}
    e^{\i \frac{\pi x}{2}} \cos\l(\frac{\pi x}{2}\r) & -\i e^{\i \pi(\frac{x}{2} - a)} \sin\l(\frac{\pi x}{2}\r)\\
    -\i e^{\i \pi(\frac{x}{2} + a + z)} \sin\l(\frac{\pi x}{2}\r) & e^{\i \pi(\frac{x}{2} + z)} \cos\l(\frac{\pi x}{2}\r)
    \end{pmatrix}\\
    &= Z^z Z^a X^x Z^{-a}.
\end{split}
\end{equation}
This describes a rotation by $\pi x$ about an axis determined by the parameter $a$ within the $xy$ plane, followed by a phasing of $\pi z$. The native two-qubit gates that we use are
\begin{equation}
\begin{split}
    \sqrt{\iSWAP} &= \begin{pmatrix}
    1 & 0 & 0 & 0\\
    0 & \frac{1}{\sqrt{2}} & \frac{\i}{\sqrt{2}} & 0\\
    0 & \frac{\i}{\sqrt{2}} & \frac{1}{\sqrt{2}} & 0\\
    0 & 0 & 0 & 1
    \end{pmatrix}\\
    &= e^{\i \frac{\pi}{4} (X \otimes X + Y \otimes Y) / 2},
\end{split}
\end{equation}
constrained to the nearest-neighbor connectivity of the chip.

\begin{figure*}
\centering
\includegraphics[scale=0.6]{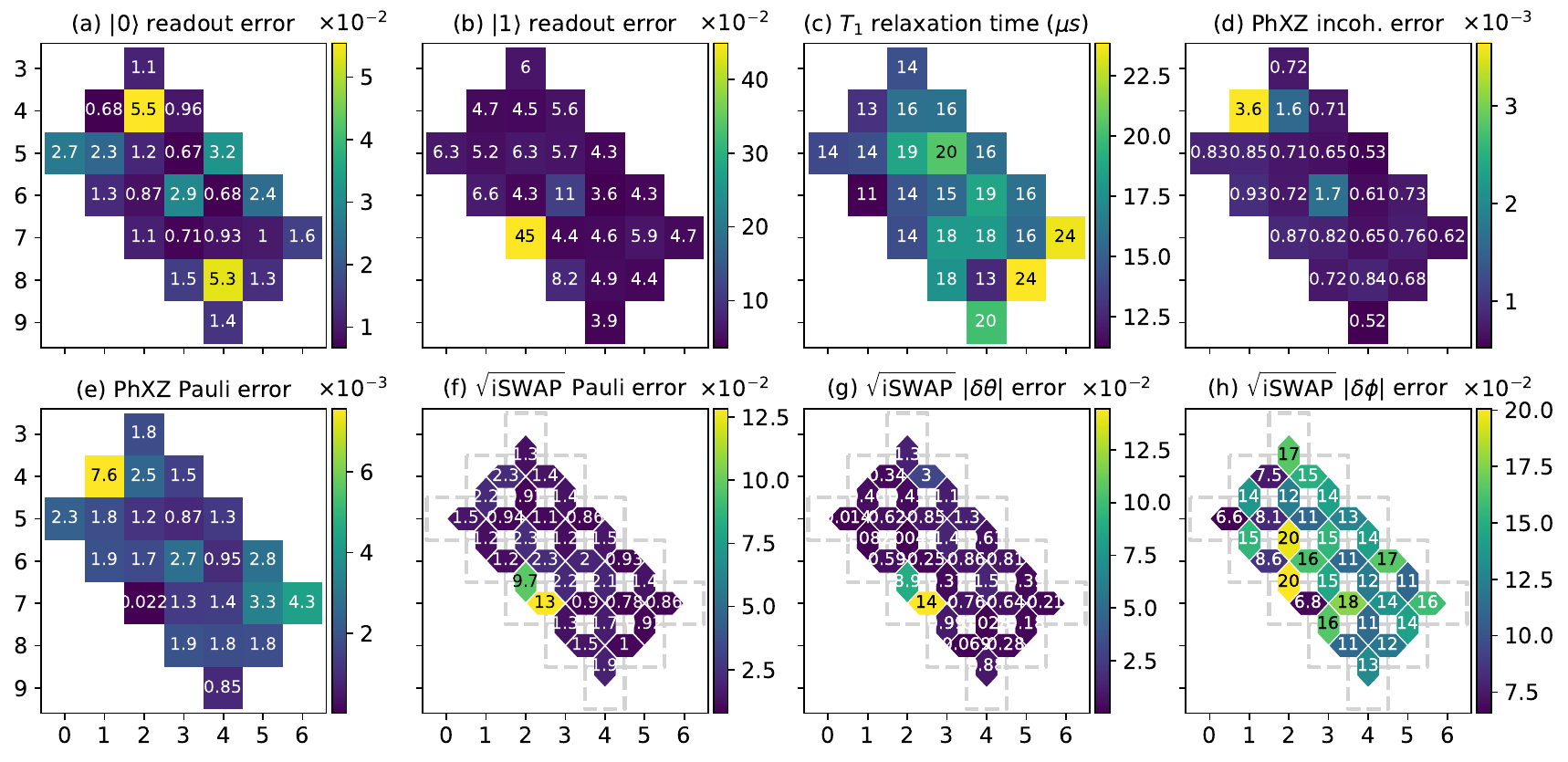}
\caption{\textbf{Chip layout, connectivity, and error rates of the Google Sycamore Rainbow processor, as simulated by the QVM.} Our eight-qubit numerical experiments use the $2 \times 4$ grid spanning from qubit $(5, 1)$ to qubit $(6, 4)$. \textbf{(a, b)} Readout error probabilities, measured in parallel. \textbf{(c)} Characteristic $T_1$ qubit decay times, measured in isolation. \textbf{(d)} Incoherent component of single-qubit gate errors, measured in isolation from RB and purity benchmarking data. Used in conjunction with $T_1$ to infer the $T_2$ dephasing time. \textbf{(e)} Total Pauli error of single-qubit gates, measured in isolation via RB. \textbf{(f)} Total Pauli error of two-qubit $\sqrt{\iSWAP}$ gates, measured in parallel via XEB. \textbf{(g, h)} Coherent errors in two-qubit gates $\sqrt{\iSWAP} = \fSim(\theta = -\frac{\pi}{4}, \phi = 0)$, measured in parallel from XEB data. We display the magnitudes $|\delta\theta|, |\delta\phi|$ for visualization purposes;~full calibration data (including signs) is available at the Cirq open-source repository~\cite{cirq}.}
\label{fig:calibration_data}
\end{figure*}

The QVM noise model that we simulate is not fully comprehensive of all types of errors occurring in an actual device, however it captures the most dominant error sources in the superconducting platform~\cite{isakov2021simulations}. It consists of four categories:
\begin{enumerate}
    \item Readout errors are modeled as asymmetric bit-flip channels on each qubit. The asymmetry reflects the fact that the probability of a $\ket{1}$ outcome being erroneously measured as $\ket{0}$ is generally higher than misreading a $\ket{0}$ outcome. Although the errors are modeled as single-qubit channels, the calibration data is taken from parallel experiments, to potentially account for effects such as readout crosstalk and other unintended interactions between qubits.

    \item Decay ($T_1$) and dephasing ($T_2$) errors occur whenever a qubit idles during a moment (layer) of a circuit. Both $T_1$ and $T_2$ relaxations are incorporated into a single channel,
    \begin{equation}
    \mathcal{E}_{\mathrm{idle}}(\rho) = \begin{pmatrix}
    1 - \rho_{11} e^{-t/T_1} & \rho_{01} e^{-t/T_2}\\
    \rho_{10} e^{-t/T_2} & \rho_{11} e^{-t/T_1}
    \end{pmatrix}.
    \end{equation}
    The decay time $T_1$ is characterized by a simple experiment that prepares $\ket{1}$ and measures the survival probability as a function of $t$. This experiment is performed in isolation, i.e., one qubit at a time while all other qubits on the chip idle.
    
    The $T_2$ time is determined from the equation
    \begin{equation}
    \frac{1}{T_2} = \frac{1}{2T_1} + \frac{1}{T_\phi},
    \end{equation}
    where $1/T_\phi$ is the pure dephasing rate that can in principle be measured by Ramsey interferometry. For simplicity, however, this noise model instead approximates $T_\phi$ from the total single-qubit incoherent error $\epsilon_{\mathrm{inc}}$, which is determined by purity benchmarking~\cite{wallman2015estimating,feng2016estimating} performed in isolation. To leading order, $T_\phi$ is approximated using the relation
    \begin{equation}
    \epsilon_{\mathrm{inc}} = \frac{t}{3T_1} + \frac{t}{3T_\phi} + \O(t^2).
    \end{equation}
    The time $t$ which appears in the model channel $\mathcal{E}_{\mathrm{idle}}$ is the longest gate duration occurring within that moment:~$\PhXZ$ gates have a duration of $25$~ns, while $\sqrt{\iSWAP}$ gates take $32$~ns.
    
    \item Single-qubit gate errors are modeled as depolarizing channels occurring after each gate. The depolarizing rate is set to match the total single-qubit Pauli error, which is measured from the device via randomized benchmarking (RB)~\cite{magesan2011scalable,magesan2012characterizing} in isolation.
    
    \item Two-qubit gate errors are modeled with both coherent and incoherent components. The coherent contribution uses the fact that $\sqrt{\iSWAP}$ is an instance of the general fermionic simulation ($\fSim$) gate,
    \begin{equation}
    \fSim(\theta, \phi) = \begin{pmatrix}
    1 & 0 & 0 & 0\\
    0 & \cos\theta & -\i\sin\theta & 0\\
    0 & -\i\sin\theta & \cos\theta & 0\\
    0 & 0 & 0 & e^{-\i\phi}
    \end{pmatrix},
    \end{equation}
    which is a native, tunable interaction on the superconducting platform. The $\sqrt{\iSWAP}$ gate is the instance $(\theta, \phi) = (-\frac{\pi}{4}, 0)$. Coherent errors are thus modeled as an overrotation by $(\delta\theta, \delta\phi)$, which are determined for each pair of connected qubits by fitting to cross-entropy benchmarking (XEB) data using random cycles of gates across the chip~\cite{boixo2018characterizing,neill2018blueprint,arute2019quantum}.
    
    After the coherent overrotation, an incoherent error follows, modeled as a two-qubit depolarizing channel. The depolarizing rate $r_{\mathrm{dep}}^{(i, j)}$ for each pair $(i, j)$ of connected qubits is inferred as follows:~from the total XEB Pauli error $r_{\mathrm{XEB}}^{(i, j)}$, we subtract off the single-qubit incoherent error rates $r_{\mathrm{inc}}^{(i)}, r_{\mathrm{inc}}^{(j)}$ (determined from RB), as well as the average entangling error rate $r_{\mathrm{ent}}^{(i, j)}$, which are calculated using the coherent errors $\delta\theta, \delta\phi$. The model's two-qubit depolarizing rate is then set to account for the remaining amount of error:
    \begin{equation}
    r_{\mathrm{dep}}^{(i, j)} = r_{\mathrm{XEB}}^{(i, j)} - r_{\mathrm{inc}}^{(i)} - r_{\mathrm{inc}}^{(j)} - r_{\mathrm{ent}}^{(i, j)}.
    \end{equation}
    Due to the nature of XEB, both two-qubit error sources are characterized by parallel experimental data.
\end{enumerate}
Further details of the noise model, its numerical implementation, and the calibration-data acquisition are described in Ref.~\cite{isakov2021simulations}, as well as in the Cirq repository~\cite{cirq}. For completeness, in Supplementary Figure~\ref{fig:calibration_data} we display a series of plots which show the chip connectivity and numerical values of the calibration data used for the various errors described above.

\subsection{Compiling circuits to the native gate set}

Single-qubit rotations are compiled into $\PhXZ$ gates according to an Euler-angle decomposition. Two-qubit unitaries are compiled into at most three $\sqrt{\iSWAP}$ gates (interleaved with single-qubit rotations) by a KAK decomposition, although most two-qubit unitaries (79\% with respect to the Haar measure) can be implemented with just two $\sqrt{\iSWAP}$ gates~\cite{huang2023quantum}. After compiling the entire circuit into this gate set, single-qubit rotations are concatenated into a single $\PhXZ$ gate whenever possible. All operations besides readout are pushed as early into the circuit as possible.

One exception we make is in the random permutation circuits $S_\pi$ appearing in the group $\SymCl{n}$ (for subsystem-symmetrized Pauli shadows). First, we decompose $\pi$ into an parallelized network of adjacent transpositions using an odd--even sorting algorithm~\cite{habermann1972parallel}. Each transposition $i \leftrightarrow j$ corresponds to a $\SWAP$ gate between qubits $i$ and $j$. However, rather than compile $\mathrm{SWAP}$ to the gate set directly (which would require three $\sqrt{\iSWAP}$ gates and four layers of $\PhXZ^{\otimes 2}$ gates), we instead implement the unitary
\begin{equation}
    \iSWAP = \sqrt{\iSWAP} \times \sqrt{\iSWAP},
\end{equation}
which uses only two $\sqrt{\iSWAP}$ gates and no single-qubit gates. The $\iSWAP$ gate differs from $\SWAP$ only by a phasing of $\i$ on the basis states $\ket{01}$ and $\ket{10}$. Such a replacement is valid because $S_\pi$ occurs only at the end of the circuit, immediately before readout. Thus while this phasing is technically unwanted, it has no observable effect on the measurement outcomes.

Finally, we note that the Trotter circuits for our Fermi--Hubbard simulations are optimized for the Sycamore architecture according to Ref.~\cite{arute2020observation}, which we follow closely. In particular, open-source code for their implementation can be found in Ref.~\cite{recirq}.

\subsection{Qubit assignment averaging}\label{subsec:qaa}

Our eight-qubit numerical experiments on the QVM utilize the $2 \times 4$ grid spanning from qubits $(5, 1)$ to $(6, 4)$ (see Supplementary Figure~\ref{fig:calibration_data}). To map these qubits to the simulated degrees of freedom (fermion modes or spin-$1/2$ particles), we employ qubit assignment averaging (QAA), which was introduced in Ref.~\cite{arute2020observation} in order to handle the issue of inhomogeneous error rates across a noisy quantum device. QAA works by identifying $N$ different assignments of the $n$ qubits and allocating $T/N$ of the experimental repetitions to each realization. Properties are estimated by averaging over all $T$ samples as usual. In principle, one can use a combination of shifting, rotating, and flipping the qubits throughout the chip;~for our simulations, we vary qubit assignments within the same fixed $2 \times 4$ grid.

For the Fermi--Hubbard model, we assign a spin sector to each of the parallel $1 \times 4$ qubit chains. We average over $N = 4$ different qubit assignments, defined by setting either the top or bottom chain as the spin-up chain, and ordering the four site labels starting either from the left or the right.

For the XXZ Heisenberg model, the eight-spin chain is embedded into the $2 \times 4$ grid of qubits. Each qubit assignment ($N = 12$) is defined by setting one of six qubits $\in \{(5, 1), (5, 2), (5, 3), (5, 4), (6, 4), (6, 1)\}$ as either the left end (ordered clockwise) or right end (ordered counterclockwise) of the spin chain.

While QAA aims to reduce device inhomogeneities, it  cannot lower the total amount of circuit noise. Thus QAA does not necessarily improve prediction accuracy with the unmitigated (standard shadow) estimators. Instead, homogenizing the noise appears to massage it into an effective form which approximately satisfies Assumptions~1 better than a single fixed configuration. We substantiate this claim with Supplementary Figure~\ref{fig:qaa_comparison}, using spin--spin correlations of the XXZ model ($R = 4$ Trotter steps) as a demonstrative example. We see that the unmitigated errors are virtually identical whether or not we perform QAA. On the other hand, the symmetry-adjusted estimates with QAA exhibits a more uniform error profile and overall improved noise suppression. Further investigation into this behavior is left as an open problem.

\begin{figure}
\centering
\includegraphics[scale=0.5]{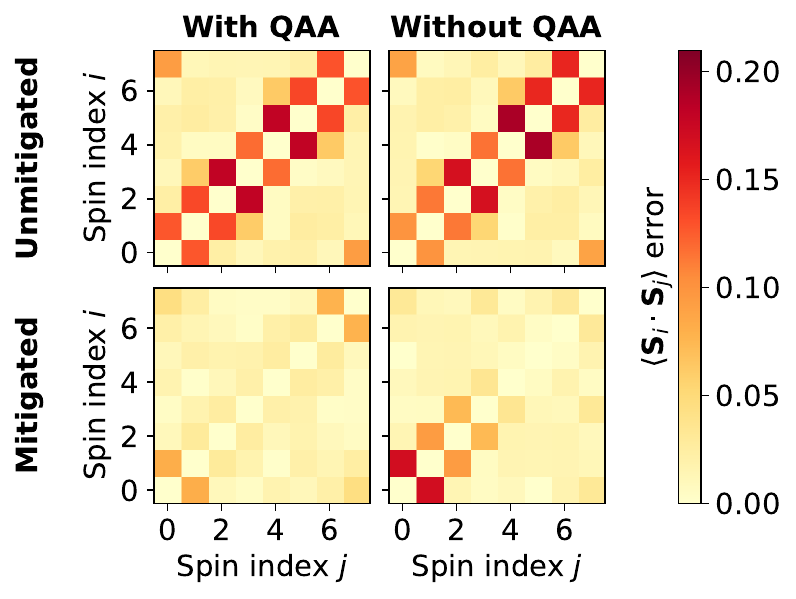}
\caption{\textbf{Behavior of qubit assignment averaging (QAA), demonstrated with the XXZ spin--spin correlations.} We illustrate using the experiment with $R = 4$ Trotter steps and $T = 4.8 \times 10^5$ subsystem-symmetrized Pauli shadows. Experiments with QAA average over the twelve configurations described in Section~\ref{subsec:qaa}, while experiments without QAA fix the qubit ordering $0 \mapsto (5, 1), \ldots, 7 \mapsto (6, 1)$.}
\label{fig:qaa_comparison}
\end{figure}

\subsection{Bootstrapping uncertainty bars}\label{subsec:bootstrap_error_bars}

To estimate uncertainty bars, we employ empirical bootstrapping~\cite{efron1992bootstrap}, modified by batching together samples. First we summarize the original method:~given $T$ classical-shadow snapshots, one resamples that data $T$ times with replacement. Then, averages $\hat{o}_j(T)$ (being either the unmitigated or mitigated estimators) are computed from that resampled data, yielding one bootstrap sample. Repeating this $B$ times and computing the standard deviation among those $B$ bootstrap samples yields the uncertainty bar.

Due to the size $T \sim 10^6$--$10^7$ from our simulations and limitations on classical compute resources, we perform bootstrapping on batches of snapshots. Split the $T$ samples into $K$ batches (each containing $T/K$ samples) and compute $\hat{o}_j^{(k)}(T/K)$ for each batch $k = 1, \ldots, K$. Because these estimates obey $\hat{o}_j(T) = (1/K) \sum_{k=1}^K \hat{o}_j^{(k)}(T/K)$, we resample the $K$ batches (rather than all $T$ shots) to bootstrap uncertainty bars for $\hat{o}_j(T)$. Depending on $T$, we set $K \sim 10^2$--$10^3$, and for all cases we take $B = 200$.

\subsection{Estimating the gate dependence of the QVM noise model}\label{subsec:gate_dep}

Here we provide an estimate of how much the QVM noise model violates Assumptions~1. We quantify this by computing a lower bound on the minimal observable error achievable by symmetry-adjusted classical shadows.

Let $U_{\mathrm{prep}}$ be the state-preparation circuit and $U_g$ a random measurement circuit. For Schur's lemma to hold (Assumptions~1), we require that the entire noisy circuit take the form $\mathcal{E} \mathcal{U}_g \mathcal{U}_{\mathrm{prep}}$, where $\mathcal{E}$ is the both time- and $g$-independent. While the noise model that we simulate is indeed time stationary and Markovian, the effective error channel $\mathcal{E} = \mathcal{E}_g$ depends on $g$. (This can be seen, for example, by commuting all the individual gate-level errors throughout $\widetilde{\mathcal{U}}_g$ and $\widetilde{\mathcal{U}}_{\mathrm{prep}}$ to the end of the circuit.)

In order to study this dependence on $g$, consider the decomposition
\begin{equation}
    \mathcal{E}_g = \mathcal{E}_0 + \Delta_g,
\end{equation}
where $\mathcal{E}_0$ is defined to be independent of $g \in G$. Although somewhat of an artificial decomposition, this is always mathematically possible with both $\mathcal{E}_0$ and $\Delta_g$ completely positive;~indeed, a trivial choice is $\mathcal{E}_0 = 0$. Our goal is to find the ``largest'' (in some sense) valid solution for $\mathcal{E}_0$. The remaining contribution $\Delta_g$ will then represent the minimal amount of assumption-violating noise in the model that our rigorous theory currently has no guarantees for.

From the decomposition above, the noisy measurement channel can be written as
\begin{equation}
\begin{split}
    \widetilde{\mathcal{M}} &= \E_{g \sim G} \mathcal{U}_g^\dagger \mathcal{M}_Z \mathcal{E}_g \mathcal{U}_g\\
    &= \widetilde{\mathcal{M}}_0 + \overline{\Delta},
\end{split}
\end{equation}
where
\begin{equation}
\begin{split}
    \widetilde{\mathcal{M}}_0 &= \E_{g \sim G} \mathcal{U}_g^\dagger \mathcal{M}_Z \mathcal{E}_0 \mathcal{U}_g\\
    &= \sum_{\lambda \in R_G} \widetilde{f}_\lambda(\mathcal{E}_0) \Pi_\lambda
\end{split} 
\end{equation}
is diagonal in the irreps of $G$, while the form of $\overline{\Delta} \coloneqq \E_{g \sim G} \mathcal{U}_g^\dagger \mathcal{M}_Z \Delta_g \mathcal{U}_g$ is unknown.

Applying $\mathcal{M}^{-1}$ and taking expectation values for the observables $\{O_j\}_{j=1}^L$ yields (assuming each $O_j \in V_\lambda$)
\begin{align}
    \vev{O_j}{\mathcal{M}^{-1} \widetilde{\mathcal{M}}}{\rho} &= \vev{O_j}{\mathcal{M}^{-1} \widetilde{\mathcal{M}}_0}{\rho} + \vev{O_j}{\mathcal{M}^{-1} \overline{\Delta}}{\rho} \notag\\
    &= \frac{\widetilde{f}_\lambda(\mathcal{E}_0)}{f_\lambda} \vip{O_j}{\rho} + \delta_j.
\end{align}
The terms $\delta_j \coloneqq \vev{O_j}{\mathcal{M}^{-1} \overline{\Delta}}{\rho}$ describe the deviation of observable estimates due to violations of the noise assumptions, which is precisely what we wish to quantify. For notation, denote the noisy expectations by $y_j \coloneqq \vev{O_j}{\mathcal{M}^{-1} \widetilde{\mathcal{M}}}{\rho}$ and noiseless expectations by $x_j \coloneqq \vip{O_j}{\rho}$. We collect these quantities into vectors of length $L$ and define the diagonal matrix $A \in \R^{L \times L}$ with eigenvalues $\widetilde{f}_\lambda(\mathcal{E}_0)/f_\lambda$ (in the appropriate positions corresponding to the irreps). This yields in the linear relationship
\begin{equation}
    \bm{\delta} = \bm{y} - A\bm{x}.
\end{equation}
This equation is underconstrained, so we opt for an estimate of $\bm{\delta}$ by bounding its norm from below. Namely, let $\hat{A}$ be a diagonal matrix of free parameters $0 \leq \xi_\lambda \leq 1$, which we optimize by nonnegative least-squares (NNLS) minimization:
\begin{equation}\label{eq:lower_bound_min}
    \| \bm{\delta} \|_2^2 \geq \min_{0 \leq \{ \xi_\lambda \}_{\lambda \in R'} \leq 1} \| \bm{y} - \hat{A}\bm{x} \|_2^2.
\end{equation}
Define $\hat{\bm{\delta}} \coloneqq \bm{y} - \hat{A}\bm{x}$ as the solution to this problem. In this sense, $\hat{\bm{\delta}}$ represents an error floor beyond which our theory for symmetry adjustment cannot mitigate due to inherent violations of Assumptions~1.

\begin{figure*}
\centering
\includegraphics[scale=0.5]{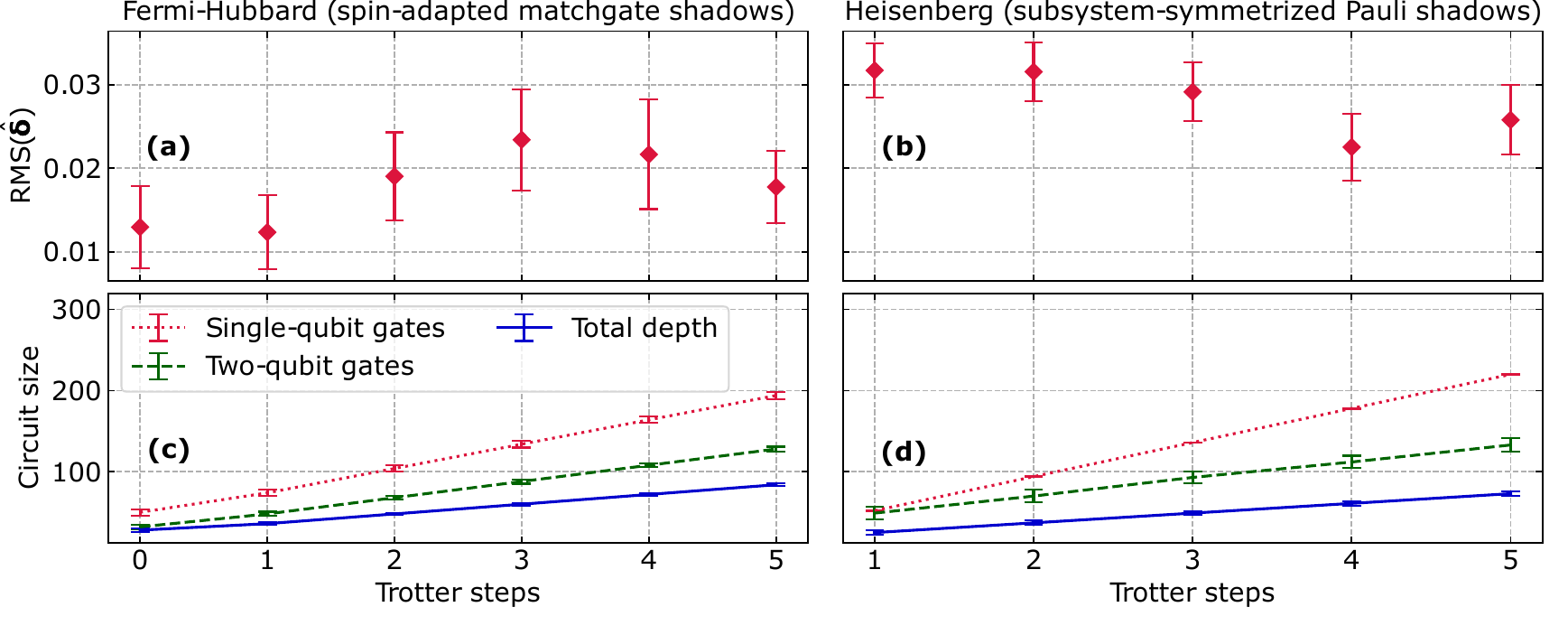}
\caption{\textbf{Estimated deviation of the QVM noise from a gate-independent model.} \textbf{(a, b)} Lower bound on the root-mean-square deviation of one- and two-body Majorana/Pauli expectation values as the number of Trotter steps (noisy circuit depth) grows. The estimated vector of deviations $\hat{\bm{\delta}}$ is described in Eq.~\eqref{eq:lower_bound_min}. Uncertainty bars are calculated by empirical bootstrapping. \textbf{(c, d)} Metrics for the size of the full circuit (state preparation via Trotterization and the random measurement unitary). We report the number of single-qubit ($\PhXZ$) and two-qubit ($\sqrt{\iSWAP}$) gates used, as well as the overall compiled circuit depth. Error bars denote one standard error of the mean.}
\label{fig:delta_study}
\end{figure*}

In Supplementary Figure~\ref{fig:delta_study} we plot the root mean square of $\hat{\bm{\delta}}$,
\begin{equation}
    \mathrm{RMS}(\hat{\bm{\delta}}) \coloneqq \frac{1}{\sqrt{L}} \| \hat{\bm{\delta}} \|_2,
\end{equation}
which quantifies the average additive error of the estimates. The observables we choose constitute local operators depending on the type of system simulated. For fermions, we consider one- and two-body Majorana operators that respect the spin adaptation. For qubits, we take strictly two-body Pauli operators. Uncertainty bars are bootstrapped as described in Section~\ref{subsec:bootstrap_error_bars}, where each bootstrap sample is obtained from the NNLS solution of the resampled data. We also show data for the Trotter circuit size:~the number of single- and two-qubit gates after compiling to the native gate set, as well as the circuit depth. Uncertainty bars here are given by one standard deviation in the size fluctuations due to the random unitaries $U_g$.

Overall, we assess that there is an error floor on the order of $10^{-2}$ per observable (recall that the observables have unit spectral norm). Interestingly, this lower bound appears roughly independent of circuit size (within uncertainty bars), perhaps indicating a saturation of the $g$-dependent contributions after a certain circuit size. In practice however, we have observed that symmetry-adjusted classical shadows only achieve mitigated errors on the order of $10^{-1}$ at the deepest circuits. We leave a closer analysis of this behavior, and whether this lower bound can actually be achieved, to future work.

\section{Review of prior symmetry-based QEM techniques}\label{sec:symmetry_qem_discussion}

In this section we review prior work on techniques broadly known as symmetry verification (SV), introduced by Bonet-Monroig \emph{et al.}~\cite{bonet2018low} and McArdle \emph{et al.}~\cite{mcardle2019error}, and generalized by Cai~\cite{cai2021quantum}. Like symmetry-adjusted classical shadows, these approaches take advantage of inherent symmetries of the quantum system, and certain formulations also accomplish this in an offline manner. However, symmetry adjustment is a fundamentally different idea, and we will discuss this distinction here. Note that, in principle many QEM strategies (and indeed quantum error correction itself) can be understood as symmetry-based techniques wherein one artificially builds large amounts of symmetry into the system;~we focus on techniques based on inherently possessed, physical symmetries here.

\subsection{Symmetry verification}

First we describe postprocessing symmetry verification (ppSV)~\cite{bonet2018low}. Let $M_S$ be the projector onto a symmetry subspace and $\rho$ an ideal state such that $M_S \rho M_S = \rho$. Given the preparation of a noisy state $\widetilde{\rho}$, ppSV aims to calculate the properties of the state
\begin{equation}
    \rho^{\mathrm{SV}} \coloneqq \frac{M_S \widetilde{\rho} M_S}{\tr(M_S \widetilde{\rho})}
\end{equation}
in an offline manner. That is, consider an observable $O$ which commutes with the symmetry. Then the error-mitigated estimate for $\langle O \rangle$ is
\begin{equation}
    \tr(O \rho^{\mathrm{SV}}) = \frac{\sum_{i=1}^m \tr(O S_i \widetilde{\rho})}{\sum_{i=1}^m \tr(S_i \widetilde{\rho})},
\end{equation}
where we have written $M_S = \frac{1}{m} \sum_{i=1}^m S_i$ in terms of the symmetry stabilizers $S_i$. This quantity requires measuring the (noisy) expectation values of all $S_i$ and $OS_i$ for $i = 1, \ldots, m$. Furthermore, the sampling cost to maintain the estimation accuracy $\epsilon$ is amplified by a factor of $\O(\tr(M_S \widetilde{\rho})^{-2})$. Bonet-Monroig \emph{et al.}~\cite{bonet2018low} show that this approach is equivalent to an instance of quantum subspace expansion~\cite{mcclean2017hybrid}, but with the ability to mitigate both coherent and incoherent errors outside the symmetry sector. Unfortunately, errors that commute with the symmetry cannot be projected out by this method. In general, SV projects $\widetilde{\rho}$ to the closest symmetry-respecting state, with no guarantee of closeness to the ideal state $\rho$.

This idea can be generalized to symmetry expansion (SE)~\cite{cai2021quantum}, which essentially replaces $M_S$ with the symmetry expansion operator
\begin{equation}
    E_S \coloneqq \sum_{i=1}^m w_i S_i,
\end{equation}
where $w_i$ are nonnegative weights that sum to unity. The symmetry-expanded pseudostate $\rho^{\mathrm{SE}} \coloneqq E_S \widetilde{\rho} / \tr(E_S \widetilde{\rho})$ is no longer guaranteed to lie in the desired symmetry sector, nor even to be a positive semidefinite operator. Hence SE exhibits estimation bias. However, by searching for weights that minimize this bias, one can heuristically achieve mitigated estimates with biased errors below that of the unmitigated noise level, while being simpler to implement than ppSV (for instance, by enforcing some of the $w_i = 0$).

Finally, direct symmetry verification (dSV)~\cite{mcardle2019error,bonet2018low} introduces additional quantum circuitry in order to check the symmetry value of the state throughout the course of the quantum computation. For example, in order to check the parity operator $S = Z_1 \cdots Z_n$, one can perform $n$ CNOT gates, controlled on each system qubit and targeting an ancilla, followed by reading off the ancilla. This encodes the parity information into the ancilla, and if at any point the ancilla returns the incorrect parity value then that circuit run is therefore discarded. Other symmetries may require more complicated circuitry;~for example, verifying the particle number $\eta$ requires $\O(n \log \eta)$ CPhase gates in total to read off the binary representation of $\eta$ from the ancilla. Although dSV requires significant additional coherent quantum control, the associated sampling overhead is $\O(\tr(M_S \widetilde{\rho})^{-1})$, a quadratic improvement over ppSV.

\subsection{Distinction from symmetry-adjusted classical shadows}

We now discuss the distinction between our method and these SV techniques. We will focus on ppSV, as it is the most comparable to symmetry-adjusted classical shadows. That is, rather dSV which requires coherent detection of symmetry violations, both ppSV and symmetry-adjusted classical shadows use classical postprocessing as their primary mechanism for error mitigation.

The most significant conceptual difference is that symmetry adjustment does not necessarily project the noisy state into the desired symmetry subspace. Instead, the symmetry is used as a reference point to calibrate the effects of the noise on the shadow measurement channel $\mathcal{M}$. Recall that we aim to estimate the eigenvalues of the channel $\widetilde{\mathcal{M}} = \sum_{\lambda} \widetilde{f}_\lambda \Pi_\lambda$ via the relation
\begin{equation}\label{eq:sacs_explanation}
    \widetilde{f}_\lambda = f_\lambda \frac{\tr(S_\lambda \widetilde{\rho})}{s_\lambda}.
\end{equation}
Thus, symmetry adjustment ultimately yields an estimate for the noisy channel, whose inverse is formally applied to produce the mitigated shadow
\begin{equation}
    \hat{\rho}^{\mathrm{EM}} = \sum_{\lambda} \frac{s_\lambda \Pi_\lambda(\hat{\rho})}{\tr(S_\lambda \rho)}.
\end{equation}
Here, $\hat{\rho}$ is the noisy shadow which converges to the noisy state $\widetilde{\rho} = \E[\hat{\rho}]$, but in fact any state could in principle be used as long as the noise channel is the same. For instance, suppose we prepare a different state $\sigma$ that experiences the same error channels that $\rho$ does. (This may hold under models where both states are prepared by circuits of the same structure but different gate angles, such as in variational quantum algorithms.) Then Eq.~\eqref{eq:sacs_explanation} tells us that the $\widetilde{f}_\lambda$ learned by the measurements of $\rho$ is equally applicable in the mitigation of properties of $\sigma$. Furthermore, the twirling nature of the randomized measurements allows the technique to account for errors that commute with the symmetry.

Beyond this conceptual difference, symmetry-adjusted classical shadows can easily incorporate symmetries such as particle number, whereas ppSV cannot. This is because the measurements of $\langle OS_i \rangle$ required of ppSV restrict the practical implementation to simple symmetries such as $\Z_2$ Pauli symmetries. More generally, it requires the decomposition of $M_S$ into a tractable number of measurable terms. However, the projector into the $\eta$-particle sector is
\begin{equation}
    M_S = \sum_{|x|=\eta} \op{x}{x},
\end{equation}
where each $\op{x}{x} = \prod_{j \in [n]} \frac{1}{2} (\I + (-1)^{x_j} Z_j)$ expands into exponentially many Pauli-$Z$ terms. While all such terms mutually commute and can therefore be measured in the same basis, the upfront classical postprocessing cost of exactly computing this ppSV estimator is nonetheless exponential.

Additionally, as a classical-shadows protocol, our method is tailored to tackle the problem of multiple observable estimation. In contrast, all forms of SV can complicate the task, namely due to the inclusion of the additional $OS_i$ operators. Recent work~\cite{jnane2023quantum} has attempted to address this by straightforwardly combining ppSV with classical shadows. However, they showed limited success with this unification. Indeed, suppose we wish to project into a parity sector. The parity operator is $n$-body, while the target observables $O_j$ are $k$-local for some small constant $k$. This means that each $O_j S$ has locality $n - k = \O(n)$, which leads to an exponentially large shadow norm (variance).

\end{document}